\documentclass[11pt,a4paper]{article}

\usepackage[T1]{fontenc}
\usepackage[utf8]{inputenc}
\usepackage{mathptmx}
\usepackage{courier}
\usepackage[margin=2.4cm]{geometry}
\usepackage{amsmath,amssymb,amsthm,mathtools}
\usepackage{graphicx}
\usepackage{booktabs}
\usepackage[font=small,labelfont=bf]{caption}
\usepackage{authblk}
\usepackage{microtype}
\usepackage{xcolor}
\usepackage[numbers,sort&compress]{natbib}
\PassOptionsToPackage{hyphens}{url}
\usepackage[colorlinks=true,allcolors=blue!55!black]{hyperref}
\hypersetup{
  pdftitle={Reciprocity can halve what a mechanical network can learn},
  pdfauthor={Thai-Son Vu; Hoang-Giang Nguyen; Quoc-Bao Nguyen;
             Sengaloun Keoalounxay; Bao-Viet Tran},
  pdfcreator={LaTeX with hyperref}
}

\newtheorem{theorem}{Theorem}
\newtheorem{proposition}[theorem]{Proposition}
\newtheorem{corollary}[theorem]{Corollary}
\theoremstyle{definition}
\newtheorem{remark}{Remark}   

\DeclareMathOperator{\rank}{rank}
\DeclareMathOperator{\Sym}{Sym}
\DeclareMathOperator{\codim}{codim}
\DeclareMathOperator{\spn}{span}
\DeclareMathOperator{\spec}{spec}
\newcommand{\T}{{\mathsf{T}}}
\newcommand{\R}{\mathbb{R}}
\newcommand{\zenodoi}[1]{\href{https://doi.org/#1}{\texttt{#1}}}
\newcommand{\bq}{\mathbf{q}}
\newcommand{\bs}{\mathbf{s}}
\newcommand{\bw}{\mathbf{w}}
\newcommand{\bu}{\mathbf{u}}
\newcommand{\bv}{\mathbf{v}}
\newcommand{\by}{\mathbf{y}}
\newcommand{\bx}{\mathbf{x}}
\newcommand{\bz}{\mathbf{z}}

\title{\bfseries Reciprocity can halve what a mechanical network can learn}

\author[1]{Thai-Son Vu\thanks{Corresponding author:
  \texttt{sonvt2@huce.edu.vn}\ \ ORCID 0009-0005-2491-6018}}
\author[1]{Hoang-Giang Nguyen\thanks{\texttt{giangnh@huce.edu.vn}\ \
  ORCID 0000-0001-7505-2233}}
\author[1]{Quoc-Bao Nguyen\thanks{\texttt{baonq@huce.edu.vn}\ \
  ORCID 0000-0003-4406-2476}}
\author[2]{\mbox{Sengaloun Keoalounxay}\thanks{\texttt{s.keoalounxay@nuol.edu.la}\ \
  ORCID 0009-0006-1995-9694}}
\author[3]{\mbox{Bao-Viet Tran}\thanks{\texttt{viettb@utc.edu.vn}\ \
  ORCID 0000-0001-9709-5699}}

\affil[1]{Hanoi University of Civil Engineering,
  55 Giai Phong Street, Bach Mai Ward, Hanoi, Viet Nam}
\affil[2]{National University of Laos,
  Dongdok Campus, Vientiane, Lao PDR}
\affil[3]{University of Transport and Communications,
  No.~3 Cau Giay Street, Lang Ward, Hanoi, Viet Nam}
\date{}

\begin{document}
\maketitle

\begin{abstract}
\noindent
A passive, reciprocal, linear network driven by forces or currents and read
at points it drives cannot be trained below an error floor computed from the target
beforehand. The usual parameter count for such networks does not charge for
Maxwell--Betti reciprocity. We prove
that when such an elastic, resistor or flow network has a symmetric response
operator and $p$ driven degrees of freedom are also read, its reachable response
blocks lie in a subspace of codimension $p(p-1)/2$, whatever its size and
topology, and the floor is the distance to it. On targets with a reachable
symmetric part, a second-order optimiser ends within $0.1\%$ of the floor in
$142$ of $144$ simulated runs, and a contrastive rule with a bond-local update direction,
within its step budget, in $22$ of $24$. Reciprocity charges a list of single drive--read tasks only for the pairs it
instruments in both directions, and a rank test on the passive network says in
advance whether, and on which bonds, odd couplings locally restore the lost
directions, at a cost in non-reciprocity that the target bounds from below. Under an
imposed-displacement drive, the one most physical-learning hardware uses,
reciprocity survives as a weaker law, an exact balance between forward and
reverse transmissions once each is weighted by the driving-point compliance at
its input, and an inequality on their product, under which no
symmetric positive-definite chain meets both targets of a published robotic
metamaterial. All nineteen published layouts we tabulated pay nothing; a
force-driven layout that asks one pair to respond differently in its two
directions will.
\end{abstract}

\section{Introduction}

Geometry, more than composition, gives a flexible mechanical metamaterial its
function~\citep{Bertoldi2017}, and a growing body of work asks such materials
to acquire that function themselves. Networks trained by directed
aging~\citep{Pashine2019}, by coupled or contrastive local
rules~\citep{ScellierBengio2017,Stern2021,Dillavou2022,%
Dillavou2024}, by in situ backpropagation~\citep{Li2024}, or by explicit
design~\citep{Rocks2017,Pashine2023} adjust internal
parameters (stiffnesses in a spring network, conductances in an electrical one)
until a prescribed input--output map is realised; the approaches are surveyed
by \citet{SternMurugan2023} and \citet{Momeni2025}. Several now exist in
hardware: lattices of tunable beams~\citep{Lee2022}, networks trained in place
by coupled learning~\citep{Altman2024}, robotic metamaterials that learn to
change shape~\citep{Du2026}.

The governing question for such a network is capacity under its own drive:
when it is driven by forces or currents and read at some of the points it
drives, which targets stay out of reach, however many parameters it has and
however long it trains? The standard first estimate answers with a count
alone, blind to how the network is driven and read. A network with
$n_\theta$ tunable parameters, asked to realise a response block with $m^2$
entries ($m$ degrees of freedom driven and $m$ read), can reach at most $\min(n_\theta,m^2)$ independent
dimensions~\citep{Stern2021}. \citet{Stern2024} halve the \emph{parameter}
side of that count for a symmetric Hessian, but the target side is still $m^2$
free entries. \citet{Zu2025} report that in disordered solids the rank of the
parameter-to-feature Jacobian does reach $\min(n_\theta,n_y)$, with $n_y$ the
number of features, across a wide range of targets, parameters, sizes and
dimensions, failing only for ordered lattices. When the features are instead
the entries of a \emph{response block} read on degrees of freedom that are
also driven, the count fails for a reason unrelated to rank deficiency: the
target space itself carries a redundancy, so the count overstates the
reachable dimension even when the Jacobian has full rank on the directions
reciprocity leaves open.

Counting arguments of this kind tacitly assume that the entries of the
response block are independent coordinates on the reachable set. For a
passive, linear elastic network they are not. Maxwell--Betti
reciprocity~\citep{Maxwell1864,Betti1872} forces the compliance operator to be
symmetric. Each \emph{pair} of degrees of freedom that are both driven and read
therefore occupies two mirrored positions of the response block, and the two
entries must be equal. With $p$ shared degrees of freedom this gives
$p(p-1)/2$ independent equalities; a single shared degree of freedom gives
none. The symmetry has been noticed in this setting without being priced:
\citet{Stern2021} remark that the input--output relations of their networks
are symmetric in linear response, and \citet{Li2025} that reciprocity couples
the forward and backward transformations of a surface-plasmonic network, a
constraint they lift with magnetically biased ferrite. The supplementary
material of \citet{Du2026} names this omission from the count,
writing that ``the above evaluation ignores the constraints according to the
Maxwell--Betti theorem'', and leaves the size of the correction open. That
reciprocity halves a count of independent measurements is itself classical,
on record in electrical impedance tomography and in network theory
(Sec.~\ref{sec:priorwork} surveys it); what these equalities cost a fixed
trainable network whose driven and read degrees of freedom overlap, fully or
in part, has not been worked out.

We derive the exact codimension that reciprocity removes at partial overlap, a
lower bound on the capacity deficit, and a target-specific reciprocity floor
for learning on a fixed tunable network. For a force-driven
block a subspace of codimension exactly $p(p-1)/2$ contains the set a fixed
reciprocal network reaches by varying its stiffnesses, at every overlap
$p\le\min(m_T,m_S)$, partial or full, with $m_S$ degrees of freedom driven and
$m_T$ read, an elementary observation
(Theorem~\ref{thm:main}, Remark~\ref{rem:global}); the dimension of that set
falls short of $m_Tm_S$ by at least that much, and by exactly that much
wherever the symmetry branch of the bound is attained.
At full overlap of a square block, $p=m_T=m_S=m$, that codimension is a
fraction $(m-1)/(2m)$ of the target space, approaching one half. For the imposed-angle
task of \citet{Du2026} the correction is a spectral inequality rather than a dimension count
(Theorem~\ref{thm:multiclamp}), which their degree-of-freedom count cannot
see. The relation to prior work at the end of this section separates the
classical ingredients from ours.

Three results go beyond counting the codimension of a symmetric response
matrix. First, the distance from the target to that subspace is an error
floor computed before training (Theorem~\ref{thm:floor}, sharpened by
passivity in Theorem~\ref{thm:floorpd}), and two learning rules, run in
simulation, end on it in nearly every run when the target's symmetric part is
reachable, the contrastive one within its step budget
(Sec.~\ref{sec:learning}). Second, reciprocity charges a list of single
drive--read tasks only for the pairs it instruments in both directions, so a
shared terminal on its own costs no dimension (Sec.~\ref{sec:tasklist},
Corollary~\ref{cor:digon}). Third, a rank test on the passive network decides,
before any odd coupling is built, whether, and on which bonds, odd couplings
restore, near the passive network, what reciprocity removes
(Theorem~\ref{thm:recovery}, Sec.~\ref{sec:recovery}).

A concrete instance fixes the scale. On the sixteen-node network of
Fig.~\ref{fig:setup}a, adding one reversed task to a three-task list sets a
floor of $9.9\%$ of the target norm; on a target built so that this floor is
attainable, training stops on it, and one odd bond chosen in advance removes it
(Sec.~\ref{sec:consequences}, Fig.~\ref{fig:example}).

Grouped by what each rests on, the same material makes four contributions. The
first rests on the containment of Theorem~\ref{thm:main}, the second and third on
the rank of the Jacobian (Sec.~\ref{sec:jacobian}), and the fourth on the form the
containment takes for a list of tasks. The first is an error
floor that is a \emph{number computed from the target in advance of
training}: a reciprocity term (Theorem~\ref{thm:floor}) and an orthogonal
passivity term that binds on targets the first cannot see. The second is
structural: a test of whether a given graph attains the ceiling on reachable
dimension, a rank ceiling for imposed displacements, and odd couplings on
$p(p-1)/2$ bonds that restore the lost directions near the passive network wherever its wedges (the antisymmetric directions an odd coupling on each bond would add to the shared block) span them
(Secs.~\ref{sec:dirichlet}--\ref{sec:certificates} and~\ref{sec:recovery}). The price is a
non-reciprocity the target bounds from below (Proposition~\ref{prop:ratio}).
The third is the layout law (Proposition~\ref{prop:budget}). At a fixed number
of accessed degrees of freedom, each of which can be both driven and read, and
where the bond count does not bind, full overlap becomes the best layout, by a
factor approaching two. At a fixed count of sensors plus actuators, full
overlap is not the best layout. The fourth replaces the block by the finite list of demands a
laboratory actually imposes (Sec.~\ref{sec:tasklist}). Such a list carries a
deficit of its own, which forces a floor and, on coordinate tasks, counts the
pairs of terminals instrumented in both directions, so a shared terminal on
its own is free and only a reversed pair is charged.

An experimenter can therefore compute, before building or training, what no
learning rule on a passive network can beat and which layout or odd bonds
avoid it. At all nineteen published layouts we tabulated, $p=0$ and the
deficit is zero (Table~\ref{tab:layouts}). This paper corrects no capacity
number now in print; it prices the overlapping layout, and we found no
in-place learning layout driven by forces or currents that pays. Two
published protocols come close. The in situ backpropagation of
\citet{Li2024} carries one Maxwell--Betti relation between its forward and
adjoint solves and uses it to obtain an exact gradient, setting no target on
it; \citet{Du2026} exchange the driven and read sets between two targets under
an imposed drive.

A second drive, the one most physical-learning hardware
uses~\citep{Du2026,Altman2024,Dillavou2022}, holds the inputs at prescribed
\emph{displacements} instead of driving them with forces. Its measured object is not a block
of $C$, and Theorem~\ref{thm:main} does not constrain it; the companion law of
Theorem~\ref{thm:dirichlet} does. There reciprocity survives as an exact
balance between the two transmissions of a pair and their two driving-point
compliances, four scalars measurable at the terminals, so it can be tested on
a sample without knowing its network. For $p\ge2$ the same theorem caps the
rank at $\min\bigl(n_b,\tfrac12(p-1)(p+2)\bigr)$, at least $(p-1)(p-2)/2$ below the $p(p-1)$
off-diagonal entries, a smaller deficit bound of a different kind. For the imposed-angle chain of \citet{Du2026} we
find in Sec.~\ref{sec:consequences} that, in the linear model as deposited,
no symmetric positive-definite chain can meet its two targets.

Sec.~\ref{sec:setup} formulates the reachable set, and
Table~\ref{tab:notation} (Sec.~\ref{sec:notation}) collects the symbols.
Sec.~\ref{sec:bound} proves the ceiling and the two floors
(Theorems~\ref{thm:main}--\ref{thm:floorpd}). Sec.~\ref{sec:universal} extends
the ceiling and the reciprocity floor to any symmetric invertible operator
perturbed symmetrically by its parameters, complex-symmetric ones included
(Theorem~\ref{thm:universal}): resistor and flow networks, three dimensions
and bending, and the tangent response about a finitely deformed state.
Sec.~\ref{sec:tasklist} gives their form for a list of demands
(Proposition~\ref{prop:tasklist}), with the list's own deficit, its floor
(Theorem~\ref{thm:taskfloor}) and, on coordinate tasks, a count of reversed
pairs (Corollary~\ref{cor:digon}), and recovers the block law as the list that asks
for everything. Sec.~\ref{sec:dirichlet} gives the imposed-displacement laws
(Theorems~\ref{thm:dirichlet} and~\ref{thm:multiclamp}).
Sec.~\ref{sec:attainment} reduces attainment on a given graph to the
transversality of two subspaces (Theorem~\ref{thm:duality}), and
Sec.~\ref{sec:certificates} turns one
exactly certified configuration into attainment at almost every configuration
of that graph (Proposition~\ref{prop:generic}). Sec.~\ref{sec:recovery} states what
non-reciprocal couplings withdraw (Proposition~\ref{prop:odd}) and proves that
an odd coupling restores the lost directions near the passive network wherever its wedges span them, exactly in projection rather than to first order (Theorem~\ref{thm:recovery},
Corollary~\ref{cor:submersion}); Sec.~\ref{sec:price} prices it.
Sec.~\ref{sec:numerics} verifies the bounds and the recovery numerically,
Sec.~\ref{sec:learning} shows that a second-order optimiser and a
contrastive rule with a bond-local update direction reach the predicted floors
in all but a few runs, and Sec.~\ref{sec:consequences}
derives the consequences for sensor--actuator layouts and imposed-displacement
learning, including an analysis of the data of \citet{Du2026}. Limitations (Sec.~\ref{sec:limitations}),
the Conclusion and the appendices (\ref{app:layouts}--\ref{app:repro})
follow.

\subsection{Relation to prior work}
\label{sec:priorwork}

This subsection supports Theorem~\ref{thm:main} and the results built on it
and proves nothing of its own: it separates the classical from the new.
The symmetry itself is classical, and on record in the form we use:
\citet{MiltonSeppecher2008} observe that the response matrix is symmetric,
being the Schur complement of the ``clearly symmetric'' all-terminal response,
and use it for \emph{synthesis}. For networks free to carry effective negative
stiffnesses and masses at a fixed frequency, symmetry is the \emph{only}
obstruction~\citep{MiltonSeppecher2008}; for non-negative springs, a class
neighbouring the strictly positive stiffnesses used here, positive
semidefiniteness and balance conditions join it~\citep{GuevaraVasquez2011},
and for planar resistor networks so does a sign condition on the circular
minors~\citep{Curtis1998}. These constructions add internal nodes, choose the
topology freely, and drive and read the same set; a material already built has
none of that freedom. Such realisability theorems return a verdict on a
target, where the floor of Theorem~\ref{thm:floor} returns a distance.

The reduction that the symmetry imposes on independent measurements is on
record too, in electrical impedance
tomography. \citet{Cheney1999} observe, for an $L$-electrode system, that
``some of this data is redundant, because the current-voltage map is
symmetric,'' leaving the degrees of freedom of a symmetric $(L-1)\times(L-1)$
matrix, $L(L-1)/2$ of them in place of the $(L-1)^{2}$ entries of an
unconstrained map, a ratio $L/(2(L-1))$ that tends to one half for many
electrodes. This is the closest classical statement of the effect priced
here. A count of the same kind is proved, and routinely used, in electrode
arrays more broadly: on $L$ electrodes only $L(L-3)/2$ four-electrode
configurations are linearly independent and synthesise all the
rest~\citep{Lehmann1995}, a theorem of DC resistivity probing. It is
\citet{AdlerLionheart2006} who name reciprocity as the reason half of an
adjacent-drive protocol's readings are redundant: at sixteen electrodes its
$208$ readings carry $104$ independent ones, and electrical resistivity
tomography spends the surplus of normal--reciprocal pairs on error estimation
rather than discarding it~\citep{LaBrecque1996,Tso2017}. \citet{Cheney1999} count how much independent data a
measurement can supply to an inverse problem, at full overlap, on a body whose
conductivity is the unknown; we count which response blocks a \emph{fixed}
network can reach as its stiffnesses vary, at partial as well as full
overlap.

The target-side halving is already present in the feature count of
\citet{Zu2025}, six independent elastic constants in two dimensions and
twenty-one in three, which has the form $n(n+1)/2$ with $n$ the number of Voigt
components. There the redundancy priced here was removed by hand, at the level
of the elasticity tensor rather than of a response block, so their Jacobian is
full-rank on what remains. Other reasons for the count $\min(n_\theta,m^2)$
to be optimistic are known: frustration of the constraint-satisfaction
problem~\citep{Rocks2019}, and the self-stresses that turn Maxwell's counting
inequality~\citep{Maxwell1864} into Calladine's index
relation~\citep{Calladine1978,Lubensky2015}; the constraint studied here is
of another kind.

The product-of-solutions form of the Jacobian is not new either.
\citet{Draper2020} derive it for matrix-valued inverse problems on graphs and
specialise it, in their Sec.~6, to the model used here, a passive spring
network of known geometry with unknown bond stiffnesses, their data being the
displacement-to-force map of Sec.~\ref{sec:dirichlet}. Their step from
injectivity of the linearisation at one stiffness vector to injectivity at
almost every other generalises \citet{Boyer2016}, who prove, for the discrete
conductivity and Schr\"odinger problems on a graph, that a linearised problem
solvable at a single admittance is solvable at almost every admittance, the
exceptional set being the zero set of a determinantal minor.
Proposition~\ref{prop:generic} runs that argument, but moves the node
positions as well as the stiffnesses and concludes about the rank of the
Jacobian, that is, about its image. Both papers ask when $k$ can be recovered
\emph{from} the response, a condition on the kernel; we ask which response
blocks are reachable \emph{by varying} $k$, a condition on the image, the side
reciprocity acts on.
Reciprocity does not enter their question. The range condition that
\citet{Boyer2016} do write is on the parameter space, the transpose reading of
injectivity, and \citet{Draper2020} carry the response throughout as an
$n^2$-vector, so neither paper sees the symmetry that costs the dimensions
counted here. Neither splits its boundary set into driven and read-out parts or
counts a codimension, an achievable set or a floor.

That a passive reciprocal $m$-port has a symmetric positive-real immittance
matrix, with $m(m+1)/2$ independent parameters at full overlap, is textbook
network theory: \citet{Belevitch1968} records the real constant case, a
symmetric positive-semidefinite matrix realised, when nondegenerate, with $m$
resistances and $m(m-1)/2$ transformer ratios, and
\citet{AndersonVongpanitlerd1973} give the positive-real and reciprocal
synthesis theory in state-space form.

That reciprocity can put part of a response space out of reach is on record in
optics as well. \citet{GuoFan2022} determine the reflection coefficients
attainable by $n$-port scattering matrices of prescribed singular values, with
and without reciprocity, and show that for $n\ge3$ the reciprocal set is a
proper subset of the non-reciprocal one, cut out by a single inequality, so
that a lossless three-port with all-zero reflections must be non-reciprocal.
There reciprocity removes a full-dimensional region from the set of all
lossless devices; here it removes a linear subspace from the responses of one
network already built, and the codimension of that subspace fixes the floor.

The error floor of Theorem~\ref{thm:floor} rests on a symmetrisation that is
not new as mathematics: \citet{FanHoffman1955} prove that a nearest Hermitian
matrix to any $A$, in every unitarily invariant norm, is $\tfrac12(A+A^{*})$,
and for real $A$ the minimiser is itself real symmetric, so the real case
follows at once. In the Frobenius norm, the one used here, it is the orthogonal
splitting performed in the proof of Theorem~\ref{thm:floor}. What is
ours is the identification: the matrix to be symmetrised is the shared block of
the \emph{target}, the whole reachable set of a fixed passive network lies on
the symmetric side of that splitting, and the distance is therefore a number
available before any network is built.

Other ingredients are not new either. The containment behind the ceiling
needs no derivative (Remark~\ref{rem:global}). The floor of
Theorem~\ref{thm:floorpd} is the distance of \citet{Higham1988} (his
Theorem~2.1) from the target's shared block to the positive-semidefinite
matrices, which applies because the realised shared block is a principal
submatrix of a positive-definite matrix, a textbook fact; its passivity term
makes exact an observation already stated in words by \citet{Stern2021}, that
non-negativity of the tunable elements ``excludes many conceivable linear
mappings''. Of the imposed-displacement laws only the rank clauses are ours,
\eqref{eq:dbound} and its extension to families of tasks,
Proposition~\ref{prop:multirank}; the spectral clause of
Theorem~\ref{thm:multiclamp} is classical and we claim only its reading.
None of the classical record just surveyed is claimed here, the halved
measurement count and the $m(m+1)/2$ parameters of a reciprocal $m$-port
included, nor are the matrix-nearness facts the floors rest on: the nearest
symmetric matrix, the distance to the positive-semidefinite cone, whose
formula is the right-hand side of Theorem~\ref{thm:floorpd} applied to the
shared block, and the numerical-range algebra of
Proposition~\ref{prop:ratio}~\citep{London1981,LiSze2014,Lin2015}. What is new is what the symmetry costs a fixed network, at partial or
full overlap, as listed above. The ceiling is an upper
bound: a network may fall further short for want of bonds or of genericity
(Sec.~\ref{sec:numerics}), but no extra tuning repairs the deficit reciprocity
contributes.

Several recent papers ask what the physics constrains in the learning process;
none derives from reciprocity a codimension of the set a reciprocal network
can reach or an error floor computable from the target. \citet{Bosch2026} show that $K^\T=K$, or an
intertwining condition some non-reciprocal systems meet, lets the same
hardware generate the adjoint field of an exact gradient, and
\citet{TuxburyLin2026} train in situ with a stochastic adjoint that trades
reciprocity for nondegenerate diffusion, on a lattice whose couplings stay
symmetric while nonlinearity and modulation break reciprocity.
\citet{Ezraty2026} train dissipative networks through boundary conditions alone
and ask which local evolution rules descend. \citet{Niu2026}, posted before
our first version, show that within their matched response class the
closed-loop learning response under non-negative spectral feedback is
symmetric positive-semidefinite and confines learning to a preconditioned
gradient flow. An effective active boundary controller adds an antisymmetric
part to that response and turns the learning path. Their reciprocity is that
of the closed-loop learning response; the edges of their forward transport
need not be bidirectional. Like them, we find that an antisymmetric part
supplies what a symmetric one cannot; in the present paper that part sits in
bonds inside the network and acts on the reachable set
(Sec.~\ref{sec:recovery}).

\citet{McGinnis2026} show that equilibrium propagation and coupled learning
conserve a conductance mass, and \citet{Dangol2026} that untrainable elements
set what these rules remember; \citet{Dangol2026b} counts how many quantities
such rules can conserve in resistive networks with a symmetric Laplacian and
no odd couplings. \citet{Stuhlmuller2026} show that a linear network computes
only linear maps and a passive one cannot amplify. McGinnis et al.\ and both
papers by Dangol drive and read disjoint node sets, as does the linear example
of Stuhlm\"uller and Dijkstra; there $p=0$, our bound predicts no obstruction,
and their results neither confirm nor contradict it.

Two of the present authors study the same Jacobian of the boundary response
and bound the blind dimension of a learning network above by a maximum-weight
spanning-forest count over its hidden components~\citep{NguyenVu2026}; that
deficit moves with the graph and with the electrode placement, whereas the
codimension $p(p-1)/2$ counted here does not depend on the graph at all and
depends on the placement only through the number $p$ of shared terminals. \citet{Zhong2021} measure how
much a many-body system has learned of a drive it did not choose, with a
variational autoencoder; the bound here limits how many independent responses
a fixed network can be trained to produce at ports the trainer selects.

Breaking the symmetry can remove the deficit. \citet{Du2026} give the
qualitative direction, noting that a non-reciprocal system ``eludes the
Maxwell--Betti theorem''; we supply the size of the effect: the set reachable by varying the
stiffnesses lies in a subspace of codimension exactly $p(p-1)/2$, so its
dimension falls short of $m_Tm_S$ by at least that much, and by exactly that
much wherever the symmetry branch of the bound of Theorem~\ref{thm:main} is
attained. For an odd,
non-reciprocal bond coupling of the kind realised in active and robotic
metamaterials~\citep{Scheibner2020,Fruchart2023,Brandenbourger2019}, we also
prove that, near the passive network and in projection onto the antisymmetric
blocks, it restores the subspace that reciprocity removed where, and only
where, the passive network's wedges span $\Lambda^2\R^p$, the space of
antisymmetric $p\times p$ blocks (Theorem~\ref{thm:recovery}); numerical tests (Sec.~\ref{sec:numerics})
confirm it. That the wedges do span is measured rather than proved. The
restoration comes at a cost: each odd bond
exerts a couple on the two nodes it joins, and the network is active;
generically the couples leave it out of moment balance, and
Sec.~\ref{sec:price} gives the exception and the least coupling to first
order.


\section{Reachable response sets}

This section sets up the one object every later bound is about: the set of
response blocks a fixed network can reach as its parameters vary. It defines
the network, its bond law with the odd coupling that breaks reciprocity, the
block read on one set of degrees of freedom and driven on another, and that
reachable set. At a generic parameter point the local dimension of the set is
the rank of a Jacobian
computed in closed form, one rank-one column per parameter. A table collects
the symbols.

\subsection{Networks and response}
\label{sec:setup}

Consider $N$ point nodes at positions $\bx_i\in\R^{\delta}$, $\delta=2$, joined by $n_b$
bonds. For bond $b=(i,j)$, of length $L_b$, let $\hat{n}_b$ be the unit vector
along the bond and $\hat{t}_b$ its in-plane perpendicular. Collect the
linearised elongation of bond $b$ into a compatibility row
$\bq_b\in\R^{\delta N}$, carrying $-\hat{n}_b$ at node $i$ and $+\hat{n}_b$ at
node $j$, so that $e_b=\bq_b^\T\bu$ for a displacement field $\bu$. Define
$\bs_b$ identically with $\hat{t}_b$ in place of $\hat{n}_b$.

We take the bond force law
\begin{equation}
  \label{eq:stiffness}
  K \;=\; \sum_{b} \bigl(k_b\,\bq_b + a_b\,\bs_b\bigr)\,\bq_b^\T ,
\end{equation}
so that the external force needed to hold a displacement $\bu$ is $K\bu$. The
term $k_b>0$ is the ordinary central-force stiffness. The term $a_b$ is an
\emph{odd} coupling: the bond responds to its own elongation with a force partly
transverse to its axis (Fig.~\ref{fig:setup}b), the odd spring
of~\citet{Scheibner2020}.

We pin three degrees of freedom to remove the planar rigid-body modes, write
$K_{\!f\!f}$ for the restriction of $K$ to the $n_{\mathrm{free}}=2N-3$ remaining
ones, and set $C=K_{\!f\!f}^{-1}$. Wherever $\bq_b$ or $\bs_b$ multiplies $C$
or $K_{\!f\!f}$ it stands for its restriction to the free degrees of freedom.

The odd coupling makes these networks active, and it has a cost. The term
$a_b\bs_b\bq_b^\T$ puts the transverse forces $-a_be_b\hat{t}_b$ and
$+a_be_b\hat{t}_b$ at the two ends of bond $b$. They are equal and opposite,
so linear momentum is still balanced; but they are separated by
$L_b\hat{n}_b$ and are perpendicular to it, so they form a couple of magnitude
$|a_be_b|L_b$. An odd bond therefore exerts a net torque on the pair of nodes
it joins. Generically the network is then out of moment balance unless an
external agency supplies both the missing angular momentum and the work the
same term does around a closed deformation cycle. The odd spring is a
microscopic origin of the antisymmetric part of the odd elastic modulus
tensor~\citep{Scheibner2020,Fruchart2023}, and in the robotic realisations the
feedback-driven elements~\citep{Brandenbourger2019,Chen2021,Veenstra2024}
supply both the angular momentum and the cycle work.

A subspace of couplings escapes the moment imbalance, although each bond's own
couple and the cycle work remain (Proposition~\ref{prop:torquefree}).
Sec.~\ref{sec:price} puts a price on the cost, and the results below that are
stated in terms of $K_{\!f\!f}$ and its symmetry are not affected by it.

\citet{Nassar2020} review the routes to breaking reciprocity in elastic and
acoustic media, and \citet{Coulais2017} demonstrate that even a passive
metamaterial can transmit asymmetrically at finite deformation. Odd couplings
are a simple rank-one way to break reciprocity while leaving the network's
geometry untouched; switching them off restores it:
\begin{equation*}
  a_b = 0 \ \ \forall b
  \quad\Longrightarrow\quad
  K = K^\T
  \quad\Longrightarrow\quad
  \text{Maxwell--Betti reciprocity holds.}
\end{equation*}

Only the forward implications hold. Non-zero odd couplings generically break
the symmetry, but special cancellations, or the restriction to the free
degrees of freedom, can leave $K_{\!f\!f}$ symmetric even when some
$a_b\neq0$; Remark~\ref{rem:odd-hyp} exhibits such a configuration. Every
result below is therefore stated directly in terms of
$K_{\!f\!f}=K_{\!f\!f}^\T$ or $K_{\!f\!f}\neq K_{\!f\!f}^\T$, never as a
hypothesis on the $a_b$, with two exceptions. Proposition~\ref{prop:torquefree}
has as its subject the condition $(a_bL_b)_b\in\ker Q^\T$, with $Q$ the matrix
of compatibility rows (Appendix~\ref{app:jacobian}), on the couplings
themselves; and the second bound \eqref{eq:multirank2} of
Proposition~\ref{prop:multirank} assumes $a\equiv0$, so that $K_{\!f\!f}$ is
homogeneous of degree one in $k$. The odd couplings do appear in eight other statements, each in a
role that the statement itself names and that leaves this rule intact:
Theorem~\ref{thm:main}, Propositions~\ref{prop:tasklist}(v),
\ref{prop:multirank}, \ref{prop:odd} and~\ref{prop:ratio}, Theorem~\ref{thm:recovery}, and
Corollaries~\ref{cor:submersion} and~\ref{cor:price}. Theorem~\ref{thm:floorpd}
assumes passivity, which includes $a\equiv0$, but its proof uses only the part
$K_{\!f\!f}=K_{\!f\!f}^\T\succ0$ of it. Theorem~\ref{thm:dirichlet} needs
invertibility and $C_{\mu\mu}\ne0$ at every $\mu\in P$ for its transmission
formula, symmetry for its balance, cycle and rank clauses, and $C\succ0$ for
its sign and product clauses.
Theorem~\ref{thm:recovery} needs only $K_{\!f\!f}$ symmetric and so holds at
the configurations of Remark~\ref{rem:odd-hyp} as well.

For the spring networks of \eqref{eq:stiffness}, \emph{passive} is shorthand
for one case of that law and nothing wider: the linear, reciprocal, stable
network with $a\equiv0$, so
that $K=K^\T$ is the Hessian of a stored elastic energy, and
$K_{\!f\!f}\succ0$, so that the pinned network is stable and $C$ exists. For
the resistor and flow networks of Sec.~\ref{sec:universal} the word means the
same, and at finite frequency it carries its usual circuit-theoretic meaning. Each
result names the part of this it uses, and Theorem~\ref{thm:floorpd} is the
one result whose floor depends on $K_{\!f\!f}\succ0$.

Choose a set $S$ of driven degrees of freedom
and a set $T$ of read-out ones, with $|S|=m_S$, $|T|=m_T$, and let
\begin{equation*}
  P \;=\; T\cap S,
  \qquad p \;=\; |P| .
\end{equation*}
A learning rule is asked to shape the response block
\begin{equation*}
  R(\theta) \;=\; C[T,S] \;\in\; \R^{m_T\times m_S},
  \qquad
  \theta = (k) \ \text{ or } \ (k,a),\ \ \text{so } n_\theta=n_b \text{ or } 2n_b,
\end{equation*}
\begin{figure}[t]
  \centering
  \includegraphics[width=\textwidth]{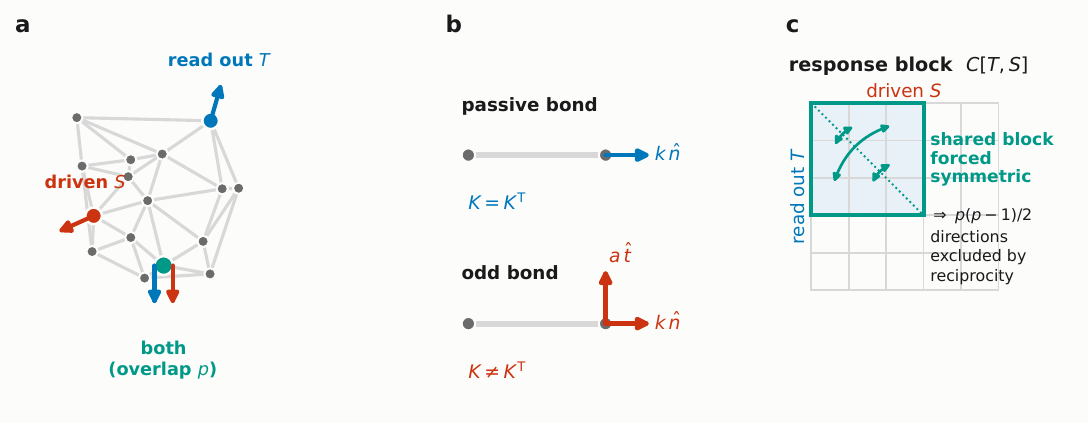}
  \caption{\textbf{Setup and the origin of the constraint.}
  \textbf{(a)}~A central-force network. Forces are applied on a set $S$ of
  degrees of freedom and displacements read on a set $T$; $p=|T\cap S|$ of them
  are both driven and read. Panels (a) and (c) are not drawn for the same $S$
  and $T$. At the node marked \emph{both}, the force (red) and the read-out
  (blue) are two parallel arrows: a shared degree of freedom is driven and read
  along the same coordinate.
  \textbf{(b)}~The bond force law. A passive bond pushes along its own axis and
  makes $K$ symmetric; an odd bond adds a transverse component $a_b\hat{t}_b$ and
  generically makes it asymmetric.
  \textbf{(c)}~The response block $C[T,S]$. Maxwell--Betti forces its $p\times p$
  sub-block on the shared degrees of freedom to be symmetric, which is
  $p(p-1)/2$ independent equalities.}
  \label{fig:setup}
\end{figure}
whose entry $R_{\tau\sigma}$ is the displacement at read-out degree of freedom
$\tau$ produced by a unit force at driven degree of freedom $\sigma$.

The \emph{reachable set} is the image $\mathcal{R}=\{R(\theta):\theta\in\Theta\}$
over the \emph{admissible} set $\Theta$: the parameters with $k_b>0$ at every
bond and $K_{\!f\!f}$ invertible, with $a\equiv0$ on the passive branch and
$a\in\R^{n_b}$ on the odd one. Individual results add hypotheses to this and
say where. At a generic $\theta$, where the rank is locally constant, the
local dimension of $\mathcal{R}$ is the rank of the
Jacobian $\partial R/\partial\theta$. That rank, rather than the number of
parameters, bounds how many independent behaviours the network can be trained
to hold.

\subsection{The Jacobian is exact and rank one per parameter}
\label{sec:jacobian}

Differentiating $C=K_{\!f\!f}^{-1}$ gives $\mathrm{d}C=-C\,(\mathrm{d}K)\,C$.
Every parameter derivative of \eqref{eq:stiffness} is a rank-one matrix,
\begin{equation*}
  \frac{\partial K}{\partial k_b}=\bq_b\bq_b^\T ,
  \qquad
  \frac{\partial K}{\partial a_b}=\bs_b\bq_b^\T ,
\end{equation*}
so for a rank-one perturbation $\mathrm{d}K=\bu\bv^\T$ the induced change in the
response block factorises. Writing $\theta_r$ for one of the scalar
parameters collected in $\theta$, the corresponding $m_T\times m_S$ derivative
(one column of the full $m_Tm_S\times n_\theta$ Jacobian, once flattened) is
\begin{equation}
  \label{eq:jaccol}
  \frac{\partial R}{\partial \theta_r}
  \;=\; -\,\bigl(C\bu\bigr)[T]\ \bigl(C^\T\bv\bigr)[S]^\T .
\end{equation}
Each Jacobian column is therefore an outer product costing $O(m_Tm_S)$ to form.
Every rank reported here is the rank of \eqref{eq:jaccol} assembled in closed
form, except in the two computations that use finite differences instead,
both listed in Appendix~\ref{app:stats}.

Equation~\eqref{eq:jaccol} already contains the whole result. For a passive
network $C=C^\T$, so with $\bu=\bv=\bq_b$ the two factors are restrictions of
one and the same vector
\begin{equation}
  \label{eq:wb}
  \bw_b \;=\; C\bq_b ,
  \qquad
  \frac{\partial R}{\partial k_b} \;=\; -\,\bw_b[T]\ \bw_b[S]^\T .
\end{equation}
With odd couplings switched on, $C\neq C^\T$ generically
(Remark~\ref{rem:odd-hyp} gives one family of exceptions), and the two factors
become restrictions of two \emph{different} vectors, $C\bq_b$ and $C^\T\bq_b$
(or $C\bs_b$ and $C^\T\bq_b$). That difference, one vector used twice against
two different vectors, is the whole of Theorem~\ref{thm:main} at the level of
the Jacobian; the containment itself needs no derivative
(Remark~\ref{rem:global}). The
local dimension of the reachable set at $\theta$ is thus the rank of
$n_\theta$ closed-form outer products, one per parameter: $n_b$ of them on the
passive branch and $2n_b$ on the odd one.

\subsection{Notation}
\label{sec:notation}

Table~\ref{tab:notation} collects the symbols used throughout, with the place
each is defined. Two conventions hold throughout. Calligraphic letters denote
subspaces, sets of blocks, sets of bonds, the stored energy and the
combinatorial objects assembled from a list of tasks (among them $\mathcal{S}_P$,
$\mathcal{A}_P$, $\mathcal{S}^{+}_P$, $\mathcal{Y}$, $\mathcal{R}$,
$\mathcal{E}$, $\mathcal{M}$, $\mathcal{S}_{\mathcal{L}}$,
$\mathcal{A}_{\mathcal{L}}$, $\mathcal{F}$, $\mathcal{G}$, $\mathcal{K}$,
$\mathcal{B}$, and the list
$\mathcal{L}$, the task family $\mathcal{T}$ and the digraph
$\mathcal{D}_{\mathcal{L}}$ themselves), never matrices. Bold lower case
denotes a vector or a row of the compatibility matrix ($\bq_b$, $\bs_b$,
$\bw_b$, $\bu$, $\mathbf{f}_j$, $\mathbf{g}_j$), never a scalar, so that the
generalised overlap $q$ of Corollary~\ref{cor:taskoverlap} is distinct from
the compatibility row $\bq_b$. Greek indices $\mu,\nu,\rho$ range over the
shared set $P$, whose elements are also written $r,c$ where a pair of them
indexes an entry, as in \eqref{eq:SP}, and $\tau$ and $\sigma$ over $T$ and over $S$, where they also
name the two injections placing $P$ inside each. Apart from the singular
values $\sigma_{\min}$, $\sigma_{\max}$ and the bond couple $\tau_b$ of
Sec.~\ref{sec:price} and the coordinates $\sigma_j$, $\tau_j$ of a coordinate
task in Sec.~\ref{sec:tasklist}, these Greek letters are used for nothing else;
the running degrees of freedom of Sec.~\ref{sec:tasklist} are the plain Latin
$u,v$. Five letters are also reused locally and defined where they occur: $u$ is
the log-stiffness of Sec.~\ref{sec:learning} and $v$ the drive amplitude of
Sec.~\ref{sec:consequences}; $B$ is the shared block of a realised response in
Sec.~\ref{sec:numerics} as well as of the target; the quarter turn $\epsilon$ of
Remark~\ref{rem:edgewedge} is distinct from the amplitude $\varepsilon$; and the
bold $\by^{*}_j$ of Sec.~\ref{sec:learning} is a read-out vector, not the task
value $y^{*}_j$. The Latin index $b$ always numbers a bond and $j$ a task, except in the
node pair $(i,j)$ of a bond and at the hinge node $j$ of an angular spring.

\begin{table}[p]
  \centering
  \scriptsize
  \renewcommand{\arraystretch}{0.80}
  \setlength{\abovecaptionskip}{4pt}
  \setlength{\belowcaptionskip}{4pt}
  \caption{\textbf{Symbols.} Subscripted symbols are listed in unsubscripted
  form, and the prime $\ell$ of Sec.~\ref{sec:certificates} is kept distinct
  from the overlap $p$. Some letters do double duty across distant sections:
  the damping $D$ against the transmission $D$, $D_{\nu\mu}$ and
  $\mathbf{D}_{\mathcal{T}}$; the mass $M$ against
  the $M$ in the proof of Thm.~\ref{thm:multiclamp} and the task dyad
  $\mathsf{M}_j$; the spatial
  dimension $\delta$ against the deficit, never written without a
  subscript; the domain
  $\Theta$ against the scaling in Prop.~\ref{prop:multirank};
  $N$ against the operator size in Thm.~\ref{thm:universal} and the count
$m_T+m_S$ of Sec.~\ref{sec:consequences}; $X$
  against the generic matrix of Sec.~\ref{sec:tasklist} and the node
  positions of Prop.~\ref{prop:generic}; the signed area
  $\Gamma$ against the $\Gamma$ of Appendix~\ref{app:jacobian};
  $J$: the task
  count here, a driven set in Sec.~\ref{sec:dirichlet},
  $\partial R/\partial\kappa$ in
  Sec.~\ref{sec:certificates};
  and $W$ with three meanings: the wedge matrix
  $\mathsf{W}$, the receptance blocks $W_{OZ},W_{JZ},W_{JJ}$ of
  Sec.~\ref{sec:dirichlet} and the scalar of \eqref{eq:duS}. The
  generic block $Z$ stands against the index set $Z$ there, $\mathsf{O}$ and
  $\mathsf{E}$ against the read-out set $O$, $O(\cdot)$, the subspace $E$ and
  the energy $\mathcal{E}$, and the bond length $L_b$ against $\mathsf{L}_a$,
  $\mathcal{L}$ and the electrode count $L$. No symbol
  carries two of these meanings inside one expression; only this table tells
  them apart.}
  \label{tab:notation}
  \begin{tabular}{@{}l@{\quad}p{0.64\textwidth}@{\quad}l@{}}
    \toprule
    Symbol & Meaning & Defined \\
    \midrule
    \multicolumn{3}{@{}l}{\emph{Network and response}}\\
    $N$, $n_b$, $\delta$ & nodes, bonds, spatial dimension
      & Sec.~\ref{sec:setup} \\
    $b$ & bond index, throughout & Sec.~\ref{sec:setup} \\
    $\bq_b$, $\bs_b$, $L_b$ & compatibility row of bond $b$; its transverse
      companion; its length & Sec.~\ref{sec:setup} \\
    $k_b$, $a_b$ & central-force stiffness; odd (non-reciprocal) coupling
      & \eqref{eq:stiffness} \\
    $\theta$, $n_\theta$ & tunable parameters $(k)$ or $(k,a)$, and their count
      & Sec.~\ref{sec:setup} \\
    $K$, $K_{\!f\!f}$ & stiffness operator; its restriction to the
      $n_{\mathrm{free}}$ free degrees of freedom & \eqref{eq:stiffness} \\
    $C$ & compliance, $C=K_{\!f\!f}^{-1}$ & Sec.~\ref{sec:setup} \\
    $\mathcal{E}$, $\mathcal{E}^{*}$ & stored energy; its value on the
      equilibrium branch & Sec.~\ref{sec:numerics} \\
    $\bw_b$ & $C\bq_b$, the vector whose outer square is a Jacobian column
      & \eqref{eq:wb} \\
    $\bz_b$ & $C\bs_b$; $\bz_b|_P\wedge\bw_b|_P$ is what an odd bond adds
      to $\mathcal{A}_P$ & \eqref{eq:wedge} \\
    $\kappa_b$ & rescaled stiffness $k_b/L_b^{2}$, rational at integer nodes
      & Sec.~\ref{sec:certificates} \\
    $\Delta_K$ & $\det K_{\!f\!f}$, cleared to certify a rank over
      $\mathbb{F}_\ell$ & Sec.~\ref{sec:certificates} \\
    $\ell$ & the certifying prime, reserved & Sec.~\ref{sec:certificates} \\
    \midrule
    \multicolumn{3}{@{}l}{\emph{Layout}}\\
    $S$, $T$ & driven and read-out index sets, $|S|=m_S$, $|T|=m_T$
      & Sec.~\ref{sec:setup} \\
    $P$, $p$ & the shared set $T\cap S$ and its size, the \emph{overlap}
      & Sec.~\ref{sec:setup} \\
    $U$, $n_U$ & the set $T\cup S$ and its size, the
      \emph{instrumentation budget}, a shared degree of freedom counted once
      & Prop.~\ref{prop:budget} \\
    $\alpha$, $\gamma$ & exclusively read and exclusively driven counts,
      $m_T-p$ and $m_S-p$ & Prop.~\ref{prop:budget} \\
    $\tau$, $\sigma$ & injections placing $P$ inside $T$ and inside $S$
      & Sec.~\ref{sec:bound} \\
    $\mu,\nu,\rho$ & terminals in $P$ under the imposed-displacement drive
      & Sec.~\ref{sec:dirichlet} \\
    $R(\theta)$, $R^{*}$ & the response block $C[T,S]$; the target block
      & Sec.~\ref{sec:setup} \\
    $\mathcal{R}$ & the reachable set $\{R(\theta)\}$ & Sec.~\ref{sec:setup} \\
    $d$ & $\dim\mathcal{S}_P=m_Tm_S-p(p-1)/2$, the symmetry ceiling & Sec.~\ref{sec:attainment} \\
    $\delta_p$ & the block deficit $p(p-1)/2$, the dimensions reciprocity
      removes & Sec.~\ref{sec:tasklist} \\
    \midrule
    \multicolumn{3}{@{}l}{\emph{Task lists}}\\
    $\mathcal{L}$, $J$, $j$ & a task list
      $\{(\mathbf{f}_j,\mathbf{g}_j,y^{*}_j)\}_{j=1}^{J}$: the pattern driven,
      the functional read and the value wanted, for each of its $J$ tasks
      & Sec.~\ref{sec:tasklist} \\
    $y(C)$, $y^{*}$ & realised value vector and target, both in $\R^{J}$
      & \eqref{eq:taskvalue} \\
    $\mathsf{M}_j$, $\mathcal{M}$, $\mathcal{M}^{s}$ & task dyad
      $\mathbf{g}_j\mathbf{f}_j^\T$; its span; the span of the symmetric parts
      & \eqref{eq:taskvalue} \\
    $\delta_{\mathcal{L}}$, $\mathcal{D}_{\mathcal{L}}$ & the deficit
      $\dim\mathcal{M}-\dim\mathcal{M}^{s}$; the drive-to-read digraph whose
      digons it counts & \eqref{eq:deltaL}, Cor.~\ref{cor:digon} \\
    $\mathcal{A}^{0}_{\mathcal{L}}$, $\mathcal{A}_{\mathcal{L}}$,
      $\mathcal{A}^{1}_{\mathcal{L}}$, $\mathcal{S}_{\mathcal{L}}$ & relations
      among the task values forced by linearity alone, by reciprocity, their
      difference (dimension $\delta_{\mathcal{L}}$), and the realisable values
      $\mathcal{A}_{\mathcal{L}}^{\perp}$ & Sec.~\ref{sec:tasklist} \\
    $\mathcal{F}$, $\mathcal{G}$, $q$ & spans of the drives and of the reads;
      the \emph{generalised overlap} $\dim(\mathcal{F}\cap\mathcal{G})$
      & Cor.~\ref{cor:taskoverlap} \\
    \midrule
    \multicolumn{3}{@{}l}{\emph{Subspaces, floors and ceilings}}\\
    $\mathcal{S}_P$ & blocks symmetric on the shared sub-block & \eqref{eq:SP} \\
    $\mathcal{A}_P$ & its orthogonal complement, $\mathcal{S}_P^{\perp}$
      & Thm.~\ref{thm:floor} \\
    $\mathcal{S}^{+}_P$ & the closed convex subset of $\mathcal{S}_P$ whose
      shared sub-block is positive semidefinite & Thm.~\ref{thm:floorpd} \\
    $B$, $B_s$, $B_a$ & the target's shared block
      $R^{*}[\tau(P),\sigma(P)]$, and its symmetric and antisymmetric parts
      & Thm.~\ref{thm:floor} \\
    $\pi_P$ & the shared sub-block, $A\mapsto A[\tau(P),\sigma(P)]$
      & Sec.~\ref{sec:floorpd} \\
    $\Pi_{\mathcal{S}_P}$, $\Pi_{\mathcal{A}_P}$ & orthogonal projections onto
      $\mathcal{S}_P$ and $\mathcal{A}_P$ & Thm.~\ref{thm:floor} \\
    $\Pi_T$, $\Pi_S$ & coordinate selections onto $T$ and $S$
      & Thm.~\ref{thm:universal} \\
    $\Lambda(p)$ & the ceiling of Theorem~\ref{thm:main} at fixed budget $n_U$,
      maximised over the split & \eqref{eq:budget} \\
    $V$, $\mathcal{Y}$ & the span of the bond dyads; the symmetrised lift of
      $\mathcal{S}_P$ & Thm.~\ref{thm:duality} \\
    \midrule
    \multicolumn{3}{@{}l}{\emph{Drives, training and dynamics}}\\
    $C_P$ & the shared principal block $C[P,P]$ & Thm.~\ref{thm:dirichlet} \\
    $D_{\nu\mu}$ & transmission: displacement at $\nu$ under a unit imposed
      displacement at $\mu$ & \eqref{eq:dtrans} \\
    $D_{B\leftarrow A}$ & transmission matrix, read on $B$ under a displacement
      pattern imposed on $A$ & \eqref{eq:mtrans} \\
    $X$, $Y$ & $C[A,A]$ and $C[B,B]$, the two driving-point compliance blocks
      & \eqref{eq:mtrans} \\
    $\mathbf{D}_{\mathcal{T}}$, $n_{\mathcal{T}}$, $\Phi_{\mathcal{T}}$ &
      the tuple of transmissions of a task family $\mathcal{T}$, its entry
      count, and the map $C_P\mapsto\mathbf{D}_{\mathcal{T}}$
      & Prop.~\ref{prop:multirank} \\
    $\varepsilon$, $\Psi$ & amplitude and unit direction of the target
      perturbation, $\|\Psi\|_F=1$ & Sec.~\ref{sec:learning} \\
    $\beta$, $\chi$ & nudge amplitude and step rate of coupled learning
      & Sec.~\ref{sec:learning} \\
    $t$ & antisymmetric contamination of the layout-task target
      & Sec.~\ref{sec:consequences} \\
    $M$, $D$, $G(\omega)$ & mass; damping; dynamic Green function
      & Sec.~\ref{sec:numerics} \\
    \midrule
    \multicolumn{3}{@{}l}{\emph{The price of non-reciprocity}}\\
    $\operatorname{sym}$, $\operatorname{anti}$ & symmetric and antisymmetric
      parts of a square matrix, $\tfrac12(Z+Z^\T)$ and $\tfrac12(Z-Z^\T)$
      & Sec.~\ref{sec:tasklist}, Sec.~\ref{sec:price} \\
    $\eta$ & non-reciprocity ratio,
      $\|(\operatorname{sym}Z)^{-1/2}(\operatorname{anti}Z)
      (\operatorname{sym}Z)^{-1/2}\|_2$ for a block $Z$ with
      $\operatorname{sym}Z\succ0$: its odd part weighed against its even part
      & Sec.~\ref{sec:price} \\
    $\operatorname{vec}_<$ & the $\delta_p$ strictly upper entries of an
      antisymmetric $p\times p$ block, in the ordering of $P$
      & Cor.~\ref{cor:price} \\
    $\mathsf{W}$, $\mathcal{B}$ & the wedge matrix, column $b$ the vectorised
      $\bz_b|_P\wedge\bw_b|_P$; a support of $\delta_p$ bonds, giving
      $\mathsf{W}_{\mathcal{B}}$ & Cor.~\ref{cor:price} \\
    $\mathsf{L}_a$ & the graph Laplacian with signed bond weights $a_b$
      & Cor.~\ref{cor:price} \\
    $\mathcal{K}$ & the torque-free couplings, $(a_bL_b)_b$ a state of
      self-stress & Prop.~\ref{prop:torquefree} \\
    $\tau_b$, $\Gamma$ & the couple $a_be_bL_b$ of odd bond $b$; the signed
      area a closed force cycle encloses
      & Prop.~\ref{prop:torquefree}, \eqref{eq:cycle} \\
    \bottomrule
  \end{tabular}
\end{table}


\section{What reciprocity forbids}
\label{sec:bound}

This section proves what reciprocity takes from a network that already exists.
The set a reciprocal network reaches by varying its stiffnesses, driven by forces, lies in a subspace of
codimension $p(p-1)/2$ (Theorem~\ref{thm:main}, Remark~\ref{rem:global};
Fig.~\ref{fig:setup}c), the distance from a target to that
subspace is an error floor, and passivity adds a second, orthogonal term. The
reciprocity statements hold for any symmetric invertible operator whose
response is a block of its inverse. For a finite list of
tasks, $p(p-1)/2$ is replaced by an invariant of the list, which on coordinate
tasks is the number of reversed pairs, and an imposed-displacement drive obeys
a law of its own. The last two subsections reduce attainment of the bound to a
transversality condition and certify it at rational configurations.

The subspace in question is the set of blocks whose sub-block on the shared
degrees of freedom is symmetric. Fix an ordering $r_1<\dots<r_p$ of the shared
set $P$ and let
$\tau:P\to\{1,\dots,m_T\}$ and $\sigma:P\to\{1,\dots,m_S\}$ be the injections
giving the position of each shared degree of freedom inside $T$ and inside $S$.
Define
\begin{equation}
  \label{eq:SP}
  \mathcal{S}_P
  \;=\;
  \bigl\{\,A\in\R^{m_T\times m_S}\ :\
  A_{\tau(r)\sigma(c)}=A_{\tau(c)\sigma(r)}\ \ \forall\, r,c\in P \,\bigr\},
\end{equation}
in which each condition equates the entry read at $r$ and driven at $c$ with
the entry read at $c$ and driven at $r$.

\begin{theorem}[Capacity under Maxwell--Betti]
  \label{thm:main}
  Let $K_{\!f\!f}$ be symmetric and invertible at every admissible $k$. This
  holds in particular for a passive network, $a\equiv 0$ with $K_{\!f\!f}\succ 0$, and it is
  the only hypothesis the proof uses. Then for
  every admissible stiffness vector $k$,
  \begin{equation}
    \label{eq:bound}
    \rank \frac{\partial R}{\partial k}
    \;\le\;
    \min\!\Bigl(\,n_b,\ \ m_Tm_S-\tfrac{1}{2}p(p-1)\,\Bigr).
  \end{equation}
  In particular, with $T=S$ and $m_T=m_S=m$ the set reachable by varying $k$
  has local dimension at most $m(m+1)/2$, so a fraction at least $(m-1)/(2m)$ of the $m^2$ target
  dimensions is unreachable, for every such network, passive ones included,
  whatever its size and its topology.
\end{theorem}

\begin{proof}
  The first argument of the minimum is the number of columns. For the second, we
  show that every column lies in $\mathcal{S}_P$ and that
  $\dim\mathcal{S}_P=m_Tm_S-p(p-1)/2$.

  \emph{Columns lie in $\mathcal{S}_P$.} The argument needs only that the
  perturbation is symmetric. If $\mathrm{d}K=\mathrm{d}K^\T$ then
  $\mathrm{d}C=-C(\mathrm{d}K)C$ is symmetric, so its
  $(\tau(r),\sigma(c))$ entry, which is the entry of $\mathrm{d}C$ at the
  pair of degrees of freedom $(r,c)\in P\times P$, is invariant under
  $r\leftrightarrow c$, which is precisely the defining condition
  \eqref{eq:SP}. Since $\partial K/\partial k_b=\bq_b\bq_b^\T$ is symmetric,
  every column qualifies; explicitly, by \eqref{eq:wb} the $b$-th column is
  $-\bw_b[T]\,\bw_b[S]^\T$ with $(\tau(r),\sigma(c))$ entry
  $-\,\bw_b(r)\,\bw_b(c)$. As $\mathcal{S}_P$ is a linear subspace, the image of
  the Jacobian lies inside it.

  \emph{Codimension.} The defining conditions are indexed by unordered pairs
  $\{r,c\}\subset P$ with $r\neq c$, of which there are $p(p-1)/2$. Because
  $\tau$ and $\sigma$ are injective, the map $(r,c)\mapsto(\tau(r),\sigma(c))$ is
  injective, so the $p(p-1)$ matrix entries appearing in these conditions are
  pairwise distinct, and each condition involves two entries that occur in no
  other condition. The conditions are therefore linearly independent and
  $\codim\mathcal{S}_P=p(p-1)/2$. Combining the two facts gives
  \eqref{eq:bound}.

  For $T=S$ we have $p=m$ and $m^2-m(m-1)/2=m(m+1)/2$; the unreachable fraction
  is $\bigl(m^2-m(m+1)/2\bigr)/m^2=(m-1)/(2m)$.

  Positive definiteness appears nowhere in the argument. All that is asked of
  $K_{\!f\!f}$ is that $C=K_{\!f\!f}^{-1}$ exist and satisfy $C=C^\T$; the
  hypothesis is therefore stated as symmetry plus invertibility, and
  passivity is kept for Theorem~\ref{thm:floorpd}, which is the result that
  genuinely uses it.
\end{proof}

\begin{remark}[the containment is global, and derivative-free]
  \label{rem:global}
  The tangent-space statement understates the situation. Since $C=C^\T$ for
  \emph{every} admissible $k$, the response block itself satisfies
  $R(k)\in\mathcal{S}_P$ for all $k$: the whole set reachable by varying $k$,
  not merely its tangent space, lies in a linear subspace of codimension
  $p(p-1)/2$. That
  containment is elementary and uses no derivative, and it is what the error
  floor of Theorem~\ref{thm:floor} rests on. The Jacobian is needed for
  everything the containment cannot decide: which dimensions inside
  $\mathcal{S}_P$ a given network actually reaches (the attainment question of
  Theorem~\ref{thm:duality}, where the bound is not obviously tight); the
  non-reciprocal case of Proposition~\ref{prop:odd}, where $C$ is not symmetric
  and the containment argument gives nothing; and the layout law of
  Proposition~\ref{prop:budget}, which is a statement about rank.
\end{remark}

\begin{remark}[the two branches have different scopes]
  The symmetry branch $m_Tm_S-p(p-1)/2$ is uniform in $k$, holding at every
  point of parameter space and not merely generically, and it is independent of
  $N$, of the coordination number, of the topology, of how far the network is
  from isostaticity, and of whether $K_{\!f\!f}$ is positive definite. The other
  branch, $n_b$, is a bond count and depends on the size, the coordination
  and the topology; in the rank sweep of experiment~03 (experiments are
  numbered as the deposited scripts of Table~\ref{tab:scripts}) it is the
  binding one in a substantial minority of cases. It is the symmetry branch that cannot be repaired by adding
  bonds.
\end{remark}

\subsection{The error floor}

Theorem~\ref{thm:main} bounds a dimension. Remark~\ref{rem:global} gives what
an experimentalist can use directly: the whole set reachable by varying $k$
lies in $\mathcal{S}_P$. A target outside $\mathcal{S}_P$ therefore has a floor below
which no training run can go.

Let $\Pi_{\mathcal{S}_P}$ and $\Pi_{\mathcal{A}_P}$ be the orthogonal
projections onto $\mathcal{S}_P$ and onto $\mathcal{A}_P=\mathcal{S}_P^{\perp}$,
the latter spanned by the $p(p-1)/2$ blocks
$E_{\tau(r)\sigma(c)}-E_{\tau(c)\sigma(r)}$, where $E_{ij}$ is the matrix unit
with a $1$ in entry $(i,j)$ and zeros elsewhere.

\begin{theorem}[error floor]
  \label{thm:floor}
  Let $R^{*}\in\R^{m_T\times m_S}$ be any target response block, and let
  $\|\cdot\|_F$ denote the Frobenius norm. For every network whose
  $K_{\!f\!f}$ is symmetric and invertible at every admissible stiffness
  vector, passive networks included, whatever its size and its topology, and
  for every admissible stiffness vector $k$,
  \begin{equation}
    \label{eq:floor}
    \bigl\|R(k)-R^{*}\bigr\|_F \;\ge\; \bigl\|\Pi_{\mathcal{A}_P}R^{*}\bigr\|_F
    \;=\; \tfrac{1}{2}\bigl\|B-B^\T\bigr\|_F ,
  \end{equation}
  where $B=R^{*}[\tau(P),\sigma(P)]$ is the target restricted to the shared
  degrees of freedom, its rows and its columns indexed by $P$ in one common
  order.
\end{theorem}

\begin{proof}
  By Remark~\ref{rem:global}, $R(k)\in\mathcal{S}_P$ for every admissible $k$,
  so it suffices to bound $\|A-R^{*}\|_F$ from below for $A\in\mathcal{S}_P$.
  Because $A\in\mathcal{S}_P$ we have $\Pi_{\mathcal{S}_P}(A-R^{*})
  =A-\Pi_{\mathcal{S}_P}R^{*}$ and $\Pi_{\mathcal{A}_P}(A-R^{*})
  =-\Pi_{\mathcal{A}_P}R^{*}$, and the orthogonal splitting
  $\R^{m_T\times m_S}=\mathcal{S}_P\oplus\mathcal{A}_P$ gives
  \[
    \|A-R^{*}\|_F^2
    = \bigl\|A-\Pi_{\mathcal{S}_P}R^{*}\bigr\|_F^2
    + \bigl\|\Pi_{\mathcal{A}_P}R^{*}\bigr\|_F^2
    \;\ge\; \bigl\|\Pi_{\mathcal{A}_P}R^{*}\bigr\|_F^2 ,
  \]
  the first term being non-negative. For the evaluation, $\mathcal{A}_P$ is
  spanned by the blocks $E_{\tau(r)\sigma(c)}-E_{\tau(c)\sigma(r)}$, so
  $\Pi_{\mathcal{A}_P}$ keeps the antisymmetric part of the shared block,
  $\tfrac12(B-B^\T)$, and annihilates everything else.
\end{proof}

The floor constrains the \emph{whole} reachable set. It therefore holds without
any genericity or attainment hypothesis. It is also quantitative: the floor is
the norm of the antisymmetric part of the target's shared block, a number
computed from the target alone before any training is attempted.

Theorem~\ref{thm:floor} is only a lower bound. The set reachable by varying $k$ is contained in
$\mathcal{S}_P$ but is in general a proper subset of it. The floor is therefore
tight only when the part of the target lying in $\mathcal{S}_P$ is itself
reachable, and loose otherwise. Both regimes occur, and
Sec.~\ref{sec:learning} measures how far apart they are.

\subsection{The floor with passivity}
\label{sec:floorpd}

Theorems~\ref{thm:main} and~\ref{thm:floor} use one hypothesis between them,
that $K_{\!f\!f}$ is symmetric and invertible. A passive network carries a
second, $K_{\!f\!f}\succ0$. It adds a term to the floor, and on a generic
target (Remark~\ref{rem:fraction-pd}) this term is comparable in size to the
reciprocity term.

If $K_{\!f\!f}$ is positive definite then so is $C=K_{\!f\!f}^{-1}$. Write
$\pi_P:A\mapsto A[\tau(P),\sigma(P)]$ for the shared sub-block of a response
block. The shared block of the \emph{realised} response,
\begin{equation*}
  \pi_P\bigl(R(k)\bigr) \;=\; R(k)[\tau(P),\sigma(P)] \;=\; C[P,P] ,
\end{equation*}
is a submatrix whose rows and columns are indexed by the same set $P$ in the
same order. It is therefore a \emph{principal} submatrix of a positive-definite
matrix, and so itself symmetric \emph{positive definite}.

Passivity thus places the reachable set inside the closed convex set
\begin{equation*}
  \mathcal{S}^{+}_P
  \;=\;
  \bigl\{\,A\in\R^{m_T\times m_S}\ :\ A[\tau(P),\sigma(P)]\in\Sym^{+}(p)\,\bigr\}
\end{equation*}
where $\Sym^{+}(p)$ is the closed cone of positive-semidefinite $p\times p$
matrices and $\Sym^{++}(p)$ its interior, the positive-definite ones; the
realised block lies in $\Sym^{++}(p)$. With $B=\pi_P(R^{*})$ the shared block
of the target, as in Theorem~\ref{thm:floor}, write $B_a=\tfrac12(B-B^\T)$ and
$B_s=\tfrac12(B+B^\T)$, with $\lambda_1,\dots,\lambda_p$ the eigenvalues of
$B_s$. The distance from $R^{*}$ to $\mathcal{S}^{+}_P$ has a reciprocity
term, set by $B_a$, and a passivity term, set by the negative eigenvalues of
$B_s$.

\begin{theorem}[error floor under reciprocity and passivity]
  \label{thm:floorpd}
  For every passive network, whatever its size and topology, and every
  admissible stiffness vector $k$,
  \begin{equation}
    \label{eq:floorpd}
    \bigl\|R(k)-R^{*}\bigr\|_F
    \;\ge\;
    \operatorname{dist}_F\!\bigl(R^{*},\mathcal{S}^{+}_P\bigr)
    \;=\;
    \Bigl(
      \underbrace{\tfrac14\bigl\|B-B^\T\bigr\|_F^{2}}_{\text{reciprocity}}
      \;+\;
      \underbrace{\textstyle\sum_{i}\min(\lambda_i,0)^{2}}_{\text{passivity}}
    \Bigr)^{1/2} .
  \end{equation}
  The two terms are orthogonal. The second vanishes exactly when
  $B_s\succeq0$, where \eqref{eq:floorpd} reduces to \eqref{eq:floor}; otherwise
  \eqref{eq:floorpd} is strictly larger.
\end{theorem}

\begin{proof}
  Because $\tau$ and $\sigma$ are injective, the $p^{2}$ positions
  $(\tau(r),\sigma(c))$ are distinct entries of the $m_T\times m_S$ block, so
  for every $A\in\R^{m_T\times m_S}$ the Frobenius norm splits over them and
  their complement,
  \begin{equation*}
    \|A-R^{*}\|_F^2
    = \bigl\|\pi_P(A)-B\bigr\|_F^2
      + \bigl\|(A-R^{*})|_{\mathrm{rest}}\bigr\|_F^2
    \;\ge\; \bigl\|\pi_P(A)-B\bigr\|_F^2 .
  \end{equation*}
  Taking $A=R(k)$: every realised $\pi_P(R(k))$ is positive definite, so the
  right-hand side is at least $\min_{Z\succeq0}\|Z-B\|_F^2$, the squared
  distance from $B$ to $\Sym^{+}(p)$.
  For symmetric $Z$ the splitting $B=B_s+B_a$ is orthogonal and $Z-B_s$ is
  symmetric, so $\|Z-B\|_F^2=\|Z-B_s\|_F^2+\|B_a\|_F^2$. Minimising
  $\|Z-B_s\|_F^2$ over the positive-semidefinite cone is the eigenvalue
  truncation, with residual $\sum_i\min(\lambda_i,0)^2$, and the resulting
  distance from $B$ to $\Sym^{+}(p)$, the right-hand side of
  \eqref{eq:floorpd}, is Theorem~2.1 of \citet{Higham1988}. Finally
  $\|B_a\|_F=\tfrac12\|B-B^\T\|_F=\|\Pi_{\mathcal{A}_P}R^{*}\|_F$, which is the
  first term. The infimum is taken over the closed cone, so it is a valid lower
  bound for the open one.
\end{proof}

\begin{remark}[what this does and does not change]
  \label{rem:pd-scope}
  Theorem~\ref{thm:main} is untouched. The blocks of $\mathcal{S}_P$ whose
  shared sub-block is positive definite form an \emph{open} subset of
  $\mathcal{S}_P$, hence a full-dimensional one, so passivity costs
  no dimensions at all and the rank bound stands exactly as proved. Only the
  metric statement changes. That is also why passivity is absent from the
  hypotheses of Theorems~\ref{thm:main} and~\ref{thm:floor} and present here:
  the sharp division is that a rank bound can see symmetry but not
  convexity, so everything passivity contributes is contained in the
  second term of \eqref{eq:floorpd} and in nothing else.

  The two terms of \eqref{eq:floorpd} are also of different kinds. The
  reciprocity term is a \emph{symmetry}: it is a linear
  condition, it removes a definite codimension, and it is what makes
  Theorem~\ref{thm:main} a rank bound at all. The passivity term is a convexity
  condition on the same block; it removes no dimensions, and it is not new in
  substance. That a passive network cannot realise every conceivable linear map
  is already stated in words by \citet{Stern2021}, who note that non-negativity
  of the tunable elements ``excludes many conceivable linear mappings''.
  Equation~\eqref{eq:floorpd} is that observation made exact, put on the same
  footing as the reciprocity term, and made computable from the target alone.
  Its two ingredients are classical too: the distance from a matrix to the nearest
  symmetric positive semidefinite one, which for $B$ is the whole right-hand
  side of \eqref{eq:floorpd}, is Theorem~2.1 of \citet{Higham1988}, and that a
  principal submatrix of a positive-definite matrix is positive definite is
  textbook. What is ours in the term is the identification alone --- which
  block, why it is principal, and that the number follows from the target
  without building a network.
\end{remark}

\begin{remark}[the forfeited norm fraction of a generic target]
  \label{rem:fraction-pd}
  Two different quantities have a claim to the name \emph{forfeited fraction},
  and they are kept apart from here on. The forfeited \emph{dimension} fraction
  is $p(p-1)/(2m_Tm_S)$, the codimension of Theorem~\ref{thm:main} as a
  share of the block; it is what a rank measurement returns, and it is the
  quantity the deficit law of Sec.~\ref{sec:numerics} reports. The forfeited
  \emph{norm} fraction is
  $\mathbb{E}\|\Pi_{\mathcal{A}_P}R^{*}\|_F^{2}/\mathbb{E}\|R^{*}\|_F^{2}$, the
  expected squared norm of a generic target that lies in forbidden directions;
  it is the quantity the last column of Table~\ref{tab:layouts} carries at the
  three force-driven layouts there, the law of Sec.~\ref{sec:dirichlet}
  supplying the entry at the sixteen layouts driven by an imposed displacement. All
  nineteen published layouts in that table have $p=0$, where both fractions
  vanish under either law. For an isotropic Gaussian
  ensemble the two coincide numerically, both being $(m-1)/(2m)$ at full
  overlap, but they are different measurements of different objects, and only
  the first is a dimension. This remark computes the second, and adds to it a
  passivity term that has no dimension counterpart at all.

  At full overlap, with $m_T=m_S=p=m$ and $R^{*}$ drawn with independent
  standard normal entries, $\mathbb{E}\|B_a\|_F^2=\tfrac12(m^2-m)$, while $B_s$
  is a Gaussian orthogonal ensemble matrix, whose spectrum is symmetric about
  zero, so
  $\mathbb{E}\sum_i\min(\lambda_i,0)^2=\tfrac12\mathbb{E}\|B_s\|_F^2
  =\tfrac14(m^2+m)$. Dividing by $\mathbb{E}\|R^{*}\|_F^2=m^2$,
  \begin{equation}
    \label{eq:fracpd}
    \frac{m-1}{2m}\ \ (\text{reciprocity})
    \;+\;
    \frac{m+1}{4m}\ \ (\text{passivity})
    \;=\;
    \frac{3m-1}{4m}\ \ (\text{together}) ,
  \end{equation}
  tending to $\tfrac12+\tfrac14=\tfrac34$. These are fractions of the
\emph{expected squared} norm; the corresponding root-mean-square fraction of
the norm is $\sqrt{(3m-1)/(4m)}$, which is $0.816$ at $m=3$, $0.829$
  at $m=4$ and $0.851$ at $m=10$, rising towards $\sqrt{3}/2$. Experiment~13,
  part~4, confirms \eqref{eq:fracpd} by Monte Carlo over $4000$ draws at each of
  the five sizes $m=3,4,6,8,10$ ($20\,000$ draws in all); the largest discrepancy between
  measurement
  and closed form is
  $1.4$ Monte Carlo standard errors.
\end{remark}

The closed form in \eqref{eq:floorpd} was checked numerically against a direct
projection onto the positive-semidefinite cone, with no network involved. Over
the $3200$ random blocks at $p=1,\dots,8$ of experiment~13, part~1, the two
agree to $2.9\times10^{-16}$ relative to $\|B\|_F$, and the eigenvalue form of
the passivity term to $2.0\times10^{-15}$.
Theorem~\ref{thm:floorpd} coincides with Theorem~\ref{thm:floor} on all $1187$
blocks whose symmetric part is already positive semidefinite, and exceeds it
strictly on all $2013$ of the blocks where it is not.

The passivity term vanishes on any target built by perturbing a realised
response by a small antisymmetric amount, since the symmetric part of the
shared block is then still positive definite. The forbidden targets of
Sec.~\ref{sec:learning} are constructed this way, so the passivity term does
not bind there.

\subsection{The bound is a statement about symmetric operators}
\label{sec:universal}

Nothing in the proofs of Theorems~\ref{thm:main} and~\ref{thm:floor} used a
spring, a bond, or two dimensions. Between them they used one fact: the
operator being inverted is symmetric and invertible at every admissible
parameter value. Theorem~\ref{thm:floor} uses no derivative at all. The proof
of Theorem~\ref{thm:main} reads a symmetric perturbation off
$\partial K/\partial k_b=\bq_b\bq_b^\T$, and that symmetry is already forced by
the same fact: an operator that is symmetric at every parameter value has a
symmetric derivative. Stated at that level, the bound covers a class much
larger than the one we simulate. Theorem~\ref{thm:floorpd} needs one fact more,
that the operator is positive definite, and transfers to any setting where that
also holds, which it does for the static compliance of every passive network.

\begin{theorem}[reciprocity bound for symmetric linear response]
  \label{thm:universal}
  Let $H(\theta)=H(\theta)^\T\in\R^{N\times N}$ be invertible and depend
  differentiably on parameters $\theta\in\R^{n_\theta}$, and let
  $R(\theta)=\Pi_T\,H(\theta)^{-1}\Pi_S^\T$ be the block read on an index set $T$
  and driven on an index set $S$, with $\Pi_T,\Pi_S$ the corresponding coordinate
  selections and $p=|T\cap S|$. Then $R(\theta)\in\mathcal{S}_P$ for every
  $\theta$, so
  \begin{equation*}
    \rank\frac{\partial R}{\partial\theta}
    \;\le\;\min\Bigl(n_\theta,\ m_Tm_S-\tfrac12 p(p-1)\Bigr),
  \end{equation*}
  and for every target $R^{*}$,
  $\|R(\theta)-R^{*}\|_F\ge\tfrac12\|B-B^\T\|_F$ with $B$ the target on the
  shared indices. The same holds verbatim for complex symmetric
  $H=H^\T\in\mathbb{C}^{N\times N}$, with $\mathcal{S}_P$ read over
  $\mathbb{C}$ and dimensions counted over $\mathbb{C}$;
  only $H=H^\T$ is used, never $H=H^{*}$. That the finite-frequency structure of
  a lossy passive network is complex symmetry with negative semidefinite
  imaginary part, rather than Hermiticity, is established
  by~\citet{MiltonSeppecher2008}; what is used here is only the symmetry.
\end{theorem}

\begin{proof}
  $H=H^\T$ gives $H^{-1}=H^{-\T}$, so the $(\tau(r),\sigma(c))$ and
  $(\tau(c),\sigma(r))$ entries of $R$ are the same entry of $H^{-1}$ and
  $R\in\mathcal{S}_P$; the two displays are then the codimension count and the
  orthogonal splitting of the proofs above, neither of which used anything else.
\end{proof}

The elastic network of \eqref{eq:stiffness} is the case
$H=K_{\!f\!f}$. A network of variable \emph{linear} resistors with
conductances $g_e$ has $H=\sum_e g_e\,\mathbf{d}_e\mathbf{d}_e^\T$, with
$\mathbf{d}_e$ the incidence vector of edge $e$, equally symmetric. A learning resistor network of that kind~\citep{Dillavou2022}
therefore \emph{would} obey the same bound on its conductance-to-voltage map if
it were driven by injected currents and read on overlapping nodes. That is a
hypothetical configuration of their hardware, not their published protocol,
which clamps input voltages and is the drive of Sec.~\ref{sec:dirichlet}. The
bound transfers likewise to a Hagen--Poiseuille flow
network~\citep{Rocks2019}, and to any linear reciprocal medium whose response
is a symmetric Green function. In that sense the constraint applies to any
reciprocal learning medium driven by injected sources and read on overlapping
terminals, with spring networks as one instance.

The later transistor-based network of~\citet{Dillavou2024} is nonlinear, which
narrows the claim without removing it: the theorem then governs the
small-signal response about the trained operating point rather than the
trained input--output map, wherever the tangent operator there is symmetric, as
it is for any element whose current is a function of the voltage across it. In
the mechanical case, where the tangent stiffness is a Hessian, the tangent
bound is verified to survive large excursions, at bond strains up to $60.8\%$
(experiment~10, Sec.~\ref{sec:numerics}). We do not settle whether a
transistor edge has a symmetric tangent conductance. Their two layouts in
Table~\ref{tab:layouts} are counted by their driven and read-out sets alone,
which is all their table entry needs.

Two warnings apply. First, the theorem
transfers the \emph{bound}. It says nothing about attainment, which depends on
the subspace $V$ spanned by the particular parameter dyads
(Sec.~\ref{sec:attainment}) and differs from one physical setting to the next. Second, the codimension is $p(p-1)/2$ only
where $R$ is genuinely a block of $H^{-1}$. Node \emph{voltages} read against
injected \emph{currents} qualify. A mixed formulation, reading a current against
a voltage source, falls outside, and so does the drive used in most
physical-learning hardware, in which the input degrees of freedom are held at
prescribed \emph{displacements} and the output is read wherever the network
puts it. That drive obeys a law of a different shape, stated in Sec.~\ref{sec:dirichlet}.
Before it, Sec.~\ref{sec:tasklist} takes up a departure of another kind, in the
target rather than in the drive.

\subsection{A list of tasks, and the deficit it carries}
\label{sec:tasklist}

A laboratory seldom trains a full response block. What it imposes is a finite
list of demands (apply this load, read that displacement, want this number),
and the block $C[T,S]$ is the special case in which all $m_Tm_S$ drive--read
pairs available on the chosen terminals are on the list at once. A list may
impose two of the four values of a $2\times2$ block, or one value of a
$10\times10$ one. The overlap $p=|T\cap S|$ is then an invariant of a block
nobody asked for.

This section replaces $p(p-1)/2$ by an invariant of the list itself, its
deficit $\delta_{\mathcal{L}}$, and shows that
Theorems~\ref{thm:main} and~\ref{thm:floor} carry over with the replacement.
The rank ceiling carries over verbatim. The floor gains a second term,
produced by the list itself, which the full block does not have. The
section then recovers both theorems as the full-block case and, for tasks that
drive and read single degrees of freedom, identifies the deficit as a count of
reversed pairs of tasks. That count
sharpens a forecast of Sec.~\ref{sec:consequences}, replacing the shared
terminal by the reversed pair as the quantity that has to rise.

Write $n=n_{\mathrm{free}}$. A \emph{task list} is a finite family
\begin{equation*}
  \mathcal{L}
  \;=\;
  \bigl\{\,(\mathbf{f}_j,\mathbf{g}_j,y^{*}_j)\,\bigr\}_{j=1}^{J},
  \qquad \mathbf{f}_j,\mathbf{g}_j\in\R^{n},\ \ y^{*}_j\in\R ,
\end{equation*}
read: apply the force pattern $\mathbf{f}_j$ on the free degrees of freedom,
read the displacement through the functional $\mathbf{g}_j$, and want the value
$y^{*}_j$. Its \emph{realised value vector} is
$y(C)=\bigl(y_1(C),\dots,y_J(C)\bigr)\in\R^{J}$ with
\begin{equation}
  \label{eq:taskvalue}
  y_j(C)\;=\;\mathbf{g}_j^\T C\,\mathbf{f}_j\;=\;\langle \mathsf{M}_j,C\rangle ,
  \qquad
  \mathsf{M}_j\;=\;\mathbf{g}_j\mathbf{f}_j^\T\in\R^{n\times n},
\end{equation}
where $\langle A,B\rangle=\operatorname{tr}(A^\T B)$ is the trace inner product
on $\R^{n\times n}$, whose norm is $\|\cdot\|_F$; the second equality is
$\langle \mathbf{g}\mathbf{f}^\T,C\rangle=\operatorname{tr}(\mathbf{f}\mathbf{g}^\T C)
=\mathbf{g}^\T C\mathbf{f}$. The list carries a target $y^{*}\in\R^{J}$ exactly
as the block carried $R^{*}$, and $y(\cdot)$ extends by the same formula to a
linear map $\R^{n\times n}\to\R^{J}$. This is the force-driven protocol of
Sec.~\ref{sec:setup} with the block replaced by a list. The
imposed-displacement drive of Sec.~\ref{sec:dirichlet} is not covered, and the
distinction matters below.

Two subspaces of $\R^{n\times n}$ and, dual to them, three spaces of relations
in $\R^{J}$ carry everything. The first,
$\mathcal{M}=\spn\{\mathsf{M}_1,\dots,\mathsf{M}_J\}$, is what the list sees of
an arbitrary linear response. The second,
$\mathcal{M}^{s}=\operatorname{sym}\mathcal{M}
=\spn\{\operatorname{sym}\mathsf{M}_1,\dots,\operatorname{sym}\mathsf{M}_J\}$
with $\operatorname{sym}A=\tfrac12(A+A^\T)$, is all it sees of a reciprocal
one. The gap between them is the \emph{deficit of the list}, the number of
dimensions of realisable task values that reciprocity removes
(Proposition~\ref{prop:tasklist}(ii)),
\begin{equation}
  \label{eq:deltaL}
  \delta_{\mathcal{L}}
  \;=\;
  \dim\mathcal{M}-\dim\mathcal{M}^{s}
  \;=\;
  \dim\bigl(\mathcal{M}\cap\operatorname{Skew}(n)\bigr),
\end{equation}
$\operatorname{Skew}(n)$ being the antisymmetric $n\times n$ matrices; the
second equality is proved in Proposition~\ref{prop:tasklist}(iii). The
notation matches the block deficit $\delta_p=p(p-1)/2$, the codimension of
Theorem~\ref{thm:main}, which Proposition~\ref{prop:taskblock} shows is the
full-block value of $\delta_{\mathcal{L}}$. On the measurement side the
analogous count gives the completeness and redundancy of a set of measurements
on a grounded electrode array~\citep{Butler2019}; here it counts what a list of
targets forfeits.

Dually, a combination $\lambda\in\R^{J}$ of the tasks probes the matrix
$\mathsf{M}(\lambda)=\sum_j\lambda_j\mathsf{M}_j$, since
$\sum_j\lambda_jy_j(C)=\langle\mathsf{M}(\lambda),C\rangle$. Let
\begin{equation*}
  \mathcal{A}^{0}_{\mathcal{L}}
  =\bigl\{\lambda\in\R^{J}:\mathsf{M}(\lambda)=0\bigr\},
  \qquad
  \mathcal{A}_{\mathcal{L}}
  =\bigl\{\lambda\in\R^{J}:\operatorname{sym}\mathsf{M}(\lambda)=0\bigr\},
  \qquad
  \mathcal{A}^{1}_{\mathcal{L}}
  =\mathcal{A}_{\mathcal{L}}\cap\bigl(\mathcal{A}^{0}_{\mathcal{L}}\bigr)^{\perp},
\end{equation*}
and $\mathcal{S}_{\mathcal{L}}=\mathcal{A}_{\mathcal{L}}^{\perp}\subseteq\R^{J}$.
A vector $\lambda\in\mathcal{A}^{0}_{\mathcal{L}}$ is a linear relation
$\sum_j\lambda_jy_j=0$ among the task values that holds for \emph{any} linear
response whatever, reciprocal or not. It is therefore a redundancy the list
was written with, and for the full block $\mathcal{A}^{0}_{\mathcal{L}}$ is
$\{0\}$. A $\lambda\in\mathcal{A}_{\mathcal{L}}$ is a relation that holds for
every reciprocal response. The difference between the two,
$\mathcal{A}^{1}_{\mathcal{L}}$, is the physics, and it has dimension
$\delta_{\mathcal{L}}$. The subspace $\mathcal{S}_{\mathcal{L}}$ is where
reciprocal responses put the task values, the task-list counterpart of
$\mathcal{S}_P$ in \eqref{eq:SP}.

\begin{proposition}[what a task list can realise]
  \label{prop:tasklist}
  Let $\mathcal{L}$ be any task list. Then:
  \emph{(i)} the values realised by arbitrary, not necessarily reciprocal,
  linear responses fill the subspace $\bigl\{y(X):X\in\R^{n\times n}\bigr\}
  =(\mathcal{A}^{0}_{\mathcal{L}})^{\perp}$, of dimension $\dim\mathcal{M}$;
  \emph{(ii)} the values realised by symmetric responses fill exactly
  $\{y(X):X=X^\T\}=\mathcal{S}_{\mathcal{L}}$, of dimension
  $\dim\mathcal{M}^{s}$ and of codimension
  $\dim\mathcal{A}^{1}_{\mathcal{L}}=\delta_{\mathcal{L}}$ inside
  $(\mathcal{A}^{0}_{\mathcal{L}})^{\perp}$, so that
  $y(C)\in\mathcal{S}_{\mathcal{L}}$ for every symmetric $C$ --- invertible or
  not, positive definite or not, realised by a network or not --- and the
  linear relations forced among the task values are exactly
  $\mathcal{A}_{\mathcal{L}}$, of which $\delta_{\mathcal{L}}$ are due to
  reciprocity rather than to the list;
  \emph{(iii)}
  $\delta_{\mathcal{L}}=\dim(\mathcal{M}\cap\operatorname{Skew}(n))$;
  \emph{(iv)} over symmetric invertible $C$ the values form an open dense
  subset of $\mathcal{S}_{\mathcal{L}}$, and over $C\succ0$ a non-empty open
  convex cone in it; both are in general proper, the second need not be closed,
  and both span $\mathcal{S}_{\mathcal{L}}$, so neither satisfies a linear
  relation that $\mathcal{S}_{\mathcal{L}}$ does not;
  \emph{(v)} writing $y(k)$ for $y\bigl(K_{\!f\!f}(k)^{-1}\bigr)$, at every
  admissible $k$ with $K_{\!f\!f}=K_{\!f\!f}^\T$,
  \begin{equation}
    \label{eq:taskbound}
    \rank\frac{\partial y}{\partial k}
    \;\le\;
    \min\bigl(n_b,\ \dim\mathcal{M}-\delta_{\mathcal{L}}\bigr),
  \end{equation}
  while differentiating in $(k,a)$ without assuming symmetry gives only
  $\min(2n_b,\dim\mathcal{M})$.
\end{proposition}

\begin{proof}
  \emph{(i)} $y(\cdot)$ is linear, and for $\lambda\in\R^{J}$ and
  $X\in\R^{n\times n}$ one has
  $\langle y(X),\lambda\rangle=\sum_j\lambda_j\langle\mathsf{M}_j,X\rangle
  =\langle\mathsf{M}(\lambda),X\rangle$, so the adjoint of $y(\cdot)$ is
  $\lambda\mapsto\mathsf{M}(\lambda)$ and the image is the orthogonal
  complement of its kernel $\mathcal{A}^{0}_{\mathcal{L}}$. Its dimension is
  $\rank\bigl(\lambda\mapsto\mathsf{M}(\lambda)\bigr)=\dim\mathcal{M}$.

  \emph{(ii)} For symmetric $X$ we have
  $\langle\mathsf{M}_j-\operatorname{sym}\mathsf{M}_j,X\rangle=0$, so
  $y_j(X)=\langle\operatorname{sym}\mathsf{M}_j,X\rangle$: a task value on a
  reciprocal response sees $\mathsf{M}_j$ only through its symmetric part. The
  adjoint of the restriction $y|_{\Sym(n)}:\Sym(n)\to\R^{J}$, with $\Sym(n)$
  carrying the trace inner product, is therefore
  $\lambda\mapsto\operatorname{sym}\mathsf{M}(\lambda)$, whose kernel is
  $\mathcal{A}_{\mathcal{L}}$ and whose rank is $\dim\mathcal{M}^{s}$. Hence
  the image is $\mathcal{A}_{\mathcal{L}}^{\perp}=\mathcal{S}_{\mathcal{L}}$,
  of dimension $J-\dim\mathcal{A}_{\mathcal{L}}=\dim\mathcal{M}^{s}$.
  Since $\mathsf{M}(\lambda)=0$ implies $\operatorname{sym}\mathsf{M}(\lambda)=0$
  we have $\mathcal{A}^{0}_{\mathcal{L}}\subseteq\mathcal{A}_{\mathcal{L}}$,
  so $\mathcal{S}_{\mathcal{L}}\subseteq(\mathcal{A}^{0}_{\mathcal{L}})^{\perp}$
  with codimension
  $\dim\mathcal{A}_{\mathcal{L}}-\dim\mathcal{A}^{0}_{\mathcal{L}}
  =(J-\dim\mathcal{M}^{s})-(J-\dim\mathcal{M})=\delta_{\mathcal{L}}$, and
  $\mathcal{A}_{\mathcal{L}}=\mathcal{A}^{0}_{\mathcal{L}}\oplus
  \mathcal{A}^{1}_{\mathcal{L}}$ orthogonally by the definition of
  $\mathcal{A}^{1}_{\mathcal{L}}$, so
  $\dim\mathcal{A}^{1}_{\mathcal{L}}=\delta_{\mathcal{L}}$.

  \emph{(iii)} Apply rank--nullity to
  $\operatorname{sym}|_{\mathcal{M}}:\mathcal{M}\to\Sym(n)$. Its image is
  $\mathcal{M}^{s}$ and its kernel is
  $\{A\in\mathcal{M}:A=-A^\T\}=\mathcal{M}\cap\operatorname{Skew}(n)$, so
  $\dim\mathcal{M}=\dim\mathcal{M}^{s}
  +\dim(\mathcal{M}\cap\operatorname{Skew}(n))$.

  \emph{(iv)} By (ii) the map $y|_{\Sym(n)}$ is linear and onto
  $\mathcal{S}_{\mathcal{L}}$, hence open onto it. The invertible symmetric
  matrices form an open subset of $\Sym(n)$, dense because $t\mapsto\det(X+tI)$
  is a non-zero polynomial for each fixed $X$; their image is therefore open,
  and dense because the preimage of a non-empty relatively open subset of
  $\mathcal{S}_{\mathcal{L}}$ is non-empty and open in $\Sym(n)$, hence meets
  that dense set. Likewise $\Sym^{++}(n)$ is a non-empty open convex cone and
  linear maps carry cones to cones and convex sets to convex sets, so its image
  is a non-empty open convex cone; a non-empty open subset of a subspace spans
  it, which gives the last clause. For properness take the full-block list
  $\mathbf{f}_{(u,v)}=e_v$, $\mathbf{g}_{(u,v)}=e_u$ over all
  $(u,v)\in\{1,\dots,n\}^{2}$, whose $\mathsf{M}_{(u,v)}$ are all $n^{2}$
  matrix units, so that $y(X)$ is the flattening of $X$ and $y|_{\Sym(n)}$ is a
  bijection onto $\mathcal{S}_{\mathcal{L}}$: the invertible symmetric
  responses then give the
  complement of the singular locus, which is proper because the zero matrix is
  symmetric and singular, and the positive definite ones give a copy of
  $\Sym^{++}(n)$, which is proper and is not closed, the zero matrix again
  lying in its closure and not in it. (At $n=1$ this is the single task
  $\mathbf{f}_1=\mathbf{g}_1=e_1$ with $\mathcal{S}_{\mathcal{L}}=\R$,
  $\R\setminus\{0\}$ and $(0,\infty)$; at $n\ge2$ that single task is not a
  witness for the invertible clause, since a symmetric invertible $C$ with $C_{11}=0$ exists and its
  image is all of $\R$.)

  \emph{(v)} Since $\partial K/\partial k_b=\bq_b\bq_b^\T$ and
  $\mathrm{d}C=-C(\mathrm{d}K)C$ with $C=C^\T$, the derivative
  $\partial C/\partial k_b=-\bw_b\bw_b^\T$ is symmetric, $\bw_b=C\bq_b$ being
  the vector of \eqref{eq:wb}; so every Jacobian column
  $\partial y/\partial k_b=y(\partial C/\partial k_b)$ lies in
  $\mathcal{S}_{\mathcal{L}}$ by (ii): the rank is at most
  $\dim\mathcal{S}_{\mathcal{L}}=\dim\mathcal{M}-\delta_{\mathcal{L}}$, and at
  most the number $n_b$ of columns. Without symmetry
  $\partial C/\partial a_b=-(C\bs_b)(C^\T\bq_b)^\T$ need not be symmetric, the
  columns lie only in $(\mathcal{A}^{0}_{\mathcal{L}})^{\perp}$, and there are
  $2n_b$ of them.
\end{proof}

Part (ii) is the containment of Remark~\ref{rem:global} for a list, and it is
again derivative-free, so it again yields a floor. We write it with $\Pi_{E}$,
the Euclidean orthogonal projection of $\R^{J}$ onto a subspace $E$.

\begin{theorem}[error floor for a task list]
  \label{thm:taskfloor}
  Let $\mathcal{L}$ be a task list with target $y^{*}\in\R^{J}$, and let
  $\|\cdot\|$ be the Euclidean norm on $\R^{J}$. For every network whose
  $K_{\!f\!f}$ is symmetric and invertible, passive networks included, whatever
  its size and its topology, and for every admissible stiffness vector $k$,
  \begin{equation}
    \label{eq:taskfloor}
    \bigl\|y(C)-y^{*}\bigr\|
    \;\ge\;
    \operatorname{dist}\bigl(y^{*},\mathcal{S}_{\mathcal{L}}\bigr)
    \;=\;
    \bigl\|\Pi_{\mathcal{A}_{\mathcal{L}}}y^{*}\bigr\| ,
  \end{equation}
  and the right-hand side splits orthogonally into a part the list imposes on
  itself and a part reciprocity imposes,
  \begin{equation}
    \label{eq:tasksplit}
    \bigl\|\Pi_{\mathcal{A}_{\mathcal{L}}}y^{*}\bigr\|^{2}
    \;=\;
    \underbrace{\bigl\|\Pi_{\mathcal{A}^{0}_{\mathcal{L}}}y^{*}\bigr\|^{2}}
      _{\text{list inconsistency}}
    \;+\;
    \underbrace{\bigl\|\Pi_{\mathcal{A}^{1}_{\mathcal{L}}}y^{*}\bigr\|^{2}}
      _{\text{reciprocity}},
    \qquad
    \dim\mathcal{A}^{1}_{\mathcal{L}}=\delta_{\mathcal{L}} .
  \end{equation}
  The first term is the error that no linear response of any kind can remove
  and the second is the additional error reciprocity alone imposes; the floor
  vanishes if and only if $\sum_j\lambda_jy^{*}_j=0$ for every $\lambda$ with
  $\operatorname{sym}\mathsf{M}(\lambda)=0$. Over all symmetric $C$ the bound is
  attained; over symmetric invertible $C$ it is the infimum and is attained for
  every $y^{*}$ outside a closed set with empty interior; over the admissible
  stiffnesses of a fixed graph it is a lower bound and is in general strict.
\end{theorem}

\begin{proof}
  By Proposition~\ref{prop:tasklist}(ii), $y(C)\in\mathcal{S}_{\mathcal{L}}$
  for every symmetric $C$, and $\mathcal{S}_{\mathcal{L}}$ is a linear subspace
  with orthogonal complement $\mathcal{A}_{\mathcal{L}}$. For any
  $w\in\mathcal{S}_{\mathcal{L}}$ the splitting
  $\R^{J}=\mathcal{S}_{\mathcal{L}}\oplus\mathcal{A}_{\mathcal{L}}$ gives
  \begin{equation*}
    \|w-y^{*}\|^{2}
    =\bigl\|w-\Pi_{\mathcal{S}_{\mathcal{L}}}y^{*}\bigr\|^{2}
     +\bigl\|\Pi_{\mathcal{A}_{\mathcal{L}}}y^{*}\bigr\|^{2}
    \;\ge\;\bigl\|\Pi_{\mathcal{A}_{\mathcal{L}}}y^{*}\bigr\|^{2},
  \end{equation*}
  which is \eqref{eq:taskfloor}, with equality at
  $w=\Pi_{\mathcal{S}_{\mathcal{L}}}y^{*}$; such a $w$ is realised by some
  symmetric $X$ because $\mathcal{S}_{\mathcal{L}}$ is by definition the image
  of $\Sym(n)$, which proves attainment in the first case. For
  \eqref{eq:tasksplit},
  $\mathcal{A}_{\mathcal{L}}
  =\mathcal{A}^{0}_{\mathcal{L}}\oplus\mathcal{A}^{1}_{\mathcal{L}}$ is
  orthogonal, so
  $\Pi_{\mathcal{A}_{\mathcal{L}}}
  =\Pi_{\mathcal{A}^{0}_{\mathcal{L}}}+\Pi_{\mathcal{A}^{1}_{\mathcal{L}}}$ with
  orthogonal summands; that
  $\|\Pi_{\mathcal{A}^{0}_{\mathcal{L}}}y^{*}\|$ is the distance from $y^{*}$
  to the values of arbitrary linear responses is the same computation applied
  to Proposition~\ref{prop:tasklist}(i), and $\dim\mathcal{A}^{1}_{\mathcal{L}}
  =\delta_{\mathcal{L}}$ is part (ii). The vanishing criterion restates
  $y^{*}\in\mathcal{S}_{\mathcal{L}}=\mathcal{A}_{\mathcal{L}}^{\perp}$. In the
  invertible case the infimum of a continuous function over a dense subset
  equals the infimum over the set, and the values realised by symmetric
  invertible $C$ are dense in $\mathcal{S}_{\mathcal{L}}$ by
  Proposition~\ref{prop:tasklist}(iv); attainment holds exactly when
  $\Pi_{\mathcal{S}_{\mathcal{L}}}y^{*}$ is one of those values, so the
  exceptional targets form the preimage under
  $\Pi_{\mathcal{S}_{\mathcal{L}}}$ of a closed nowhere dense set. On a fixed
  graph the realised values form a subset of $\mathcal{S}_{\mathcal{L}}$ that is
  in general proper, so \eqref{eq:taskfloor} remains a lower bound and need not
  be tight.
\end{proof}

Theorems~\ref{thm:main} and~\ref{thm:floor} are the case of a list that asks
for everything.

\begin{proposition}[the block law is the full-block case]
  \label{prop:taskblock}
  Take the task list indexed by $(u,v)\in T\times S$ with
  $\mathbf{f}_{(u,v)}=e_{v}$ and
  $\mathbf{g}_{(u,v)}=e_{u}$, so that
  $\mathsf{M}_{(u,v)}=e_{u}e_{v}^\T=E_{uv}$, the matrix
  unit of $\R^{n\times n}$ at that position, and $J=m_Tm_S$; here $u$ and $v$
  are running degrees of freedom, while $\tau$ and $\sigma$ keep the meaning
  they have from Sec.~\ref{sec:bound} onwards, the injections placing $P$
  inside $T$ and inside $S$. Then
  $\dim\mathcal{M}=m_Tm_S$, $\mathcal{A}^{0}_{\mathcal{L}}=\{0\}$,
  $\dim\mathcal{M}^{s}=m_Tm_S-\tfrac12p(p-1)$ and
  $\delta_{\mathcal{L}}=\tfrac12p(p-1)=\dim\mathcal{A}^{1}_{\mathcal{L}}$;
  and under the isometry $\R^{J}\cong\R^{m_T\times m_S}$,
  $y\leftrightarrow R$, one has $\mathcal{S}_{\mathcal{L}}=\mathcal{S}_P$,
  $\mathcal{A}_{\mathcal{L}}=\mathcal{A}_P$ and
  $\|\Pi_{\mathcal{A}_{\mathcal{L}}}y^{*}\|=\tfrac12\|B-B^\T\|_F$, so that
  \eqref{eq:taskbound} and \eqref{eq:taskfloor} are \eqref{eq:bound} and
  \eqref{eq:floor}.
\end{proposition}

\begin{proof}
  The index pairs $(u,v)$ are distinct, so the $E_{uv}$ are pairwise
  distinct matrix units, hence independent: $\dim\mathcal{M}=J$ and
  $\mathcal{A}^{0}_{\mathcal{L}}=\{0\}$. Next,
  $\operatorname{sym}E_{uv}$ depends only on the unordered pair
  $\{u,v\}$, and distinct pairs give matrices on disjoint sets of
  entries, hence independent, so $\dim\mathcal{M}^{s}$ counts the distinct
  pairs arising. Two distinct tasks give the same pair only if they are
  $(u,v)$ and $(v,u)$ with $u\neq v$, and both of those
  are tasks exactly when $v\in T$ and $u\in S$, that is when
  $u,v\in P$; so the fibres of
  $(u,v)\mapsto\{u,v\}$ have two elements over the $\binom p2$
  pairs inside $P$ and one elsewhere, giving
  $\dim\mathcal{M}^{s}=m_Tm_S-\binom p2$ and $\delta_{\mathcal{L}}=\binom p2$.
  A $\lambda$ lies in $\mathcal{A}_{\mathcal{L}}$ exactly when its sum over
  each fibre vanishes, which kills it on singleton fibres and forces
  $\lambda_{(u,v)}+\lambda_{(v,u)}=0$ on the others, so the
  vectors $e_{(u,v)}-e_{(v,u)}$ for $u\neq v$ in $P$,
  having disjoint supports, are an orthogonal basis of
  $\mathcal{A}_{\mathcal{L}}=\mathcal{A}^{1}_{\mathcal{L}}$. Flattening
  $\R^{m_T\times m_S}\to\R^{J}$ entry by entry carries the Frobenius norm to
  the Euclidean one and the conditions \eqref{eq:SP} defining $\mathcal{S}_P$
  to orthogonality against those same vectors, so
  $\mathcal{S}_P\leftrightarrow\mathcal{S}_{\mathcal{L}}$ and
  $\mathcal{A}_P\leftrightarrow\mathcal{A}_{\mathcal{L}}$; projecting on that
  basis, each member of which has squared norm $2$, gives
  $\|\Pi_{\mathcal{A}_{\mathcal{L}}}y^{*}\|^{2}
  =\sum_{\{r,c\}\subseteq P,\,r\neq c}(B_{rc}-B_{cr})^{2}/2
  =\tfrac14\|B-B^\T\|_F^{2}$ with $B=R^{*}[\tau(P),\sigma(P)]$, as in
  Theorem~\ref{thm:floor}.
\end{proof}

What replaces the overlap count is a count of reversed pairs of tasks, each
driving where the other reads. To count them on a list of coordinate tasks,
each driving one degree of freedom $\sigma_j$ and reading one, $\tau_j$, call
its \emph{drive-to-read digraph} $\mathcal{D}_{\mathcal{L}}$ the digraph on
$\{1,\dots,n\}$ carrying an arc $\sigma_j\to\tau_j$ for each task, repeated
tasks contributing the same arc. A reversed pair is then a two-cycle, or
digon, of $\mathcal{D}_{\mathcal{L}}$.

\begin{corollary}[coordinate tasks: the deficit counts digons]
  \label{cor:digon}
  Suppose every task drives one degree of freedom and reads one,
  $\mathbf{f}_j=e_{\sigma_j}$ and $\mathbf{g}_j=e_{\tau_j}$. Then
  \begin{equation}
    \label{eq:digon}
    \delta_{\mathcal{L}}
    \;=\;
    \#\bigl\{\,\{u,v\}:u\neq v,\ \text{both } u\to v
    \text{ and } v\to u \text{ lie in }\mathcal{D}_{\mathcal{L}}\,\bigr\},
  \end{equation}
  the number of two-cycles, or digons, of $\mathcal{D}_{\mathcal{L}}$. A
  self-loop $u\to u$, that is a driving-point task, costs nothing, and
  $\delta_{\mathcal{L}}=0$ exactly when no pair of degrees of freedom is
  instrumented in both directions.
\end{corollary}

\begin{proof}
  Here $\mathcal{M}$ is the span of the matrix units at the position set
  $\mathcal{I}=\{(\tau_j,\sigma_j)\}$. Let
  $X\in\mathcal{M}\cap\operatorname{Skew}(n)$. Lying in $\mathcal{M}$ means
  $X_{uv}=0$ whenever $(u,v)\notin\mathcal{I}$, and antisymmetry means
  $X_{uu}=0$ and $X_{uv}=-X_{vu}$; so if $(u,v)\in\mathcal{I}$ but
  $(v,u)\notin\mathcal{I}$ then $X_{vu}=0$ and hence $X_{uv}=0$. Thus $X$ is
  supported on the off-diagonal positions $(u,v)$ with both $(u,v)$ and $(v,u)$
  in $\mathcal{I}$, and on those it is free subject to $X_{uv}=-X_{vu}$, so the
  matrices $E_{uv}-E_{vu}$, one for each such unordered pair, are a basis of
  $\mathcal{M}\cap\operatorname{Skew}(n)$. Since $(u,v)\in\mathcal{I}$ says the
  arc $v\to u$ is present, the condition on a pair says both arcs are, which is
  a digon; and a self-loop contributes only $(u,u)$, which the diagonal
  condition $X_{uu}=0$ removes. Now apply \eqref{eq:deltaL}.
\end{proof}

A digon is the normal--reciprocal pair of resistivity
surveying~\citep{LaBrecque1996} and the reciprocity check of modal testing,
where the two readings are measured and compared; the corollary adds that a
digon on which targets are imposed is the only thing a coordinate list pays
for.

\begin{corollary}[the generalised overlap]
  \label{cor:taskoverlap}
  Let $\mathcal{F}=\spn\{\mathbf{f}_j\}$ and $\mathcal{G}=\spn\{\mathbf{g}_j\}$
  and $q=\dim(\mathcal{F}\cap\mathcal{G})$. Then
  \begin{equation}
    \label{eq:qbound}
    \delta_{\mathcal{L}}\;\le\;\tfrac12\,q(q-1),
  \end{equation}
  with equality if and only if every antisymmetric matrix with image inside
  $\mathcal{F}\cap\mathcal{G}$ belongs to $\mathcal{M}$; equality holds for the
  full block of Proposition~\ref{prop:taskblock}, where $q=p$. Moreover, if
  every $\mathbf{f}_j$ is supported in a set $S$ of degrees of freedom and every
  $\mathbf{g}_j$ in a set $T$, then $q\le p=|T\cap S|$ and hence
  $\delta_{\mathcal{L}}\le\tfrac12p(p-1)$: the block law never underestimates
  the deficit of a list drawn on the same terminals. In particular
  $\delta_{\mathcal{L}}=0$ when $q\le1$, hence for a
  single task and for every layout that drives one set of degrees of freedom
  and reads a disjoint one; and also, whatever $q$, for a list of self-adjoint
  tasks $\mathbf{g}_j\in\R\mathbf{f}_j$.
\end{corollary}

\begin{proof}
  Let $A=\mathsf{M}(\lambda)\in\mathcal{M}\cap\operatorname{Skew}(n)$. Then
  $\operatorname{im}A\subseteq\mathcal{G}$ and
  $\operatorname{im}A^\T\subseteq\mathcal{F}$, while antisymmetry gives
  $\operatorname{im}A=\operatorname{im}A^\T$, so
  $\operatorname{im}A\subseteq\mathcal{F}\cap\mathcal{G}=:\mathcal{V}$. Since $A$ is
  antisymmetric, $\ker A=(\operatorname{im}A)^{\perp}\supseteq \mathcal{V}^{\perp}$, so
  $A$ vanishes on $\mathcal{V}^{\perp}$ and restricts to an antisymmetric operator on
  $\mathcal{V}$; the space of such $A$ has dimension $\binom q2$, and \eqref{eq:qbound}
  and its equality case follow from \eqref{eq:deltaL}. For the block,
  $\mathcal{F}=\spn\{e_\sigma\}_{\sigma\in S}$ and
  $\mathcal{G}=\spn\{e_\tau\}_{\tau\in T}$ give $\mathcal{V}=\spn\{e_\mu\}_{\mu\in P}$
  and $q=p$, and every $E_{\mu\nu}$ with $\mu,\nu\in P$ is one of the
  $\mathsf{M}_j$, so the $E_{\mu\nu}-E_{\nu\mu}$ spanning the antisymmetric
  operators on $\mathcal{V}$ lie in $\mathcal{M}$ and equality holds, consistently with
  $\delta_{\mathcal{L}}=p(p-1)/2$ there. For the support statement, write
  $\R^{S}$ for the coordinate subspace spanned by $\{e_\sigma\}_{\sigma\in S}$;
  then $\mathcal{F}\subseteq\R^{S}$ and $\mathcal{G}\subseteq\R^{T}$, so
  $\mathcal{V}\subseteq\R^{S}\cap\R^{T}=\R^{P}$ and $q\le p$. Finally
  $q\le1$ makes $\mathcal{V}$ carry no non-zero antisymmetric operator, which covers a
  single task, where $\mathcal{V}$ is at most a line, and a disjoint layout, where $\mathcal{V}=0$;
  and if every $\mathsf{M}_j$
  is symmetric then $\mathcal{M}\subseteq\Sym(n)$ meets
  $\operatorname{Skew}(n)$ only in $0$.
\end{proof}

Experiment~21 checks all of this numerically, over the $7245$ task lists of
its first three parts and the $1637$ floor instances of its fourth, the
full-block values in exact rational arithmetic; Appendix~\ref{app:stats} lists
what each part tests. Sec.~\ref{sec:consequences} applies this
section to a single network: one reversed task with an unequal demand, added to a
three-task cycle, creates the floor of Theorem~\ref{thm:taskfloor}, and one
odd bond removes it.

Three consequences follow.

The first is that $p$ is the invariant of a block and not of a list, and
misreports in both directions. Take two tasks: drive at $a$ and read at $b$,
then drive at $b$ and read at $a$. Each, taken alone as a one-task block, has
$T\cap S=\emptyset$,
so a per-task reading of Theorem~\ref{thm:main} charges nothing. Yet the two
are a reversed pair, $\delta_{\mathcal{L}}=1$, and the forced relation $y_1=y_2$ is
Maxwell--Betti read on the pair $(a,b)$. The enclosing block $T=S=\{a,b\}$ does
see it, at $p(p-1)/2=1$, but charges that lost dimension against four values
where two were imposed. The forfeited \emph{fraction} of what was asked for is
therefore one half and not one quarter (both lists are rows of experiment~21,
part~2). In the other direction, the driving-point pair
$\mathbf{f}_1=\mathbf{g}_1=e_a$, $\mathbf{f}_2=\mathbf{g}_2=e_b$ has both
$\mathsf{M}_j$ symmetric and $\delta_{\mathcal{L}}=0$, while the enclosing
block has $p=2$ and charges one dimension, a relation between two values the
list never asked for.

What survives is the one-sided bound of Corollary~\ref{cor:taskoverlap}, and it
strengthens Sec.~\ref{sec:consequences} under a weaker hypothesis. The nineteen
layouts of Table~\ref{tab:layouts} drive one set of degrees of freedom and read
a disjoint one, so $q=0$ and $\delta_{\mathcal{L}}=0$ at each. This holds not
merely for the block each would generate but for whatever sub-list its authors
imposed, and whether the patterns are coordinate vectors, ghost-bond dipoles or
principal components. It is a theorem for the three force-driven rows of the
table; for the sixteen driven by an imposed displacement it is a count on the
force-driven analogue list, not an instance of the theorem. Experiment~21, part~5, evaluates it layout by layout and finds
$q=0$ and $\delta_{\mathcal{L}}=0$ in all $323$ rows the nineteen generate: the
base list, eight coordinate sub-lists drawn at random (a single-task layout
repeats its one task) and eight dense-pattern variants of each.

The second consequence is a number for the one published task whose driven and
read-out sets are exchanged. What is charged must be the list imposed and not
the block that encloses it. Fig.~2 of \citet{Du2026}, read as the deposited script implements
it and as Sec.~\ref{sec:consequences} states it, clamps units $1,2,3$
\emph{simultaneously} and asks for the three angles of units $4,5,6$, then
clamps $4,5,6$ simultaneously and asks for the three angles of $1,2,3$. That is
six scalar demands, not eighteen. Its force-driven analogue is therefore the
list of $J=6$ tasks with the two combined drives
$\mathbf{f}_{+}=e_1+e_2+e_3$ and $\mathbf{f}_{-}=e_4+e_5+e_6$: the three tasks
$(\mathbf{f}_{+},e_4),(\mathbf{f}_{+},e_5),(\mathbf{f}_{+},e_6)$ and the three
tasks $(\mathbf{f}_{-},e_1),(\mathbf{f}_{-},e_2),(\mathbf{f}_{-},e_3)$. It has
$\mathcal{F}=\spn\{\mathbf{f}_{+},\mathbf{f}_{-}\}$ and $q=2$, so
Corollary~\ref{cor:taskoverlap} caps its deficit at one. The cap is attained,
because $\mathsf{M}$ summed over the first three tasks is
$\mathbf{f}_{-}\mathbf{f}_{+}^{\T}$ and over the last three is
$\mathbf{f}_{+}\mathbf{f}_{-}^{\T}$, whose difference is a non-zero
antisymmetric member of $\mathcal{M}$. So $\dim\mathcal{M}=6$,
$\dim\mathcal{M}^{s}=5$ and $\delta_{\mathcal{L}}=1$ (experiment~21, part~5,
which computes both dimensions and recovers the relation vector). The single
forced relation is that the three read angles of the forward clamp sum to the
same number as the three of the reverse clamp. It is Maxwell--Betti between the
two combined drives, which is exactly the exchange the two targets ask to break.

The enclosing pair of full blocks, in which each of the three clamped units is
driven separately, is a different and larger list: $J=18$ coordinate tasks
whose digraph is the complete bipartite graph on $\{1,2,3\}$ and $\{4,5,6\}$
in both directions. All nine edges are reversed pairs; $\dim\mathcal{M}=18$,
$\dim\mathcal{M}^{s}=9$ and $\delta_{\mathcal{L}}=9$ (experiment~21, part~5,
which deposits both readings as separate rows). Nine is the deficit of that
enclosure and not of the protocol; the twelve extra demands were never
imposed. Both counts are of a force-driven analogue: the drive of Fig.~2 is an
imposed angle, which neither Theorem~\ref{thm:main} nor
Theorem~\ref{thm:taskfloor} covers, and $\delta_{\mathcal{L}}=1$ says only
what a force-driven list of the same shape would forfeit. The obstruction for
the imposed angle is the spectral condition \eqref{eq:mcone} of
Theorem~\ref{thm:multiclamp} (Sec.~\ref{sec:dirichlet}), with the identity
\eqref{eq:duS} and the floor \eqref{eq:dufloor} of
Sec.~\ref{sec:consequences}, tested against the deposited trajectory in
experiment~17.

The third consequence is that a positive deficit is a property of a list and
not a charge against it. The cleanest illustration is a method that profits
from one. The in situ backpropagation of \citet{Li2024} performs two solves on
the same network: a forward one driven by $\mathbf{f}_{\mathrm{in}}$, in which
every bond elongation $\bq_b^\T\bu$ is read, and an adjoint one driven by
$\mathbf{f}_{\mathrm{out}}$, in which the same elongations are read. That list
has $\mathcal{G}=\spn\{\bq_b\}=\R^{n}$, since $K_{\!f\!f}$ is invertible, and
$\mathcal{F}=\spn\{\mathbf{f}_{\mathrm{in}},\mathbf{f}_{\mathrm{out}}\}$, so
$q=2$ and \eqref{eq:qbound} caps the deficit at one. The cap is attained,
because $\mathcal{M}=\R^{n}\mathbf{f}_{\mathrm{in}}^\T
+\R^{n}\mathbf{f}_{\mathrm{out}}^\T$ contains the non-zero antisymmetric matrix
$\mathbf{f}_{\mathrm{out}}\mathbf{f}_{\mathrm{in}}^\T
-\mathbf{f}_{\mathrm{in}}\mathbf{f}_{\mathrm{out}}^\T$ whenever the two drives
are independent. So $\delta_{\mathcal{L}}=1$ exactly, and the one relation is
Maxwell--Betti between the drives. Experiment~21, part~6, finds
$\delta_{\mathcal{L}}=1$ in all $460$ instances it builds, $400$ from random
compatibility matrices and $60$ on Delaunay networks whose read functionals
are the actual elongation covectors $\bq_b$. On those $60$ it recovers
$\mathbf{f}_{\mathrm{out}}$ and $-\mathbf{f}_{\mathrm{in}}$ from the relation
vector, to $1.5\times10^{-14}$.

No target is prescribed anywhere on that list. The elongation reads of both
solves are instrumentation, from which the gradient is assembled, and not
demands. There is therefore no $y^{*}$, and \eqref{eq:taskfloor} is not
engaged at all: the method uses the relation to obtain an exact gradient from a
second measurement on the same hardware rather than paying for it. Its deficit
would be a cost only against a list on which targets were imposed. The targets
of \citet{Li2024} are imposed on a different list, the disjoint drive-and-read
layout of Table~\ref{tab:layouts}, where $q=0$ and $\delta_{\mathcal{L}}=0$.

The forecast of Sec.~\ref{sec:consequences} for platforms that share a
terminal sharpens accordingly. A shared terminal alone is a self-loop of
$\mathcal{D}_{\mathcal{L}}$, a driving-point task whose
$\mathsf{M}_j=\mathbf{f}_j\mathbf{f}_j^\T$ is already symmetric, and by
Corollary~\ref{cor:digon} it is free, however many there are; the cost begins
only with a reversed pair of tasks, which the whole block on the accessed
degrees of freedom supplies and a list of tasks need not. In the two published
demonstrations with a shared terminal, both outside Table~\ref{tab:layouts}
(the same-node demonstration of \citet{Li2024} and the locomotion
demonstration of \citet{Du2026}), the shared degree of freedom enters one
drive and one read-out and no pair of terminals is instrumented in both
directions, so $q\le1$ and $\delta_{\mathcal{L}}=0$ (experiment~21,
part~5; Sec.~\ref{sec:consequences} gives the reading of each, including the
assumption made for \citet{Li2024}). Nothing in Sec.~\ref{sec:consequences} that concerns
a full block is affected: Proposition~\ref{prop:budget} is a statement about
$\mathcal{S}_P$ and stands as proved.

\paragraph{What the list version does not deliver.} Passivity still adds an
orthogonal second term. Since $\{y(C):C\succ0\}$ lies in
$\mathcal{S}_{\mathcal{L}}$, the splitting of Theorem~\ref{thm:taskfloor}
gives $\|y(C)-y^{*}\|^{2}\ge\|\Pi_{\mathcal{A}_{\mathcal{L}}}y^{*}\|^{2}
+\operatorname{dist}\bigl(\Pi_{\mathcal{S}_{\mathcal{L}}}y^{*},
\overline{y(\Sym^{++}(n))}\bigr)^{2}$. The right-hand side is the value of
the convex program $\inf\{\|y(X)-y^{*}\|^{2}:X\succeq0\}$, since $\Sym^{++}(n)$
is dense in $\Sym^{+}(n)$ and $y$ is continuous. For a task list the eigenvalue
formula of \eqref{eq:floorpd} does not carry over, and the passivity term is
computed from that convex program instead. Take $n=2$ and the three
self-adjoint tasks $\mathbf{f}_j=\mathbf{g}_j=e_1,e_2,e_1+e_2$, so that
$y=(C_{11},C_{22},C_{11}+2C_{12}+C_{22})$ and $\delta_{\mathcal{L}}=0$. There
is no reciprocity relation at all, and the whole floor is the passivity term.
For $y^{*}=(1,1,-3)$ the exact floor is $3$. Semidefiniteness gives
$C_{12}\ge-\sqrt{C_{11}C_{22}}$, hence
$y_3=y_1+y_2+2C_{12}\ge(\sqrt{y_1}-\sqrt{y_2})^{2}\ge0$ and
$\|y-y^{*}\|\ge y_3+3\ge3$, with equality at $y=(1,1,0)$. That is the value of
the convex program over $C\succeq0$, and it is attained there, at
$C=\bigl(\begin{smallmatrix}\ \ 1&-1\\-1&\ \ 1\end{smallmatrix}\bigr)$, which is
positive semidefinite with eigenvalues $0$ and $2$ and hence singular. It is
attained by no $C\succ0$ and therefore by no passive network: definiteness
makes $C_{12}>-\sqrt{C_{11}C_{22}}$ strict, so $y_3>0$ and $\|y-y^{*}\|>3$ at
every one of them. The identity is not the minimiser and is not close to it:
$y(I)=(1,1,2)$, an error of $5$.

Reading a symmetric $B$ off the target in the obvious way, $B_{11}=y^{*}_1$,
$B_{22}=y^{*}_2$ and $B_{12}=B_{21}=(y^{*}_3-y^{*}_1-y^{*}_2)/2=-5/2$, and
applying the eigenvalue truncation of \citet{Higham1988} to it returns $3/2$,
wrong by a factor of two. The reason is structural. On a block, the shared
part $C\mapsto C[P,P]$ of $y$ is an isometry onto a copy of $\Sym(p)$ carrying the
Frobenius metric, so the eigenvalue projection is taken in the metric that
measures the error, whereas $y|_{\Sym(n)}$ is not in general a similarity. Nothing is
lost on a block, where a Schur-complement construction realises every block
with a positive definite shared sub-block by some $C\succ0$, so
\eqref{eq:floorpd} is there the exact distance to the closure of the
achievable set. For a list, that closure is needed, because the cone is not
always closed: for the tasks $(e_1,e_1)$ and $(e_2,e_1)$ it is the open
half-space $\{y_1>0\}$, and at $y^{*}=(-1,5)$ the infimum $1$ is not attained.
Attainment on a fixed graph is likewise left open. Whether a graph reaches
\eqref{eq:taskbound} is the transversality question of
Theorem~\ref{thm:duality}, which for a list reads
$\dim\mathcal{M}^{s}-\rank\partial y/\partial k
=\dim(\mathcal{M}^{s}\cap K_{\!f\!f}V^{\perp}K_{\!f\!f})$ by the same proof, with
$V^{\perp}$ as defined before that theorem.

Two cautions remain. Tasks may carry different units, so the Euclidean norm on
$\R^{J}$ is a choice: $\mathcal{S}_{\mathcal{L}}$, $\delta_{\mathcal{L}}$ and
the vanishing of the floor do not depend on it, but its value does. And of the
two terms of \eqref{eq:tasksplit} only the second is physics; in a passive
network the passivity term above is physics too. The first is produced by
repeated tasks, null tasks and tasks that combine others, so a reciprocity cost
must never be read off the total, on pain of inflating it by padding the list.

A list, like the block of Theorem~\ref{thm:main}, is read on a network driven
by applied forces. Most hardware instead holds its inputs at prescribed
displacements; the next subsection gives the law of that drive, promised at
the end of Sec.~\ref{sec:universal}.

\subsection{The imposed-displacement drive obeys a different law}
\label{sec:dirichlet}

Under an imposed-displacement drive only the input is imposed, and the output
is read wherever the network puts it. Fix a set $P$ of $p$ degrees of freedom, each driven in one test and read in
the others. Impose a unit
displacement at one $\mu\in P$, leave every other degree of freedom, the rest
of $P$ included, free to equilibrate, and read the displacement at $\nu\in P$.
Call that number $D_{\nu\mu}$. This is the protocol of \citet{Altman2024} and
of the analogue networks of~\citet{Dillavou2022,Dillavou2024}, and, with
several units clamped at once (Theorem~\ref{thm:multiclamp}), of
\citet{Du2026}. It is also the drive of sixteen of the nineteen layouts of
Table~\ref{tab:layouts}.

The matrix $D$ of these transmissions is not a block of $C$, so
Theorem~\ref{thm:main} does not constrain it. Reciprocity still does, but not
by a codimension count. Theorem~\ref{thm:dirichlet} writes each transmission
as a ratio of two entries of $C$; what symmetry and passivity then impose, and
how that differs from the force-driven law, is set out after its proof.

\begin{theorem}[reciprocity under an imposed-displacement drive]
  \label{thm:dirichlet}
  Let $K_{\!f\!f}$ be invertible, $C=K_{\!f\!f}^{-1}$, and $C_{\mu\mu}\ne0$ for
  every $\mu\in P$. Then for every $\mu\ne\nu$,
  \begin{equation}
    \label{eq:dtrans}
    D_{\nu\mu} \;=\; \frac{C_{\nu\mu}}{C_{\mu\mu}} .
  \end{equation}
  Whenever $C=C^\T$, and for all $\mu,\nu,\rho\in P$,
  \begin{equation}
    \label{eq:dbalance}
    D_{\nu\mu}\,C_{\mu\mu} \;=\; D_{\mu\nu}\,C_{\nu\nu} \;=\; C_{\mu\nu} ,
    \qquad
    D_{\nu\mu}D_{\rho\nu}D_{\mu\rho} \;=\; D_{\mu\nu}D_{\nu\rho}D_{\rho\mu} ,
  \end{equation}
  both sides of the cycle identity being
  $C_{\nu\mu}C_{\rho\nu}C_{\mu\rho}/(C_{\mu\mu}C_{\nu\nu}C_{\rho\rho})$. If in
  addition the network is passive, so that $C\succ0$, then the diagonal entries
  are positive and, with the inequality for $\mu\ne\nu$,
  \begin{equation}
    \label{eq:dcycle}
    \operatorname{sign} D_{\nu\mu} = \operatorname{sign} D_{\mu\nu},
    \qquad
    D_{\nu\mu}D_{\mu\nu} < 1 .
  \end{equation}
    The diagonal of $D$ is identically one, so at $p=1$ there is nothing to
  train. For $p\ge2$ and \emph{whenever $C=C^\T$},
  \begin{equation}
    \label{eq:dbound}
    \rank\frac{\partial D}{\partial k}
    \;\le\;
    \min\Bigl(n_b,\ \tfrac12 (p-1)(p+2)\Bigr) ,
  \end{equation}
  which is $p(p-1)-\tfrac12(p-1)(p-2)$: a rank ceiling $(p-1)(p-2)/2$ below the
  $p(p-1)$ off-diagonal entries, hence a deficit of at least that many. The
  claim is a bound on rank, hence on the
  local dimension of the reachable set, and not a containment: unlike
  Theorem~\ref{thm:main}, the proof does not exhibit a linear subspace of that
  codimension holding the reachable set, and we do not assert one.
\end{theorem}

\begin{proof}
  Partition the free degrees of freedom into $\{\mu\}$ and its complement $F$.
  By the cofactor formula for the inverse, $\det K_{FF}=C_{\mu\mu}\det K_{\!f\!f}$,
  so the hypothesis $C_{\mu\mu}\ne0$ is exactly the invertibility of $K_{FF}$.
  Imposing $u_\mu=1$ with $K_{\!f\!f}u=0$ on $F$ then has the unique solution
  $u_F=-K_{FF}^{-1}K_{F\mu}$,
  and the bordered-inverse identity
  $C_{\nu\mu}/C_{\mu\mu}=-(K_{FF}^{-1}K_{F\mu})_\nu$ is
  \eqref{eq:dtrans}. No symmetry is used, so \eqref{eq:dtrans} holds for an odd
  network too; the hypothesis $C_{\mu\mu}\ne0$ is what makes the clamped problem
  well posed and both sides defined, and
  it is automatic once $C\succ0$.
  If $C=C^\T$ then $C_{\nu\mu}=C_{\mu\nu}$, and multiplying \eqref{eq:dtrans} by
  $C_{\mu\mu}$ gives the first equality of \eqref{eq:dbalance}. Multiplying it
  around a triple leaves both sides of the second equal to
  $C_{\nu\mu}C_{\rho\nu}C_{\mu\rho}/(C_{\mu\mu}C_{\nu\nu}C_{\rho\rho})$, the three
  diagonal factors cancelling whatever their signs, so the cycle identity needs
  symmetry alone. Positive definiteness is what \eqref{eq:dcycle} adds: the
  diagonal entries are then positive, so $D_{\nu\mu}$ and $D_{\mu\nu}$ carry the
  sign of the single number $C_{\mu\nu}$, and their product is
  $C_{\mu\nu}^2/(C_{\mu\mu}C_{\nu\nu})$, the squared correlation of the pair,
  which is less than one by the strict $2\times2$ principal minor. Only the
  positivity of the $2\times2$ principal minors of $C_P=C[P,P]$ is used, so a
  negative-definite $C$ satisfies both clauses too; each can fail on a suitable
  symmetric indefinite $C$, though never both at the same pair, since a sign
  disagreement forces $C_{\mu\mu}C_{\nu\nu}<0$ and hence
  $D_{\nu\mu}D_{\mu\nu}<0<1$.
    For \eqref{eq:dbound}, note that $D$ depends on $k$ only through the principal
  submatrix $C_P=C[P,P]$, by way of $\Phi:C_P\mapsto C_P\,\mathrm{diag}(C_P)^{-1}$, defined on the
  $p\times p$ matrices with invertible diagonal and restricted to the
  $p(p+1)/2$-dimensional space of symmetric ones. $\Phi$ is invariant under $C_P\mapsto tC_P$, so
  differentiating in $t$ at $t=1$ gives $\mathrm{d}\Phi_{C_P}(C_P)=0$ with
  $C_P\ne0$; hence $\rank\mathrm{d}\Phi\le p(p+1)/2-1=(p-1)(p+2)/2$ at every such
  point, and $\rank\partial D/\partial k\le\rank\mathrm{d}\Phi$. Only symmetry of
  $C_P$ is used: without it the domain has dimension $p^2$, the same step gives only
  $\rank\mathrm{d}\Phi\le p^2-1$, and the ceiling falls back to the vacuous
  $p(p-1)$, the full count of off-diagonal entries. If moreover $a\equiv0$, the parameter vector itself
  lies in the kernel and the first argument of the minimum sharpens to
  $n_b-1$ (Proposition~\ref{prop:multirank}, \eqref{eq:multirank2}).
\end{proof}

The first identity of \eqref{eq:dbalance} is the form Maxwell--Betti reciprocity takes under
this drive: the two transmissions of a pair agree once each is multiplied by
the driving-point compliance at its own input.

Only the last clause, the rank ceiling \eqref{eq:dbound}, is ours, together
with the reading of any of this as a constraint on what a physically learning
network can be trained to do; we have not found either stated anywhere. Identity \eqref{eq:dtrans} is classical several
times over: it is the influence-line construction of
M\"uller-Breslau~\citep{MuellerBreslau1886}, the circuit relation that an
open-circuit voltage transfer ratio is a transimpedance over a driving-point
impedance, and, for a Gaussian field with precision matrix $K$, the simple
regression coefficient $\mathbb{E}[u_\nu\mid u_\mu]/u_\mu$. Equation
\eqref{eq:dbalance} implies that the two regression coefficients of a
pair multiply to the squared correlation. The sign condition of
\eqref{eq:dcycle} together with the cycle identity of \eqref{eq:dbalance} is
exactly the criterion of \citet{EngelSchneider1973} for a matrix to be
diagonally similar to a symmetric one; see also \citet{EngelSchneider1980} for
the minor-based form and \citet{FominZelevinsky2003}, Lemma 7.4, for the skew
analogue. In the Markov-chain setting the same cycle condition is Kolmogorov's
criterion for reversibility~\citep{Kolmogorov1936}.

The Engel--Schneider criterion is stated for all cycles. Triples suffice whenever every
off-diagonal entry of $D$ is non-zero, since a diagonal similarity can then be
assembled from the cycles through one fixed terminal. That is the generic case,
$C_{\nu\mu}\ne0$ for every pair in $P$ failing only on a lower-dimensional set
of stiffnesses unless it fails identically. It is not guaranteed, and at a
configuration where some $C_{\nu\mu}$ vanishes the criterion must be checked on
every cycle of the support of $D$. Our own proof of the cycle identity does not
pass through the criterion and holds for every triple regardless. In the
passive case the similarity is explicit,
$D=\mathrm{diag}(C_{\mu\mu})^{1/2}\,\Omega\,\mathrm{diag}(C_{\mu\mu})^{-1/2}$
with $\Omega$ the correlation matrix of $C_P$.

The conditions \eqref{eq:dbalance} and \eqref{eq:dcycle} are necessary but not
sufficient. Reachability requires $D\,\mathrm{diag}(w)$ to be symmetric
\emph{and positive definite} for some positive $w$, and \eqref{eq:dbound}
bounds the dimension of that set without giving a complete semialgebraic
description of it.

Set against Theorem~\ref{thm:main}, the two laws differ in three ways. First,
the deficit bound is smaller and of a different kind. Driving with forces
confines the whole set reachable by varying $k$ to a linear subspace of codimension exactly
$p(p-1)/2$. Imposing displacements leaves a deficit of \emph{at least}
$(p-1)(p-2)/2$ below the $p(p-1)$ off-diagonal entries, and only as a bound on
$\rank\partial D/\partial k$, that is, as dimension the image lacks. At $p=2$
the first removes one dimension, while the rank ceiling of the second is the
full $p(p-1)=2$: with two shared degrees of freedom the bound sees nothing, and
reciprocity acts there only through \eqref{eq:dbalance} and \eqref{eq:dcycle}.

Second, what remains is a set of inequalities inherited from the
positive-definite cone, not a subspace, and more than one condition. The sign law and
$D_{\nu\mu}D_{\mu\nu}<1$ each exclude an open set of targets while removing no
dimension. That is the same shape of statement as the passivity term of
Theorem~\ref{thm:floorpd}, and no count of parameters against constraints
registers it. The two are independent: $D_{\nu\mu}=D_{\mu\nu}=2$ passes the
first and fails the second, and $D_{\nu\mu}=1$, $D_{\mu\nu}=-\tfrac12$ fails
the first and passes the second.

Third, the magnitudes are free. Nothing forces $|D_{\nu\mu}|=|D_{\mu\nu}|$;
\eqref{eq:dbalance} fixes only their ratio, at $C_{\nu\nu}/C_{\mu\mu}$.

The first identity of \eqref{eq:dbalance} is also the simplest statement in
this paper that can be tested on a sample without knowing its network. Measure
the two transmissions and the two driving-point compliances, four scalars all
accessible at the terminals: a reciprocal linear sample must return
$D_{\nu\mu}C_{\mu\mu}-D_{\mu\nu}C_{\nu\nu}=0$.

Experiment~16 checks the clauses on our own networks. Identity
\eqref{eq:dtrans} holds to $1.3\times10^{-15}$, and to $1.3\times10^{-14}$ on
unstructured accretive operators, symmetric or not; the first identity of
\eqref{eq:dbalance} holds to a relative $6.2\times10^{-13}$. The cycle identity,
the second of \eqref{eq:dbalance}, holds to $4.5\times10^{-13}$, against a violation of order one, $1.0$, once odd
couplings are on, and there is no sign disagreement in $2000$ passive cases,
against $603$ of $2000$ odd ones.

The rank of \eqref{eq:dbound} is attained in $20$ of $20$ configurations, with
a singular-value gap at the cut of at least $10^{13}$. Rank is measured from the
exact Jacobian, since central differences put a noise floor at the cutoff and
inflate the count (Appendix~\ref{app:stats}). Equation~\eqref{eq:dbound} needs
$C=C^\T$ and nothing more; positive definiteness plays no part in it. It does
need that much: with odd couplings the identical measurement gives the full
off-diagonal $p(p-1)$ at $p=3$, in $20$ of $20$ configurations, against
$(p-1)(p+2)/2$ for the passive control.

Theorem~\ref{thm:dirichlet} clamps one terminal and leaves the rest of $P$
free. Most hardware clamps several at once, and the two targets of
\citet{Du2026} clamp three. That case obeys a law of the same shape, in which
the two inequality clauses of \eqref{eq:dcycle} become a single spectral one.

\begin{theorem}[reciprocity under a multi-terminal imposed drive]
  \label{thm:multiclamp}
  Let $A$ and $B$ be disjoint sets of degrees of freedom, let every remaining
  free degree of freedom equilibrate, and let $D_{B\leftarrow A}$ be the matrix
  of read-out displacements produced by a prescribed displacement pattern on
  $A$. Write $X=C[A,A]$ and $Y=C[B,B]$, and assume both are invertible. Then
  \begin{equation}
    \label{eq:mtrans}
    D_{B\leftarrow A} \;=\; C[B,A]\,X^{-1} ,
  \end{equation}
  with no symmetry hypothesis, the degrees of freedom outside $A\cup B$ having
  been eliminated by the Schur complement. If $K_{\!f\!f}=K_{\!f\!f}^\T$ then,
  writing $D_1=D_{B\leftarrow A}$ and $D_2=D_{A\leftarrow B}$,
  \begin{equation}
    \label{eq:mbalance}
    D_1X \;=\; \bigl(D_2Y\bigr)^\T \;=\; C[B,A] ,
  \end{equation}
  and if in addition $K_{\!f\!f}\succ0$ then
  \begin{equation}
    \label{eq:mcone}
    \spec\bigl(D_2D_1\bigr) \;\subset\; [0,1) .
  \end{equation}
\end{theorem}

\begin{proof}
  Partition the free degrees of freedom into $A$ and its complement $F$, which
  contains $B$. By Jacobi's complementary minor identity,
  $\det K_{FF}=\det X\cdot\det K_{\!f\!f}$, which is the block form of the
  cofactor relation used in the proof of Theorem~\ref{thm:dirichlet}, so the hypothesis that $X$ is
  invertible is exactly the invertibility of $K_{FF}$, and the clamped problem
  $K_{FF}\bu_F=-K_{FA}\bu_A$ has the unique solution
  $\bu_F=-K_{FF}^{-1}K_{FA}\bu_A$. The block-inverse identity
  $-K_{FF}^{-1}K_{FA}=C[F,A]\,X^{-1}$ restricted to the rows $B$ is
  \eqref{eq:mtrans}. No symmetry is used, and nothing beyond $X$ is inverted,
  exactly as in \eqref{eq:dtrans}; the invertibility of $Y$ is what makes the
  same statement available with the two sets exchanged.
  Equation~\eqref{eq:mbalance} is $C[A,B]=C[B,A]^\T$ substituted into the two
  definitions. For \eqref{eq:mcone} note first that $K_{\!f\!f}\succ0$ gives
  $C\succ0$, hence $C[U,U]\succ0$ for $U=A\cup B$ and in particular $X\succ0$,
  since a principal submatrix of a positive-definite matrix is positive definite.
  Let $X^{1/2}$ be its unique positive-definite square root and put
  $M=C[A,B]\,Y^{-1}C[B,A]\succeq0$. Then $D_2D_1=MX^{-1}$, and conjugation by
  $X^{-1/2}$ gives
  \[
    X^{-1/2}\bigl(D_2D_1\bigr)X^{1/2} \;=\; X^{-1/2}MX^{-1/2} ,
  \]
  so the two have the same spectrum, which is therefore real and non-negative.
  The Schur condition for $C[U,U]\succ0$ is $X-M\succ0$, hence
  $X^{-1/2}MX^{-1/2}\prec I$ and every eigenvalue is strictly below one.
\end{proof}

All three clauses of Theorem~\ref{thm:multiclamp} specialise at $|A|=|B|=1$ to
statements of Theorem~\ref{thm:dirichlet}: \eqref{eq:mtrans} becomes
\eqref{eq:dtrans}, \eqref{eq:mbalance} becomes
$D_{\nu\mu}C_{\mu\mu}=D_{\mu\nu}C_{\nu\nu}$, and \eqref{eq:mcone} becomes both
clauses of \eqref{eq:dcycle} at once, the sign law and $D_{\nu\mu}D_{\mu\nu}<1$
being the two statements that the single eigenvalue is non-negative and below
one. The containment does not run the other way: the three-index cycle identity
of \eqref{eq:dbalance} and the rank ceiling \eqref{eq:dbound} have no
counterpart among the three clauses above, and for the ceiling
Proposition~\ref{prop:multirank} below supplies the multi-terminal form.

The invertibility hypothesis is on $X$ and $Y$ and on nothing else. It does
not imply that $C[U,U]$ is
invertible, and the proof does not need it to be: for $A=\{1\}$, $B=\{2\}$ and
$C[U,U]$ the all-ones $2\times2$ matrix, both diagonal blocks are invertible
while $C[U,U]$ is singular. Under the hypothesis of \eqref{eq:mcone} the
question does not arise, since $C[U,U]$ is then a principal submatrix of a
positive-definite matrix.

The mathematics here is classical to the same extent as that of
\eqref{eq:dtrans} and \eqref{eq:dbalance} above, and we claim none of it. The
eigenvalues of $D_2D_1$ are the squared canonical correlations between the two
terminal sets under the Gaussian field with precision matrix $K_{\!f\!f}$, and
that they lie in $[0,1)$ is the defining property of canonical
correlations~\citep{Hotelling1936,Anderson2003}. The object is not new either:
structural dynamics writes the multi-degree-of-freedom transmissibility as a
product of receptance blocks, $W_{OZ}W_{JZ}^{+}$, carrying the responses at a
driven set $J$ to the responses at a read-out set $O$ when the forces act on a
third set $Z$ alone~\citep{Ribeiro2000,Maia2001}. With the forces on the driven
set itself, $Z=J$, so that $W_{JJ}$ is square and the pseudo-inverse is an
inverse, at zero frequency, where the receptance is the static compliance, and
with $J=A$ and $O=B$, the result is $C[B,A]X^{-1}$, entry for entry.
Equation~\eqref{eq:mtrans} is that matrix, and the step
$\bu_F=-K_{FF}^{-1}K_{FA}\bu_A$ behind it is Guyan reduction, or static
condensation, in its original form~\citep{Guyan1965}. That literature uses
transmissibility to identify and to reconstruct a structure taken as given,
and as far as we have found it contains none of the three statements we draw
from the object: the reciprocity identity \eqref{eq:mbalance} between the
forward and the reverse transmissibility, the spectral bound \eqref{eq:mcone}
on their product, or the rank ceiling of Proposition~\ref{prop:multirank} on
what tuning a fixed network can make them do. What is ours is that reading:
\eqref{eq:mbalance} and \eqref{eq:mcone} constrain the transmission pairs a
passive reciprocal learning network can be \emph{trained} to realise, they are
stated in four blocks all measurable at the terminals, and they do not depend
on the network behind them.

If $A$ and $B$ intersect, a shared terminal is both clamped and read, so its
row of $D_{B\leftarrow A}$ is a unit vector and carries no information;
deleting the shared terminals from $B$ returns the disjoint case, with $A$
unchanged and $B\setminus A$ in place of $B$. A single task carries no
constraint at all: given any target $D$, set $X=I$, $C[B,A]=D$ and
$Y=\lambda I$ with $\lambda>\|D\|^2$, and the result is positive definite and
realises $D$ exactly. There is no dimension deficit either: a fibre count over
the pair $(D_1,D_2)$ returns the full $2|A||B|$.
Remark~\ref{rem:multirank} below says why no count through $C_P$ could return
less. The content of Theorem~\ref{thm:multiclamp}
appears only when the same terminals occupy both roles across two tasks, and it
is then an inequality rather than a codimension.

Experiment~18 checks \eqref{eq:mbalance} and \eqref{eq:mcone} on the networks
of Sec.~\ref{sec:numerics} at $|A|,|B|\le3$, in a multi-clamp table separate
from the wedge comparison reported there: the identity holds to
$3.7\times10^{-15}$ and the spectrum lies in $[0,1)$ in $243$ of $243$ cases.

The scale-invariance count behind \eqref{eq:dbound} extends to every family of
imposed-displacement tasks on a common set of terminals.

\begin{proposition}[rank ceiling for a family of imposed-displacement tasks]
  \label{prop:multirank}
  Let $\mathcal{T}$ be a non-empty finite collection of tasks
  $\iota=(A_\iota,B_\iota)$, each a pair of non-empty disjoint sets of free
  degrees of freedom with $X_\iota=C[A_\iota,A_\iota]$ invertible; let
  $P=\bigcup_\iota(A_\iota\cup B_\iota)$, $p=|P|$, and let
  $\mathbf{D}_{\mathcal{T}}=(D_{B_\iota\leftarrow A_\iota})_{\iota\in\mathcal{T}}$
  be the tuple of the transmissions \eqref{eq:mtrans}, with
  $n_{\mathcal{T}}=\sum_\iota|B_\iota|\,|A_\iota|$ entries. Whenever $C=C^\T$,
  \begin{equation}
    \label{eq:multirank}
    \rank\frac{\partial\mathbf{D}_{\mathcal{T}}}{\partial k}
    \;\le\;
    \min\Bigl(n_b,\ \tfrac12p(p+1)-1\Bigr) .
  \end{equation}
  If moreover $a\equiv0$, so that
  $K_{\!f\!f}=\sum_bk_b\bq_b\bq_b^\T$ on the free degrees of freedom, then
  \begin{equation}
    \label{eq:multirank2}
    \rank\frac{\partial\mathbf{D}_{\mathcal{T}}}{\partial k}
    \;\le\;
    \rank\frac{\partial C_P}{\partial k}-1
    \;\le\;
    \min\Bigl(n_b,\ \tfrac12p(p+1)\Bigr)-1 ,
  \end{equation}
  the $-1$ because $k\mapsto tk$ leaves $\mathbf{D}_{\mathcal{T}}$ fixed.
  Without symmetry, for the Jacobian with respect to $(k,a)$, the same
  argument gives $\min\bigl(2n_b-1,\ p(p-1)\bigr)$.
\end{proposition}

\begin{proof}
  Each $D_{B_\iota\leftarrow A_\iota}=C[B_\iota,A_\iota]X_\iota^{-1}$ is a
  function of $C_P=C[P,P]$ alone, so
  $\mathbf{D}_{\mathcal{T}}=\Phi_{\mathcal{T}}(C_P)$ with
  $\Phi_{\mathcal{T}}:\mathsf{H}\mapsto\bigl(\mathsf{H}[B_\iota,A_\iota]\,
  \mathsf{H}[A_\iota,A_\iota]^{-1}\bigr)_{\iota\in\mathcal{T}}$, defined and
  smooth on the open set $\mathcal{O}$ of $p\times p$ matrices on which every
  $\mathsf{H}[A_\iota,A_\iota]$ is invertible. For $t>0$,
  $\Phi_{\mathcal{T}}(t\mathsf{H})=\Phi_{\mathcal{T}}(\mathsf{H})$;
  differentiating in $t$ at $t=1$ shows that the differential of
  $\Phi_{\mathcal{T}}$ at $\mathsf{H}$ annihilates $\mathsf{H}$, and
  $\mathsf{H}\ne0$ because it has an invertible block. If
  $\mathsf{H}\in\Sym(p)$, then $\mathsf{H}$ is a non-zero vector of $\Sym(p)$
  in the kernel of the restriction of that differential to $\Sym(p)$, so the
  restriction has rank at most $\tfrac12p(p+1)-1$. At a $k$ with $C=C^\T$,
  \eqref{eq:wb} gives
  $\partial C_P/\partial k_b=-\bw_b[P]\,\bw_b[P]^\T\in\Sym(p)$, so the chain
  rule $\partial\mathbf{D}_{\mathcal{T}}/\partial k_b
  =\mathrm{d}\Phi_{\mathcal{T}}(\partial C_P/\partial k_b)$, with the
  differential taken at $C_P$, puts every one of the $n_b$ columns of
  $\partial\mathbf{D}_{\mathcal{T}}/\partial k$ in the image of $\Sym(p)$
  under that differential, and \eqref{eq:multirank} follows. If $a\equiv0$
  then $\sum_bk_b\,\partial C_P/\partial k_b=-(CK_{\!f\!f}C)[P,P]=-C_P$, so
  $C_P$ lies in the span $\mathcal{W}$ of the columns
  $\partial C_P/\partial k_b$ and is annihilated by the differential at
  $C_P$; hence $\rank\partial\mathbf{D}_{\mathcal{T}}/\partial k
  =\dim\mathrm{d}\Phi_{\mathcal{T}}(\mathcal{W})\le\dim\mathcal{W}-1$, which
  is \eqref{eq:multirank2}. Without symmetry $\Phi_{\mathcal{T}}$ on
  $\mathcal{O}$ is invariant under every column scaling
  $\mathsf{H}\mapsto\mathsf{H}\Theta$, $\Theta$ invertible diagonal, since
  $(\mathsf{H}\Theta)[B_\iota,A_\iota]\bigl((\mathsf{H}\Theta)[A_\iota,A_\iota]\bigr)^{-1}
  =\mathsf{H}[B_\iota,A_\iota]\Theta_{A_\iota}\Theta_{A_\iota}^{-1}
  \mathsf{H}[A_\iota,A_\iota]^{-1}
  =\mathsf{H}[B_\iota,A_\iota]\mathsf{H}[A_\iota,A_\iota]^{-1}$.
  The kernel of the differential at $\mathsf{H}$ therefore contains
  $\mathsf{H}e_ie_i^\T$, with $e_i$ the standard basis vector, for every
  column $i$ contained in some $A_\iota$, non-zero and in distinct positions,
  and every $e_je_i^\T$ for a column $i$ contained in no $A_\iota$, so it has
  dimension at least $p$ and the rank is at most $p^2-p$; and $K_{\!f\!f}$
  is linear in $(k,a)$, so $(k,a)\mapsto t(k,a)$ leaves
  $\mathbf{D}_{\mathcal{T}}$ fixed and the non-zero vector $(k,a)$ lies in
  the kernel of the $(k,a)$-Jacobian, which therefore has rank at most
  $2n_b-1$.
\end{proof}

Equation~\eqref{eq:dbound} is the singleton family
$\mathcal{T}=\{(\{\mu\},\{\nu\}):\mu\ne\nu\in P\}$, $p\ge2$, with
$X_\iota=C_{\mu\mu}$, $n_{\mathcal{T}}=p(p-1)$ and
$\tfrac12p(p+1)-1=\tfrac12(p-1)(p+2)$; the family
$\{(\{\mu\},P\setminus\{\mu\}):\mu\in P\}$, one terminal clamped and all
others read, returns the same rows and has the same rank. The ceiling bites
only when $n_{\mathcal{T}}>\tfrac12p(p+1)-1$, and the deficit is then at
least $n_{\mathcal{T}}-\tfrac12p(p+1)+1$. At $p=4$ with $\mathcal{T}$ the six
ordered pairs of disjoint two-element subsets, $n_{\mathcal{T}}=24$ against a
ceiling of $9$. Of the deficit of at least $15$, $(p-1)(p-2)/2=3$ is the shortfall of the symmetric
ceiling against the non-symmetric $p(p-1)=12$; the remaining $12$ lies beyond
any parametrisation through $C_P$, symmetric or not. No odd coupling restores
those $12$ dimensions.

Where the ceiling bites, experiment~20 finds it attained in $264$ of $264$
configurations on Delaunay networks of $16$ to $40$ nodes of the kind used in
Sec.~\ref{sec:numerics} ($p=3$ to $6$; the singleton, singleton-to-rest and
disjoint-subset families), with a singular-value gap of at least $10^{12}$ at
the cut. It is never exceeded in any of the $720$ rows, which also cover $8$-node
networks; there, with $n_b$ close to binding, it is missed in $34$ of $66$
configurations (Appendix~\ref{app:stats}).

For the pair of full $3\times3$ blocks that encloses the six scalar demands of
\citet{Du2026}, $\mathcal{T}=\{(A_1,B_1),(B_1,A_1)\}$ with
$A_1=\{1,2,3\}$ and $B_1=\{4,5,6\}$, the ceiling does not bite: $p=6$ and
$n_{\mathcal{T}}=18\le20=\tfrac12p(p+1)-1$, so the count constrains nothing
beyond the bond count $n_b$. Within Theorem~\ref{thm:multiclamp} the
obstruction to that task is the spectral clause \eqref{eq:mcone} alone: it excludes
an open set of pairs $(D_1,D_2)$, the target $D_2D_1\mathbf{1}=-\mathbf{1}$ of
Sec.~\ref{sec:consequences} among them, and removes no dimension, so no count
of parameters against constraints can see it.

\begin{remark}[the ceiling is sharp for the argument]
  \label{rem:multirank}
  For the singleton family the scalar direction is the whole kernel: with
  $\Phi$ the map of the proof of \eqref{eq:dbound}, if $H\in\Sym(p)$ and
  $\mathrm{d}\Phi_{C_P}(H)=0$ then $H_{ij}=C_{ij}H_{jj}/C_{jj}$ for all
  $i,j$, so $C_{ij}(H_{ii}/C_{ii}-H_{jj}/C_{jj})=0$, and $H$ is a multiple of
  $C_P$ whenever the off-diagonal entries of $C_P$ that vanish do not
  disconnect $P$. There $\rank\mathrm{d}\Phi=\tfrac12p(p+1)-1$ exactly, and
  \eqref{eq:dbound} is all that any argument through $C_P$ can give. For the
  pair of blocks enclosing the task of \citet{Du2026} the kernel is
  three-dimensional, tangent to the fibre of the map
  $C_P\mapsto(D_1,D_2)$, and the same argument returns $18=n_{\mathcal{T}}$:
  no constraint.
\end{remark}

\subsection{Attainment is a transversality condition}
\label{sec:attainment}

We now leave the change of drive and return to the force-driven block
$C[T,S]$.
Theorems~\ref{thm:main} and~\ref{thm:floor} hold at every admissible
parameter value, with no genericity assumption; the first
bounds the rank from above and the second the error from below. Whether they
are \emph{tight} is a separate question. This subsection reduces it to whether
two explicit subspaces meet in general position.

There is a coordinate-free way to see \eqref{eq:bound}. Let
$\pi_{T,S}:\Sym(n_{\mathrm{free}})\to\R^{m_T\times m_S}$ extract rows $T$ and
columns $S$ of a symmetric matrix. Two entries $(i,j)$ and $(i',j')$ of the
block come from the same entry of the symmetric matrix if and only if $i=j'$ and
$j=i'$, which forces $i,j\in P$. Identifications therefore occur only inside
the shared sub-block, and
\begin{equation*}
  \pi_{T,S}\bigl(\Sym(n_{\mathrm{free}})\bigr) \;=\; \mathcal{S}_P ,
  \qquad
  \dim = m_Tm_S-\tfrac{1}{2}p(p-1) \;=:\; d .
\end{equation*}
Let $V=\spn\{\bq_b\bq_b^\T\}\subseteq\Sym$ be the space of
stiffness perturbations the bonds can produce. Since $M\mapsto CMC$ is a linear
automorphism of $\Sym$ when $C$ is symmetric and invertible,
\begin{equation}
  \label{eq:rankid}
  \rank\frac{\partial R}{\partial k}
  \;=\;
  \dim \pi_{T,S}\bigl(CVC\bigr).
\end{equation}
This identity recovers Theorem~\ref{thm:main}. It amounts to a change of
variables. Theorem~\ref{thm:duality}
states exactly how much is lost.

Write $V^{\perp}=\{M\in\Sym:\bq_b^\T M\bq_b=0\ \forall b\}$ for the orthogonal
complement of $V$; in $KV^{\perp}K$, here and below, $K$ abbreviates
$K_{\!f\!f}$. For $Y\in\R^{m_T\times m_S}$, let
$\hat{Y}\in\R^{n_{\mathrm{free}}\times n_{\mathrm{free}}}$ denote $Y$ placed in
rows $T$ and columns $S$ with zeros elsewhere. The matrix $\hat{Y}$ is not
symmetric, so we symmetrise it and put
\begin{equation*}
  \mathcal{Y} \;=\; \bigl\{\,\tfrac12(\hat{Y}+\hat{Y}^\T)\ :\ Y\in\mathcal{S}_P\,\bigr\}
  \;\subseteq\;\Sym(n_{\mathrm{free}}) .
\end{equation*}

\begin{theorem}[the deficit is an intersection]
  \label{thm:duality}
  For every $K_{\!f\!f}=K_{\!f\!f}^\T$ invertible and every admissible $k$,
  \begin{equation}
    \label{eq:duality}
    \dim\mathcal{Y}=d
    \qquad\text{and}\qquad
    d-\rank\frac{\partial R}{\partial k}
    \;=\;
    \dim\Bigl(\mathcal{Y}\ \cap\ KV^{\perp}K\Bigr).
  \end{equation}
\end{theorem}

\begin{proof}
  The linear map $Y\mapsto\tfrac12(\hat{Y}+\hat{Y}^\T)$ kills $Y$ exactly when
  $\hat{Y}+\hat{Y}^\T=0$. As $\hat{Y}$ is supported in rows $T$ and columns $S$
  and $\hat{Y}^\T$ in rows $S$ and columns $T$, this happens precisely when $Y$
  is supported on the shared block and antisymmetric there, i.e.\
  $Y\in\mathcal{A}_P$. Since $\mathcal{S}_P\cap\mathcal{A}_P=0$ the map is
  injective on $\mathcal{S}_P$, giving $\dim\mathcal{Y}=d$.

  By Theorem~\ref{thm:main} the Jacobian image lies in $\mathcal{S}_P$, so its
  orthogonal complement \emph{within} $\mathcal{S}_P$ has dimension
  $d-\rank$. A functional $Y\in\mathcal{S}_P$ annihilates every column iff
  $\langle Y,\bw_b[T]\bw_b[S]^\T\rangle=0$ for all $b$, that is
  $\bw_b[T]^\T Y\bw_b[S]=0$, that is $\bw_b^\T\hat{Y}\bw_b=0$, and a quadratic
  form sees only the symmetric part, so this reads
  $\bw_b^\T\,\tfrac12(\hat{Y}+\hat{Y}^\T)\,\bw_b=0$. Substituting
  $\bw_b=C\bq_b$,
  \[
    \bq_b^\T\,C\,\tfrac12(\hat{Y}+\hat{Y}^\T)\,C\,\bq_b=0 \ \ \forall b
    \iff
    C\,\tfrac12(\hat{Y}+\hat{Y}^\T)\,C\in V^{\perp}
    \iff
    \tfrac12(\hat{Y}+\hat{Y}^\T)\in KV^{\perp}K .
  \]
  Each step is an equivalence, so the correspondence is onto as well as into:
  every $Z\in\mathcal{Y}\cap KV^{\perp}K$ is $\tfrac12(\hat{Y}+\hat{Y}^\T)$ for
  a unique $Y\in\mathcal{S}_P$, and that $Y$ annihilates every column. The
  annihilator is therefore carried isomorphically onto
  $\mathcal{Y}\cap KV^{\perp}K$, which gives \eqref{eq:duality}.
\end{proof}

\begin{remark}[the general-position answer, restated]
  \label{cor:generic}
  Under the hypotheses of Theorem~\ref{thm:duality}, suppose the bond dyads are
  independent, $\dim V=n_b$, so that
  $\dim KV^{\perp}K=\dim\Sym-n_b$. If $\mathcal{Y}$ and $KV^{\perp}K$ meet in
  general position, then
  $\dim(\mathcal{Y}\cap KV^{\perp}K)=\max(0,\,d-n_b)$ and
  \begin{equation*}
    \rank\frac{\partial R}{\partial k}=\min(n_b,\,d),
  \end{equation*}
  the bound of Theorem~\ref{thm:main}, attained.
\end{remark}

For two \emph{specific} subspaces, ``meeting in general position'' is precisely
the conclusion restated, so invoking it as a property proves nothing. What
Theorem~\ref{thm:duality} does buy is a change of question. Failures of
attainment are failures of transversality between $\mathcal{Y}$ and
$KV^{\perp}K$, and nothing else. $V$ is generated by the rigidity dyads of a graph
and is highly structured; $\mathcal{Y}$ is supported on a prescribed block.
The rigidity matroid does not decide whether they are transverse, and we leave
it open.

\subsection{Certified attainment at rational configurations}
\label{sec:certificates}

Transversality can nevertheless be \emph{proved} instance by instance, and the
proof then transfers to almost every configuration. Two observations make this
work.

First, the problem is rational. The compatibility row carries the unit vector
along the bond, whose entries are in general irrational at rational node
positions. Let $L_b$ be the bond length and $\check{\bq}_b$ the row built from
the unnormalised difference $\bx_j-\bx_i$. Then
$\bq_b\bq_b^\T=\check{\bq}_b\check{\bq}_b^\T/L_b^{2}$, so
$K=\sum_b\kappa_b\,\check{\bq}_b\check{\bq}_b^\T$ with
$\kappa_b=k_b/L_b^{2}$. Reparametrising by $\kappa$ rescales each Jacobian
column by a non-zero constant and cannot change the rank. At integer node
positions and integer $\kappa$, every quantity is rational.

Second, a rank can be certified from below in a finite field. In this
subsection $J$ denotes $\partial R/\partial\kappa$. With
$\Delta_K=\det K_{\!f\!f}\in\mathbb{Z}$, the matrix $\Delta_K^{2}J$ has integer
entries. For every prime $\ell$, reduction modulo $\ell$ can only lower the
rank, $\rank_{\mathbb{F}_\ell}\le\rank_{\mathbb{Q}}$ (we write $\ell$ for the
prime throughout, reserving $p$ for the overlap $|T\cap S|$), and for
$\ell\nmid\Delta_K$ the reduction can be computed from $K_{\!f\!f}\bmod\ell$
alone.
Theorem~\ref{thm:main} bounds the rank from above, so a mod-$\ell$ computation
that returns $\min(n_b,d)$ forces equality over $\mathbb{Q}$.

\begin{proposition}[from one certificate to almost every configuration]
  \label{prop:generic}
  Fix a graph, a choice of pinned degrees of freedom, and index sets $T,S$ among
  the remaining ones, and write $X$ for the stacked node positions
  $(\bx_1,\dots,\bx_N)$. If a single $(X_0,\kappa_0)$ with $\kappa_0>0$ satisfies
  $\rank J=\min(n_b,d)$, then there is a non-zero polynomial
  $F$ on $\R^{\delta N}\times\R^{n_b}$ ($\delta=2$, the spatial dimension of
  Sec.~\ref{sec:setup}) such that $\rank J=\min(n_b,d)$ at every
  admissible $(X,\kappa)$ with $F(X,\kappa)\neq0$. The excluded set has Lebesgue
  measure zero and empty interior.
\end{proposition}

\begin{proof}
  Fixing the pinning makes $T$ and $S$ index a set of coordinates that does not
  move with $X$. The entries of $J$ are rational in $(X,\kappa)$, regular where
  $\det K_{\!f\!f}\neq0$, so $\det(K_{\!f\!f})^{2}J$ has polynomial entries. Let
  $r=\min(n_b,d)$ and let $F$ be the product of $\det K_{\!f\!f}$ with an
  $r\times r$ minor of $\det(K_{\!f\!f})^{2}J$ that is non-zero at
  $(X_0,\kappa_0)$; such a minor exists by hypothesis, so $F$ is a non-zero
  polynomial. Where $F\neq0$ the rank is at least $r$, and
  Theorem~\ref{thm:main} bounds it above by $r$. The zero set of a non-zero
  polynomial on $\R^{\delta N}\times\R^{n_b}$ is closed with empty interior and Lebesgue measure
  zero, and it therefore meets the open admissible region $\{\kappa>0\}$ in a
  set of measure zero.
\end{proof}

\begin{remark}
  The statement is deliberately phrased through one explicit polynomial rather
  than in Zariski-topological terms. Over $\R$ a non-empty Zariski-open subset
  of $\R^{\delta N}\times\R^{n_b}$ is open and dense in the Euclidean topology, with a complement of
  measure zero, so the two phrasings carry the same content; ``outside the zero
  set of one explicit non-zero polynomial'' is the easier one to check, and it is
  what the certificate delivers.
\end{remark}

We carried this out with $\ell=2^{61}-1$ for $432$ instances: Delaunay
triangulations of random integer points, $10$ to $36$ nodes, $m_T=m_S=m$ from
$2$ to $6$, and overlaps $p\in\{0,1,\lfloor m/2\rfloor,m\}$. The room criterion
asks for room for the predicted rank both among the free degrees of freedom and
among the parameters, $d\le\tfrac12 n_{\mathrm{free}}$ and
$d\le\tfrac12 n_\theta$. The criterion is empirical, read off Sec.~\ref{sec:numerics}, and it fails
without a genericity clause, which we do not have; here it only sorts the
instances.

\emph{Every one of the $237$ instances satisfying the room criterion is
certified.} Of the remaining $195$, $133$ certify, which makes $370$ certified
configurations in all. Each certified instance is a theorem in the sense of
Proposition~\ref{prop:generic}: for that graph, that pinning and that choice
of driven and read-out degrees of freedom, attainment holds off the zero set
of an explicit polynomial. Each certified configuration (node coordinates,
bond list, $T$, $S$ and $\kappa$) is supplied with the code, so that it can be
restated and rechecked independently.

The other $62$ do not certify. Failing to certify proves nothing on its own,
since $\rank_{\mathbb{F}_\ell}$ is only a lower bound, so we recomputed those
$62$ \emph{exactly over} $\mathbb{Q}$ with rational arithmetic. In $62$ of
$62$ the exact rank falls strictly below $\min(n_b,d)$; those are genuine
failures of attainment. The floating-point rank agrees with the mod-$\ell$ rank
in all $432$ instances, but both ranks are one-sided, so we record the
agreement as a consistency check only and not as evidence either way.

Two caveats belong with any computer-assisted result. Every computation
reported in this paper is a single implementation, unverified by an independent one, and
the certificates are for particular graphs. The one exception to the first
caveat is experiment~18, whose two claims were also checked by
a second implementation, experiment~18b, that shares no code with the first:
it rebuilds the multi-terminal transmission by solving the clamped equilibrium
directly rather than through \eqref{eq:mtrans}, and takes the Jacobian whose
projection \eqref{eq:wedge} describes by central finite differences
(Sec.~\ref{sec:numerics}). The certificate computation itself has not been
reimplemented.

What is still missing is a proof valid for \emph{all} graphs at once. We have
reduced that to the transversality of $\mathcal{Y}$ and $KV^{\perp}K$
(Theorem~\ref{thm:duality}), verified the reduction and certified a family of
instances, but we have not established the transversality in general. The
result of this section as a whole stands at two levels. The containment in a subspace of codimension $p(p-1)/2$ is proved for every fixed
network; attainment of the bound is certified at $370$ of the $432$
configurations tested, and exact rational arithmetic shows it failing at the
other $62$.


\section{What non-reciprocity buys}
\label{sec:recovery}

The containment of Theorem~\ref{thm:main} rests on reciprocity, and this
section asks what an odd coupling, which breaks it, buys back. The answer is
that, whenever the passive network's wedges span, odd couplings on
$p(p-1)/2$ bonds chosen in advance realise every small enough antisymmetric
part of the shared block exactly, in projection and at unchanged stiffnesses.
Proposition~\ref{prop:odd} records what dropping reciprocity delivers on its own;
Theorem~\ref{thm:recovery} with Corollary~\ref{cor:submersion} identifies which
directions an odd coupling restores, and from how few bonds.
Sec.~\ref{sec:price} then asks what the network pays for them.

\subsection{What the wedges restore}
\label{sec:wedges}

\begin{proposition}[the reciprocity obstruction is no longer guaranteed]
  \label{prop:odd}
  Suppose $K_{\!f\!f}\neq K_{\!f\!f}^\T$, with $K_{\!f\!f}$ invertible. Then
  the containment $R\in\mathcal{S}_P$ of Theorem~\ref{thm:main} and
  Remark~\ref{rem:global} is no longer guaranteed (the step that produces
  it, the symmetry of $\mathrm{d}C$, is unavailable), and the bound that
  holds uniformly over such networks is the count
  \begin{equation*}
    \rank \frac{\partial R}{\partial(k,a)} \;\le\; \min\bigl(2n_b,\ m_Tm_S\bigr).
  \end{equation*}
  Nothing more is asserted. The hypothesis withdraws the \emph{universal}
  guarantee of codimension $p(p-1)/2$; it does not force the realised block
  out of $\mathcal{S}_P$ at any parameter value, and it says nothing about how
  much of $\mathcal{A}_P$ a given construction reaches. That is a property of
  the construction; for the odd coupling of \eqref{eq:stiffness}
  Theorem~\ref{thm:recovery} computes it, and Sec.~\ref{sec:numerics}
  measures it for those networks and for no others.
\end{proposition}

\begin{proof}
  If $K_{\!f\!f}\neq K_{\!f\!f}^\T$ then $C\neq C^\T$, and by \eqref{eq:jaccol}
  the columns are $-\bx[T]\,\by[S]^\T$ with $\bx\in\{C\bq_b,C\bs_b\}$ and
  $\by=C^\T\bq_b$. At least one of these pairs is distinct: the range of
  $K_{\!f\!f}^\T$ lies in the span of the restricted $\bq_b$, so invertibility
  makes them span $\R^{n_{\mathrm{free}}}$, and $C-C^\T\ne0$ then fails to
  annihilate some $\bq_b$. The step in
  Theorem~\ref{thm:main} that used $\bx=\by$, equivalently the symmetry of
  $\mathrm{d}C$, is therefore unavailable, so \eqref{eq:SP} is no longer implied.
  That is all the argument gives: $\bx\ne\by$ does not make the shared block of
  $\bx[T]\by[S]^\T$ asymmetric, and nothing here excludes further entry
  relations or rank deficiencies arising
  from the particular geometry, topology or parametrisation; the claim is only
  that the reciprocity argument is absent. The stated inequality is the column count against the dimension of
  the codomain.
\end{proof}

\begin{remark}[why the hypothesis is on $K_{\!f\!f}$ and not on $a$]
  \label{rem:odd-hyp}
  The natural-looking hypothesis ``some $a_b\neq0$'' is not enough, and we state
  the stronger one because the weaker one is false. A bond whose two endpoints
  are pinned except along a single free direction contributes a $1\times1$ block
  to $K_{\!f\!f}$, and a $1\times1$ block is symmetric whatever $a_b$ is. With
  the pinning used here such a bond is constructible: one joining the fully
  pinned node
  to the partly pinned one touches exactly one free coordinate, and setting
  $a_b=3.7$ on it leaves $\|K_{\!f\!f}-K_{\!f\!f}^\T\|_F$ at machine zero, below
  $10^{-15}$, so $C$ is symmetric and the realised block lies in
  $\mathcal{S}_P$ after all.

  Nor is the counting bound always the best available, which is why we do not
  call it the only one: for such a bond $\partial K/\partial k_b$ and
  $\partial K/\partial a_b$ restrict to proportional $1\times1$ blocks, so two
  of the $2n_b$ columns are parallel and counting alone already gives $2n_b-1$.
  Neither claim enters a proof: the configuration is cited elsewhere only
  to explain the hypothesis on $K_{\!f\!f}$, to note that
  Theorem~\ref{thm:recovery} covers it, and to list it among the one-off checks under Code
  availability. The proposition is deliberately
  negative.
\end{remark}

\begin{remark}[the limits of the proposition]
  Proposition~\ref{prop:odd} is a negative result. It withdraws a guarantee,
  and leaves open how much of $\R^{m_T\times m_S}$ the network's own vectors
  actually reach. It would be easy, and wrong, to argue from the fact that
  unconstrained outer products $\bx[T]\by[S]^\T$ span everything; the columns
  are not unconstrained. Theorem~\ref{thm:recovery} answers the part of that
  question which concerns $\mathcal{A}_P$, and answers it by proof rather than
  by measurement.
\end{remark}

For the odd coupling of \eqref{eq:stiffness}, the question the remark leaves
open has an exact answer, computable at the passive network before any odd
bond is switched on.

The spaces $\mathcal{A}_P$ and $\Lambda^2\R^p$ have the same dimension,
$\dim\mathcal{A}_P=p(p-1)/2=\dim\Lambda^2\R^p$. Fixing the ordering of $P$ used in \eqref{eq:SP} identifies
them, by sending $E_{\tau(r)\sigma(c)}-E_{\tau(c)\sigma(r)}$ to
$e_r\wedge e_c$, with the convention
$\bu\wedge\bv=\bu\bv^\T-\bv\bu^\T$ used throughout. The identification depends
on that ordering; the spans below do not.

\begin{theorem}[what non-reciprocity recovers]
  \label{thm:recovery}
  Let $K_{\!f\!f}=K_{\!f\!f}^\T$ be invertible, $C=K_{\!f\!f}^{-1}$, and set
  $\bw_b=C\bq_b$ and $\bz_b=C\bs_b$. At any such configuration, whether $a=0$
  or one of the exceptional $a\neq0$ of Remark~\ref{rem:odd-hyp},
  \begin{equation}
    \label{eq:wedge}
    \Pi_{\mathcal{A}_P}\Bigl(\operatorname{im}
      \frac{\partial R}{\partial(k,a)}\Bigr)
    \;=\;
    \spn\bigl\{\,\bz_b|_P\wedge\bw_b|_P \;:\; b=1,\dots,n_b\,\bigr\}
    \;\subseteq\; \Lambda^2\R^p .
  \end{equation}
  The odd branch therefore recovers the whole of the subspace reciprocity
  removes if and only if those $n_b$ wedges span $\Lambda^2\R^p$. A single bond
  contributes nothing exactly when $\bz_b|_P$ and $\bw_b|_P$ are linearly
  dependent, which includes the case where one of them vanishes, and every
  contribution vanishes when $p\le1$.
\end{theorem}

\begin{proof}
  By Theorem~\ref{thm:main} every $\partial R/\partial k_b$ lies in
  $\mathcal{S}_P$ and projects to zero. By hypothesis $C=C^\T$, so
  \eqref{eq:jaccol} with $\bu=\bs_b$ and $\bv=\bq_b$ gives
  $\partial R/\partial a_b=-\bz_b[T]\,\bw_b[S]^\T$, whose shared sub-block is
  $-\bz_b|_P\,\bw_b|_P^\T$ and whose antisymmetric part is therefore
  $-\tfrac12\bigl(\bz_b|_P\bw_b|_P^\T-\bw_b|_P\bz_b|_P^\T\bigr)
  =-\tfrac12\,\bz_b|_P\wedge\bw_b|_P$. $\Pi_{\mathcal{A}_P}$ is linear, so the projection of a span is the span of
  the projections.
\end{proof}

\begin{remark}[the raw wedge is an edge, not a direction]
  \label{rem:edgewedge}
  In two dimensions the bond's direction and length drop out of the wedge
  before the compliance acts on it. With $\epsilon=\left(\begin{smallmatrix}0&1\\-1&0\end{smallmatrix}
  \right)$ and the degrees of freedom ordered node by node,
  \begin{equation}
    \label{eq:edgewedge}
    \bs_b\wedge\bq_b\;=\;-\,\bigl[(e_i-e_j)(e_i-e_j)^\T\bigr]\otimes\epsilon
    \qquad\text{for a bond } b=\{i,j\},
  \end{equation}
  since $\bq_b=(e_j-e_i)\otimes\hat n_b$ and $\bs_b=(e_j-e_i)\otimes\hat t_b$,
  whose outer product is even in that sign, with $\hat t_b=\epsilon^\T\hat n_b$,
  and
  $\hat t\hat n^\T-\hat n\hat t^\T=-\epsilon$ whatever the direction, because
  $SO(2)$ acts trivially on $\Lambda^2\R^2$. Hence
  $\bz_b\wedge\bw_b=C(\bs_b\wedge\bq_b)C$, and what separates one odd bond from
  another in \eqref{eq:wedge} is which pair of nodes it joins, conjugated by the
  compliance; the bond's own direction and length do not enter, though $C$
  carries the geometry of the whole network. The cancellation is two-dimensional:
  in three dimensions $\bs_b$ is not determined by the bond, and as the
  transverse direction turns, the wedge sweeps the whole two-dimensional family
  of $2$-forms $\hat n_b\wedge\hat t$ with $\hat t\perp\hat n_b$, in place of
  the single fixed $\epsilon$.
\end{remark}

Breaking reciprocity to regain what it forbids has been done in another
medium: \citet{Li2025} break it throughout a surface-plasmonic network, with
magnetically biased ferrite, so that its forward and backward paths carry
independently programmed functions, rather than on bonds selected in advance.
An odd bond is the distributed counterpart of the gyrator of network
synthesis~\citep{Tellegen1948,Belevitch1968}, which likewise adds one wedge to
an immittance matrix; synthesis places its gyrators freely, whereas here the
wedges are fixed by the response fields of the built network.

Four consequences follow: two counts, a selection rule and a decoupling.

\emph{At least $p(p-1)/2$ bonds must contribute, with the stiffnesses held.} A
wedge is a single element of
$\Lambda^2\R^p$, so fewer bonds with a non-zero wedge than $\dim\Lambda^2\R^p$
cannot span it. The count is read off \eqref{eq:wedge}, so it bounds the odd
bonds needed to reach $\mathcal{A}_P$ from a passive network at fixed $k$. The
count is not a bound on the finite problem with $k$ free as well. Once
$a\neq0$, the derivatives in $k$ no longer lie in $\mathcal{S}_P$. We are not
aware of a lower bound of this kind on the non-reciprocal side of the problem.

Experiment~21, part~8, runs both cases (Appendix~\ref{app:stats}). In its
stiffness-free arm, on $40$-node networks at $m=p=5$, where the count asks for
ten bonds, three odd bonds realise a prescribed antisymmetric block of relative
size $10^{-2}$ in nine of nine (network, target) pairs, with every $k_b>0$ and
the symmetric part of $K_{\!f\!f}$ still positive definite; three bonds suffice,
but not every triple does. Its negative control runs the pivoted-$QR$ triple
at fixed $k$, at the smaller target $10^{-4}$, and solves none of the nine: it
stalls at the component of the target orthogonal to the span of its three
wedges, to within a relative $8.4\times10^{-5}$. A second fixed-$k$ control, on the ten bonds the count does ask
for, is run up the full ladder to $10^{-2}$, with mixed results;
Remark~\ref{rem:submersion} reports them.

\emph{Conversely, $p(p-1)/2$ bonds suffice to recover the antisymmetric
complement, though not necessarily the full block.} The part of the block
inside $\mathcal{S}_P$ is the ordinary attainment question that the room
criterion sorts; switching on $p(p-1)/2$ odd bonds adds only $p(p-1)/2$
parameters to that question. When that question is answered affirmatively at
$k_0$, Corollary~\ref{cor:submersion} reaches the full block exactly near the passive
response. Those bonds do settle $\mathcal{A}_P$: whenever $p(p-1)/2$ of
them have independent wedges, \eqref{eq:wedge} delivers the whole of the
antisymmetric complement from those bonds alone, and every other bond may stay
passive. Corollary~\ref{cor:submersion} makes that exact rather than first
order, on a neighbourhood whose size is governed in part by the conditioning of
the selected wedges (Remark~\ref{rem:submersion}). The \emph{number} of active
elements is therefore bounded by the codimension rather than by the size of the
network. Keeping that number small matters because each odd bond exerts its
own couple (Sec.~\ref{sec:setup}); Proposition~\ref{prop:torquefree} cancels
their net torque but not the couples themselves. Corollary~\ref{cor:price} gives, to first order, the
least Euclidean norm of the coupling $a$ that reaches a small antisymmetric
target; how small the couples themselves can be made is left open.

\emph{The choice can be made before anything is built.} Both $\bw_b$ and $\bz_b$
are response fields of the passive network, obtained from one factorisation of
$K_{\!f\!f}$ followed by a batched set of right-hand-side solves, with no active
element present. Like the error floor of Theorem~\ref{thm:floor}, the choice
is computed in advance rather than measured afterwards. Selection is a
rank-revealing problem rather than a ranking one: each wedge is vectorised by
its $p(p-1)/2$ strictly upper entries, as one column of the wedge matrix
$\mathsf{W}$ (Sec.~\ref{sec:price}), and a pivoted $QR$ of that (unnormalised) matrix
selects $p(p-1)/2$ of them, the greedy selection also used to place
sensors~\citep{Manohar2018}. The smallest singular value $\sigma_{\min}$ of the
selected set reports how well conditioned the choice is. It is also one of the
two quantities that set the size of the neighbourhood on which
Corollary~\ref{cor:submersion} applies (Remark~\ref{rem:submersion}). Ordering
the bonds by $\|\bz_b|_P\wedge\bw_b|_P\|$ alone would not do, since several
large wedges can be nearly parallel. Nor would the rank-greedy rule, which
accepts any wedge not yet in the span and so certifies independence but not
conditioning. Experiment~18 uses the pivoted $QR$, a heuristic that
Sec.~\ref{sec:price} measures against the optimal selection.

\emph{The theorem decouples the two rank questions.} The comparison of
Sec.~\ref{sec:numerics} needed two rank equalities, that the passive image
reaches $\mathcal{S}_P$ and that the odd image reaches all of
$\R^{m_T\times m_S}$, and it excluded the configurations where the first fails.
Equation~\eqref{eq:wedge} is a statement about $\Pi_{\mathcal{A}_P}$ of the
image and needs neither, so it covers the room-limited configurations as well.

On a given graph, the wedges span generically as soon as they span at one
configuration, as the argument below shows. That they span at all is measured
(Sec.~\ref{sec:numerics}) and not proved: the rationalisation of
Sec.~\ref{sec:certificates} would carry over to it, but we have not run those
certificates on the wedge condition. The genericity argument carries over to
the odd branch as follows. Writing $\check{\bs}_b$ for the row built from the
unnormalised perpendicular of $\bx_j-\bx_i$, one has
$\bs_b\bq_b^\T=\check{\bs}_b\check{\bq}_b^\T/L_b^{2}$, so that with
$\check{a}_b=a_b/L_b^{2}$ the operator
$K=\sum_b(\kappa_b\check{\bq}_b+\check{a}_b\check{\bs}_b)\check{\bq}_b^\T$ is
polynomial in the node positions and $(\kappa,\check{a})$; the rescaling
multiplies each odd Jacobian column, and each wedge, by a non-zero constant,
so no rank and no span changes.
The wedge entries are rational in the node positions and
$(\kappa,\check{a})$, so a $\delta_p\times\delta_p$ minor of the wedge matrix
that is non-zero at one configuration clears to a non-zero polynomial, and, by
the argument of Proposition~\ref{prop:generic}, the wedges span at every
configuration off its zero set. The target there has dimension
$p(p-1)/2$ rather than $m_Tm_S$, so the certificate is correspondingly cheaper
and the room criterion does not bind.

Theorem~\ref{thm:recovery} is a statement about the tangent space wherever
$K_{\!f\!f}$ is symmetric, and in particular at $a=0$.
The inverse function theorem turns it into one about the response itself.

\begin{corollary}[the recovery is exact in projection, not only to first order]
  \label{cor:submersion}
  Fix an admissible $k_0$ with $K_{\!f\!f}(k_0,0)$ invertible, and suppose
  the wedges $\bz_b|_P\wedge\bw_b|_P$, $b=1,\dots,n_b$, formed at $(k_0,0)$,
  span $\Lambda^2\R^p$. Then $F:a\mapsto\Pi_{\mathcal{A}_P}R(k_0,a)$ is a
  submersion at $a=0$ with $F(0)=0$, so the image of every neighbourhood of
  $a=0$ contains a neighbourhood of $0$ in $\mathcal{A}_P$. That is, every
  antisymmetric shared block of sufficiently small norm is realised exactly,
  at $k=k_0$ unchanged, by some odd-coupling vector $a$ near $0$, which may be
  supported on any $p(p-1)/2$ bonds whose wedges are independent. The
  symmetric part of the response is not controlled by this clause
  (Remark~\ref{rem:submersion}). If in addition
  $\rank\partial R/\partial(k,a)=m_Tm_S$ at $(k_0,0)$, which holds in
  particular when $\rank\partial R/\partial k=d=m_Tm_S-p(p-1)/2$ at $k_0$,
  the $d$ of Sec.~\ref{sec:attainment}, then $R$ itself is a submersion at $(k_0,0)$
  and the image of every neighbourhood of $(k_0,0)$ contains a neighbourhood
  of $R(k_0,0)$. Every target sufficiently close to the passive response,
  antisymmetric part included, is reached exactly.
\end{corollary}

\begin{proof}
  The entries of $C=K_{\!f\!f}^{-1}$ are rational in $(k,a)$, regular on the
  open set $\{\det K_{\!f\!f}\neq0\}\ni(k_0,0)$, so $R$ and $F$ are smooth
  there and $\mathrm{d}F_0=\Pi_{\mathcal{A}_P}\circ\partial R/\partial a$. By
  the proof of Theorem~\ref{thm:recovery} its image is the span of the
  wedges, identified with a subspace of $\mathcal{A}_P$, which is all of
  $\mathcal{A}_P$ by hypothesis; $F$ is a submersion at $0$. The local
  submersion theorem (the inverse function theorem on a slice through $0$
  complementary to $\ker\mathrm{d}F_0$) then gives that the image under $F$
  of every neighbourhood of $0$ contains a neighbourhood of $F(0)$. The slice
  may be taken to be the coordinate subspace of any $p(p-1)/2$ bonds whose
  wedges are independent, since $\mathrm{d}F_0$ restricted to those
  coordinates is then invertible; the realising $a$ can therefore be
  supported on those bonds alone, with every other bond passive.
  The radius this produces depends on the slice, but the admissible supports
  are the subsets of $\{1,\dots,n_b\}$ of size $p(p-1)/2$ whose wedges are
  independent, of which there are at most $\binom{n_b}{p(p-1)/2}$; taking the
  smallest radius over that finite collection makes the statement uniform in
  the choice of support, as stated. Finally
  $F(0)=0$ because $C$ is symmetric at $a=0$, so $R(k_0,0)\in\mathcal{S}_P$
  (Remark~\ref{rem:global}).

  For the second clause let $J=\partial R/\partial(k,a)$ at $(k_0,0)$. The
  parameter set $\{k>0\}\times\R^{n_b}$ is open and $(k_0,0)$ is interior to
  it, so if $\rank J=m_Tm_S$ the same local submersion theorem applies to $R$
  and gives the conclusion. It remains to see that the rank hypothesis on
  $\partial R/\partial k$ forces $\rank J=m_Tm_S$.
  Theorem~\ref{thm:main} places $\operatorname{im}\,\partial R/\partial k$
  inside $\mathcal{S}_P$, whose dimension is $d$, so that hypothesis gives
  $\mathcal{S}_P=\operatorname{im}\,\partial R/\partial k\subseteq
  \operatorname{im}J$. For $A\in\mathcal{A}_P$ pick $w\in\operatorname{im}J$
  with $\Pi_{\mathcal{A}_P}w=A$, as \eqref{eq:wedge} and the wedge hypothesis
  allow; then $A=w-\Pi_{\mathcal{S}_P}w\in\operatorname{im}J+\mathcal{S}_P
  =\operatorname{im}J$. Hence $\operatorname{im}J\supseteq\mathcal{S}_P
  \oplus\mathcal{A}_P=\R^{m_T\times m_S}$.
\end{proof}

\begin{remark}[the limits of the corollary]
  \label{rem:submersion}
  Five limits belong with it. \emph{No uniform radius.} ``Sufficiently
  small'' is local: the radius is set by $\sigma_{\min}$ of $\mathrm{d}F_0$,
  the wedge matrix, and by the curvature of $R$, and nothing here bounds
  either from below. Stability lives inside that same unquantified radius and
  is not separately guaranteed. Positive
  definiteness of the symmetric part of $K_{\!f\!f}$ is an open condition, so
  it survives on some neighbourhood of $a=0$ along with everything else; but
  that neighbourhood is not bounded here either, and it is not the same for
  every construction. On the sparse construction this corollary supplies (odd
  couplings on the $p(p-1)/2$ bonds of the pivoted-$QR$ selection, every
  other bond passive), experiment~19 finds positive definiteness lost in
  none of the $336$ configurations at $\varepsilon=10^{-4}$ or $10^{-3}$, in
  $28$ at $10^{-2}$ and in $261$ at $10^{-1}$; the rank-greedy selection of
  the same size loses it in $21$, $121$, $238$ and $311$, already at the
  smallest target; and letting every odd coupling be free, which is not the sparse
  construction the corollary supplies, loses it in none, none, one and $115$. Concentrating
  the couplings on the few bonds the count allows is what costs the stability,
  and the selection rule decides how much. Experiment~21, part~8, also has a
  ten-bond fixed-$k$ arm, which runs this corollary's own count (the
  three-bond arms quoted earlier in this section do not). It sees the
  same thing from the other side: on $40$-node networks at
  $m=p=5$, where the count asks for ten bonds, the couplings those ten need at
  fixed $k$ grow to the size of the stiffnesses they sit on, and the $10^{-2}$
  target is missed in $5$ of the $9$ (network, target) pairs. Four of the five
  halt below it, at a median target size of $4.4\times10^{-3}$, two of those
  with $\operatorname{sym}K_{\!f\!f}$ driven to the edge of positive
  definiteness; the fifth carries the ladder to $10^{-2}$ and fails to converge
  there with the symmetric part still definite. The remaining $4$ of the $9$
  carry the ladder through and solve the
  $10^{-2}$ target to a relative residual below $10^{-12}$.
  What is missing is a radius, and the stability is one of the things inside
  it.
  \emph{The symmetric part moves.} The first clause fixes
  $\Pi_{\mathcal{A}_P}R$ and says nothing about $\Pi_{\mathcal{S}_P}R$;
  undoing that change is what the $k$-branch of the second clause is for.
  \emph{The second clause is conditional.} It assumes that the passive
  network attains the bound of Theorem~\ref{thm:main} at $k_0$ (more
  generally, that $\rank\partial R/\partial(k,a)=m_Tm_S$ there), which this
  paper certifies instance by instance (Sec.~\ref{sec:certificates}) and does
  not prove in general. \emph{Reaching is not learning.} The floor of
  Theorem~\ref{thm:floor} is removed by the odd branch, which moves $R$ out
  of $\mathcal{S}_P$, not lowered by any tuning of $k$;
  Sec.~\ref{sec:learning} measures trainability separately. \emph{The network
  is active.} The cost is the one stated after \eqref{eq:stiffness}; the
  sparse form of the first clause confines it to $p(p-1)/2$ bonds, and the
  selection rule matters. Experiment~19 gives an empirical lower bound on the radius for
  its ensemble, and puts numbers on the stability, the symmetric part, the conditional clause and the
  selection rule (Sec.~\ref{sec:numerics}).
\end{remark}

Theorem~\ref{thm:recovery} is also the easiest statement in this paper to
refute: a network whose wedges span $\Lambda^2\R^p$ but whose odd branch fails
to reach some direction of $\mathcal{A}_P$, or one whose odd branch reaches a
direction its wedges do not span, would falsify \eqref{eq:wedge} outright.
Experiment~18 looks for both and finds neither (Sec.~\ref{sec:numerics}).


\subsection{The price of non-reciprocity}
\label{sec:price}

Theorem~\ref{thm:recovery} and its consequences bound the number of odd bonds
by the codimension, but leave open how strongly those bonds must act and what
the network pays for it. The answer has three parts: a lower bound on the
non-reciprocity of any realising network with a positive-definite symmetric
part, fixed by the target alone
(Proposition~\ref{prop:ratio}); to first order, the least coupling that
delivers a small antisymmetric block, together with the bond selection that
minimises its worst case over targets (Corollary~\ref{cor:price}); and a class
of realisations that exerts no net torque (Proposition~\ref{prop:torquefree}).

For a real square matrix $Z$ write $\operatorname{anti}Z=\tfrac12(Z-Z^\T)$
beside $\operatorname{sym}Z$, and, when $\operatorname{sym}Z\succ0$,
\begin{equation*}
  \eta(Z) \;=\; \bigl\|(\operatorname{sym}Z)^{-1/2}\,(\operatorname{anti}Z)\,
  (\operatorname{sym}Z)^{-1/2}\bigr\|_2 .
\end{equation*}
We call $\eta(K_{\!f\!f})$ the non-reciprocity ratio of the free block. It
weighs odd against even stiffness over the whole block; the ratio $|a_b|/k_b$
of a single bond is a different quantity. Compared with the bond's axial force,
the couple $\tau_b$ of an odd bond (Sec.~\ref{sec:setup}) gives
$|\tau_b|/(|k_be_b|L_b)=|a_b|/k_b$ for every drive with $e_b\neq0$: the ratio
of transverse to axial force, the tangent of the angle by which the bond force
tilts off the axis, and no more. $L_b$ cancels from it, and $|k_be_b|L_b$ is
not a torque, since the axial pair is collinear and exerts no moment.
Norms are spectral unless marked $F$, and
$C_P=C[P,P]$ is the shared sub-block of $R$.

\begin{proposition}[the non-reciprocity a target demands]
  \label{prop:ratio}
  Every network with a positive-definite symmetric part that realises a target
  is at least as non-reciprocal as the target's shared block, as measured by
  the ratio $\eta$ of its odd to its even part.
  Let $K_s=\operatorname{sym}K_{\!f\!f}(k,a)$ and
  $K_a=\operatorname{anti}K_{\!f\!f}(k,a)$ be the two parts of the active
  operator whose inverse gives $R$, so that $K_s$ is $K_{\!f\!f}(k,0)$ plus
  the restriction to the free degrees of freedom of
  $\tfrac12\sum_ba_b(\bs_b\bq_b^\T+\bq_b\bs_b^\T)$, and suppose $K_s\succ0$.
  Then $K_{\!f\!f}$ is invertible, $\operatorname{sym}C_P\succ0$, and
  \begin{equation}
    \label{eq:ratio}
    \eta(C_P) \;\le\; \eta(K_{\!f\!f})
    \;=\; \bigl\|K_s^{-1/2}K_aK_s^{-1/2}\bigr\|_2 ,
  \end{equation}
  with equality when $P$ is the whole free set. In particular every such
  network whose shared block equals a target $B$ has
  $\eta(K_{\!f\!f})\ge\eta(B)$. Nothing is linearised in $a$.
\end{proposition}

\begin{proof}
  Set $\widehat K_a=K_s^{-1/2}K_aK_s^{-1/2}$, which is antisymmetric, hence
  normal, with $\|\widehat K_a\|_2=\eta(K_{\!f\!f})$, so that
  $K_{\!f\!f}=K_s^{1/2}(I+\widehat K_a)K_s^{1/2}$; the argument uses the
  congruence normal form of~\citet{London1981}. Since
  $\bx^\T K_{\!f\!f}\bx=\bx^\T K_s\bx>0$ for $\bx\ne0$, $K_{\!f\!f}$ is
  invertible. From
  $(I+\widehat K_a)(I-\widehat K_a)=I+\widehat K_a^\T\widehat K_a$ and the
  commutation of every factor, $(I+\widehat K_a)^{-1}=(I-\widehat K_a)\Sigma$
  with $\Sigma=(I+\widehat K_a^\T\widehat K_a)^{-1}\succ0$ commuting with
  $\widehat K_a$. Congruence by $K_s^{-1/2}$ preserves symmetry and
  antisymmetry, so
  \begin{equation*}
    \operatorname{sym}C=K_s^{-1/2}\Sigma K_s^{-1/2},
    \qquad
    \operatorname{anti}C=-K_s^{-1/2}\widehat K_a\Sigma K_s^{-1/2}.
  \end{equation*}
  Restricting to $P$,
  $\operatorname{sym}C_P=K_s^{-1/2}[P,:]\,\Sigma\,K_s^{-1/2}[:,P]$ and
  $\operatorname{anti}C_P=-K_s^{-1/2}[P,:]\,\widehat K_a\Sigma\,K_s^{-1/2}[:,P]$.
  Since $K_s^{-1/2}$ is symmetric, $K_s^{-1/2}[P,:]$ is the transpose of
  $K_s^{-1/2}[:,P]$, so $\bx^\T(\operatorname{sym}C_P)\,\bx$ is the
  $\Sigma$-form of $K_s^{-1/2}[:,P]\,\bx$. The $p$ columns $K_s^{-1/2}[:,P]$ of
  the invertible matrix $K_s^{-1/2}$ are linearly independent, so that vector
  is non-zero for $\bx\neq0$, and $\Sigma\succ0$ gives
  $\bx^\T(\operatorname{sym}C_P)\,\bx>0$, so $\operatorname{sym}C_P\succ0$.
  For $\bx,\by\in\R^p$, as $\Sigma^{1/2}$ commutes with $\widehat K_a$,
  Cauchy--Schwarz gives
  \begin{equation*}
    |\bx^\T(\operatorname{anti}C_P)\by|
    \;=\; \bigl|(\Sigma^{1/2}K_s^{-1/2}[:,P]\,\bx)^\T\,\widehat K_a\,
      (\Sigma^{1/2}K_s^{-1/2}[:,P]\,\by)\bigr|
    \;\le\; \|\widehat K_a\|_2\,
    (\bx^\T\operatorname{sym}C_P\,\bx)^{1/2}\,
    (\by^\T\operatorname{sym}C_P\,\by)^{1/2},
  \end{equation*}
  and substituting $\bx\mapsto(\operatorname{sym}C_P)^{-1/2}\bx$,
  $\by\mapsto(\operatorname{sym}C_P)^{-1/2}\by$ gives \eqref{eq:ratio}. When
  $P$ is the whole free set,
  $\mathsf{O}=\Sigma^{1/2}K_s^{-1/2}(\operatorname{sym}C)^{-1/2}$
  is orthogonal and
  $(\operatorname{sym}C)^{-1/2}(\operatorname{anti}C)(\operatorname{sym}C)^{-1/2}
  =-\mathsf{O}^\T\widehat K_a\mathsf{O}$, whose norm is $\|\widehat K_a\|_2$. The last clause is
  \eqref{eq:ratio} at $C_P=B$.
\end{proof}

Proposition~\ref{prop:ratio} is the counterpart of Theorem~\ref{thm:floor}:
the floor says what a passive network cannot reach, and \eqref{eq:ratio} says
how strongly non-reciprocal an active network with a positive-definite
symmetric part must be to reach a target $B$.
The algebra is classical. In the language of accretive matrices,
$\eta(Z)=\tan\alpha$, with $\alpha$ the half-angle of the smallest sector of
the complex plane that contains the numerical range of $Z$, and
\eqref{eq:ratio} says only that this sector widens under neither inversion nor
compression to a principal
sub-block~\citep{London1981,LiSze2014,Lin2015}. What \eqref{eq:ratio} adds
here is a reading: its left-hand side is computed from the target, so it is a lower
bound on the non-reciprocity of every network with a positive-definite
symmetric part that realises the target.

Experiment~22 checks \eqref{eq:ratio} on $3000$ abstract matrices with
$K_s\succ0$ and on $3360$ cases inside the odd-bond model. Each abstract draw
carries one choice of $P$, $1872$ of them a proper subset and $1128$ the whole
set, whereas every odd-bond case is read at both choices; the $5232$
proper-subset readings are those $1872$ together with all $3360$
(Appendix~\ref{app:stats}). Over them the largest ratio
$\eta(C_P)/\eta(K_{\!f\!f})$ is $0.9997$. The $1128$ abstract cases with $P$
the whole set are the equality clause, and sit at $1$ to within
$2.64\times10^{-5}$, the round-off of an ill-conditioned evaluation (median
$4\times10^{-15}$; Appendix~\ref{app:stats}).

The $K_s$ in \eqref{eq:ratio} is that of the active operator, and the
symmetric part $K_{\!f\!f}(k,0)$ of the passive network with the same
stiffnesses cannot stand in for it. With this passive substitution the
inequality fails on the whole free set, where \eqref{eq:ratio} is an
equality, in $3066$ of the $3360$ odd-bond cases; only on a small shared set,
where the inequality has slack of its own, does the substitution mostly
survive (Appendix~\ref{app:stats}). The substitution would be safe if the
passive operator were dominated by the
active symmetric part, $K_{\!f\!f}(k,0)\preceq K_s$; it is that order which
fails. The term the substitution drops is
$\mathsf{E}=K_s-K_{\!f\!f}(k,0)$, the restriction to the free degrees of freedom of
$\tfrac12\sum_ba_b(\bs_b\bq_b^\T+\bq_b\bs_b^\T)$. On the whole operator that
term is traceless, since $\bs_b\cdot\bq_b=0$ bond by bond. On the free block,
which is where the comparison with $K_s$ is made, tracelessness need not
survive the restriction, and what experiment~22 measures is the free block
itself: $\lambda_{\min}(\mathsf{E})<0$ in $3360$ of $3360$ cases. So
$\mathsf{E}\not\succeq0$,
that is $K_s$ dominates $K_{\!f\!f}(k,0)$ in none of the $3360$ cases, and no ordering is available
to rescue the substitution at any coupling strength.

At first order the coupling itself can be priced. Write
$\operatorname{vec}_<(Z)\in\R^{\delta_p}$ for the strictly upper entries of an
antisymmetric $p\times p$ matrix $Z$, in the ordering of $P$ fixed before
Theorem~\ref{thm:recovery}, and let $\mathsf{W}$ be the $\delta_p\times n_b$
wedge matrix on which Sec.~\ref{sec:wedges} runs its pivoted $QR$. Its column $b$
is $\operatorname{vec}_<(\bz_b|_P\wedge\bw_b|_P)$, formed at $(k_0,0)$, and
its singular values are
$\sigma_{\max}(\mathsf{W})\ge\dots\ge\sigma_{\min}(\mathsf{W})$. The same
ordering identifies $\mathcal{A}_P$ with the antisymmetric $p\times p$
matrices, and $F(a)$ of Corollary~\ref{cor:submersion} with
$\operatorname{anti}C_P(k_0,a)$. The corollary below follows from
Theorem~\ref{thm:recovery} and a first-order lemma, proved in Step~1 of the
proof of Corollary~\ref{cor:price}.

\begin{corollary}[the least odd coupling, to first order]
  \label{cor:price}
  To first order, the odd coupling a small antisymmetric shared block requires
  is at worst $\sqrt2\|B_a\|_F/\sigma_{\min}(\mathsf{W})$: the larger the
  smallest singular value of the wedges, the weaker it is.
  Let $p\ge2$, fix an admissible $k_0$ with $K_{\!f\!f}(k_0,0)$ invertible and
  suppose $\rank\mathsf{W}=\delta_p=p(p-1)/2$. For an antisymmetric $p\times p$ block
  $B_a$, the odd coupling of least Euclidean norm with
  $\mathrm{d}F_0[a]=B_a$ is
  \begin{equation}
    \label{eq:astar}
    a^{*} \;=\; -2\,\mathsf{W}^{+}\operatorname{vec}_<(B_a),
    \qquad
    \frac{\sqrt2\,\|B_a\|_F}{\sigma_{\max}(\mathsf{W})}
    \;\le\; \|a^{*}\|_2 \;\le\;
    \frac{\sqrt2\,\|B_a\|_F}{\sigma_{\min}(\mathsf{W})},
  \end{equation}
  with $\mathsf{W}^{+}$ the Moore--Penrose pseudoinverse, and each bound is
  attained on the corresponding singular direction. On a support
  $\mathcal{B}$ of $\delta_p$ bonds with $\mathsf{W}_{\mathcal{B}}$ invertible
  the coupling is unique,
  $a_{\mathcal{B}}=-2\mathsf{W}_{\mathcal{B}}^{-1}\operatorname{vec}_<(B_a)$,
  and obeys \eqref{eq:astar} with $\mathsf{W}_{\mathcal{B}}$ in place of
  $\mathsf{W}$. The worst case over $\|B_a\|_F=1$ is therefore
  $\sqrt2/\sigma_{\min}(\mathsf{W}_{\mathcal{B}})$, and minimising it over
  supports of size $\delta_p$ is exactly maximising
  $\sigma_{\min}(\mathsf{W}_{\mathcal{B}})$.
\end{corollary}

\begin{proof}
  \emph{Step~1, the first-order lemma.} Only the antisymmetric part of the odd
  couplings acts at first order, and
  $\operatorname{vec}_<(\mathrm{d}F_0[a])=-\tfrac12\,\mathsf{W}a$. At $a=0$, $C$ is symmetric and $\mathrm{d}C=-C(\mathrm{d}K)C$
  gives $\mathrm{d}C_P=-C[P,:]\,(\mathrm{d}K)\,C[:,P]$. The odd part of
  $\mathrm{d}K$ is $\sum_ba_b\bs_b\bq_b^\T$ on the free degrees of freedom,
  and its antisymmetric part is, by \eqref{eq:edgewedge},
  \begin{equation*}
    K_a \;=\; \tfrac12\sum_ba_b\,\bs_b\wedge\bq_b
    \;=\; -\tfrac12\,(\mathsf{L}_a\otimes\epsilon)_{f\!f},
    \qquad
    \mathsf{L}_a=\sum_{b=\{i,j\}}a_b\,(e_i-e_j)(e_i-e_j)^\T ,
  \end{equation*}
  $\mathsf{L}_a$ being the graph Laplacian with signed bond weights $a_b$. It
  depends on neither bond direction nor bond length. Congruence preserves
  symmetry, so the symmetric part of $\mathrm{d}K$ leaves the antisymmetric
  part of $\mathrm{d}C_P$ alone, and
  \begin{equation*}
    \mathrm{d}F_0[a] \;=\; \operatorname{anti}\mathrm{d}C_P
    \;=\; -\,C[P,:]\,K_a\,C[:,P] .
  \end{equation*}
  This holds wherever $K_{\!f\!f}$ is symmetric, not only at $a=0$; elsewhere
  the symmetric increment also moves $\operatorname{anti}R$ at first order.
  The formula for $K_a$ needs the node-by-node ordering and
  $\hat t_b=\epsilon^\T\hat n_b$ of Remark~\ref{rem:edgewedge}, since the
  other perpendicular flips the sign of $\epsilon$. The construction is
  two-dimensional, $\epsilon$ being the quarter turn in the plane, but only the
  Laplacian form of $K_a$ uses the plane: the first equality of the display is
  $\operatorname{anti}\sum_ba_b\bs_b\bq_b^\T$ written out, and the step this
  corollary goes on to use is the one Theorem~\ref{thm:recovery} proves from
  $C=C^\T$ and nothing else, so neither it nor \eqref{eq:astar} carries a
  hypothesis on the dimension.
  Bond by bond, $C[P,:]\,\bs_b=\bz_b|_P$ and $C[P,:]\,\bq_b=\bw_b|_P$, so
  bond $b$ contributes $-\tfrac12\,\bz_b|_P\wedge\bw_b|_P$, as in the proof of
  Theorem~\ref{thm:recovery}, and
  $\operatorname{vec}_<(\mathrm{d}F_0[a])=-\tfrac12\,\mathsf{W}a$.

  \emph{Step~2, the least-norm coupling.} With $\rank\mathsf{W}=\delta_p$ the system
  $-\tfrac12\mathsf{W}a=\operatorname{vec}_<(B_a)$ is solvable for every
  $B_a$, and its least-norm solution is $a^{*}$. On $\R^{\delta_p}$ the
  singular values of $\mathsf{W}^{+}$ are the reciprocals of those of
  $\mathsf{W}$, so $\|\mathsf{W}^{+}\operatorname{vec}_<(B_a)\|_2$ lies between
  $\|\operatorname{vec}_<(B_a)\|_2/\sigma_{\max}(\mathsf{W})$ and
  $\|\operatorname{vec}_<(B_a)\|_2/\sigma_{\min}(\mathsf{W})$, with the ends
  reached when $\operatorname{vec}_<(B_a)$ lies along the corresponding left
  singular vector. Each strictly upper entry of $B_a$ appears twice in
  $\|B_a\|_F^2$, so $\|\operatorname{vec}_<(B_a)\|_2=\|B_a\|_F/\sqrt2$, and
  multiplying by $2$ gives \eqref{eq:astar}. For the support $\mathcal{B}$ the
  same argument runs with $\mathsf{W}_{\mathcal{B}}^{-1}$. The supremum of
  $\|a_{\mathcal{B}}\|_2$ over $\|B_a\|_F=1$ is
  $\sqrt2/\sigma_{\min}(\mathsf{W}_{\mathcal{B}})$ and is attained, so the
  support with the least worst case is the one with the largest
  $\sigma_{\min}(\mathsf{W}_{\mathcal{B}})$.
\end{proof}

In Sec.~\ref{sec:wedges}, $\sigma_{\min}$ of the selected wedges reports how
well conditioned the selection is. By Corollary~\ref{cor:price} it also sets
the worst-case first-order price of that selection, and choosing the support that
maximises it is a column subset-selection problem. Strong rank-revealing $QR$
solves that problem to within a factor polynomial in the number of
bonds~\citep{GuEisenstat1996}, by making the volume of the selected block
large. Experiment~22 compares the pivoted $QR$ of experiment~18, a heuristic
whose choice need not be the optimum, with the optimum. It finds the optimum by
brute force over all supports of size $\delta_p$ on ten nine-node Delaunay networks at
$p=3$, one per seed, each with its own geometry, stiffnesses and terminal set.
The pivoted $QR$ finds the optimum in $3$ of the ten, and the median of
$\sigma_{\min}(\mathsf{W}_{QR})/\max_{\mathcal{B}}\sigma_{\min}(\mathsf{W}_{\mathcal{B}})$
is $0.93$.

The hypothesis $\rank\mathsf{W}=\delta_p$ is not automatic even when
$n_b\ge\delta_p$: it can fail on non-generic networks. Without it, only the
$B_a$ with $\operatorname{vec}_<(B_a)\in\operatorname{im}\mathsf{W}$ are
reached at first order, and $\sigma_{\min}(\mathsf{W})$ becomes the smallest
non-zero singular value, on that subspace. The equivalence with maximising
$\sigma_{\min}(\mathsf{W}_{\mathcal{B}})$ holds at first order, for
Frobenius-normalised $B_a$ and Euclidean $\|a\|$; Appendix~\ref{app:stats}
gives the origin of the factor $\sqrt2$ and what other normalisations change.

The third part is the exception to the statement of Sec.~\ref{sec:setup} that
an odd network is generically out of moment balance. Let $Q$ be the matrix of
compatibility rows of Appendix~\ref{app:jacobian}, $L_b$ the length of bond
$b$, and $\mathcal{K}=\{a\in\R^{n_b}:(a_bL_b)_b\in\ker Q^\T\}$, the couplings
that exert no net torque by clause~(i) of Proposition~\ref{prop:torquefree}.
Write
$\mathsf{W}\mathcal{K}=\{\mathsf{W}a:a\in\mathcal{K}\}\subseteq\R^{\delta_p}$
for the image of the subspace $\mathcal{K}$ under $\mathsf{W}$. Unlike
$\mathsf{W}_{\mathcal{B}}$, which keeps the columns of $\mathsf{W}$ indexed by
a bond set $\mathcal{B}$, $\mathsf{W}\mathcal{K}$ is a subspace, not a
matrix. The reactions at the three pinned degrees of freedom are external
forces like the applied loads.

\begin{proposition}[a torque-free realisation]
  \label{prop:torquefree}
  Odd couplings whose products $(a_bL_b)_b$ form a state of self-stress realise
  any small enough
  antisymmetric shared block when $\mathsf{W}\mathcal{K}=\R^{\delta_p}$, with no
  net torque on the network as a whole. Each odd bond that changes length
  still exerts its own couple, and the symmetric part of the response is not
  controlled.
  \emph{(i)}~The net moment of all external forces holding a displacement
  $\bu$, applied loads and pin reactions together, is
  $\sum_b\tau_b$ with $\tau_b=a_be_bL_b$, and it vanishes for every $\bu$ if
  and only if $(a_bL_b)_b$ is a state of self-stress, $a\in\mathcal{K}$.
  \emph{(ii)}~Fix an admissible $k_0$ with $K_{\!f\!f}(k_0,0)$ invertible and
  suppose $\mathsf{W}\mathcal{K}=\R^{\delta_p}$, which requires
  $\dim\ker Q^\T\ge\delta_p$ and so fails for isostatic networks when $p\ge2$.
  Then every
  antisymmetric shared block $B_a$ of sufficiently small norm is realised
  exactly in projection, $\Pi_{\mathcal{A}_P}R(k_0,a)=B_a$, by some
  $a\in\mathcal{K}$ near $0$.
  \emph{(iii)}~For such an $a$ realising $B_a\neq0$, each odd bond with
  $a_be_b\neq0$ still exerts its own couple $\tau_b$, $K_a\neq0$, and a closed
  force cycle on $P$ still exchanges work.
\end{proposition}

\begin{proof}
  The term $k_b\bq_b\bq_b^\T$ puts collinear forces on the ends of bond $b$
  and contributes no moment. The term $a_b\bs_b\bq_b^\T$ puts
  $\mp a_be_b\hat t_b$ at nodes $i$ and $j$, separated by $L_b\hat n_b$; with
  $\hat t_b=\epsilon^\T\hat n_b$ their moment is $a_be_bL_b$. The external
  force $K\bu$ is the sum of these bond contributions, so its net moment is
  $\sum_ba_bL_b\,\bq_b^\T\bu$, the pairing of $(a_bL_b)_b$ with $Q\bu$, and it
  vanishes for every free $\bu$ exactly when $Q^\T(a_bL_b)_b=0$. The three
  pins are statically determinate, since a rigid motion vanishing on them
  would be a null vector of the invertible $K_{\!f\!f}$. Hence $\ker Q^\T$ is
  the same whether or not the pinned coordinates are kept, and the self-stress
  may be taken in either framework. This is (i).

  For (ii), $\mathcal{K}$ is a linear subspace, and by the proof of
  Corollary~\ref{cor:price} the differential at $0$ of $F$ restricted to
  $\mathcal{K}$ has image $-\tfrac12\mathsf{W}\mathcal{K}$ under
  $\operatorname{vec}_<$, which is all of $\R^{\delta_p}$. The local
  submersion theorem, applied as in the proof of
  Corollary~\ref{cor:submersion} on a slice of $\mathcal{K}$, gives the
  realisation; the realising $a$ lies in $\mathcal{K}$, so by (i) the net
  moment vanishes identically, not only to first order. Since
  $\dim\mathsf{W}\mathcal{K}\le\dim\mathcal{K}=\dim\ker Q^\T$, the hypothesis
  needs $\dim\ker Q^\T\ge\delta_p$, and an isostatic network has no
  self-stress, so it fails the hypothesis when $p\ge2$.

  For (iii), $\tau_b=a_be_bL_b$ is non-zero whenever $a_be_b\neq0$. If
  $K_{\!f\!f}$ were symmetric, $C$ and $C_P$ would be too, against
  $B_a\neq0$; so $K_a=-\tfrac12(\mathsf{L}_a\otimes\epsilon)_{f\!f}\neq0$. By
  \eqref{eq:cycle} below, a cycle with $\mathbf{h}_1^\T B_a\mathbf{h}_2\neq0$,
  which exists when $B_a\neq0$, exchanges non-zero work.
\end{proof}

Experiment~22 runs the construction on $336$ configurations. The
torque-free couplings span the target space, $\mathsf{W}\mathcal{K}=\R^{\delta_p}$,
in $316$ of them, and every configuration where they do not has
$\dim\ker Q^\T<\delta_p$. Clause (ii) is local, and spanning and solving are
counted apart: at target sizes $\varepsilon=10^{-4}$, $10^{-3}$ and $10^{-2}$,
the solve inside $\mathcal{K}$ reaches the antisymmetric target in $316$, $314$
and $314$ of the spanning configurations. At $\varepsilon=10^{-4}$ the
torque-free coupling that realises the target exactly costs a median $4.73$
times the $\|a^{*}\|_2$ of Corollary~\ref{cor:price}, which is a first-order
quantity and is not itself required to realise the target.

Odd elasticity without a net torque density has a precedent: the odd shear
modulus carries none, and a honeycomb lattice whose two bond families are
tuned against each other has zero net torque density for every linear
deformation while each spring still exerts its own
couple~\citep[Fig.~S1]{Scheibner2020}. What is added here is that on an
arbitrary graph the condition is linear and reads $(a_bL_b)_b\in\ker Q^\T$, a
state of self-stress. The couples $\tau_b$ cancel only in sum, so the
actuators still exchange angular momentum through the frame that carries them.

For force cycles supported on the terminals $P$, the work exchanged around a
closed cycle is fixed by the target. That the work per cycle is set by the odd
modulus and the enclosed area is known, in continuum
form~\citep{Scheibner2020} and for an experimental unit
cell~\citep{Chen2021}; what follows is its form in terms of the response
block of a discrete network. For loads
$\mathbf{h}=c_1\mathbf{h}_1+c_2\mathbf{h}_2$ on $P$, with
$\mathbf{h}_1,\mathbf{h}_2\in\R^p$ and $(c_1,c_2)$ tracing a closed loop that
encloses the signed area $\Gamma$,
\begin{equation}
  \label{eq:cycle}
  \oint\mathbf{h}^\T C_P\,\mathrm{d}\mathbf{h}
  \;=\; 2\Gamma\,\mathbf{h}_1^\T(\operatorname{anti}C_P)\,\mathbf{h}_2 ,
\end{equation}
since $\operatorname{sym}C_P$ integrates to zero around a closed loop and
$\oint(c_1\mathrm{d}c_2-c_2\mathrm{d}c_1)=2\Gamma$. Nothing is linearised in
$a$. Once the target is met, $\operatorname{anti}C_P=B_a$, and no choice of
the odd couplings reduces the work. Outside $P$ the target no longer fixes the
work: loads there bring in entries of $C[S,S]$ that $R$ does not contain. The
other conditions on \eqref{eq:cycle} are listed in Sec.~\ref{sec:limitations}.

The directions Theorem~\ref{thm:recovery} restores come at a price that the target
alone bounds from below: every realising network with a
positive-definite symmetric part has $\eta(K_{\!f\!f})\ge\eta(B)$, and for
force cycles on $P$ the work per cycle is fixed. In experiment~22 the ratio
$\eta(C_P)/\eta(K_{\!f\!f})$ is at most $0.9997$ over all $5232$
proper-subset readings, and at $\varepsilon=10^{-4}$ a torque-free realisation
costs a median $4.73$ times the least first-order coupling.


\section{Numerical verification}
\label{sec:numerics}

We test four things separately, so that a failure would identify which claim is
wrong: the mechanism (are the columns symmetric on the shared block?), the
subspace (is what non-reciprocity adds the antisymmetric complement, or merely
the right \emph{number} of dimensions?), the bound (is it ever violated?), and
attainment (when is it equality?).
We then repeat the tests in three dimensions, with angular springs, about a
deformed state and at finite frequency, and close with a control that
separates symmetry from a generic rank shortfall.

All ranks come from the singular-value spectrum of the exact Jacobian
\eqref{eq:jaccol} with a relative cutoff $10^{-9}$, with two exceptions. Both
take the Jacobian by central finite differences, whose noise floor is set by
truncation error in place of machine epsilon: the large-deformation sweep cuts
at $10^{-7}$, and the second implementation, experiment~18b, at $10^{-6}$
(Appendix~\ref{app:stats} gives the gap distributions and what happens at
$10^{-9}$).
On the exact Jacobian, Fig.~\ref{fig:robust}a shows the gap in which the
$10^{-9}$ cutoff falls for one $30$-node network at full overlap, $m=6$: the
passive spectrum drops by thirteen orders of magnitude after index
$m^2-p(p-1)/2=21$, and with odd couplings it continues to $m^2=36$.

In the main rank sweep, networks are random Delaunay
triangulations of $12$--$56$ points and
triangular lattices with positional disorder $0$, $0.05$ and $0.15$ lattice
units. Stiffnesses are drawn log-normally, $k_b=e^{\xi_b}$ with
$\xi_b\sim\mathcal{N}(0,0.4^2)$, and odd couplings as
$a_b\sim\mathcal{N}(0,0.3^2)$; Appendix~\ref{app:params} notes the experiments
that depart from these draws. There is no material scale to normalise against:
$\xi_b$ has mean zero, so the median stiffness $k=1$ \emph{defines} the unit of
stiffness, and $a_b$, in the same unit, has spread $0.3$ of it. Nothing
reported depends on this choice: ranks are invariant under a global rescaling
$(k,a)\mapsto c(k,a)$, which sends $C\mapsto C/c$ and rescales every Jacobian column by
$c^{-2}$, and every error and floor is a ratio of two quantities in the same
unit. The spread $0.4$ does matter: it sets the stiffness disorder at a factor
$e^{0.4}=1.49$ per standard deviation.

In the main rank sweep (experiment~03)
seeds depend on every loop variable, so no two of its cases share a random
draw. The control sweeps
deliberately share index sets and starting draws between the arms they compare
(the two arms of experiment~01, experiments~04 and~13, and experiments~11
and~12 at fixed network). The counts from the main rank sweep, in the text and
in Fig.~\ref{fig:robust}, are medians over three independent stiffness draws
per configuration; for attainment, a yes-or-no outcome, the median is a
majority vote, and the convention matters for reading the counts.
Appendix~\ref{app:stats} gives the per-draw figures; the bound is violated in
zero cases under either convention.

\paragraph{The mechanism.}
For $120$ network configurations of experiment~03, each evaluated under both
bond laws, we formed
every Jacobian column, extracted its $p\times p$ sub-block $B$ on the shared
degrees of freedom and measured $\|B-B^\T\|_F/\|B\|_F$. Across all $6579$
passive columns the median is $1.4\times10^{-15}$ and the largest
$3.3\times10^{-12}$: machine zero, as Theorem~\ref{thm:main} requires, and
\emph{not one} passive column exceeds the $10^{-9}$ rank cutoff.
Across the $13158$ columns of the same networks with odd couplings the ratio has
median $0.755$ and maximum $1.414$. Both features are predicted. For a rank-one
block $B=\bx\by^\T$,
$\|B-B^\T\|_F/\|B\|_F=\sqrt{2}\,\lvert\sin\angle(\bx,\by)\rvert$, so
$\sqrt{2}$ is the ceiling, reached when the two factors are orthogonal, and the
odd population must spread below it as the angle varies, as
Fig.~\ref{fig:robust}d shows. The smallest odd value anywhere is
$1.8\times10^{-4}$, so even in the worst case the two populations are separated
by nearly eight orders of magnitude.

\paragraph{Which subspace non-reciprocity buys.}
Theorem~\ref{thm:recovery} makes this a direct measurement rather than an
inference. For each configuration we form the wedges $\bz_b|_P\wedge\bw_b|_P$ at
$a=0$ and compare their rank with the dimension of the projection of the full
Jacobian image onto $\mathcal{A}_P$; \eqref{eq:wedge} says the two are equal.
They agree in all $1008$ rows of experiment~18, which cover Delaunay networks
and disordered triangular lattices at $m=3$ to $6$ and overlaps
$p\in\{0,1,\lfloor m/2\rfloor,m\}$ ($945$ distinct configurations;
Appendix~\ref{app:stats}).

At full overlap the wedges span the whole of $\Lambda^2\R^p$ in $251$ of $252$
configurations. The exception is a $5\%$-disorder lattice configuration inside
the room criterion \eqref{eq:criterion} stated below, of the near-lattice kind
discussed in Appendix~\ref{app:stats}: its passive rank is $11$ against $d=15$,
its wedge rank is $8$ against $10$, and \eqref{eq:wedge} holds there too. So in
this construction non-reciprocity buys the antisymmetric complement and nothing
beyond it. In the $251$ configurations whose wedges span, keeping odd couplings
on only the $p(p-1)/2$ bonds chosen by a pivoted $QR$ of the unnormalised wedge
matrix, and leaving every other bond passive, spans $\mathcal{A}_P$ to first
order in $251$ of $251$, as Theorem~\ref{thm:recovery} requires. In this sweep
it also gives the full rank $m^2$ in $251$ of $251$, which is a statement about room rather
than about symmetry. Those $251$ are exactly the configurations in which the odd
Jacobian has full rank $m^2$, the hypothesis of the second clause of
Corollary~\ref{cor:submersion}, so the corollary applies to each.

The measurement is not confined to two dimensions. Repeated on
three-dimensional bond networks in the second implementation, experiment~18b,
with a pinning of its own (one node fully, two coordinates of a second, one of
a third) and the finite-difference Jacobian and $10^{-6}$ cutoff given above,
applied to the column-normalised projection, the two ranks agree in $60$ of
$60$ configurations.

The measurement also separates two effects that the rank gain conflates. At
$p=0$, where $\mathcal{A}_P=\{0\}$ ($252$ configurations), the odd network
still gains rank over the passive one, by up to $13$ dimensions. That gain comes from parameters, not symmetry: odd
couplings double $n_\theta$, and a room-limited network spends the extra
columns inside $\mathcal{S}_P$.

A second argument reaches the same conclusion from two rank equalities alone,
without Theorem~\ref{thm:recovery}; it counts dimensions where the wedge test
measures a projection.
In $117$ of the $120$ configurations of experiment~03 the passive image has
$\dim=\dim\mathcal{S}_P$. By Theorem~\ref{thm:main} that image lies inside
$\mathcal{S}_P$, and a subspace of $\mathcal{S}_P$ of the same dimension is
$\mathcal{S}_P$. In all $117$ the odd network reaches the full $m_Tm_S$, so the
part of the odd image orthogonal to the passive image is $\mathcal{A}_P$. What
\eqref{eq:wedge} adds is that the conclusion no longer needs either equality,
and so covers the room-limited configurations that argument had to exclude.

\paragraph{Exact realisation (experiment~19).}
On experiment~18's Delaunay geometries, re-run at every overlap from $2$ to $m$
and at two stiffness draws ($336$ configurations), Newton's method from $a=0$
at fixed $k=k_0$ realises a prescribed antisymmetric shared block of relative
size $\varepsilon$, from $10^{-4}$ to $10^{-1}$, in $336$ of $336$
configurations at every $\varepsilon$, with $\|a\|/\|k_0\|\approx2\varepsilon$
in the median. That count measures one thing only: the residual of
$\Pi_{\mathcal{A}_P}R$ against the target, below $2\times10^{-12}$ relative in
every solve.

That count says nothing about the rest of the block, and the rest of the block
moves.
On the same solves $\Pi_{\mathcal{S}_P}R$, the part of the response that
reciprocity already allowed and that nothing asked to change, is displaced by
$2.9\,\|A\|$ in the median and by up to $549\,\|A\|$, with $A$ the prescribed
block.

Positive definiteness goes as well. With every odd coupling free, the
symmetric part of $K_{\!f\!f}$ loses it in $115$ of the $336$ at
$\varepsilon=10^{-1}$, in one at $10^{-2}$ and in none below. At the largest
target, then, more than a third of these solves sit outside the neighbourhood
of $a=0$ on which Remark~\ref{rem:submersion} can promise anything about
stability. The sparse selections below lose it far more often, and
Remark~\ref{rem:submersion} gives their counts.

Where the passive rank attains $d$ ($327$ configurations), the joint solve in
$(k,a)$ reaches a prescribed full block of relative size $10^{-3}$ in $327$ of
$327$, with every $k_b>0$ and $\operatorname{sym}K_{\!f\!f}\succ0$, and in
$313$ of $327$ at $10^{-2}$. Of the $14$ misses, twelve drive a stiffness below
$10^{-4}$ (one of these also runs to the iteration cap), one other runs to the
iteration cap, and one stalls after five damped steps; in all $14$ the
positivity guard intercepts a Newton step that would cross $k_b=0$. That is a limit of the solver
and not a statement about existence.

With odd couplings free on only the $p(p-1)/2$ bonds of the pivoted-$QR$
selection, the target is met in $336$, $335$, $312$ and $210$ of $336$
configurations at $\varepsilon=10^{-4}$, $10^{-3}$, $10^{-2}$ and $10^{-1}$.
The rank-greedy selection, which certifies independence but not conditioning,
manages $287$, $218$, $127$ and $75$.

\begin{figure}[t]
  \centering
  \includegraphics[width=\textwidth]{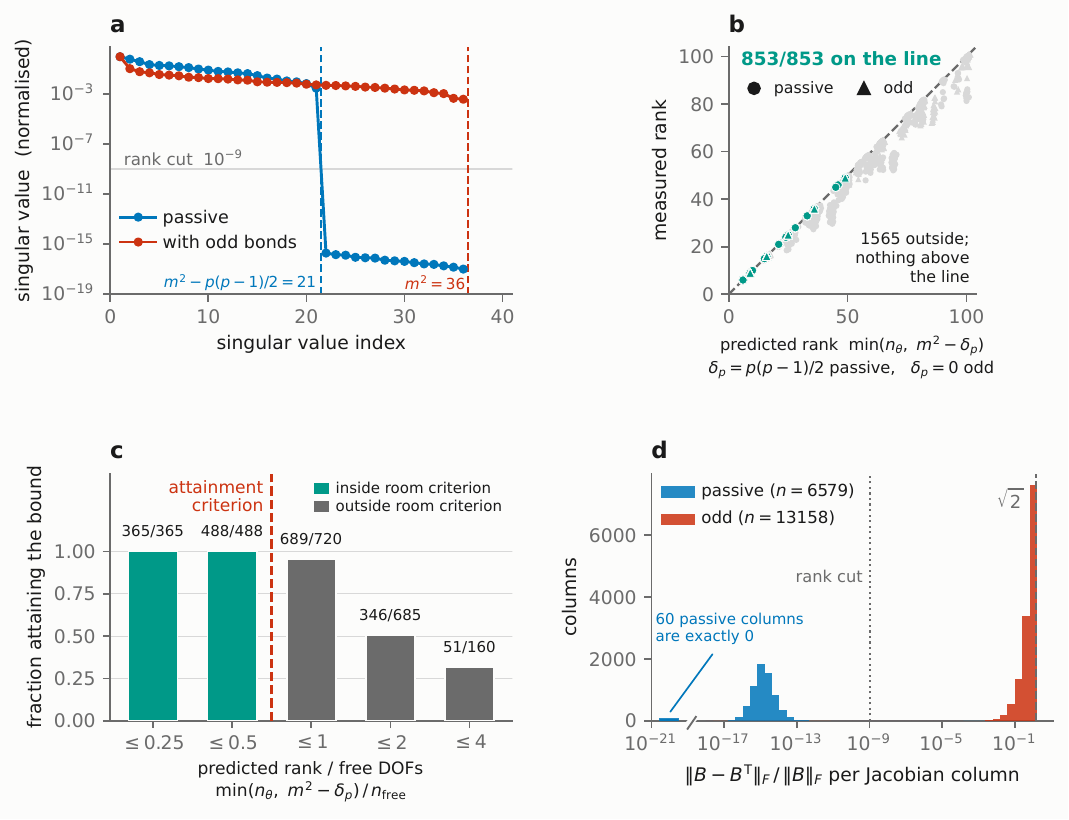}
  \caption{\textbf{Spectra, ranks and the column mechanism.}
  \textbf{(a)}~Singular values of the exact Jacobian, normalised by the
  largest, for one $30$-node network at full overlap, $m=6$, passive and with
  odd couplings. The dashed lines mark the indices $m^2-p(p-1)/2=21$ and
  $m^2=36$; the horizontal line is the rank cutoff $10^{-9}$.
  \textbf{(b)}~Measured against predicted rank over all $2418$ rank cases;
  circles are passive and triangles the odd control. Cases inside
  criterion~\eqref{eq:criterion} are coloured. Cases outside it are grey and
  jittered, with the jitter kept below the dashed diagonal.
  \textbf{(c)}~Fraction attaining the bound over the same $2418$ cases, binned
  by predicted rank per free degree of freedom. Each bar carries its count of
  attaining over total cases, and the dashed line separates the bins at or
  below $\tfrac12$. The legend's ``room criterion'' is
  criterion~\eqref{eq:criterion}.
  \textbf{(d)}~Distribution of $\|B-B^\T\|_F/\|B\|_F$, $B$ the shared block of
  one Jacobian column, over $120$ configurations ($6579$ passive and $13158$
  odd columns). The dotted line is the rank cutoff and the dashed line the
  ceiling $\sqrt{2}$; columns with a ratio of exactly zero are drawn left of
  the axis break.}
  \label{fig:robust}
\end{figure}

\paragraph{The bound.}
Over $2418$ cases spanning both network families, $m$ from $3$ to $10$, overlaps
$0$, $1$, $\lfloor m/2\rfloor$ and $m$, and both bond laws, including networks far too small for the
target and networks comfortably large, the measured rank exceeded its bound,
\eqref{eq:bound} for the passive law and the count of
Proposition~\ref{prop:odd} for the odd one, in \emph{zero} cases
(Fig.~\ref{fig:robust}b).

\paragraph{The law at fixed room.}
Where a network has enough parameters for the target, the measured deficit
equals $p(p-1)/2$ in $216$ of $216$ cases (Fig.~\ref{fig:law}; $12$ network
builds, each with its own random draw, at every overlap $p=0,\dots,m$ for
$m=4,5,6$, which is $5+6+7=18$ settings per build), with no fitted quantity.
It is independent of $m$: the three curves for $m=4,5,6$ collapse onto one, as
predicted. Of the $216$, the
$144$ with $p\ge2$ test a non-zero prediction; the other $72$ verify that the
deficit is zero when it should be. The non-reciprocal control reaches the full
$m^2$ in $216$ of $216$ cases at every overlap.

These two sets of $216$ are the reciprocal and non-reciprocal halves of the
$432$-case overlap sweep. Experiment~02 deposits $540$ cases in all, the
further $108$ being the large-$m$, full-overlap arm on disordered lattices from
which the $(m-1)/(2m)$ comparison is drawn. At full overlap the forfeited
\emph{dimension} fraction matches $(m-1)/(2m)$ at each of $m=4,6,8,10$: $0.375$,
$0.417$, $0.438$ and $0.450$. With odd
couplings the median forfeited fraction is zero at each of these $m$.

The $(m-1)/(2m)$ comparison is restricted to the cases with $n_\theta\ge m^2$,
where the rank is set by the symmetry branch and the bond count does not bind.
The restriction removes real data: the reciprocal $m=12$ and $m=14$ arms of the
sweep entirely, and three of the nine reciprocal networks at $m=10$. No
excluded case violates the bound or loses less than predicted, and where excluded
cases lose more, the losses are not small; Appendix~\ref{app:stats} reports what they measure.

\begin{figure}[t]
  \centering
  \includegraphics[width=\textwidth]{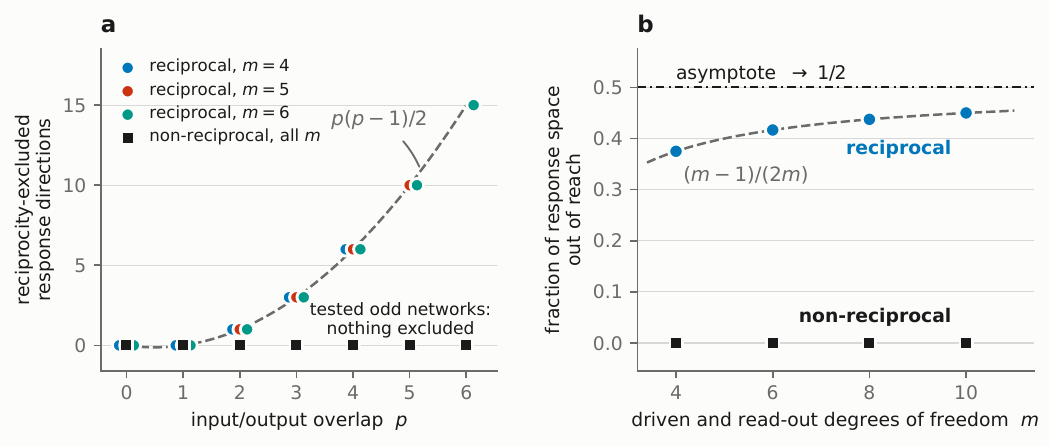}
  \caption{\textbf{The overlap law.}
  \textbf{(a)}~Lost response dimensions against the input/output overlap $p$.
  Coloured points are the reciprocal medians over $12$ network builds for
  $m=4,5,6$, dodged horizontally; black squares are the medians with odd
  couplings, pooled over the values of $m$ that admit that overlap. Each arm
  has $216$ cases. The dashed curve is $p(p-1)/2$.
  \textbf{(b)}~Fraction of the response space out of reach at full overlap,
  against $m$, reciprocal and with odd couplings. Each point is the median over
  the networks with $n_\theta\ge m^{2}$: nine, except six on the reciprocal arm
  at $m=10$; no reciprocal network at $m=12$ or $14$ meets the condition
  (Sec.~\ref{sec:numerics}). The dashed curve is
  $(m-1)/(2m)$ and the dash-dotted line its asymptote $\tfrac12$.}
  \label{fig:law}
\end{figure}

\paragraph{Attainment, and where it stops.}
Equality is not universal. On a given graph, pinning and layout it holds at
almost every configuration or at none, since one attaining configuration
implies attainment at almost every configuration
(Proposition~\ref{prop:generic}). Write $d=m_Tm_S-p(p-1)/2$ for the dimension
of $\mathcal{S}_P$, the symmetry branch of the ceiling of Theorem~\ref{thm:main}. On the non-reciprocal control arm there is no deficit to lose, and
$d$ is read as $m_Tm_S$; this convention holds throughout the counts below and
in the deposited table. On the square layouts of this sweep,
$d=m^{2}-\delta_p$, with $\delta_p=p(p-1)/2$ passive and $\delta_p=0$ odd,
the convention of the axis labels of Fig.~\ref{fig:robust} and of this sweep
only; elsewhere $\delta_p$ is always $p(p-1)/2$.

The bound is attained in $853$ of $853$ cases satisfying
\begin{equation}
  \label{eq:criterion}
  d \;\le\; \tfrac{1}{2}\,n_{\mathrm{free}}
  \qquad\text{and}\qquad
  d \;\le\; \tfrac{1}{2}\,n_\theta ,
\end{equation}
and the fraction attaining it then falls smoothly. Binned by the predicted rank
per free degree of freedom, $\min(n_\theta,\,m^{2}-\delta_p)/n_{\mathrm{free}}$,
the abscissa of Fig.~\ref{fig:robust}c, it is $365/365$, $488/488$,
$689/720$, $346/685$ and $51/160$ in the bins that end at
$0.25$, $0.5$, $1$, $2$ and $4$. On this sweep the bars at or below
$\tfrac12$ hold exactly the $853$ cases that satisfy \eqref{eq:criterion}
(Appendix~\ref{app:stats}).

This fall-off is consistent with
\eqref{eq:rankid}: when $d$ is small the projection $\pi_{T,S}$ sees
only a small window of $\Sym$, through which the structured subspace $CVC$ is
indistinguishable from a generic one; as $d$ grows towards the size of the whole
space, the structure of $V$ shows through and costs additional rank. Every such
case is a further \emph{loss}, never a violation.
The counts above are conditioned on an empirical criterion, not a proven one.
Criterion~\eqref{eq:criterion} is read off this sweep and needs a genericity
clause, which we do not have; what is \emph{proved} about attainment is the
certified family of Sec.~\ref{sec:certificates} and nothing wider.

Without a genericity clause the criterion is false. On a perfect square
lattice with diagonals and uniform stiffnesses (integer coordinates, so the
rank is exact over $\mathbb{Q}$ with no floating-point cutoff at all),
attainment fails \emph{inside} the criterion: at $m=4$, $p=3$, $d=13$ the exact
rank is $11$. Appendix~\ref{app:stats} gives a second case, at $m=3$, $p=2$,
and a sweep of random index sets, and experiment~18 a near-lattice third: the
$5\%$-disorder configuration above, at passive rank $11$ against $d=15$.

That clause cannot be about positions and stiffnesses alone. Attainment can fail
inside \eqref{eq:criterion} at configurations generic in both, for a reason that
is local and provable. Let $v$ be a vertex of degree two whose two bonds
have independent directions $\hat n_1,\hat n_2$, and let one of its two
neighbours, $w$, be instrumented. In any response to a load applied away from
$v$, equilibrium at $v$ reads
$k_1(\hat n_1\cdot\Delta_1)\hat n_1+k_2(\hat n_2\cdot\Delta_2)\hat n_2=0$, with
$\Delta_i$ the relative displacement across bond $i$. Because the two directions
are independent, both extensions vanish, so $\bu_v$ is a fixed geometric
function of its neighbours' displacements, the same for every stiffness vector.
Take $T=S$ to be the two degrees of freedom of $v$ together with the two of $w$.
The two columns loaded at $w$ then give one identity each. A third comes from
the columns loaded at $v$ itself: a load at $v$ along $\hat n_2$ is carried by
the second bond alone, so the first is again unstrained, which is one relation
between those two columns with coefficients fixed by the geometry. All three
hold identically in $k$, so the set reachable by varying $k$ lies in a
subspace of dimension $m(m+1)/2-3$, however large and however densely bonded the rest of the network
is.

Hanging such a vertex on a complete graph $K_r$ gives, at $m=4$ and
$d=10$, exact rank $7$ over $\mathbb{Q}$ at $r=11,15,19$, where $n_b=57,107,173$
and $n_{\mathrm{free}}=21,29,37$ and both halves of \eqref{eq:criterion} hold by
a margin that grows with $r$ (a one-off exact computation, outside the
deposited sweeps). What
\eqref{eq:criterion} is missing is a condition on the graph and the index sets
together; the counts above describe the ensembles sampled here.

We also expected perfect triangular lattices, being non-generic, to lose extra
rank. The counts lean that way, $72.6\%$ attained at zero disorder against $78.2\%$
over the two non-zero disorders, but the gap of $5.6$ percentage points is not
established (two-sided observed significance level $0.10$;
Appendix~\ref{app:stats}); room remains the main determinant of attainment.

\paragraph{The deficit is the intersection of Theorem~\ref{thm:duality}.}
An identity that only reproduced the answer where the bound is attained would
carry no information, since there both sides are forced. We therefore computed
both sides
of \eqref{eq:duality} directly, assembling $\mathcal{Y}$ and $KV^{\perp}K$ as
explicit subspaces of $\Sym(n_{\mathrm{free}})$ and taking their intersection
from the rank of the stacked bases, over $336$ configurations of which $40$
do \emph{not} attain the bound. The identity holds in $336/336$, including
$40/40$ of the non-attained cases, where the excess intersection
$\dim(\mathcal{Y}\cap KV^{\perp}K)-\max(0,d-n_b)$ runs from $1$ to $7$. The
sweep includes $144$ cases with $m_T\neq m_S$ and $112$ with $p=0$, both
$100\%$. In all $336$ the independence hypothesis $\dim V=n_b$ of
Remark~\ref{cor:generic} holds, so it is not the weak link; general position
is.

\paragraph{The bound in three dimensions and with angular springs.}
The proof of Theorem~\ref{thm:main} uses only that $K$ is symmetric at every
$k$, hence that $\mathrm{d}K$ is symmetric,
so it should survive a change of spatial dimension and a change of element, and
we checked both. For a bond, and for any element with one generalised strain
per stiffness, the linearised strain is a linear functional
$s_r=\bq_r\cdot\bu$ of the displacement, the energy is
$\sum_r(\theta_r/2)s_r^{2}$, and $\partial K/\partial\theta_r=\bq_r\bq_r^\T$,
symmetric and rank one. An angular spring at a hinge, whose $s_r$ is the change
of the angle between two bonds meeting there, qualifies just as a central-force
bond does. Over $432$ cases in three settings (two-dimensional bonds,
three-dimensional bonds on Delaunay-tetrahedralised points, and two-dimensional
bonds with angular springs added), the bound was violated in \emph{zero} cases
and was attained in $102/102$, $126/126$ and $102/102$ of the cases inside
criterion~\eqref{eq:criterion}. The shared block $B$ of the response, on whose
symmetry the error floor rests, stayed symmetric to within
$7.3\times10^{-16}$, $3.0\times10^{-15}$ and $1.9\times10^{-14}$ respectively,
measured as $\|B-B^\T\|_F/\|B\|_F$ and maximised over the cases with $p\ge2$.

\paragraph{The bound survives in the tangent response about a deformed state.}
Everything above is the linear response about an undeformed reference state,
whereas physical-learning experiments operate at large strain. For a
hyperelastic network the tangent stiffness at \emph{any} configuration is the
Hessian of the energy,
\begin{equation*}
  H \;=\; \sum_b\Bigl[\, k_b\,\bq_b\bq_b^\T
      \;+\; \frac{k_b(L_b'-L_b)}{L_b'}\,\bs_b\bs_b^\T \Bigr],
\end{equation*}
with $L_b'$ the current bond length, $L_b$ its rest length and $\bs_b$ the
transverse companion of $\bq_b$ (Table~\ref{tab:notation}), carrying the
in-plane unit perpendicular at each end of the bond. In
$\delta$ dimensions $\bs_b\bs_b^\T$ is replaced by the sum of $\delta-1$ such
terms over an orthonormal frame transverse to the bond. The two terms are the
material and the geometric contribution, both symmetric, so $H$ takes the place
of $K$ and the tangent compliance is again symmetric. Symmetry of
$H$ at every parameter value is the \emph{only} property Theorems~\ref{thm:main}
and~\ref{thm:floor} use. They therefore apply to the tangent response about any
pre-loaded state, which is what a small-signal measurement on a deformed sample
returns.

We checked this where it could fail. The nonlinear equilibrium is solved by
Newton with load stepping, and $\partial R/\partial k$ is taken by central
finite differences, so that the change of the equilibrium \emph{configuration}
with $k$ is included. Over $59$ cases at bond strains
up to $60.8\%$, the shared block of the tangent response stays symmetric to
$5\times10^{-16}$, and, at the $10^{-7}$ cutoff of this sweep, the bound is violated
in \emph{zero} cases and attained in all of them. All $59$ lie inside
criterion~\eqref{eq:criterion}, so the sweep is not evidence about attainment
outside it. The
rank is identical across three finite-difference step sizes in every case, so it
is not an artefact of the differencing. At the largest load the median peak
bond strain is $22\%$. Of $60$ planned cases, $59$ are reported; the one
dropped failed to converge in the perturbed finite-difference equilibrium
(Appendix~\ref{app:stats}).

That cutoff is set by the finite-difference noise floor, and the choice
determines the outcome: at the $10^{-9}$ of the
closed-form sweeps the measured ranks exceed the bound in $21$ of the $59$
cases, and in none of the $59$ do they agree across the three step sizes, while
at $10^{-7}$ they agree in all $59$. That is a statement about differencing and
not about Theorem~\ref{thm:main}. The spectrum has a gap where $10^{-7}$ is
placed; Appendix~\ref{app:stats} gives its size and names the deposited
columns that carry it.

\paragraph{The finite map inherits the symmetry, on a branch.}
Hold the loaded degrees of freedom at $\bu[S]$,
let every other one equilibrate, and write $\mathcal{E}^{*}(\bu[S])$ for the
stored energy in that state. On any branch along which the equilibrium is
single-valued and moves smoothly with $\bu[S]$, $\mathcal{E}^{*}$ is one
differentiable function and the envelope identity gives
$\mathbf{f}[S]=\partial\mathcal{E}^{*}/\partial\bu[S]$: the finite port map is
a gradient map. Its Jacobian is then symmetric at every amplitude. This is the
Hessian above read at finite deformation, not a second argument. That Jacobian
is the Schur complement of $H$ onto $S$. Where it is invertible, its inverse is
the tangent compliance block $C[S,S]$, of which the shared block is a principal
submatrix; it stops being invertible at a limit load, where a load-controlled
path snaps through.
Integrated, the same identity says that no closed port path inside the branch
does net work, so the state reached depends on where a load path ends and not
on how it got there. In that sense the constraint is not an artefact of
linearisation, though $p(p-1)/2$ counts tangent dimensions throughout and we do
not turn the gradient property into a count for the finite map.

What the argument uses, and Theorem~\ref{thm:main} does not, is that the
network is conservative along the whole branch; a tangent operator that happens
to be symmetric at one configuration is not enough. It also uses
single-valuedness, and invertibility of the Hessian block on the equilibrating
degrees of freedom, which is what makes the branch smooth and is a condition
distinct from the one that fails at a limit load. Where that Hessian block
becomes singular the branch is lost and $\mathcal{E}^{*}$ stops being one
function of $\bu[S]$, so the
symmetry is inherited branch by branch and not across the jump. On a branch
the finite map does not escape the symmetry; what remains open is what a path
gains by crossing branches.

\paragraph{The bound survives at finite frequency.}
A material driven periodically responds through the dynamic
Green function $G(\omega)=\bigl[K-\omega^{2}M+i\omega D\bigr]^{-1}$ instead of
through $C$, with $M$
the mass and $D$ the damping. If $M$ and $D$ are symmetric (lumped or
consistent mass, and viscous damping, including the Rayleigh form $\alpha_R
K+\beta_R M$), the bracket is \emph{complex symmetric}, so $G=G^{\T}$ (the
transpose, not the conjugate transpose). Theorem~\ref{thm:universal}
(Sec.~\ref{sec:universal}) then applies over $\mathbb{C}$, and the bound holds
at every frequency with the dimensions counted over $\mathbb{C}$. Each column
is again a symmetric rank-one outer product. For damping independent of the
stiffnesses, $\partial G/\partial k_b=-G\bq_b\bq_b^\T G$. For the Rayleigh
damping used in the sweep below, $D=\alpha_R K+\beta_R M$ depends on $k_b$
through $K$, and the derivative acquires a scalar factor,
$\partial G/\partial k_b=-(1+i\alpha_R\omega)\,G\bq_b\bq_b^\T G$, which changes
neither the rank nor the symmetry of the column.

The sweep covers $216$ cases: three network sizes, two geometries each,
$m=3,4,5$, at $p=0$ and at full overlap, and six frequencies from $\omega=0$ to $\omega=2$ in
the reduced units of Appendix~\ref{app:params}, in which the median bond
stiffness and the mean nodal mass are both $1$, so that $\omega$ is measured in
units of the square root of their ratio. Of these, $180$ have $\omega>0$. The bound
is violated in \emph{zero} cases and attained in $216$ of $216$, and at full
overlap the deficit equals $p(p-1)/2$ in all $108$. In those $108$
full-overlap cases, the only ones in which a shared sub-block exists, it stays
symmetric to a median $4\times10^{-16}$ and at worst $4.4\times10^{-15}$. (At
$\omega=0$ the damping drops out and $G$ is the static compliance, so the $36$
cases at $\omega=0$ repeat the static problem; excluding the $18$ of them at
full overlap moves the median to $5\times10^{-16}$ and leaves the maximum
unchanged.)

The control is the frequency-domain analogue of the odd coupling: a
circulatory or feedback term makes $D$ non-symmetric. The shared block $B$ is then
asymmetric, with $\|B-B^\T\|_F/\|B\|_F$ at a median $0.46$, and the image reaches the full $m_Tm_S$ in all
$18$ cases of that control, a sweep of its own (three networks at each of
$m=3,4,5$ and $\omega=0.2,1$), not a sub-count of the $216$.

At $\omega>0$ the stiffnesses are real while $G$ is complex, so the image of the
parameter-to-response map is a \emph{real} subspace, and its real dimension is
not fixed by the complex rank: a complex rank of $r$ permits any real rank from
$r$ to $2r$ (the $1\times2$ Jacobians $[1\ \ i]$ and $[1\ \ 2]$ both have
complex rank $1$, with real ranks $2$ and $1$ respectively). We therefore
measure it separately, as the rank of the stacked real Jacobian
$[\Re J;\,\Im J]$ under the same relative cutoff, and deposit it alongside the
complex rank. Generically the real and imaginary parts of a column are
independent directions and the real dimension is $\min(n_\theta,2d)$, twice the
static value wherever the bond count does not bind, that is wherever
$n_\theta\ge2d$. The sweep finds this in $179$ of the $180$ cases with
$\omega>0$.

The one exception (sixteen nodes, $m=5$, full overlap,
$\omega=0.2$) has real rank $29$ against $30$: near a resonance the leading
singular value grows to $\sim\!7\times10^{2}$ and the thirtieth falls to
$6\times10^{-10}$ of it, just under the cutoff, while the numerical floor sits
at $10^{-17}$. The thirtieth direction is present but weakly conditioned, and
the case shows that the doubling is generic rather than forced. Where it holds,
the forbidden fraction is unchanged, because the reachable and the forbidden
parts double together. Driving at a frequency generically buys independent
directions; it never buys any of the ones reciprocity forbids.

\paragraph{A control that separates symmetry from a generic rank shortfall.}
If the deficit were an artefact of small networks rather than of Maxwell--Betti,
it would appear with disjoint drive and read-out sets too. A separate sweep of
$640$ cases isolates the two extremes. With $T\cap S=\emptyset$ the passive
networks reach the full $m^2$ in $32/32$ cases at $m=2,3,4$, whereas with $T=S$
they reach only $m(m+1)/2$ in $32/32$; this is the symmetry effect, present
only when $p>0$. At $m=5$ and $m=6$ the disjoint passive case reaches $m^2$ in
$27/32$ and $19/32$. This is the room effect of \eqref{eq:criterion}, and at
$m=6$ it reaches the odd networks too, far more weakly: they reach $m^2$ in
$32/32$ at $m=5$ and $31/32$ at $m=6$ over the same $32$ configurations. The
two effects are therefore separable by their signature: the symmetry effect
switches on with $p$ and never touches an odd network, and the room effect
grows with $m$ and by $m=6$ touches both. Over the $2418$
rank cases of experiment~03 the bound is exceeded in none and attained in all
$853$ inside \eqref{eq:criterion}; in the overlap sweep of experiment~02, where
the network has room, the deficit equals $p(p-1)/2$ in $216$ of $216$ cases.


\section{The floor, seen by a learning rule}
\label{sec:learning}

A capacity statement earns its keep only if a learning rule actually runs into
it. We therefore trained networks on targets with a known forbidden part, first
with a second-order optimiser and then with a contrastive rule whose update
direction is bond-local, and asked where training stops relative to the floor
of Theorem~\ref{thm:floor}. Both rules end at it in almost every run; the
allowed-target control, and for the optimiser also the odd-coupling control,
goes below it. On an arbitrary target that floor is not
attained, and the passivity term of Theorem~\ref{thm:floorpd} accounts for most
of the remaining gap.

\paragraph{Protocol.}
We fit the log-stiffnesses $u=\log k$ to a target block by Levenberg--Marquardt
on the residual $\mathrm{vec}(R(k)-R^{*})$, using the exact Jacobian
\eqref{eq:jaccol}, so that a plateau is a property of the reachable set and not
of a noisy gradient. Targets are built with a known forbidden part: starting
from a realisable block $R_0=R(k_0)$ at full overlap ($T=S$, $p=m$), we set
\begin{equation*}
  R^{*} \;=\; R_0 \;+\; \varepsilon\,\Psi , \qquad \|\Psi\|_F = 1 ,
\end{equation*}
with $\Psi$ either \emph{forbidden} (antisymmetric on the shared block, hence in
$\mathcal{A}_P$) or \emph{allowed} ($\Psi\in\mathcal{S}_P$). The floor is then
evaluated from the target by the formula of Theorem~\ref{thm:floor}, not reused
from $\varepsilon$; the two agree to $2.3\times10^{-12}$ relative, and the floor
of the allowed targets comes out at $4.3\times10^{-12}$ relative to
$\varepsilon$, i.e.\ zero. Nothing is fitted.

\paragraph{Why the starting point matters.}
Starting the training at $k_0$ would be circular. For a forbidden target, $k_0$
\emph{is} the global minimiser. Every Jacobian column lies in $\mathcal{S}_P$
while the residual $-\varepsilon\Psi$ lies in $\mathcal{A}_P$, so the
Gauss--Newton gradient $J^\T\mathbf{r}$ vanishes identically there: the
optimiser could not move, and ``the error equals the floor'' would hold by
construction rather than by measurement. Every run reported below therefore
starts from an \emph{independent} random stiffness vector, unrelated to $k_0$,
and takes the best of four restarts. The optimiser has to find its way to the
constrained optimum, and arriving at the floor is a result.

\paragraph{A second-order optimiser stops at the floor.}
The runs number $144$, but they are not $144$ independent samples: three network
sizes $\times$ two geometries each $\times$ three values of $m$ $\times$ two
target draws $=36$ configurations, each trained at four amplitudes of
$\varepsilon$ spanning three decades. Over those runs the passive network on a
forbidden target reaches
\begin{equation*}
  \begin{gathered}
  \frac{\text{error reached}}{\text{floor predicted}} :
  \quad \text{median } 1-7.8\times10^{-10},
  \quad \text{within } 10^{-6} \text{ in } 139/144 \text{ runs,}\\
  \text{within } 0.1\% \text{ in } 142/144,
  \quad \text{within } 1\% \text{ in } 143/144 .
  \end{gathered}
\end{equation*}
No run ends meaningfully below the floor. Of the $144$, $126$ end below it, by a
relative amount of median $1.1\times10^{-9}$ and worst case
$2.65\times10^{-7}$, which is round-off in the solve at the conditioning these
runs reach (an exactly symmetric block cannot fall below the floor) and so no
violation (Fig.~\ref{fig:learning}b). The largest ratio of error reached to
floor predicted is $1.0228$, one run of the $144$ that stopped $2.3\%$ above
the floor.

The two controls locate the plateau in the reachable set, away from the
optimiser and away from the network. On an allowed perturbation of the
\emph{same size}, on the same network, with the same optimiser and the same
starts, the error falls to a median $8.8\times10^{-9}$ of the floor. With odd
couplings switched on, the error on the forbidden target falls to a median
$2.9\times10^{-11}$ of the floor, and below $10^{-6}$ in \emph{all} $144$ runs
(Fig.~\ref{fig:learning}a,c).

\paragraph{Where the controls are weaker, and why.}
The allowed-target control is decisive at small perturbations and less so at the
largest. Relative to the floor it falls below $10^{-6}$ in $116$ of $144$ runs,
with a worst case of $0.44$, and the failures concentrate at large
$\varepsilon$: $18/36$ fall below $10^{-6}$ of the floor at $\varepsilon/\|R_0\|_F=10^{-1}$
against $34/36$ at $10^{-3}$ and $10^{-4}$. We attribute these failures to the
optimiser: where the perturbation is a tenth of the block, the map
$k\mapsto R(k)$ is far from linear and Levenberg--Marquardt sometimes stops
early. Whatever their cause, they cannot explain the forbidden arm, which stops
at the floor at \emph{every} $\varepsilon$, including the ones where the
control succeeds outright. The odd control has no
such tail: $144/144$.

\paragraph{The same floor, found by coupled learning.}
A plateau found by a second-order optimiser may still be an artefact of that
optimiser, so we repeated the test with \emph{coupled learning}
\citep{Stern2021} in its contrastive form, in which each bond sees only its own
extension in two physical states,
\begin{equation*}
  \Delta k_b \;\propto\; \sum_j\Bigl[(e_b^{\mathrm{fr},j})^{2}-(e_b^{\mathrm{cl},j})^{2}\Bigr],
\end{equation*}
where $j$ indexes the training pairs $(\mathbf{f}_j,\by^{*}_j)$ of applied load
and desired read-out vector. Here $e_b^{\mathrm{fr},j}$ is the bond extension in the free
state solving $K_{\!f\!f}\bu=\mathbf{f}_j$ and $e_b^{\mathrm{cl},j}$ that in the nudged
state
$(K_{\!f\!f}+\beta\Pi_T^\T\Pi_T)\bu=\mathbf{f}_j+\beta\Pi_T^\T\by^{*}_j$, with $\Pi_T$
selecting the read-out coordinates $T$. The update itself uses no Jacobian and
no global loss gradient; this is a rule a material could plausibly implement.

The update \emph{direction} is local. Its magnitude is not: the step is
rescaled by the largest relative change over all bonds
(Appendix~\ref{app:params}), and the rate schedule, the stopping test and the
choice of which state to return read the global loss $\|R-R^{*}\|_F$, so a
supervisor measuring that loss is assumed.

Theorem~\ref{thm:floor} must bind this rule too, because the floor is a
property of the reachable set and not of an algorithm. The open question is
whether a local rule actually \emph{reaches} the floor or stalls somewhere above
it. It reaches it. Over $24$ runs on six networks, started from stiffnesses
independent of the target, the error reached divided by the floor predicted has
median $1.000001$, minimum $1.000000$ and maximum $1.088$; within the step
budget, $22$ of the $24$ runs finished within $0.1\%$ of the floor and none
below it (Fig.~\ref{fig:coupled}b). The trajectories of Fig.~\ref{fig:coupled}a are
one illustrative pair: its ratio on the deposited run, $1.0000003$, is close to
but not equal to the median over the $24$ runs, $1.0000012$, and every
conclusion rests on all $24$.

Because a slow rule and a blocked rule can end at similar errors, we also
measured which one has \emph{stopped}: the fractional improvement over the final
fifth of each run. On a forbidden target it is $3.7\times10^{-7}$ at the median,
below $1\%$ in $23$ of $24$ runs, and below the $10^{-3}$ that
Fig.~\ref{fig:coupled}c marks in $20$ of $24$, so the typical run is flat by
either threshold. On an allowed perturbation of the same size, on the same network,
from the same start, it is $0.43$: still learning, and by then already a median
factor of $3\times10^{3}$ below where the forbidden run stopped
(Fig.~\ref{fig:coupled}a,c); in the pair of panel~a the allowed run's
error descends $6.2$ decades
in all, $5.0$ of them below the floor. Yet $6$ of the $24$ allowed runs land
within a factor $20$ of the \emph{forbidden} arm's floor, which is the only
floor in the comparison (the allowed arm's own floor is zero). The two arms therefore separate cleanly in
three quarters of the runs.

The two runs that end more than $1\%$ above the floor, at $1.020$ and $1.088$,
are both on the $30$-node network where the rule converged most slowly, and
both exhausted the step budget. They
are not the same kind of run. The $8.8\%$ one has a tail improvement of $0.38$:
it had not stopped, and we count it as
budget-limited rather than as a stall. The $2.0\%$ one has a tail improvement
of $2.3\times10^{-3}$, close to flat. The count of $22$ in $24$ is therefore
finite-budget performance and not evidence of convergence in general. Eleven of
the $24$ runs reached the budget of $40\,000$ steps and thirteen stopped early;
for the thirteen the reported value is where the rule settles rather than where
we chose to stop it, and for the eleven it is not.

Because $\beta$ is a parameter of the rule and not of the network, we also
varied the nudge amplitude over two decades, $\beta=10^{-2}$, $10^{-3}$,
$10^{-4}$, on one network and one target: the stall value moves by less than
$4\times10^{-5}$ relative. The update is rescaled to a fixed relative rate
$\chi$, which cancels the $1/(2\beta)$ prefactor of the update
(Appendix~\ref{app:params}), so the sweep probes only the
$O(\beta)$ error in the step \emph{direction}, and shows only that the stall
value is not an artefact of the particular $\beta$ we chose.

\begin{figure}[t]
  \centering
  \includegraphics[width=\textwidth]{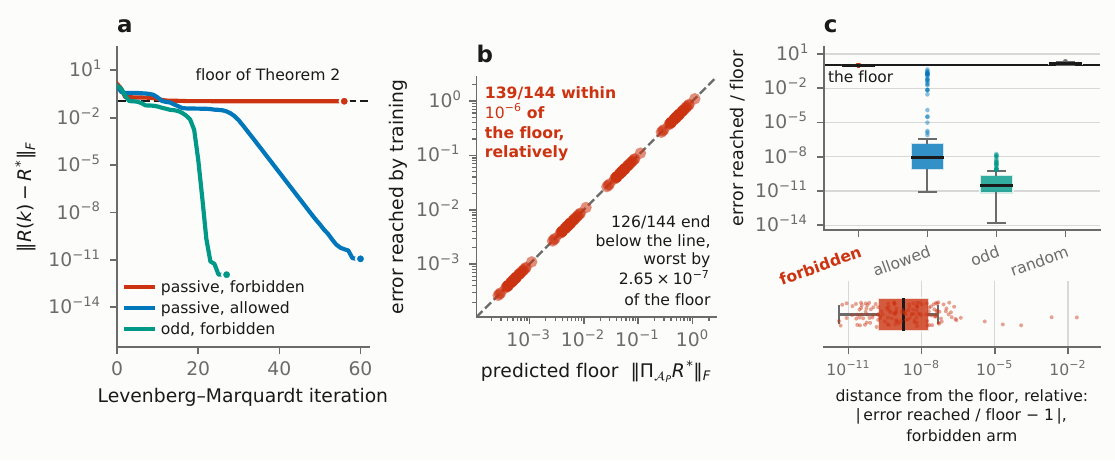}
  \caption{\textbf{The error floor, seen by a learning rule.}
  \textbf{(a)}~Training curves for one network and one target ($30$ nodes,
  geometry seed $0$, $m=5$, $\varepsilon/\|R_0\|_F=10^{-2}$), all arms started
  from the same independent random stiffnesses. Each curve is drawn only for as
  many iterations as it actually ran, with a dot at its last. Arms: the passive
  network on the forbidden target, the same optimiser on an allowed
  perturbation of the same size, and the odd network on the forbidden target.
  \textbf{(b)}~Error reached against the floor Theorem~\ref{thm:floor} predicts,
  over $144$ runs and three decades of perturbation size; the dashed diagonal
  is equality. The in-panel count uses a relative tolerance of $10^{-6}$; at
  the $0.1\%$ of the text it is $142/144$.
  \textbf{(c)}~Error reached relative to the floor, by arm, drawn as box plots:
  the box is the interquartile range, the line inside it the median, the
  whiskers reach $1.5$ times the interquartile range, and every point beyond
  them is drawn individually. The strip beneath the axis resolves the forbidden
  arm, whose box collapses onto the floor line, by plotting
  $\lvert\text{error reached}/\text{floor}-1\rvert$ for its $144$ runs on a
  logarithmic scale. All four boxes are experiment~04. The first three arms
  are divided by the floor of the forbidden target on the same configuration.
  The fourth box, labelled \emph{random}, is the random-target ensemble
  on the same configurations, divided by its own floor, and is not the
  experiment~13 generic-target ensemble of the text. Its $144$ plotted points
  carry only $36$ distinct values, one random target per configuration, fitted once and copied into the four
  amplitude rows.}
  \label{fig:learning}
\end{figure}

\paragraph{How tight is the floor for an arbitrary target?}
Theorem~\ref{thm:floor} is a lower bound, and the targets above were built so
that it is attained. A fourth arm trains the same networks on a \emph{random}
block of the same scale, for which nothing guarantees tightness: the
\emph{experiment~04 random-target ensemble}. It draws one random target per
configuration, fits it once outside the amplitude loop and copies the result
into each of the four amplitude rows, so its $144$ rows in the deposited table
hold only $36$ distinct measurements. Over those $36$, which also make up the
fourth box of Fig.~\ref{fig:learning}c, the error reached is $1.400$ times that
floor at the median, ranging from $1.075$ to $2.192$.

\paragraph{The origin of the slack.}
On a second, independent ensemble of generic targets, most of such a gap is
\eqref{eq:floorpd} at work rather than a limitation of the networks. A random block has an indefinite symmetric part. Its passivity term is
therefore non-zero, and Theorem~\ref{thm:floor} cannot see it. To separate the
two terms we retrain on a second, independent set of generic blocks, the
\emph{experiment~13 generic-target ensemble}: $36$ configurations over the same
six networks, with the same optimiser and four restarts drawn the same way but
its own targets, so its numbers are not expected to reproduce those of the
experiment~04 ensemble, and do not. There the ratio of error reached to floor
predicted has median $1.369$, from $1.093$ to $2.003$, against
Theorem~\ref{thm:floor}, and median $1.016$, from $1.002$ to $1.099$, against
Theorem~\ref{thm:floorpd}. Both are best-of-four-restart figures, and
quadrupling that budget barely moves them: on the same targets, held fixed, the
best of sixteen restarts puts the two medians at $1.365$ and $1.010$
(experiment~21, part~7), which points to the reachable set rather than the
search as the source of the slack. The slack against Theorem~\ref{thm:floor}
was mostly the second term.

The separation is sharpest where reciprocity says nothing at all. Symmetrise the
shared block of a realised response, so that Theorem~\ref{thm:floor}'s floor is
zero up to round-off, then flip the sign of its smallest eigenvalue. The block
is still symmetric to round-off, so Theorem~\ref{thm:floor} predicts a floor of
$1.5\times10^{-16}$ at the median, which is nothing at all. Training stalls at a
median of $0.526$, in the same absolute units, where Theorem~\ref{thm:floorpd}
predicts $0.449$; relative to the target norm, the medians are an error reached
of $9.3\%$ against a floor of $8.8\%$. Each of these medians is taken separately
over the $36$ runs, so neither quotient, including the $1.05$ of the relative
pair (taken from the unrounded medians), is the slack of any one run. The run-by-run ratio of error reached to
floor predicted has median $1.08$. Passivity is the floor that binds here, and
Theorem~\ref{thm:floor}, which sees only reciprocity, misses it.

Across all $108$ runs of experiment~13, spanning its three target classes
(constructed, generic and indefinite; Appendix~\ref{app:params}), no
run ends below the full floor by more than $3.4\times10^{-11}$ of the target
norm. The hypothesis behind the full floor was measured rather than assumed:
over $72$ realised shared blocks on the same six networks (three sizes, two
geometries each, $m=4,5,6$ and four index-set draws), the shared block is
positive definite in $72$ of $72$, with a smallest eigenvalue never closer to
zero than $2.4\%$ of the largest, and with an antisymmetric part at most
$1.1\times10^{-15}$ of its norm.

A run that stalls \emph{above} \eqref{eq:floorpd} may be blocked by the size
of the network, or by the optimiser, or by nothing but the slack of the bound;
the floor alone does
not distinguish those (Sec.~\ref{sec:consequences} reads a stall \emph{at}
it). The median gap above the floor falls from $36.9\%$ to $1.6\%$ when the
floor includes the passivity term. The medians are over the $36$ targets of the
experiment~13 generic-target ensemble, each the best of four restarts, and each
gap is taken as a fraction of the floor it stands above. A median gap of
$1.6\%$ makes the floor usable as a diagnostic.

On targets built to attain it,
the floor of Theorem~\ref{thm:floor} is also where training ends. From starts
independent of the target, a second-order optimiser reaches it within $0.1\%$ in
$142$ of $144$ runs, and a contrastive rule with a bond-local update direction,
under a budget of $40\,000$ steps, reaches it within $0.1\%$ in $22$ of $24$.

\begin{figure}[t]
  \centering
  \includegraphics[width=\textwidth]{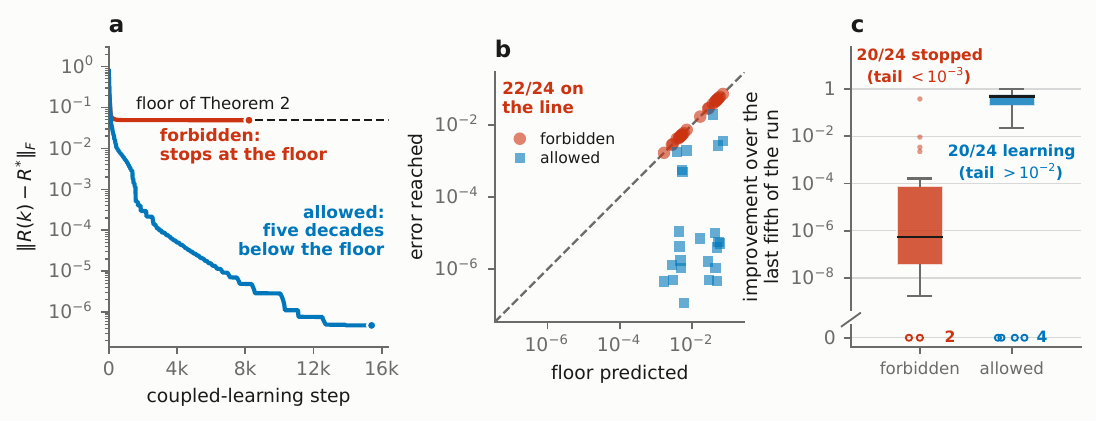}
  \caption{\textbf{The floor found by a contrastive rule with a bond-local update
  direction.}
  \textbf{(a)}~Coupled learning on one network: the $30$-node network at $m=5$,
  geometry seed $1$, the pair whose trajectories the script deposits. That
  pair is hard-coded in the script, chosen after a first run as one near
  the median (eight of the other $23$ lie closer to it). The curves are the best error reached so far, and so
  monotone by construction, on a forbidden target and on an allowed
  perturbation of the same size from the same start; the dashed line is the
  floor of Theorem~\ref{thm:floor}.
  \textbf{(b)}~Error reached against the floor predicted, all $24$ runs and both
  arms; the dashed diagonal is equality.
  \textbf{(c)}~Fractional improvement over the last fifth
  of each run, as box plots (interquartile range, median, $1.5\times$
  interquartile whiskers, and every outlying run drawn), with the runs that
  improved by exactly nothing counted on their own row below the axis break.
  The boxes summarise the runs with a non-zero improvement ($22$ forbidden,
  $20$ allowed), so their medians ($5.5\times10^{-7}$ and $0.50$) differ from
  the medians over all $24$ runs quoted in the text ($3.7\times10^{-7}$ and
  $0.43$).}
  \label{fig:coupled}
\end{figure}


\section{Consequences for physical learning}
\label{sec:consequences}

This section draws the consequences of
Theorems~\ref{thm:main}--\ref{thm:floorpd} and of Sec.~\ref{sec:recovery} for
physical learning: what can be
computed before training, where to place sensors and actuators, and what the
published layouts pay. Each consequence is testable on an existing
experimental platform, at layouts none of them currently uses
(Table~\ref{tab:layouts}, where every layout has $p=0$). In the tabulated
layouts a terminal is an actuator or a sensor, not both, and the two published
demonstrations with a shared terminal, discussed below, instrument no pair of
terminals in both directions (the locomotion demonstration of \citet{Du2026}
as we read it, the same-node demonstration of \citet{Li2024} as we assume).
Reaching a layout that pays therefore means adding instrumentation on hardware,
not only moving it; the comparison of $p=m$ with
$p=0$ below is run in simulation.

\paragraph{Three of these results can be used before anything is trained.}
First, for a force-driven block the floor of Theorem~\ref{thm:floorpd}
follows from the target and the layout alone. Second, the layout that
maximises the reachable ceiling follows from which resource is scarce, before
any network is built (Proposition~\ref{prop:budget}; both layout rules are
stated together below), and no choice of overlap buys a target's antisymmetric
part. Third, odd couplings on only $\delta_p=p(p-1)/2$
bonds, chosen at the passive network in advance, realise that part exactly
when it is small enough
(Corollary~\ref{cor:submersion}), provided the passive wedges span, which is
measured instance by instance and not proved in general.
Proposition~\ref{prop:ratio} and Corollary~\ref{cor:price} bound the cost of
that recovery from the target and the wedges, the corollary to first order.

\paragraph{Testing the floor needs pairs of terminals instrumented in both directions.}
For each pair $\mu\neq\nu$ in $P$, a force-driven test drives $\mu$ and reads $\nu$,
and then the reverse. A passive, linear network trained on a target whose two
directions differ cannot bring $\|R-R^{*}\|_F$ below the floor of
Theorem~\ref{thm:floorpd}. In simulation, on targets with a reachable
symmetric part, both learning rules end on this floor in nearly every run,
within $0.1\%$ (Sec.~\ref{sec:learning}).
A layout with $p\le1$, or a list of coordinate tasks without such a pair
(Corollary~\ref{cor:digon}), shows no reciprocity floor. Under an imposed
displacement Theorem~\ref{thm:dirichlet} predicts instead that $D_{\nu\mu}$ and
$D_{\mu\nu}$ share a sign and multiply to less than one. Neither prediction
covers non-reciprocal materials, which the odd-coupling controls take below the
floor, or finite port paths that cross equilibrium branches.

\paragraph{Self-measurement is expensive at a fixed block size.}
A force-driven learning material that senses at the same locations it
actuates has $p=m$
and forfeits a fraction $(m-1)/(2m)$ of its nominal capacity: $37.5\%$ at $m=4$,
$45\%$ at $m=10$, approaching one half. Separating the actuated and sensed
degrees of freedom removes that loss without changing a single stiffness; any
limit set by the number of tunable bonds still applies. This separation is the
design rule, and we tested it directly.

Keeping the network, starting stiffnesses, number of parameters, learning rule,
step budget and target construction (a realisable block plus a $5\%$ generic
perturbation) the same, we moved only the $m$ read-out degrees of freedom,
from the driven ones ($p=m$) to $m$ others ($p=0$). In all twelve matched
pairs the overlapping layout stalled at or near
the floor of Theorem~\ref{thm:floor} (reached/floor: median $1.00000$, none
below, the worst $1.062$, on the $24$-node pair at $m=3$). The disjoint
layout, whose floor is identically zero, reached a median factor
$7\times10^{5}$ lower (Fig.~\ref{fig:design}).

Fig.~\ref{fig:design}a shows one pair for illustration only; the conclusions
rest on all $12$ pairs. On reached/floor, the quantity quoted above, that pair
ends at $1.012$, away from the median. With sensors on the actuators its error approaches the floor
from above and never crosses it: the best error reached so far stands $7.2\%$
above the floor at step $100$ ($7.6\%$ for the raw error at that step), $2.8\%$
at step $10^{4}$, and $1.2\%$ when the run exhausts its budget of $40\,000$
steps, still creeping down. With the sensors moved elsewhere, at no cost in
stiffnesses or parameters, the same rule descends six and a half decades before stopping
at step $32\,046$ on the early-stopping rule of Appendix~\ref{app:params}. Only
this pair's training curves are deposited, so the twelve cannot be ranked by
speed of convergence.

The comparison establishes a floor on one side and its absence on the other.
The median factor of
$7\times10^{5}$ measures the step budget as much as the capacity. With
no floor to settle on, the disjoint runs stop only on the step budget or on the
early-stopping rule (the table records no step count for the eleven pairs not
plotted), so a longer budget or patience would enlarge the factor. In four of
the twelve pairs the gap was only $2$--$5\times$. We attribute this to the
disjoint run not having converged; the table cannot confirm it, since it records
no step count for those runs.

\paragraph{Under a fixed budget of accessed degrees of freedom the design rule inverts.}
The comparison above holds $m$ fixed: the same $m\times m$ target, the sensors
moved, with separation paid for in instrumentation. An experimentalist more
often faces a fixed number of degrees of freedom they can afford to
instrument. Write $n_U=|T\cup S|$ for that number, and $\alpha=m_T-p$,
$\gamma=m_S-p$ for the exclusively read and exclusively driven counts, so
$n_U=\alpha+\gamma+p$. This is a budget of access points, each of which can be
driven and read (a terminal of an electrical network, a self-sensing actuator,
a robotic unit with motor and encoder), and a degree of freedom used for both
is charged once, unlike in a budget of separate sensors and actuators,
$m_T+m_S$.

\begin{proposition}[optimal layout at fixed instrumentation]
  \label{prop:budget}
  Let $n_U=|T\cup S|\ge2$ be fixed, so that a degree of freedom that is both
  driven and read is charged once, and suppose the bond count does not bind.
  Then the ceiling of
  Theorem~\ref{thm:main}, maximised over the split $\alpha+\gamma=n_U-p$, is
  \begin{equation}
    \label{eq:budget}
    \Lambda(p) \;=\; \Bigl\lfloor \tfrac{1}{4}(n_U+p)^{2} \Bigr\rfloor
                \;-\; \tfrac{1}{2}p(p-1) ,
  \end{equation}
  which increases strictly with $p$ at every integer step, by
  $\Lambda(p+1)-\Lambda(p)=\lceil (n_U+p)/2\rceil-p\ge 1$. The maximum is at \emph{full}
  overlap, $\Lambda(n_U)=n_U(n_U+1)/2$, against
  $\Lambda(0)=\lfloor n_U/2\rfloor\lceil n_U/2\rceil$ for the disjoint layout, a
  factor $2+1/\lfloor n_U/2\rfloor$, approaching two from above.
\end{proposition}

\begin{proof}
  For fixed $p$ the integers $\alpha,\gamma\ge0$ with $\alpha+\gamma=n_U-p$ maximise
  $(\alpha+p)(\gamma+p)$ when
  they are as equal as possible, and then
  $(\alpha+p)(\gamma+p)=\lfloor(n_U+p)/2\rfloor\lceil(n_U+p)/2\rceil=\lfloor(n_U+p)^2/4\rfloor$;
  subtracting the codimension $p(p-1)/2$ gives \eqref{eq:budget}. The stated
  increment follows by cases on the parity of $n_U+p$, and is at least $1$ for
  $p<n_U$, so $\Lambda$ is strictly increasing and maximal at $p=n_U$, where $\alpha=\gamma=0$
  and
  $\Lambda=n_U^{2}-n_U(n_U-1)/2=n_U(n_U+1)/2$. At $p=0$ the formula reads
  $\lfloor n_U^{2}/4\rfloor=\lfloor n_U/2\rfloor\lceil n_U/2\rceil$, and the ratio
  of the two is $2+1/\lfloor n_U/2\rfloor$ for either parity of $n_U$.
\end{proof}

The maximum $\Lambda(n_U)=n_U(n_U+1)/2$ is itself the classical count of
independent self- and mutual resistances of a grounded electrode
array~\citep{Butler2019}; what the proposition adds is the comparison across
overlaps at a fixed budget.

If sensors and actuators are instead counted separately, with $N=m_T+m_S$
fixed and a shared degree of freedom charged twice, the ceiling is
$\lfloor N^2/4\rfloor-p(p-1)/2$, maximal at $p\le1$ and
decreasing thereafter: for even $N$, full overlap gives $N^2/8+N/4$ against
$N^2/4$ for the disjoint layout, a factor $2N/(N+2)$ the other way, and the
design rule stands.

The inversion is arithmetic. Reciprocity removes a fraction approaching $1/2$
of $m^2$, whereas splitting a budget into two disjoint halves divides the count
by $4$, because it halves \emph{both} sides of the block.

We realised every overlap $p=0,\dots,n_U$ at the balanced split of fixed $n_U$
(Appendix~\ref{app:params}) on nine networks of $30$ to $56$ nodes. The
measured rank equals $\min(n_\theta,\Lambda(p))$ in $315$ of $315$ cases (this
arm is passive, so $n_\theta=n_b$), and the maximum over layouts is at full
overlap in $45$ of $45$ network--budget pairs (Fig.~\ref{fig:phase}a). At
$n_U=8$ the measured ranks give the predicted ratio
$\Lambda(n_U)/\Lambda(0)=36/16=2.25$.

\paragraph{Both counts are of the same object.}
It is fair to object that $\Lambda(n_U)$ counts dimensions inside an
$n_U\times n_U$ target space and $\Lambda(0)$ inside a much smaller one, so
that their ratio means nothing.

\begin{proposition}[$\Lambda(p)$ counts constraints on one fixed object]
  \label{prop:count}
  Let $U=T\cup S$ with $|U|=n_U$. The distinct entries of the compliance
  submatrix $C[U,U]$ that the layout $(T,S)$ prescribes are the unordered pairs
  $\{u,v\}$ with $u\in T$ and $v\in S$, including the singletons $u=v$, which
  index the diagonal entries of $C[U,U]$; there are exactly
  $(\alpha+p)(\gamma+p)-p(p-1)/2$ of them, which is $\Lambda(p)$ at the split that
  maximises it
  and the value used throughout.
\end{proposition}

\begin{proof}
  $C[u,v]$ and $C[v,u]$ are the same number, so a prescribed entry is an
  unordered pair. There are $(\alpha+p)(\gamma+p)$ ordered choices $(u,v)\in T\times S$. A
  pair is obtained twice exactly when both $(u,v)$ and $(v,u)$ are admissible,
  which requires $u,v\in T\cap S=P$ and $u\neq v$: that is $p(p-1)$ ordered
  choices forming $p(p-1)/2$ duplicated pairs. Subtracting gives
  $(\alpha+p)(\gamma+p)-p(p-1)/2$, which is $\Lambda(p)$ at the split that
  maximises it.
\end{proof}

Whatever the layout, the instrumented set carries one object, the compliance
submatrix $C[U,U]$, symmetric and positive definite because
$C=K_{\!f\!f}^{-1}$ is. A layout does not create a target space of its own;
it selects which entries of that object the experimentalist
can see and therefore prescribe. So $p(p-1)/2$ is the number of would-be
prescriptions that reciprocity turns into duplicates instead of new
information, and $\Lambda(n_U)$ and $\Lambda(0)$ count entries of the same
matrix.

We checked the count by enumeration, with no network involved. Over
$454$ distinct layouts (budgets $n_U=1,\dots,12$, every overlap, every split of
the remainder between exclusive read-out and exclusive drive; $544$ checks in
all, the balanced split visited twice), the number of distinct pairs equals
$(\alpha+p)(\gamma+p)-p(p-1)/2$ with no mismatch. This is experiment~14,
part~1, which deposits no table; the $360$ layouts listed for experiment~14
under Data availability are those trained in part~2.

\paragraph{On a task, too, the overlapping layout wins.}
Take as target the compliance among $n_U$ instrumented degrees of freedom of
the same network at an \emph{independent} stiffness draw, contaminated by an antisymmetric part of relative
size $t$. Each layout is trained on the block it can see and then evaluated on
the whole of $C[U,U]$ under one normalisation, so that no layout is scored
against a target space of its own. Over six networks of $30$ to $44$ nodes,
budgets $n_U=4,6,8$ and $t=0$,
$0.02$, $0.05$, $0.10$, $0.20$, full overlap is the best layout in $90$ of
$90$ network--budget--asymmetry cases. The disjoint layout leaves $56$ to
$60\%$ of the object uninstrumented over these budgets, and that costs more
than the floor does.

At $t=0$, where the target is realisable at every layout by construction, the
layouts meet every constraint they can see, at a fit error below $10^{-6}$ of
the target norm, in $67$ of $72$ runs. The five misses are the optimiser
stopping short, not forfeited constraints: each has a predicted floor below
$4\times10^{-16}$ and sits at the largest overlaps and budgets tried ($p=n_U=6$, $p=n_U=8$ and $p=6$ at $n_U=8$), and their fit errors run from $6\times10^{-6}$ to
$1.8\times10^{-4}$, up to a hundred and eighty times the threshold.

The error of the full-overlap layout is also \emph{predicted}, not merely
bounded. Over all $90$ full-overlap runs the end-to-end error equals
\eqref{eq:floorpd}, computed from the target alone before a network is chosen,
to a median of $8\times10^{-12}$ of the target norm and at worst
$1.4\times10^{-4}$; the largest of the five misses above, $1.8\times10^{-4}$,
is at a partial overlap.

\paragraph{Which rule applies depends on what is scarce.}
The two rules answer different questions, and neither supersedes the other.
(i)~The design rule: if the target block is fixed and reachable except for its
antisymmetric part, sensors should be separated from actuators, which removes
the floor at no cost in parameters. (ii)~The layout law: if instead the number
of accessible degrees of freedom is what is scarce, each of them able to be
both driven and read, and the target is one the symmetric subspace can hold,
overlapping the sensors onto the actuators buys about twice the ceiling on
reachable dimension (Proposition~\ref{prop:budget}). Overlap can never buy the
antisymmetric part itself: at any $p\ge2$ a target whose shared block is not
symmetric has a non-zero floor by Theorem~\ref{thm:floor}, so a target with an
irreducible antisymmetric component needs disjoint read-out or broken
reciprocity, however many dimensions the overlapping layout offers. Dimension
and reachability are different currencies, and Fig.~\ref{fig:phase} separates
them.

Together, \eqref{eq:bound} and \eqref{eq:budget} give the phase diagram of
Fig.~\ref{fig:phase}b. At small $n_U$ the symmetry branch binds and overlap is
worth a factor $2+1/\lfloor n_U/2\rfloor$. Past $\Lambda(p)=n_\theta$ the
parameter count binds instead and the gain collapses: a network that has run
out of bonds to move has no use for extra reachable directions. A network too
small for the enlarged target loses dimensions for the unrelated reason of the
empirical criterion~\eqref{eq:criterion}, which needs a genericity clause
(Sec.~\ref{sec:numerics}); here and below it serves as a filter, never as a
theorem.

\paragraph{Where the bound bites, and where it does not.}
We read $m_S$, $m_T$ and $p$ off every
platform cited here whose driven and read-out sets are separately
identifiable, from the papers and, where they release it, the code, and
evaluated the forfeited \emph{norm} fraction (Remark~\ref{rem:fraction-pd}) at
each. For a target with independent
standard normal entries the expected squared fraction that reciprocity puts out
of reach is
\begin{equation}
  \label{eq:fracgen}
  \frac{\mathbb{E}\|B_a\|_F^{2}}{\mathbb{E}\|R^{*}\|_F^{2}}
  \;=\; \frac{p(p-1)}{2\,m_Tm_S} ,
\end{equation}
which is the codimension of Theorem~\ref{thm:main} divided by the dimension of
the block, since a Gaussian ensemble is isotropic; at full overlap it reduces
to $(m-1)/(2m)$. Table~\ref{tab:layouts} (Appendix~\ref{app:layouts},
reproduced by experiment~15) reports it at the force-driven layouts and, at
those driven by an imposed displacement, the quantity that replaces it under
the law of Sec.~\ref{sec:dirichlet}. We count
independent scalar channels: a strain imposed across a pair of nodes through a
ghost bond is one, and so is a clamped node voltage, including the grounded
node that several electrical tasks use to fix the global voltage shift.

\emph{At all nineteen published layouts we tabulated, the forfeited norm
fraction is exactly zero}. All nineteen drive one
set of degrees of freedom and read another, sixteen under an imposed drive and
three under applied forces, and at $p=0$
equation~\eqref{eq:fracgen} vanishes. The imposed-displacement law of
Theorem~\ref{thm:dirichlet} constrains nothing there either, because with no
shared terminal the matrix $D$ has no entries. Theorem~\ref{thm:main} does not
correct a single capacity number in this literature. The table is a sample of
the cited demonstrations, not a census: a layout is admitted only when its
source fixes $m_S$ and $m_T$, and further layouts in the same papers also have
$p=0$, among them the one-input, two-output symmetry network of
\citet{Altman2024} (their Fig.~4), whose source node is driven along $y$ and
read along $x$. The five layouts excluded from the table are named at the end
of this discussion.

The theorem locates the constraint instead, and two things follow. First,
$p=0$ is not free. By Proposition~\ref{prop:budget}, at a fixed number $n_U$
of accessed degrees of freedom, each of which could be both driven and read,
the disjoint layout is the \emph{worst} available and costs a factor
$2+1/\lfloor n_U/2\rfloor$ in the ceiling on reachable dimension. The three
force-driven layouts in Table~\ref{tab:layouts}, all at $p=0$, forfeit no norm fraction,
but they would pay that factor in reachable dimension if their hardware were
counted in access points. The
sixteen imposed-drive layouts lie outside Proposition~\ref{prop:budget}
altogether. Only one source in the table, covering two of its layouts, reports a
hardware budget at all, and that budget is additive: \citet{Dillavou2024} state the restriction ``input nodes $+$
output nodes $\le 8$'', which caps $N=m_T+m_S$ rather than fixing it and
charges a shared terminal twice. This is the second budget above
(experiment~15). The platform's
drive is an imposed voltage, so Proposition~\ref{prop:budget} does not govern
it either. For a layout the proposition does cover, a budget of this shape
puts the maximum of the ceiling at $p\le1$, so the disjoint layout is then the
right choice rather than a forfeited one. We have found no in-place learning
demonstration reporting a
budget of the first kind, on access points rather than channels, which is the
kind Proposition~\ref{prop:budget} governs.

Second, two published demonstrations already point towards overlap.
\citet{Li2024} demonstrate a node serving as both input and output and
advertise it as enabling ``more compact design''; \citet{Du2026} drive and read
the same unit in their locomotion demonstration. Whether the field follows them
is not ours to assert, and a shared terminal is not by itself enough. Under the
task-list law of Sec.~\ref{sec:tasklist} a terminal that is both driven and
read is a self-loop and costs nothing; on a list of coordinate tasks the
deficit counts reversed \emph{pairs} (Corollary~\ref{cor:digon}). In each
demonstration, as reported or, for Li and Mao, as assumed below, the one shared
degree of freedom enters one drive and one read-out and no pair of terminals is
instrumented in both directions, so each has $q\le1$ and $\delta_{\mathcal{L}}=0$ (experiment~21, part~5). For the
locomotion demonstration this is read off the reported protocol: Du et al.'s
shared unit is unit~$4$, driven by a sinusoidal torque and read along with four
other units, which gives a self-loop and four ordinary arcs. For the
supplementary same-node demonstration of Li and Mao it is an assumption and not
a reading, since that source states no counts (see below): the deposited reading treats
it as a single task ($J=1$), one drive read where it is applied, and were that
node in fact instrumented in both directions against a second shared terminal,
the count would not be zero. The locomotion demonstration carries a
caveat stronger than that on Fig.~2 of \citet{Du2026}, treated below. Its drive
is a prescribed sinusoidal actuation rather than an
applied static force, and two on-site stiffnesses are trained negative there to
make the chain multistable, the observed motion being a limit cycle among four
stable shapes, so there is no linear response block to bound at all. Its $p=1$
is a channel count, and its $\delta_{\mathcal{L}}=0$ is the count of the
analogue list and in no reading an instance of Theorem~\ref{thm:taskfloor}.

A compact self-sensing platform asked for the whole block on its accessed
degrees of freedom does have $p$ rising and \eqref{eq:fracgen} turning on as
written; one asked only for a list of tasks pays $\delta_{\mathcal{L}}$
instead. Fig.~2 of \citet{Du2026}, treated next, is the one published list of
demands whose force-driven analogue has a non-zero deficit: its two combined
drive patterns give $q=2$, and
$\delta_{\mathcal{L}}\le\tfrac12q(q-1)=1$ (Corollary~\ref{cor:taskoverlap})
holds with equality on that analogue. The analogue is needed because the
drive itself is an imposed angle.

The chain of \citet{Du2026} is the one place in this literature where the
driven and the read-out sets are exchanged between two targets, and the only
published case in which the obstruction can be tested against a deposited
training trajectory. Fig.~2 of that work trains a six-unit chain on two
targets at once. A reciprocal chain fails at this task
and a non-reciprocal one succeeds, and the authors attribute the failure to
Maxwell--Betti (their supplementary remark, quoted in the introduction).
The source describes the task in two ways. The main text has one driven and
one read unit: ``applying a positive
curvature to unit 2 leads to a positive curvature to unit 5, whereas applying a
positive curvature to unit 5 leads to a negative curvature to unit 2''. The
Supplementary Information, Sec.~8.1, describes the same figure in halves:
``applying a positive curvature to the left part leads to a positive curvature to
the right part, whereas applying a positive curvature to the right part leads to
a negative curvature to the left part''. The deposited script agrees with the
second reading: clamp units
$1,2,3$ at $+30^\circ$ and ask units $4,5,6$ for $+30^\circ$; then clamp units
$4,5,6$ at $+30^\circ$ and ask units $1,2,3$ for $-30^\circ$. We work from the
code, which produced the deposited stiffness matrices analysed below; its
three-in, three-out reading was confirmed by one of the authors (Y.~Du, private
communication). Under the main-text reading the obstruction would instead be
the sign law of \eqref{eq:dcycle} at $p=2$; both readings forbid the task, and
neither obstruction is a parameter count. Angles enter their model, and every
floor we quote below, in radians: $30^\circ=0.5236\,\mathrm{rad}$, so the
squared scale that sets those floors is $v^{2}=0.2742\,\mathrm{rad}^{2}$.

In this task every unit is driven in one target and read in the other, and the
drive is an imposed angle, so Theorem~\ref{thm:main} does not apply. Nor is
the task an instance of Theorem~\ref{thm:dirichlet}, which clamps
a single terminal and leaves the rest of $P$ free; here three units are clamped
at once, so no entry of that $D$ is measured. It is an instance of
Theorem~\ref{thm:multiclamp}, and one line settles it. With $A=\{1,2,3\}$ and
$B=\{4,5,6\}$, and since both conditions are linear in $v$, the two targets read
$D_1\mathbf{1}=\mathbf{1}$ and $D_2\mathbf{1}=-\mathbf{1}$. Then
$D_2D_1\mathbf{1}=-\mathbf{1}$, so $-1$ is an eigenvalue of $D_2D_1$, which
\eqref{eq:mcone} forbids. No rank count does: the ceiling of
Proposition~\ref{prop:multirank} does not bite on this task
(Sec.~\ref{sec:dirichlet}), and no bound on dimension, from symmetry or from
the parameter count of either chain, excludes a particular pair of targets. The argument uses no property of a chain:
it holds for any network on any six terminals, and for any two targets whose
composition has an eigenvalue outside $[0,1)$.

Write $I=A$ and $O=B$ for the input and output units,
$K_{II}=K[I,I]$, $K_{OO}=K[O,O]$, $K_{OI}=K[O,I]$, and let $\mathbf{1}$ be the
all-ones vector on three units. The two free-state responses are
$y_1=-K_{OO}^{-1}K_{OI}(v\mathbf{1})$ and $y_2=-K_{II}^{-1}K_{OI}^\T(v\mathbf{1})$,
with $v=30^\circ=0.5236\,\mathrm{rad}$. For a symmetric $K$, contracting each
response with its own block gives the exact identity
$\mathbf{1}^\T K_{OO}y_1=-v\,\mathbf{1}^\T K_{OI}\mathbf{1}=\mathbf{1}^\T K_{II}y_2$;
symmetry is used here, and at no other step, to write $K[I,O]=K_{OI}^\T$.
Meeting both targets exactly would need $y_1=v\mathbf{1}$ and
$y_2=-v\mathbf{1}$, and the identity would then read
$v\,\mathbf{1}^\T K_{OO}\mathbf{1}=-v\,\mathbf{1}^\T K_{II}\mathbf{1}$, which
requires
\begin{equation}
  \label{eq:duS}
  W \;:=\; z_1^\T K z_1 + z_2^\T K z_2 \;=\; 0 ,
  \qquad z_1=(1,1,1,0,0,0)^\T,\ z_2=(0,0,0,1,1,1)^\T ,
\end{equation}
whereas $W$ is a sum of two quadratic forms of $K$, so $W>0$ for every
positive-definite $K$. We therefore find that, in the linear model as
deposited, \emph{no passive, positive-definite chain can meet those
two targets}, whatever its size, at any stiffness and under any learning rule.
That model is the stable, positive-definite one that their simulation
integrates, under the imposed-displacement protocol that their script
implements.

If symmetry is dropped, the two cross terms $z_1^\T Kz_2$ and $z_2^\T Kz_1$
become independent, and the identity and the contradiction both dissolve. The
obstruction therefore comes from reciprocity, and positive definiteness enters
only to sign $W$. Equations \eqref{eq:duS} and \eqref{eq:mcone} are the same
obstruction seen from two sides: $W>0$ says the two quadratic forms cannot
cancel, and $\spec(D_2D_1)\subset[0,1)$ says the two transmissions cannot
compose to a sign reversal on any direction. The identity also bounds their
training loss from below. Writing $e_1=y_1-v\mathbf{1}$, $e_2=y_2+v\mathbf{1}$,
substituting $y_1=e_1+v\mathbf{1}$, $y_2=e_2-v\mathbf{1}$ into the identity and
applying Cauchy--Schwarz gives
\begin{equation}
  \label{eq:dufloor}
  \mathrm{MSE} \;=\; \frac{\|e_1\|^2+\|e_2\|^2}{6}
  \;\ge\; \frac{v^{2}W^{2}}{6\bigl(\|K_{II}\mathbf{1}\|^{2}+\|K_{OO}\mathbf{1}\|^{2}\bigr)} .
\end{equation}
The floor \eqref{eq:dufloor} is a bound at each $K$ and no universal
number: it is small when $W$ is small relative to $\|K_{II}\mathbf{1}\|$ and
$\|K_{OO}\mathbf{1}\|$, that is, when the network is close to losing positive
definiteness in the $z_1,z_2$ directions.

Their deposited data bear this out (experiment~17). At the initial condition
both runs share, $W=0.770$ and \eqref{eq:dufloor} gives $0.274\,\mathrm{rad}^2$
against a measured $0.275$. Their chain starts uniform, with diagonal $0.115$ and nearest-neighbour
coupling $0.010$, and so transmits almost nothing: the free-state responses are
below $0.05$ against a target of $v=0.524$, and both the measured error and the
bound land within $0.3\%$ of the $v^2$ that a chain transmitting nothing at all
would return. Uniformity also makes the Cauchy--Schwarz step nearly tight,
since a uniform tridiagonal block has $\mathbf{1}$ for a near-eigenvector.
Both sides are near $v^2$ for the same reason, so the agreement is to that
extent forced, though it is not an identity: over $4\times10^{5}$
random positive-definite $K$ the ratio of the bound to $v^2$ has median $0.69$, and
two draws in $4\times10^{5}$ reach the value their initial stiffness takes.
Epoch one, where the reciprocal run comes closest to its floor, sits next to
epoch zero and carries the same caveat.

What carries the argument is \eqref{eq:duS}, which is exact, and what the two
runs do after the start. Over the deposited epochs the two runs give

\begin{center}
\small
\begin{tabular*}{\textwidth}{@{\extracolsep{\fill}}p{0.25\textwidth}p{0.36\textwidth}p{0.32\textwidth}@{}}
\toprule
 & reciprocal run & non-reciprocal run \tabularnewline
\midrule
\raggedright floor \eqref{eq:dufloor}
 & \raggedright respected at all $21$ deposited epochs, most tightly at epoch
   one, at $1.0028$ times the floor
 & \raggedright no longer applies once the run breaks the symmetry \tabularnewline\addlinespace[2pt]
\raggedright identity $\mathbf{1}^\T K_{OO}y_1=\mathbf{1}^\T K_{II}y_2$
 & \raggedright holds to at worst $9\times10^{-17}$ at every epoch
 & \raggedright violated by $9.8\times10^{-2}$ \tabularnewline\addlinespace[2pt]
\raggedright smallest eigenvalue of the symmetric part of $K$
 & \raggedright first negative at epoch five, positive again at epochs six
   and seven, negative from epoch eight onwards; $-0.273$ at the last deposited epoch
 & \raggedright never below $0.011$, at epoch six \tabularnewline\addlinespace[2pt]
\raggedright $W$
 & \raggedright negative at eight of the last twelve epochs
 & \raggedright not an obstruction: with the identity violated, its sign
   no longer decides the task \tabularnewline\addlinespace[2pt]
\raggedright final error
 & \raggedright $1.21$ times its initial error; never falls to its own floor
 & \raggedright $8.7\times10^{-7}$, five orders of magnitude below the value
   \eqref{eq:dufloor} would take at its own $K$, which no longer bounds it \tabularnewline\addlinespace[2pt]
\raggedright epochs failing their Gershgorin monostability test
 & \raggedright thirteen of twenty-one
 & \raggedright seventeen of the twenty-one, from epoch four \tabularnewline
\bottomrule
\end{tabular*}
\end{center}

\noindent
The last row counts the epochs that fail the monostability test of
\citet{Du2026}, which they call a Gershgorin test but which compares each
unit's on-site stiffness with a single coupling rather than with a row of $K$;
we evaluate it on their deposited stiffnesses. It is therefore not the
eigenvalue criterion of the third row: at epoch five it still passes
although the smallest eigenvalue is already negative, so the reciprocal run
counts thirteen epochs here against fourteen there.
The non-reciprocal run breaks the symmetry and keeps $K$ positive definite
throughout. The reciprocal run cannot break the symmetry, and it does what we
did not predict: it breaks positive definiteness. Nothing in their script stops
it, since its Gershgorin monostability guard never acts on the update (next
paragraph). What is unexpected is only that the unconstrained contrastive
update goes there at all. We report the
symmetric part because their deposited eigenvalue array holds the eigenvalues
of $K$ itself, whose smallest value over the non-reciprocal run, $0.0838$, is
about seven times the smallest eigenvalue of the symmetric part, $0.0114$. From epoch eight onwards the reciprocal run therefore sits outside the
class the obstruction covers, and outside the monostable region. The
non-reciprocal run, which succeeds, spends more epochs outside the monostable
region than the reciprocal one, so leaving that region does not mark the
failing run. What separates the two is the symmetric part of $K$, which stays
positive definite in one and does not in the other.

We find the guard inert as deposited, because of the order of operations in
their script: the update increments are formed from the gradient arrays before
the guard runs, and the guard zeroes the gradient arrays and not the
increments. Experiment~17, part~3 reimplements their learning loop from the
equations of \texttt{CL\_LearningNRshape\_\{p,a\}.py} in their deposit
(doi:10.5281/zenodo.18544299) and runs it three ways. With the guard as
written, and with the guard deleted, it reproduces the deposited
\texttt{K\_step.npy} of both runs bit for bit. Only when the test is applied to the increments do the trajectories part, by
$\max|\Delta K|=10.9$ in the reciprocal run and $5.9\times10^{-2}$ in the
non-reciprocal one; the test would have fired $87$ and $166$ times. We locate a
behaviour in the code and read no intention into it. The argument above needs
only that the deposited trajectories are the unconstrained ones, and that is
what the three runs establish.

The obstruction independently supports the central finding of
\citet{Du2026}: two targets of this kind require non-reciprocity in a stable,
linear chain. The two hypotheses behind \eqref{eq:duS} and
Theorem~\ref{thm:multiclamp} are symmetry and positive definiteness, and the
result concerns that model, not the physical chain, which has a bounded
actuator range and a noisier experimental run; nonlinear finite paths,
hysteresis, multistability and a different protocol all remain open to the
chain. What is universal within that model
is \eqref{eq:duS}: the task is unreachable there, and there the only exits
are non-reciprocity or instability.

Five layouts are excluded from the count rather than guessed at, and are
supplied with the data. The first is Fig.~2 of \citet{Du2026}, treated above,
which carries no single $(m_S,m_T,p)$ and is tabulated under neither law.
\citet{Lee2022}
drive two input nodes in a locked equal-magnitude pattern, alternating between
a horizontal and a vertical scenario, and never state a count of independent
input channels. Their $p=0$ is unambiguous, since inputs and outputs are
distinct nodes on opposite sides of the lattice, so the fraction is zero
whatever $m_S$ turns out to be. The supplementary same-node demonstration of
\citet{Li2024} has $p\ge1$ by construction but states no counts, so we cannot
say how large $p$ is there and assert no bound on it.
\citet{Rocks2019} use a single source throughout and do not say whether the
target edges are drawn excluding the source edge, though a target edge coinciding
with the source edge would make their constraint unsatisfiable. The locomotion demonstration of
\citet{Du2026}, discussed above, has no linear response block to bound. The remaining cited works on trained or designed
networks have no separately identifiable driven and read-out sets and fall
outside the table's scope: the globally strained samples
of \citet{Pashine2019}, the parameter-to-feature Jacobian of \citet{Zu2025},
and the theory papers and surveys that report no layout of their own.

\begin{figure}[t]
  \centering
  \includegraphics[width=\textwidth]{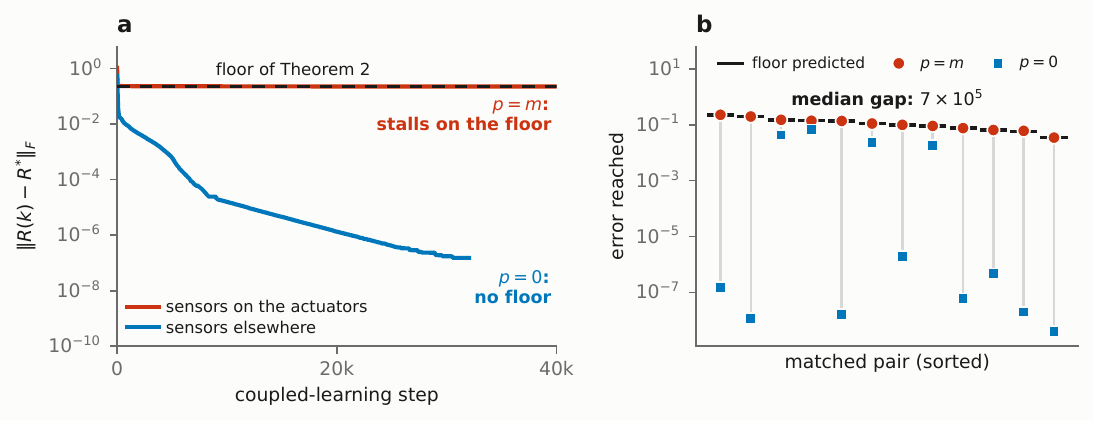}
  \caption{\textbf{The design rule, run.}
  \textbf{(a)}~Coupled learning on one network, differing only in where the
  read-out degrees of freedom sit: the $30$-node network at $m=4$, geometry
  seed $0$, hard-coded in the script after a first run as a pair whose gap is
  close to the median in log terms. It is the only pair whose training curves
  are deposited. Error against coupled-learning step, with sensors on the
  actuators ($p=m$) and moved elsewhere ($p=0$); the plotted curve is the best
  error reached so far, monotone by construction, and the dashed line is the
  floor of Theorem~\ref{thm:floor}. The $p=m$ run comes within $5\%$ of the
  floor by step $438$ and stalls there, ending $1.2\%$ above it after its whole
  budget of $40\,000$ steps; the $p=0$ run, which has no floor, ends at step
  $32\,046$ on the early-stopping rule.
  \textbf{(b)}~All twelve matched pairs, sorted by the error reached at
  $p=m$: that error, the error at $p=0$ joined to it by a grey line, and the
  floor predicted for each pair.}
  \label{fig:design}
\end{figure}

\begin{figure}[t]
  \centering
  \includegraphics[width=\textwidth]{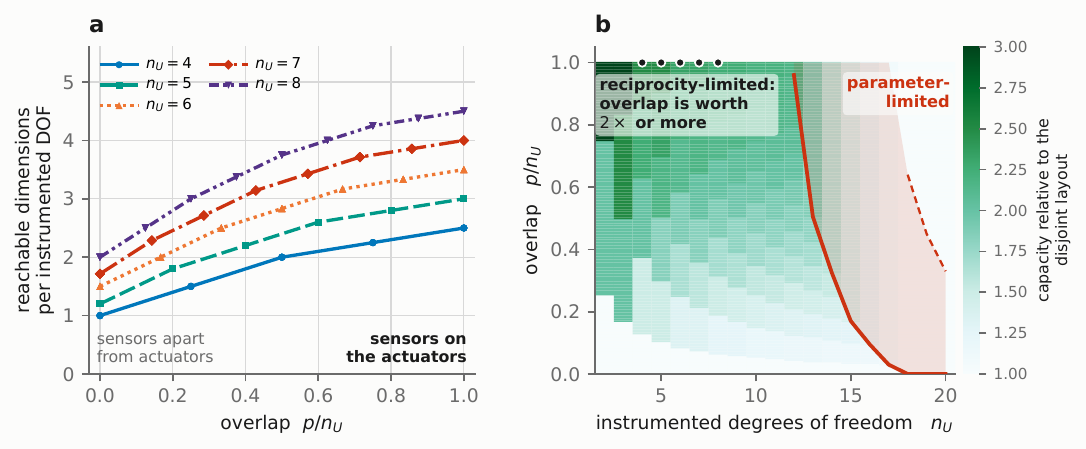}
  \caption{\textbf{The capacity phase diagram.}
  \textbf{(a)}~Reachable dimensions per instrumented degree of freedom against
  overlap, at fixed instrumentation budget $n_U$ (accessed degrees of freedom,
  a shared one charged once). Lines are \eqref{eq:budget};
  points are measured ranks over $315$ layouts on nine networks.
  \textbf{(b)}~Where capacity is set, computed from the closed form rather than
  from the sweep. Colour is the ceiling at overlap $p$ relative to the disjoint
  layout of the same budget, normalised at the smallest $n_\theta$; the
  stepped edges come from the integer part $\lfloor(n_U+p)^2/4\rfloor$ in
  \eqref{eq:budget}. The red band is where
  the two branches of $\min(n_\theta,\Lambda(p))$ cross; inside it, which
  branch binds depends on the network. Its position depends
  on $n_\theta$, so it is drawn over the range of bond counts of the
  nine networks of the sweep, $n_\theta=78$ to $156$: the solid lower edge is
  $n_\theta=78$, where the parameter branch binds soonest, from $n_U=12$;
  the dashed upper edge is $n_\theta=156$, whose crossing does
  not enter the panel until $n_U=18$. Markers on the top edge are the budgets of
  Fig.~\ref{fig:phase}a.}
  \label{fig:phase}
\end{figure}

\paragraph{Part of the residual error of a training run is known in advance.}
Before training, compute \eqref{eq:floorpd} on the shared sub-block of the
target: the antisymmetric part $\tfrac12\|B-B^\T\|_F$ that reciprocity forbids
and the negative eigenvalues of the symmetric part that passivity forbids,
added in quadrature. A run that stalls at that value is being asked for
something no passive reciprocal network can represent. It is therefore neither
under-parametrised nor stuck in a local minimum, and more bonds, more data or a
better optimiser will not help. What does help depends on which term binds.
The reciprocity term is removed by eliminating the overlap or by breaking
reciprocity. The passivity term is of a different kind: it follows from the
shared block being a \emph{principal} submatrix of a positive-definite
reciprocal compliance, which is a property of the layout as much as of the
target. At $p=0$ the visible block is no longer principal, and under an active
or odd constitutive law the hypothesis of Theorem~\ref{thm:floorpd} no longer
holds automatically; in either case the fate of that term has to be assessed
separately, and \eqref{eq:floorpd} on its own does not establish that it
survives. Stalls \emph{above} \eqref{eq:floorpd} are taken up in
Sec.~\ref{sec:learning}, the size of the network
(criterion~\eqref{eq:criterion}) among the candidates.

\paragraph{One reversed pair creates a floor on one network.}
The list form of that computation is Theorem~\ref{thm:taskfloor}.
Fig.~\ref{fig:example} runs it on the sixteen-node network drawn at the start
of the paper (experiment~23). Its terminals $1$, $2$ and $3$ are the
$x$-coordinates of three of its nodes, each driven and read. A task $u\to v$
applies a force at $u$ and reads the displacement at $v$. List~A is the cycle
$1\to2$, $2\to3$, $3\to1$. Its enclosing block has $p=3$ and would forfeit
three dimensions, but the list has no reversed pair, so $\delta_{\mathcal{L}}=0$
by Corollary~\ref{cor:digon}. List~B adds the one reversed task $2\to1$, a
single digon, so $\delta_{\mathcal{L}}=1$. Both lists ask for the values of a
reference passive network, with $1\to2$ raised by $10\%$ and, on list~B,
$2\to1$ lowered by $10\%$.

Before any training, Theorem~\ref{thm:taskfloor} bounds the error below by zero
on list~A and by $0.2334$ on list~B, which is $9.9\%$ of its target norm. The
symmetric part of list~B's target is the reference network's own response, so
on this network the bound is attained by construction, at the reference
stiffnesses. Levenberg--Marquardt training from $20$ independent starts per
list attains the bound on both lists in every run. List~A ends at the
round-off level of the solve, below $3\times10^{-10}$ of its target norm, and
list~B stops on its floor, within a relative $2.4\times10^{-10}$ of it. The
part of the error that training can still remove ends below $10^{-6}$ of the
target norm in every run; it enters the relative distance to the floor only at
second order, whereas round-off enters at first order. The fourth task adds a value but no
dimension:
$\partial y/\partial k$ has rank $3$ on both lists, at the reference
stiffnesses and at the end of every run. A control on list~B that raises both
$1\to2$ and $2\to1$ by $10\%$ also ends below $3\times10^{-10}$ of its target
norm in $20$ of $20$ runs. On this network the extra task raises the error only
when its demand breaks Maxwell--Betti.

One odd bond removes that floor. The bond is chosen before training: its odd
Jacobian column has the largest component in magnitude along list~B's single
relation, the list form of the wedges of Theorem~\ref{thm:recovery}. With that
bond added and every stiffness still free, training ends below
$3\times10^{-10}$ of the target norm in $20$ of $20$ runs. The bond's coupling
$a_b$ ends at a median of $-1.30$, over a range from $-16.3$ to $-0.58$. The
symmetric part of $K_{\!f\!f}$ is positive definite at the end of $18$ of those
$20$ runs; the training does not constrain it, and in the other two every
eigenvalue of $K_{\!f\!f}$ keeps a positive real part. Almost any bond would
have served. At $\delta_{\mathcal{L}}=1$ any bond with a non-zero wedge entry
spans the one missing direction, and here every bond has one: used alone as
the odd bond, $33$ of the $38$ bonds remove the floor in $20$ of $20$ runs, and
every bond in at least $17$. That count ignores stability. We call a trained
network stable when every eigenvalue of $K_{\!f\!f}$ has a positive real
part, which is stability of the overdamped dynamics
$\dot{\bu}=-K_{\!f\!f}\bu$; under inertial dynamics a non-symmetric
$K_{\!f\!f}$ can flutter even then, and that case is not tested. The test is
the sign of the least real part of an eigenvalue of $K_{\!f\!f}$ relative to
its largest eigenvalue modulus (column \texttt{Kff\_min\_re\_eig\_rel} of
\texttt{exp23\_allbonds.csv}). In $69$ of those $760$ runs ($38$ bonds, $20$ starts each) that part is negative, $64$ of them among the $751$ that remove the floor,
and the symmetric part of $K_{\!f\!f}$ is positive definite in only $286$ of the $760$.
Counting only stable runs, $23$ of the $38$ bonds remove the floor in $20$ of
$20$ runs and every bond in at least $8$. The count depends on the margin: in
$22$ runs the relative least real part is positive by less than $10^{-9}$, and
requiring a margin of $10^{-9}$ leaves $15$ bonds at $20$ of $20$, the worst
still at $8$. The selected bond keeps $20$ of $20$ at that margin, its
smallest relative least real part being $8.8\times10^{-9}$. Any bond with a
non-zero wedge entry restores the single missing direction to first order; at
the couplings that training reaches, stability is not guaranteed. On this network the reversed pair predicts the
stall before training, and the wedge rule picks, also before training, one of
the many bonds that remove it.

\begin{figure}[t]
  \centering
  \includegraphics[width=\textwidth]{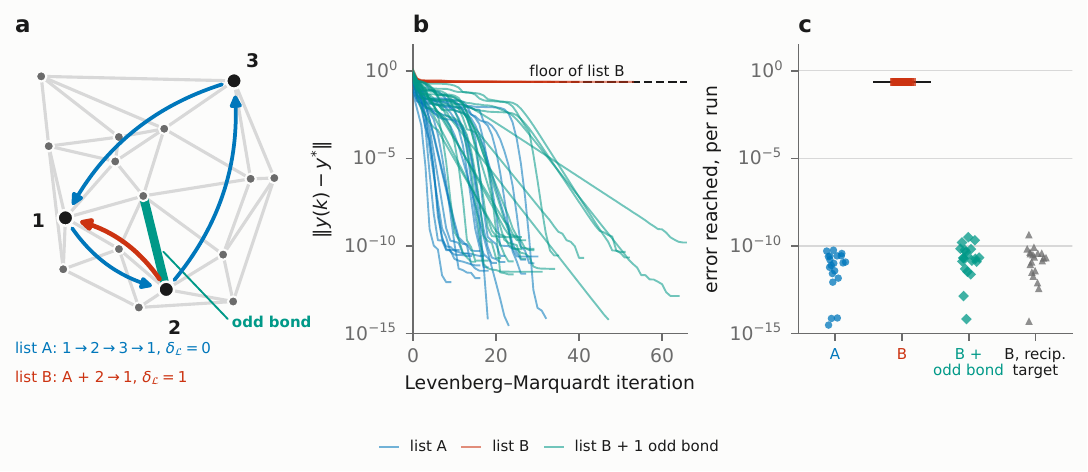}
  \caption{\textbf{Two task lists on one network, with and without one odd
  bond.}
  \textbf{(a)}~The network of Fig.~\ref{fig:setup}a, with two task lists on the
  $x$-coordinates of its three marked nodes. An arrow $u\to v$ is one task:
  force at $u$, displacement read at $v$. List~A is the cycle
  $1\to2\to3\to1$, with $\delta_{\mathcal{L}}=0$; list~B adds $2\to1$, one
  digon, so $\delta_{\mathcal{L}}=1$. Each list asks for the responses of a
  reference passive network with $1\to2$ raised by $10\%$ and, on list~B,
  $2\to1$ lowered by $10\%$. The thick bond is the single odd bond,
  chosen before training.
  \textbf{(b)}~Error $\|y(k)-y^{*}\|$, with $y(k)$ the values the tasks read
  at stiffnesses $k$, against Levenberg--Marquardt iteration
  for all $60$ runs of three arms, $20$ independent starts each: list~A,
  list~B, and list~B with the odd bond. The dashed line is the floor of
  Theorem~\ref{thm:taskfloor} for list~B, computed from the target before
  training.
  \textbf{(c)}~Final absolute error $\|y(k)-y^{*}\|$ of each of the $80$ runs
  of experiment~23, one column per arm (target norms $1.84$ for list~A, $2.37$
  for list~B and $2.59$ for the control). The fourth column is a control on list~B that raises both
  $1\to2$ and $2\to1$ by $10\%$; the short line across the second column marks
  the floor of panel b.}
  \label{fig:example}
\end{figure}

\paragraph{Non-reciprocity restores what parameters cannot.}
The dimensions reciprocity forbids are the antisymmetric ones. Adding bonds
therefore cannot buy them, and adding odd couplings can, near the passive
network, wherever the passive wedges span them. Non-reciprocal metamaterials
\citep{Scheibner2020,Brandenbourger2019,Fruchart2023} can therefore restore
capacity, and not only extend it. The same fact
also bounds the gain in advance: a network that already attained its passive bound at the
symmetry branch of \eqref{eq:bound} can gain at most the $p(p-1)/2$ forbidden
dimensions from breaking reciprocity, and in the odd-coupling construction
tested here it gains exactly those. A network limited instead by its bond
count may gain more, for the unrelated reason that odd couplings also add
parameters. On the sixteen-node network of experiment~23 the restoring
element is one odd bond, which removes a floor of $9.9\%$ of the target norm
in $20$ of $20$ runs.


\section{Limitations}
\label{sec:limitations}

Each item below states what remains open and, where we can, how to test it.

\begin{itemize}\itemsep2pt
\item \emph{Attainment is proved instance by instance, not in general.}
  Theorem~\ref{thm:main} is an upper bound and holds at every admissible
  parameter value, with no genericity condition.
  Theorem~\ref{thm:duality} reduces equality to the transversality of
  $\mathcal{Y}$ and $KV^{\perp}K$, and Proposition~\ref{prop:generic} turns one
  exact certificate into a statement about almost every configuration of that
  graph; we supply such certificates for $370$ instances, $237$ of them inside
  criterion~\eqref{eq:criterion}. No proof covers all graphs at once; it would
  have to establish that transversality. The graph enters through $V$,
  generated by the rigidity
  dyads~\citep{Laman1970,AsimowRoth1978,Jacobs1995}, but the rigidity matroid
  alone does not decide the question. Since $\bw_b=C\bq_b$ with $C$
  invertible, $\{\bw_b\}$ and $\{\bq_b\}$ carry the same linear matroid, whereas
  the deficit is a rank condition on the quadratic images
  $\bw_b[T]\bw_b[S]^\T$, and passing to that lift is not a matroid operation:
  vectors with the same matroid can have outer-product spans of different
  dimension. The deficit also moves with the stiffnesses and with the choice of
  $T$ and $S$, which no generic-rigidity invariant sees. The framework does
  settle the question by itself in the isostatic case
  $n_b=n_{\mathrm{free}}$. There the matrix with rows $\bq_b^\T$, restricted to
  the free degrees of freedom, is square, $\bw_b$ is $1/k_b$ times the
  corresponding column of its inverse, and each passive Jacobian column is
  $1/k_b^{2}$ times a fixed matrix, so their span, and with it attainment, does
  not depend on $k$ at all. We regard the general question as the natural next
  step, which the reduction above makes well-posed.
\item \emph{Two further statements are measured, not proved.} That the passive
  wedges of Theorem~\ref{thm:recovery} span $\Lambda^2\R^p$, on which
  Corollary~\ref{cor:submersion} rests, is verified in $251$ of $252$
  full-overlap configurations (Sec.~\ref{sec:numerics}) and proved for no class
  of graphs. The ceiling of
  Proposition~\ref{prop:multirank} is attained in $264$ of $264$ configurations
  on networks whose bond count is far from binding, and missed in $34$ of $66$
  on the $8$-node networks, where it is close to binding.
\item \emph{Two learning procedures, both in software.} We tested
  Levenberg--Marquardt, with log-stiffnesses clamped at $\pm12$, which the
  winning fits of $134$ of the $144$ forbidden-target runs reach to within
  $10^{-9}$ (Appendix~\ref{app:stats}), and coupled learning, whose bond-local
  update direction a material could plausibly implement; the second finds the
  same floor. Neither is a material, though the
  platforms exist: coupled learning has been run in a physical elastic
  network~\citep{Altman2024} and in analogue electrical
  ones~\citep{Dillavou2022,Dillavou2024}, and reprogrammable disordered spring
  networks with settable bond stiffnesses have been
  fabricated~\citep{Pashine2023}, alongside architected lattices of tunable
  beams~\citep{Lee2022}. Whether a physical realisation, with noise, hysteresis
  and a bounded stiffness range, also settles on the floor is untested.
  Theorems~\ref{thm:floor} and~\ref{thm:floorpd} bound it regardless, since
  both are properties of the reachable set.

  That measurement can also go against us without refuting a theorem. An
  antisymmetric part significantly above the measurement error, in a sample
  believed passive and linear, would refute the modelling hypothesis that the
  sample has a symmetric stiffness operator, rather than
  Theorem~\ref{thm:main}. A measured error \emph{below} the predicted floor
  would indicate that the sample is not the passive, linear, reciprocal object
  the floor is computed for, or that the measurement is not resolving the
  shared block. Only the first of these two explanations is informative about
the physics, and
  the measurement that produced the discrepancy, not the floor, decides which
  applies.
\item \emph{Coupled learning converges slowly.} Two of the $24$
  forbidden-target runs ended more than $1\%$ above the floor within our step
  budget, and eleven of the $24$ reached that budget instead of stopping on
  their own (Sec.~\ref{sec:learning}). We report the two as they stand, without
  running them longer.
\item \emph{Coupled learning assumes a supervisor.} Only its update direction
  is bond-local: the step is rescaled by the largest relative change over all
  bonds, and the rate schedule, the stopping test and the choice of which
  state to return read the global loss (Sec.~\ref{sec:learning}). A material
  with no external observer is not tested.
\item \emph{The forbidden targets were built so that the floor is attainable.}
  $R^{*}=R(k_0)+\varepsilon\Psi$ has a reachable symmetric part by construction,
  and its shared block stays positive definite, so Fig.~\ref{fig:learning}b
  tests that training finds the constrained optimum rather than that
  Theorem~\ref{thm:floor} is tight in general. It is not tight on the generic
  targets of experiment~13 (Sec.~\ref{sec:learning}).
  Theorem~\ref{thm:floorpd} accounts for most of the difference, but it too is
  a lower bound, and it is loose whenever the part of the target inside
  $\mathcal{S}^{+}_P$ is itself unreachable. Both are reported, and the
  diagnostic of Sec.~\ref{sec:consequences} is stated only in the direction the
  bound supports.
\item \emph{The allowed-target control degrades at large perturbations}
  (Sec.~\ref{sec:learning}). We attribute this to the nonlinearity of
  $k\mapsto R(k)$ at such amplitudes; we have not proved that attribution.
\item \emph{Local, not global.} The rank of the Jacobian is the local dimension
  of the reachable set at a generic parameter point. Positive definiteness costs
  no dimensions (Remark~\ref{rem:pd-scope}). Because the count is local, it
  does not characterise the boundary of the reachable set, or whether a
  learning rule can approach a target that lies on it.
\item \emph{Tangent response at any pre-load; the finite map on a branch;
  branch crossings open.} The bound applies to the tangent response at every
  configuration, deformed or not, because the tangent stiffness is a Hessian;
  this is verified at bond strains up to $60.8\%$. Because the network is
  conservative, the symmetry of the tangent response extends to the finite port
  map on any branch along
  which the relaxed energy is a single smooth function of the held port
  displacements (Sec.~\ref{sec:numerics}). That extension is a corollary of the
  Hessian property and not an independent test; the finite map itself remains
  unmeasured. What a path gains by crossing branches, where the relaxed energy
  stops being one function and the argument does not reach, is open, and it is
  the natural next experiment.
\item \emph{One realisation of oddness.} The odd coupling used here is a
  transverse force proportional to bond elongation. Other realisations of
  non-reciprocity may reach a different subspace; we do not claim ours is
  canonical, and Proposition~\ref{prop:odd} shows only that no symmetry is
  forced, not that every construction achieves the full space. Recovery was
  tested only in this realisation; the non-symmetric damping term of
  Sec.~\ref{sec:numerics} was checked for rank only.
\item \emph{The work per cycle, \eqref{eq:cycle}, is a linear, quasi-static
  law for force cycles.}
  In the notation of Sec.~\ref{sec:price}, $\Gamma$ is the signed area in the
  $(c_1,c_2)$ plane, and equals the area in
  force space only when $\mathbf{h}_1,\mathbf{h}_2$ are orthonormal. The
  constitutive law is linear, so \eqref{eq:cycle} is not exact for a
  geometrically nonlinear network, and a quasi-static loop presumes that the
  loaded equilibrium is dynamically stable. A displacement-controlled cycle at
  $P$ is governed by $\operatorname{anti}(C_P^{-1})$, which depends on
  $\operatorname{sym}C_P$ as well as on $B_a$.
\item \emph{Elements, dimension and size.} In our tests, three-dimensional
  networks and angular (bending) springs give the same result as
  two-dimensional bonds (Sec.~\ref{sec:numerics}). The rank-one structure of the stiffness derivatives
  $\partial K/\partial\theta_r=\bq_r\bq_r^\T$ of Sec.~\ref{sec:numerics}, on which the cheap Jacobian and
  the attainment analysis rest, holds for any element whose generalised strain
  is a linear functional of the displacement, which covers every element tested
  here. An element parametrised by a single stiffness acting on several
  independent strains would not be rank one, and we have not tested one. Size is
  untested beyond the networks used here, which run up to $56$ nodes. Every
  computation except the two tests of experiment~18 has a single
  implementation (Sec.~\ref{sec:certificates}).
\item \emph{One frequency at a time.} The frequency-domain result bounds the
  block at each $\omega$ separately. Learning at several frequencies at once
  stacks blocks drawn from different $G(\omega)$; whether the stacked problem
  carries its own constraint, weaker, stronger or none, we have not examined.
\item \emph{The layout law is about dimension, not reachability, and about
  one budget.}
  Proposition~\ref{prop:budget} maximises the \emph{ceiling} at a fixed
  number of accessed degrees of freedom, a shared one charged once; under a
  sensors-plus-actuators budget a layout with $p\le1$, the disjoint one
  included, stays best (see after its proof). Neither it nor Proposition~\ref{prop:count}, which lets the two
  layout rules be compared on one fixed object, says whether a particular
  target inside that ceiling is reachable, and at full overlap the floor of
  Theorem~\ref{thm:floorpd} still forbids the antisymmetric part.
  Experiment~14 (Sec.~\ref{sec:consequences}) tests the comparison only at
  contaminations up to $t=0.20$. How the errors of the full-overlap and
  disjoint layouts compare at larger $t$, and whether they cross, we have not
  measured.
\item \emph{Comparison with published bounds is with the method, not with any
  particular system.} The counting bounds quoted in the literature are stated for
  specific systems, including feedback-controlled robotic ones. We compare with
  the counting \emph{method}; we have not re-derived any particular published
  bound.
\end{itemize}

\section{Conclusion}

That a passive network's response matrix is symmetric is classical; the
results above add three consequences for a fixed network trained in place
under a drive by forces or currents. First, the floor of
Theorem~\ref{thm:floor} is written down from the target and the layout before
training and binds every passive linear network, at any size, stiffness or
learning rule; both rules tested here, in simulation, end on it in nearly
every run when the target's symmetric part is reachable, the contrastive one
within its step budget (Sec.~\ref{sec:learning}). Second, what reciprocity
charges a list of single drive--read tasks is set by its reversed pairs, not
by its shared terminals (Sec.~\ref{sec:tasklist}). Third, odd
couplings~\citep{Scheibner2020,Brandenbourger2019,Fruchart2023} restore the
forbidden directions, which adding bonds cannot, near the passive network
wherever it passes the rank test of Theorem~\ref{thm:recovery}, on bonds
chosen from the passive network before any odd bond is built.

The capacity of a passive mechanical network driven by forces is at most
$\min(n_b,\,m_Tm_S-p(p-1)/2)$, where $p=|T\cap S|$ counts the degrees of
freedom that are both driven and read. When $p\ge2$ and the bond count does not
bind first, this is below the $\min(n_b,m_Tm_S)$ that a parameter count gives. The
shortfall is a consequence of Maxwell--Betti reciprocity alone, and at full
overlap it approaches half the space. What a laboratory usually imposes is a
finite list of demands. For a list, an invariant $\delta_{\mathcal{L}}$ of the
list itself replaces $p(p-1)/2$ and is never larger; on coordinate tasks it
counts the pairs of terminals instrumented in both directions
(Sec.~\ref{sec:tasklist}), and a full block is the list that contains every
drive--read task.

Adding parameters does not recover the forbidden subspace; breaking reciprocity
withdraws the obstruction. In the odd-coupling networks tested here the
directions recovered are exactly that subspace wherever the wedges of
Theorem~\ref{thm:recovery} span, as they do in nearly every full-overlap
configuration measured. Where they span, every sufficiently small
antisymmetric part of a shared block is realised exactly in projection, at
unchanged passive stiffnesses, by odd couplings that can be confined to $p(p-1)/2$ bonds
(Corollary~\ref{cor:submersion}); the network is then active. The symmetric
part of the response is not held fixed with it, and at the largest target
tested the symmetric part of $K_{\!f\!f}$ loses positive definiteness in $261$
of the $336$ solves with the couplings confined to the $p(p-1)/2$ bonds of the
pivoted-$QR$ selection, and in $115$ with every coupling free
(Sec.~\ref{sec:numerics}, Remark~\ref{rem:submersion}). How much
non-reciprocity the network must carry is bounded from below by the target
alone: any realising network whose symmetric part is positive definite is at
least as non-reciprocal as the target's shared block, as measured by the ratio
of odd to even part (Proposition~\ref{prop:ratio}).

Theorem~\ref{thm:floorpd} is the practical form: for any target, one line of
arithmetic gives a lower bound on the error that reciprocity and passivity
leave behind, and a run that reaches that bound has finished
(Sec.~\ref{sec:consequences}). On targets with a reachable symmetric part, a
second-order optimiser and a contrastive rule with a bond-local update
direction both reach the floor in nearly every run, within $0.1\%$
(Sec.~\ref{sec:learning}). Only the reciprocity term comes from a symmetry, but on a
generic target the passivity term is comparable: adding it to the floor cuts
the median gap between the error reached and the floor from $36.9\%$ to
$1.6\%$ (experiment~13).

The reciprocity term applies to any response block of a symmetric operator
(Theorem~\ref{thm:universal}), so it survives three dimensions, bending, the
tangent response about finitely deformed states, and finite frequency. It
reaches the finite port map on a single equilibrium branch as well, but by an
envelope argument that needs the network to be conservative along the branch,
not symmetry alone (Sec.~\ref{sec:numerics}). It holds in resistor
and flow networks as well, whenever the measured block is a block of the
inverse operator, which the imposed-displacement drive that most hardware uses
is not (Sec.~\ref{sec:dirichlet}). The passivity term of
Theorem~\ref{thm:floorpd} is stated for the positive-definite static setting
and is not carried over to the dynamic Green function. Symmetry and positive
definiteness together also rule out
the two-target task of \citet{Du2026}, through the spectral bound of
Theorem~\ref{thm:multiclamp} and the identity \eqref{eq:duS}: within the linear model their deposited
simulation integrates, no symmetric positive-definite chain can meet both
targets (Sec.~\ref{sec:consequences}).

The best layout depends on what binds and on what is counted
(Sec.~\ref{sec:consequences}). Where reciprocity binds, overlapping the sensors
onto the actuators buys about twice the reachable ceiling per accessible degree
of freedom, a shared one counted once (Proposition~\ref{prop:budget}). The gain
collapses where the parameter count binds, and under a fixed count of sensors
plus actuators, with a shared terminal charged twice, the ceiling is instead
maximal at $p\le1$. The overlap recommendation concerns a platform asked for
the whole block: on a list of tasks a shared terminal is by itself free, and
on coordinate tasks only pairs instrumented in both directions cost
(Sec.~\ref{sec:tasklist}). Overlap never buys the antisymmetric part itself; a target that needs it needs disjoint
read-out or broken reciprocity.

At all nineteen published layouts we tabulated, all of them disjoint, the
forfeited norm fraction is zero, and Theorem~\ref{thm:main} corrects no
capacity number in this literature. We found no in-place learning layout
driven by forces or currents that pays; the one that will instruments a pair
of terminals in both directions and asks its two responses to differ. A parameter count, the usual
estimate of capacity in this field, cannot see a symmetry of the response
operator. That symmetry bounds what the material can represent before the
parameter count becomes the binding constraint, and it does so before
training: the floor is written down from the target and the layout alone. A
target whose predicted floor is non-zero is certified unreachable, under a
drive by forces on $S$ read on $T$, by every passive, linear network of any
size, at any stiffness, under any learning rule.

Three of the largest questions are left open. Attainment on all graphs at once
reduces, through Theorem~\ref{thm:duality}, to a transversality that the
rigidity matroid alone does not decide. The rank we compute is local, so the
boundary of the positive-definite reachable set, and whether a learning rule
can approach a target on it, are not characterised. And both learning procedures ran in software; whether a physical
network, with noise and hysteresis, also settles on the floor is a measurement
that existing coupled-learning platforms are equipped to make. The direct test
is such a network driven by forces or currents, read at each shared pair in
both directions and trained on a target whose symmetric part it can realise
and whose predicted floor is non-zero. We predict that training stalls within
measurement error of that floor and not below it. A stall below it would show
that the sample is not the passive, linear, reciprocal object the floor is
computed for, or that the measurement does not resolve the shared block; one
well above it would count against the
claim that learning rules stop on the floor.


\section*{Data availability}
The raw result tables and the certified configurations generated in this study
are archived on Zenodo at \zenodoi{10.5281/zenodo.22952768} and also
included with this preprint as ancillary files, downloadable from its arXiv
abstract page. They comprise $640$ cases from
experiment~01; $540$ from experiment~02; $19737$ Jacobian columns, $120$ subspace comparisons
and $2418$ rank cases from experiment~03; and $1872$ training fits from experiment~04. Those fits are $36$ configurations at $52$ fits each: four amplitudes
$\times$ three arms $\times$ four restarts, which is $48$, plus the four
restarts of the random-target arm, fitted once per configuration outside the
amplitude loop. They are reported as the $144$ rows of the deposited table, one
per configuration and amplitude, and the random-target result of a
configuration is copied unchanged into each of its four rows, not recomputed.
A companion table, \texttt{exp04\_curves.csv}, holds the three convergence
histories drawn in Fig.~\ref{fig:learning}a in $61$ rows, each to its own last
iteration ($57$, $61$ and $28$ recorded iterations). The deposit further holds
$336$ subspace
intersections from experiment~05; $432$
certificate instances from experiment~06, together with the $370$ certified
configurations themselves, each carrying node coordinates, bond list, index sets
and
integer stiffnesses, so that each certificate can be rechecked without rerunning
anything; and $24$ matched coupled-learning pairs from
experiment~07, whose \texttt{exp07\_beta.csv} holds the end points of the three runs of
the nudge-amplitude sweep, one of them repeating a forbidden-target run
already among the pairs, and whose
\texttt{exp07\_curves.csv} holds the two histories of Fig.~\ref{fig:coupled}a,
the allowed one over $15\,411$ steps and the forbidden one to step $8\,221$,
where it stops. It holds $432$ cases in three dimensions and with bending from
experiment~08; $12$ matched layout pairs from experiment~09, with the two
histories of Fig.~\ref{fig:design}a in \texttt{exp09\_curves.csv}, the
overlapping one over all $40\,000$ steps and the disjoint one to step
$32\,046$, where it stops; $59$ finitely deformed
states from experiment~10; and $216$ frequency-domain cases and $18$
non-symmetric
damping controls from experiment~11. It holds $315$ layouts at fixed
instrumentation
budget from experiment~12; and, from experiment~13, $108$ training runs in
three target classes (part~3), $72$
realised shared blocks (part~2) and the five rows of
\texttt{exp13\_fraction.csv}, the Monte Carlo of part~4 over $4000$ draws at
each of $m=3,4,6,8,10$. Two parts of experiment~13 deposit no table and
instead regenerate themselves exactly from their fixed integer seeds when the
script is run: part~1, the $3200$ random blocks at $p=1,\dots,8$ on which the
closed form of \eqref{eq:floorpd} is checked against a direct projection, and
part~5, the $322$ exact-rank
index-set draws on the regular square lattice of which $9$ fail to attain the
bound. The deposit holds $360$ trained layouts at fixed
instrumentation budget on a prescribed task from experiment~14. It holds the
$19$ tabulated layouts, and the $5$ excluded ones, read from the
literature in
experiment~15, and the $14$ reported quantities of experiment~16, over
$4024$ identity checks, $1200$ balance checks, $4000$ sign cases, $20$ rank
configurations, $600$ cycle triples of each kind and $4\times10^{5}$ chain draws.
Then come the
$531$ rows of experiment~17: eleven diagnostics at all twenty-one deposited
epochs of each of the two runs of \citet{Du2026}, Fig.~2 ($462$ rows); the
fifteen quantities of the shared initial condition; their own epoch-zero error
for each run ($2$); twenty-two per-run summaries of each trajectory, of the
Gershgorin guard and of the reproduction of their learning loop ($44$); and the
$8$ rows of the random-chain ensemble that puts the epoch-zero agreement in
context, three recording its size, seed and positivity check and five its
ratio statistics. Experiment~17 needs their deposited data, which we do
not redistribute; it is at \zenodoi{10.5281/zenodo.18544299} and the script
names the directory it expects. Experiment~18 deposits $1008$ wedge-comparison
rows ($945$ distinct configurations, Sec.~\ref{sec:numerics}), $251$
sparse-odd selections and $243$ multi-clamp cases, and
experiment~18b re-derives two of those claims from a second implementation.
Experiment~19 deposits the $4786$ Newton solves behind
Corollary~\ref{cor:submersion}, together with the counts quoted in
Remark~\ref{rem:submersion}, and experiment~20 the $720$ rank rows behind
Proposition~\ref{prop:multirank}. Experiment~21 deposits the material behind
Sec.~\ref{sec:tasklist} in five tables: the $7245$ task lists of its first
three parts, carrying the deficit $\delta_{\mathcal{L}}$, the digon count and
the full-block identity in exact rational arithmetic; the $1637$ floor
instances of its fourth part; the $790$ rows of its fifth and sixth, which read
the $19$ tabulated layouts, the excluded ones, one further imposed list
(Fig.~4 of \citet{Altman2024}) and Li and Mao's measurement list as task lists; the $36$ rows of its seventh, the restart sweep of
experiment~13's part~3; and the $43$ rows of its eighth, which carry all three
arms of that eighth part (the prescribed
antisymmetric blocks realised by three odd bonds with the stiffnesses free, and
the two fixed-$k$ controls, on three and on ten bonds). Experiment~22 deposits
the material behind Sec.~\ref{sec:price} in five tables and a machine-readable
summary: the $40$ cases of its zeroth part, which check the two identities the
proofs rest on; the $3000$ abstract matrices and the $3360$ odd-bond
configurations on which \eqref{eq:ratio} and the passive substitution are
measured; the ten exhaustive support searches behind
Corollary~\ref{cor:price}; and the $336$ configurations of
Proposition~\ref{prop:torquefree}, on the ensemble of experiment~19.
Experiment~23 deposits the $80$ runs of Fig.~\ref{fig:example}, their
training histories ($2042$ recorded points: the $80$ starting points and $1962$
iterations), a summary of the two lists and their floors, the $760$ runs with
each bond in turn as the single odd bond, the $400$ runs of the clamp
re-run, and the pre-registration of the $80$ runs. Every case
excluded from a reported count (Appendix~\ref{app:stats}) is supplied with the
data, save the one large-deformation case that produced no row because its
equilibrium did not converge.

\section*{Code availability}
The five source modules, the twenty-four experiment scripts, and the eight
plotting scripts that produce the figures are included as ancillary files
with this preprint. The scripts produce every number in this paper except
nine groups, produced by one-line variants of the deposited scripts or by
one-off computations outside them and not deposited as separate outputs, the
first six reported in Appendix~\ref{app:stats} (one of them also in
Appendix~\ref{app:params}) and the last three where the text uses them: the clamp and rate-floor sensitivity re-runs (the training
classes made to log saturation, condition numbers and the largest stiffness
reached; the log-stiffness clamp moved from $12$ to $8$ and to $24$ in
experiment~04; the stiffness clamp widened from $[10^{-3},10^{3}]$ to
$[10^{-6},10^{6}]$ in experiments~07 and~09, and the rate floor lowered from
$10^{-5}$ to $10^{-8}$ in experiment~07); the per-draw counts of experiment~03 (the three draws retained
instead of their median); the recount of experiment~10 at the cutoff
$10^{-9}$ (its count of $21$ in $59$, though not its per-step agreement, is
also readable off the deposited gap columns of Appendix~\ref{app:stats}); the natural-frequency
range of experiment~11 (the generalised eigenvalues of $(K,M)$); the
null-column threshold sweep of experiment~18; the genericity counter-example
of Appendix~\ref{app:stats}, that is the Delaunay arm of experiment~08
enlarged from three seeds to eight and, on the twelve-node network at seed~$7$
it produces, the exhaustive sweep of the $5985$ full-overlap index sets of
four free coordinates over five stiffness draws; the exact ranks of the
degree-two vertex hung on the complete graphs $K_{11}$, $K_{15}$ and $K_{19}$
(Sec.~\ref{sec:numerics}); the $1\times1$-block check at $a_b=3.7$
(Remark~\ref{rem:odd-hyp}); and the singular values of the one frequency-sweep
exception of experiment~11 (Sec.~\ref{sec:numerics}). The counts quoted for
release 1.3.0 and for the undeposited three-dimensional pass in
Appendix~\ref{app:stats} describe superseded code and are not reproducible
from this deposit. The deposited scripts are
archived at the version
used here on Zenodo at \zenodoi{10.5281/zenodo.22952768}, which is the citable
snapshot; \zenodoi{10.5281/zenodo.22238334} resolves to the most recent version.
The releases of that record are numbered, and those numbers are what
Appendix~\ref{app:stats} names when it distinguishes one deposited run from
another: the archive's \texttt{README} opens with the release it belongs to and
gives a test on the files themselves, so a reader holding the archive can say
which release it is without going back to the record. The only third-party
dependencies are NumPy, SciPy and Matplotlib; this was verified by running all
thirty-two scripts in a clean environment containing those three packages and
nothing else. Every experiment is deterministic: on a fixed machine and software
stack each reproduces its own result files byte for byte on repeated runs, and
the twenty-four experiment scripts together take about forty-four minutes on one
core.
Across machines, LAPACK builds or BLAS threadings, we claim reproducibility of
the reported quantities, not byte-identical floating-point output; the two
caveats that this entails are stated in the ancillary \texttt{README} and in
Appendix~\ref{app:repro}.

\section*{Acknowledgements}
This work grew out of a remark in the supplementary material of
\citet{Du2026}, whose authors flagged the omission that this paper quantifies.
We thank Yao Du for confirming the reading of the protocol of their Fig.~2
used in Sec.~\ref{sec:consequences}, and Corentin Coulais and Jonas Veenstra
for their help in posting the preprint. All three are authors of
\citet{Du2026}. Both kinds of help were technical, and neither implies that
they endorse the analysis built on that reading or its conclusions.

\section*{Funding}
The authors declare that no specific funding was received for this work.

\section*{Author contributions}
T.-S.V.\ conceived the study, developed the theory, wrote the code and drafted
the manuscript. H.-G.N.\ and B.-V.T.\ supervised the work and contributed to the
mechanical formulation. Q.-B.N.\ and S.K.\ contributed to the numerical design
and verification. All authors discussed the results and revised the
manuscript.

\section*{Competing interests}
The authors declare no competing interests.

\section*{Use of generative AI}
A generative AI assistant (Anthropic Claude) was used for language editing, code
development and numerical checking. The authors formulated the research questions
and mathematical claims, independently verified all proofs, references and
numerical outputs, and take full responsibility for the content of this work.

\bibliographystyle{unsrtnat}
\bibliography{refs}

\appendix
\section{Published layouts}
\label{app:layouts}

This appendix gives the layout-by-layout data behind the analysis of published
layouts in Sec.~\ref{sec:consequences}. Table~\ref{tab:layouts} lists
nineteen layouts cited in this paper for which a source fixes the driven and
read-out counts separately, with the fraction of a generic target that reciprocity puts
out of reach at each. The entry is zero at all nineteen under either law, for
the reason given in Sec.~\ref{sec:consequences}: every one of them drives one
set of degrees of freedom and reads another. That section also names the five
layouts excluded from the count.

\begin{table}[ht]
  \centering
  \small
  \caption{\textbf{The forfeited norm fraction at published layouts.}
  $m_S$ driven degrees of freedom, $m_T$ read out, $p=|T\cap S|$ both, and the
  fraction of a generic target that reciprocity puts out of reach, which is
  \eqref{eq:fracgen} under a force drive. The drive column records how the inputs enter. \emph{Force}: the
  inputs are applied forces and the measured block is $C[T,S]$, the object of
  Theorem~\ref{thm:main}. \emph{Imposed}: the inputs are held at prescribed
  displacements, strains, pressures or voltages and the output is read wherever
  the network puts it, the drive of Sec.~\ref{sec:dirichlet}, whose law
  replaces \eqref{eq:fracgen} and constrains nothing at $p=0$. Every row was read from the paper cited,
  and from its released code where there is one; independent scalar channels
  are counted throughout, a node clamped at a fixed ground or bias voltage
  counting as one. The table is a sample of the cited demonstrations, not a
  census (Sec.~\ref{sec:consequences}). Fig.~2 of \citet{Du2026}, which exchanges the driven
  and read-out sets between its two targets, is treated in
  Sec.~\ref{sec:consequences} instead.}
  \label{tab:layouts}
  \begin{tabular}{lllccc r}
    \toprule
    Layout & Platform & Drive & $m_S$ & $m_T$ & $p$ & fraction \\
    \midrule
    \citet{Rocks2017}, Fig.~1        & tuned spring network      & imposed & 1  & 1  & 0 & $0$ \\
    \citet{Rocks2017}, Fig.~3A       & tuned spring network      & imposed & 1  & 3  & 0 & $0$ \\
    \citet{Rocks2017}, Fig.~3B       & tuned spring network      & imposed & 2  & 2  & 0 & $0$ \\
    \citet{Stern2021}, Fig.~2        & flow network              & imposed & 10 & 10 & 0 & $0$ \\
    \citet{Stern2021}, Fig.~3        & elastic network           & imposed & 10 & 3  & 0 & $0$ \\
    \citet{Stern2021}, Fig.~4        & flow network              & imposed & 25 & 2  & 0 & $0$ \\
    \citet{Altman2024}, Fig.~2       & elastic network (expt.)   & imposed & 1  & 1  & 0 & $0$ \\
    \citet{Altman2024}, Fig.~5       & elastic network (expt.)   & imposed & 2  & 2  & 0 & $0$ \\
    \citet{Dillavou2022}, Fig.~3A    & resistor network          & imposed & 3  & 3  & 0 & $0$ \\
    \citet{Dillavou2022}, Fig.~3B    & resistor network          & imposed & 3  & 2  & 0 & $0$ \\
    \citet{Dillavou2022}, Fig.~3D    & resistor network          & imposed & 5  & 3  & 0 & $0$ \\
    \citet{Dillavou2024}, Fig.~3     & nonlinear analogue net.   & imposed & 4  & 2  & 0 & $0$ \\
    \citet{Dillavou2024}, Fig.~4     & nonlinear analogue net.   & imposed & 3  & 1  & 0 & $0$ \\
    \citet{Li2024}, Fig.~2           & mechanical neural net.    & force   & 1  & 2  & 0 & $0$ \\
    \citet{Li2024}, Fig.~3           & mechanical neural net.    & force   & 1  & 4  & 0 & $0$ \\
    \citet{Li2024}, Fig.~4           & mechanical neural net.    & force   & 4  & 3  & 0 & $0$ \\
    \citet{Pashine2023}, Fig.~2      & allosteric metamaterial   & imposed & 1  & 1  & 0 & $0$ \\
    \citet{Pashine2023}, Fig.~5      & allosteric metamaterial   & imposed & 1  & 2  & 0 & $0$ \\
    \citet{Du2026}, Fig.~1           & robotic metamaterial      & imposed & 1  & 5  & 0 & $0$ \\
    \bottomrule
  \end{tabular}
\end{table}

\section{Assembling the Jacobian}
\label{app:jacobian}

The rank counts of Sec.~\ref{sec:numerics}, outside its two finite-difference
exceptions, are taken from the Jacobian \eqref{eq:jaccol}, which is assembled
for all bonds at once. Let $Q$ be the $n_b\times n_{\mathrm{free}}$ matrix of
compatibility rows restricted to free degrees of freedom and $Q_{t}$ the
corresponding matrix of transverse rows. Three products are needed,
\begin{equation*}
  L_n = C[T,:]\,Q^\T, \qquad
  L_t = C[T,:]\,Q_{t}^\T, \qquad
  \Gamma = Q\,C[:,S] ,
\end{equation*}
after which the columns are the outer products
$-L_n[:,b]\,\Gamma[b,:]$ for $\partial/\partial k_b$ and
$-L_t[:,b]\,\Gamma[b,:]$ for $\partial/\partial a_b$. The cost is
$O\!\left(n_b\,m_Tm_S\right)$ after one inversion. The result is accurate to
machine precision, which is what makes a rank cut at $10^{-9}$ meaningful:
across all $2418$ rank cases the ratio of the last retained to the first
discarded singular value exceeds $10^{9}$ in $99.50\%$ of cases and $10^{3}$ in
$99.88\%$. The twelve marginal cases occur at $m=8$, $9$ and $10$, and every one
of them lies outside criterion \eqref{eq:criterion}, where we do not claim
attainment. These figures are medians over the three stiffness draws at each
configuration; Appendix~\ref{app:stats} gives the per-draw counterpart.


\section{Parameters of the numerical experiments}
\label{app:params}

This appendix collects the settings that the main text uses but does not state,
so that Secs.~\ref{sec:price},~\ref{sec:numerics},~\ref{sec:learning}
and~\ref{sec:consequences} can be
reproduced from the paper alone. None of these values is fitted to anything.

\paragraph{Pinning.}
In two dimensions the three pinned degrees of freedom are both coordinates of
the leftmost node and the $y$-coordinate of the rightmost node
(experiment~18b, written independently, pins both coordinates of its first
node and the $x$-coordinate of its second), leaving $n_{\mathrm{free}}=2N-3$
free coordinates. In
three dimensions six are pinned: one node fully, two coordinates of a second
and one of a third, the three nodes being the extremes along three directions,
leaving $n_{\mathrm{free}}=3N-6$. Every drawn passive configuration satisfies
$K_{\!f\!f}\succ0$; configurations reached with odd couplings or by training are
checked where stated.

\paragraph{Network ensembles.}
Outside the main rank sweep of Sec.~\ref{sec:numerics}, some experiments use
other networks: Delaunay triangulations of ten points in part of
experiments~05, 06 and~19, of nine in part of experiment~22 and of eight in
part of experiment~20; triangular lattices with positional disorder $0.08$ in
experiment~05, and $0$ or $0.1$ on the small lattices of part of
experiment~22; and a regular square lattice with diagonals in part~5 of
experiment~13. The driven and read-out sets are drawn at random from the free
degrees of freedom subject to the prescribed overlap $p$. In the three-dimensional arm of
the generality sweep the points are drawn uniformly in the unit cube before
being tetrahedralised. An angular spring acts on the change of the angle
between two bonds $j\to i$ and $j\to k$ meeting at node $j$, which to first
order is
\begin{equation*}
  \hat{t}_{jk}\cdot(\bu_k-\bu_j)/L_{jk}
  \;-\;
  \hat{t}_{ji}\cdot(\bu_i-\bu_j)/L_{ji} ,
\end{equation*}
where $\hat{t}$ is, as in Sec.~\ref{sec:setup}, the in-plane perpendicular to
the direction it carries and not the tangent along it, so that each term is a
first-order rotation angle and their difference is the first-order change of
the angle. It is
again linear in $\bu$, so that the perturbation it contributes is again
symmetric and rank one. Every pair of bonds meeting at a node gives one such
spring; where a network has more of them than it has bonds, a random subset of
$n_b$ is kept, and their stiffnesses come from the same log-normal law as the bonds.

Experiment~16 departs from the draws of Sec.~\ref{sec:numerics}: it draws
stiffnesses uniformly on $[0.5,1.5]$ (parts 1, 2, 4 and 6) and, in its
sign-law part~3, on $[0.05,3]$; its odd couplings are uniform on $[-1,1]$
in parts 3 and 6 and $\mathcal{N}(0,0.3^2)$ in its rank part~4. Its part~5 is
not a network draw at all but the six-unit chain of \citet{Du2026}, whose
on-site stiffnesses are uniform on $[10^{-3},1]$, whose coupling stiffnesses are
uniform on $[10^{-3},0.5]$ and whose odd couplings, in the control arm, are
uniform on $[-0.5,0.5]$.

\paragraph{Levenberg--Marquardt.}
The log-parametrisation keeps the stiffnesses positive, and $u$ is clipped to
$[-12,12]$ on acceptance of a step, so that the stored state and the state the
Jacobian refers to are always the same vector. Each fit runs at most $600$
iterations and stops when no damping value in an iteration's inner search
improves the residual, when the
residual falls below $10^{-14}$, or when the residual has fallen by less than
$10^{-13}$ of its starting value over six consecutive accepted steps. The
damping starts at $\lambda_0=10^{-3}$, and each iteration tries at most $30$
damping values; $\lambda\leftarrow0.3\lambda$ (floored at $10^{-12}$) on
acceptance and $10\lambda$ on rejection. The
best of four restarts means four starts $u\sim\mathcal{N}(0,0.5^2)$ drawn
independently of the $k_0$ that built the target, of which the one with the
smallest final residual is reported. The forbidden and allowed directions $\Psi$
are built by drawing a $p\times p$ Gaussian matrix $G$, taking $G-G^\T$ or
$G+G^\T$ respectively, embedding the result in the shared block with zeros
elsewhere and normalising. The amplitudes are
$\varepsilon/\|R_0\|_F\in\{10^{-1},10^{-2},10^{-3},10^{-4}\}$, on networks of
$24$, $30$ and $36$ nodes. The third arm, the \emph{odd control}, fits the same
forbidden targets on the same networks from the same starts, with the $n_b$ odd
couplings fitted alongside the log-stiffnesses, started at $a=0$ and left
unclamped; it is the arm Sec.~\ref{sec:learning} reports as the forbidden
target with odd couplings switched on. In the fourth, the random-target arm,
$R^{*}$ is instead a
Gaussian block rescaled to $\|R_0\|_F$, for which nothing guarantees that the
floor is tight.

\paragraph{Coupled learning.}
The constant of proportionality in the bond update is $1/(2\beta)$, which is
what makes the update the $\beta\to0$ contrastive limit. The training pairs are
the $m_S$ unit loads on the driven degrees of freedom, so the free states are
the columns of $C[:,S]$ and the target is the whole block. The step size is set
by a cap on the largest relative change of any stiffness: $\Delta k$ is rescaled so that
$\max_b|\Delta k_b|/k_b$ equals a rate $\chi$, starting at $\chi_0=3\times10^{-2}$ and
halved every $120$ non-improving steps down to $10^{-5}$, and the stiffnesses
are clipped to $[10^{-3},10^{3}]$. Training runs
to a budget of $40\,000$ steps and stops early when the loss has not improved
for $1500$ consecutive steps. The rescaling uses the largest relative change over all
bonds, and the rate schedule, the stopping test and the choice of the returned
state read the global loss. Eleven of the $24$ forbidden-target runs of
Sec.~\ref{sec:learning} exhausted the budget instead, among them the two that end
more than $1\%$ above the floor. The nudge amplitude is $\beta=10^{-3}$
throughout, the sweep over $10^{-2}$, $10^{-3}$ and $10^{-4}$ being a check on
that choice rather than a source of any reported number.

\paragraph{The design-rule experiment.}
From a pool of free degrees of freedom the first $m$ are driven, and the
read-out set is either those same $m$ (overlapping, $p=m$) or $m$ others (disjoint, $p=0$).
Each layout's target is built from \emph{its own} realisable block plus the
$5\%$ generic perturbation, so that neither layout is handed a target defined by
the other's reachable set.

\paragraph{Frequency domain.} The mass is lumped and diagonal with entries drawn
uniformly on $[0.5,1.5]$; the damping is Rayleigh, $D=0.02K+0.01M$. The six
frequencies are $0$, $0.05$, $0.2$, $0.5$, $1.0$ and $2.0$ in units fixed by
two choices: the median bond stiffness is $\bar{k}=1$ and the mean nodal mass is
$\bar{m}=1$, so $\omega$ is in units of $\sqrt{\bar{k}/\bar{m}}=1$. That range
overlaps the networks' own spectrum rather than lying beside it: over the six
networks of the sweep the undamped natural frequencies, the square roots of the
generalised eigenvalues of $(K,M)$, run from $0.150$ to $3.61$ (a range taken by a
one-line variant of the deposited script; Appendix~\ref{app:stats}), so
$\omega=2$ sits just above the middle of that band and $\omega=0.05$ below its
lowest mode. The non-symmetric control adds a full Gaussian matrix of scale
$0.05$ to $D$. The ranks there are ranks of complex matrices, taken at the same
relative cutoff.

\paragraph{Layout at fixed budget.} For each network and each budget $n_U$, one
set of $n_U$ free degrees of freedom is drawn and every overlap $p=0,\dots,n_U$ is
then realised on that same set, splitting the remaining $n_U-p$ as evenly as
possible between exclusive read-out and exclusive drive, which is the split that
maximises $(\alpha+p)(\gamma+p)$. Networks have $30$, $40$ and $56$ nodes, so the bond
count never binds over the range of $n_U$ used.

\paragraph{Large deformation.} The five load amplitudes are $0$, $0.02$, $0.05$,
$0.10$ and $0.20$, applied
through a fixed random load pattern; the stiffness spread is $0.3$ rather
than $0.4$. The central finite differences that stand
in for the closed form there use relative steps $10^{-5}$, $10^{-6}$ and
$10^{-7}$ on $\log k$, and the rank is required to come out identical across all
three. The rank cutoff is $10^{-7}$, not the $10^{-9}$ of the closed-form
sweeps; see Appendix~\ref{app:stats}.

\paragraph{The floor with passivity.} Experiment~13 runs on the six networks of the
training experiment ($24$, $30$ and $36$ nodes, two independently drawn
geometries each) at $m=4,5,6$ and two index-set draws per configuration, at
full overlap, with the same Levenberg--Marquardt settings and the same best of
four restarts. Three target classes are used on each configuration. The
\emph{constructed} target is $R_0$ plus an antisymmetric perturbation of
$1\%$ of $\|R_0\|_F$ on the shared block, the class of
Sec.~\ref{sec:learning}. The \emph{generic} target is a Gaussian block rescaled
to $\|R_0\|_F$. The \emph{indefinite} target is built from the shared block of $R_0$,
symmetrised to remove round-off so that the floor of Theorem~\ref{thm:floor} is
zero up to round-off, with the sign of its smallest eigenvalue then flipped, which makes the block
indefinite and therefore unreachable. The reported floors are evaluated from the
target by the formulae of Theorems~\ref{thm:floor} and~\ref{thm:floorpd}, never
reused from the construction. The separate hypothesis check draws four index
sets per configuration on the same six networks and measures the eigenvalues of
the realised shared block directly.

\paragraph{The layout on a task.} Experiment~14 uses six networks of $30$, $36$
and $44$ nodes, two geometries each, budgets $n_U=4,6,8$ and overlaps
$p=n_U,n_U-2,\dots$ down to $0$, so that the split of $n_U-p$ is always even. The
target $\Xi$ is the compliance $C[U,U]$ of an \emph{independent} stiffness draw
on the
same network, plus $t$ times an antisymmetric Gaussian matrix normalised to the
same Frobenius norm, with $t=0$, $0.02$, $0.05$, $0.10$, $0.20$. Each layout is
trained on its own visible block $\Xi[T,S]$ from three shared restarts, and
the trained stiffnesses are then evaluated on the whole of $C[U,U]$, so that
every layout is scored on one object with one normalisation.

\paragraph{The price of non-reciprocity.} Experiment~22 uses five ensembles,
four of its own and one rebuilt from experiment~19. The first-order lemma behind
Corollary~\ref{cor:price} is checked on $40$ cases: Delaunay networks of $9$ to
$32$ nodes and the small triangular lattices of disorder $0$ or $0.1$, at
$p=2$ to $6$. The abstract check of \eqref{eq:ratio} draws $3000$ operators
$K_{\!f\!f}=K_s+K_a$ of sizes $n=2$ to $8$: $K_s$ has a random orthogonal
eigenbasis and a log-uniform spectrum spread over up to six decades, $K_a$ is
the antisymmetric part of a Gaussian matrix rescaled so that
$\|K_a\|_2/\|K_s\|_2$ is log-uniform from $10^{-3}$ to $10^{3}$, and $P$ is a
uniform subset of size $2$ to $n$, so that the equality clause of
Proposition~\ref{prop:ratio} is sampled as well as the inequality: $1128$ of
the draws have $P$ the whole set and $1872$ a proper subset.

The odd-bond
sweep uses twelve Delaunay networks, of $9$, $16$, $24$ and $32$ nodes with
three geometries each, and on each of them $40$ draws of a log-normal $k$ of
spread $0.4$, of a shared set $P$ of size $2$ to $6$, and of a unit direction
$a_{\mathrm{dir}}$. Each draw is evaluated at seven coupling sizes
$a=\phi\,t_{\mathrm{edge}}\,a_{\mathrm{dir}}$, where $t_{\mathrm{edge}}$ is the
distance from $a=0$ along $a_{\mathrm{dir}}$ to the boundary of
$\{\operatorname{sym}K_{\!f\!f}\succ0\}$, taken from the generalised spectrum
since $\operatorname{sym}K_{\!f\!f}$ is affine in $a$, and $\phi=0.1$, $0.3$,
$0.5$, $0.7$, $0.9$, $0.99$ and $0.999$. Every one of them is below $1$, so
$K_s=\operatorname{sym}K_{\!f\!f}\succ0$ holds by construction and the coupling
size is a stated spread rather than one arbitrary draw: $12\times40\times7=3360$
cases, $480$ at each $\phi$. Every one of the twelve
networks has $n_{\mathrm{free}}\ge15$, so a $P$ of at most six coordinates is
always a small proper sub-block; the whole-free-set comparison of
Sec.~\ref{sec:price} is made on the same $3360$ draws with $P$ the whole free
set.

The brute-force comparison behind Corollary~\ref{cor:price} uses ten
nine-node Delaunay networks, one per seed, each seed fixing the geometry, the
stiffnesses and the three shared coordinates together, so that the ten are ten
networks and not ten stiffness draws on one; $n_b$ is $19$ at six of the seeds
and $18$ at the other four, so the exhaustive search runs over $969$ or $816$
supports of three bonds.

The torque-free construction runs on the $336$
configurations of experiment~19 (the same sweep, the same seeds and the
same targets, rebuilt from the same recipe) at that experiment's first
three target sizes $\varepsilon=10^{-4}$, $10^{-3}$ and $10^{-2}$, a target
being a random antisymmetric shared block of Frobenius norm
$\varepsilon\|R(k_0,0)\|_F$. The solve is experiment~19's Newton with
minimum-norm steps, run on an orthonormal basis of $\mathcal{K}$, with a
halving line search, at most $30$ iterations, and convergence declared at a
relative residual of $10^{-13}$.

\paragraph{The worked example.} Experiment~23 uses one network, the
sixteen-node Delaunay network of Fig.~\ref{fig:setup}a at geometry seed~$3$,
with $38$ bonds and $29$ free coordinates. Its terminals $1$, $2$ and $3$ are
the $x$-coordinates of the three nodes that figure marks, chosen by the formula
of its plotting script. The reference stiffnesses are log-normal with spread
$0.4$, drawn at seed~$2300$. The target of list~B takes the reference values and
multiplies $1\to2$ by $1.1$ and $2\to1$ by $0.9$; the target of list~A omits
$2\to1$. A control arm multiplies both tasks of the pair by $1.1$. Every
arm starts from the same $20$ log-stiffness vectors, Gaussian with spread
$0.5$ at seeds $1000$ to $1019$. Training uses the Levenberg--Marquardt rule of
experiment~04 with its default settings, and every run is reported, none as
the best of several. The odd bond is the first pivot of the pivoted $QR$ of the
row $\langle v,\partial y/\partial a_b\rangle$, taken on the uniform passive
network, $k=1$ and $a=0$, not at the reference stiffnesses, with
$v$ the unit vector spanning $\mathcal{A}^{1}_{\mathcal{L}}$. It joins
the node carrying terminal~$2$ to a neighbour, and its entry is $0.196$ against a median of
$0.036$ over the $38$ bonds.

A run of experiment~23 reaches zero when its error is at most
$10^{-8}$ of the target norm. It reaches the floor when its error lies within
a relative $10^{-6}$ of the floor. Both thresholds, and every parameter above,
were fixed before the first run in a pre-registration deposited with the
script. At that first run the scoring code assigned the control arm, whose
computed floor is round-off at $2.2\times10^{-16}$, to the floor criterion; it
was corrected to the pre-set assignment by arm, and the output of the first run
was not kept.

Most runs end with some log-stiffness within $0.01$ of the clamp of
$\pm12$: $18$, $20$, $20$ and $19$ of the $20$ runs of the four arms.
Experiment~04 saturates in the same way, counted there within $10^{-9}$ of the
clamp.
The criteria do not refer to the clamp. Two sweeps were added after the
pre-registered run and are deposited with it: list~B with each of the $38$
bonds in turn as the single odd bond, from the same $20$ starts, and all four
arms re-run with the clamp at $\pm8$, $\pm10$, $\pm12$, $\pm14$ and $\pm16$. The runs at $\pm12$ reproduce the
pre-registered ones exactly. At $\pm8$,
$\pm10$ and $\pm14$ every arm again meets its criterion in $20$ of $20$; at
$\pm16$ two runs of the odd arm stop at the conditioning guard of the solver.
The error reached by the three arms whose floor is zero is the round-off of the
solve and grows with the clamp. The script runs in about $55$\,s, $2$\,s of it the four arms.


\section{Conventions, cutoffs and exclusions}
\label{app:stats}

This appendix states the conventions, cutoffs and exclusions behind the counts
reported in the main text. One statistical comparison is made in this paper:
the lattice disorder comparison of Sec.~\ref{sec:numerics}, for which a
two-proportion $z$ test and Fisher's exact test are quoted below. Every other
quantity reported is a direct measurement over an enumerated case set, and
counts are given as successes over total; sample sizes are set by the sweep
design and not by a power calculation. Fifteen headed paragraphs follow (five
conventions, two cutoff choices, three records of what particular experiments
check, three limits on where the counts hold, one that corrects two earlier
readings, and one that lists three exclusions).

\paragraph{Median over stiffness draws.}
Where a rank or a ratio is quoted once per configuration it is the median over
the independent stiffness draws at that configuration, of which there are
three in the main
rank sweep of Sec.~\ref{sec:numerics}. Per individual draw rather than per
median, the bound is attained in $2541$ of $2559$ measurements inside
criterion~\eqref{eq:criterion}, and $835$ of the $853$ configurations attain it
at every draw. The $18$ configurations with one failing draw have
singular-value gaps from $8.6\times10^{10}$ to $6.5\times10^{14}$, so those are
genuine rank deficits and not cutoff noise. The bound is violated in zero cases
under either convention: $0$ of $2418$ medians and $0$ of $7254$ draws.

The rank-cut statistics of Appendix~\ref{app:jacobian} are median-scoped in the
same way. Per draw they read $99.01\%$ and $99.70\%$, against $99.50\%$ and
$99.88\%$, with $72$ marginal draws at $m=5$ to $10$; as with the twelve
marginal medians, not one of them lies inside criterion~\eqref{eq:criterion}.

\paragraph{The rank cutoff, and the two places it differs.}
Sec.~\ref{sec:numerics} states the cutoffs: $10^{-9}$ times the largest
singular value, except for the two finite-difference Jacobians, where
experiment~18b cuts its column-normalised projection at $10^{-6}$ and the
large-deformation sweep, whose noise floor comes from truncation error, at
$10^{-7}$. All $59$ large-deformation cases have a
last-retained to first-discarded singular-value ratio, taken as the median over
the three step sizes, between $1.6\times10^{6}$ and $2.1\times10^{8}$, median
$3.3\times10^{7}$, so the
looser cut is the appropriate one and the spectrum has a clear gap where it is
placed. The gap is deposited: the columns
\texttt{sv\_at\_bound}, \texttt{sv\_first\_discarded}, \texttt{sv\_gap},
\texttt{sv\_at\_bound\_rel} and \texttt{sv\_first\_discarded\_rel} of
\texttt{exp10\_large\_deform.csv} give, for each of the $59$ cases, the
singular value at the predicted rank, the first one below it, their ratio, and
the two normalised by the largest. Each is the median over the three step sizes
taken independently, as the rank itself is, so the ratio column is the median
of the per-step ratios and not the quotient of the other two. The first
discarded value lies between $1.7\times10^{-10}$ and $2.1\times10^{-9}$ of the
largest, so $10^{-7}$ clears it by a factor of about $50$, or $1.7$ decades,
even in the worst case.
The choice is outcome-determining and we report what the alternative
gives: at $10^{-9}$ the measured ranks exceed the bound in $21$ of the $59$
cases, and in none of the $59$ do they agree across the three step sizes. That is a
statement about finite differencing, not about Theorem~\ref{thm:main}. The
$21$ are exactly the cases in which \texttt{sv\_first\_discarded\_rel} exceeds
$10^{-9}$, so the recount can be read off the deposited table.

\paragraph{The factor $\tfrac12$ in criterion \eqref{eq:criterion}.}
That factor is a choice too, and it filters data in several places. Replacing
$\tfrac12$ by $c$ in
both clauses of \eqref{eq:criterion} and recounting over the $2418$ cases of the
deposited table gives

\begin{center}
\small
\begin{tabular}{@{}lcc@{}}
\toprule
$c$ & attained inside & attained outside \\
\midrule
$0.4$ & $613/613$ \ ($100\%$)  & $1326/1805$ \ ($73.5\%$) \\
$0.5$ & $853/853$ \ ($100\%$)  & $1086/1565$ \ ($69.4\%$) \\
$0.6$ & $991/993$ \ ($99.8\%$) & $\ 948/1425$ \ ($66.5\%$) \\
\bottomrule
\end{tabular}
\end{center}

\noindent
so the criterion is not balanced on a knife edge, though it is not symmetric
either: it survives being tightened by a fifth, to $c=0.4$ ($613/613$), and
loosened by a tenth, to $c=0.55$ ($895/895$), and $\tfrac12$ is close to the
largest factor at which it holds without exception on this sweep. Loosening it
by a fifth is already too far: the first failures appear at $c=0.6$, and both are
triangular lattices, the less generic of the two ensembles, at
$m=6$, $p=3$ and $m=9$, $p=9$, with singular-value gaps of $2.3\times10^{12}$
and $7.5\times10^{10}$, so they are genuine rank deficits and not cutoff noise.
By $c=0.65$ there are five. Nothing in the paper turns on which of $0.4$, $0.5$
or $0.55$ is used: the counts reported inside the criterion are $100\%$ at all
three, and what changes is only how many cases the filter admits.

The first two bars of Fig.~\ref{fig:robust}c, which Sec.~\ref{sec:numerics}
reads as the cases inside \eqref{eq:criterion}, sum to
the $853$, and that arithmetic is only half an identity: the criterion forces
the abscissa to or below $\tfrac12$, since $\min(n_\theta,d)\le d$, but not
conversely, so it needs checking before it is used.
Their abscissa is neither clause of \eqref{eq:criterion}: it is
$\min(n_\theta,d)/n_{\mathrm{free}}$, and a bound on that implies neither
$d\le\tfrac12 n_{\mathrm{free}}$ nor $d\le\tfrac12 n_\theta$. The legend of
Fig.~\ref{fig:robust}c still labels the bars at or below $\tfrac12$
``inside room criterion'', correctly, because on this sweep
three sets have the same $853$ members: the cases at or below
$\tfrac12$ on the abscissa, the cases satisfying \eqref{eq:criterion}, and the
cases satisfying its first clause $d\le\tfrac12 n_{\mathrm{free}}$ on its own.
On the deposited table all $853$ cases with
$\min(n_\theta,\,m^{2}-\delta_p)/n_{\mathrm{free}}\le\tfrac12$ satisfy
$d\le n_\theta/2$, hence $\min(n_\theta,d)=d\le\tfrac12 n_{\mathrm{free}}$, so
they satisfy \eqref{eq:criterion}; and every case satisfying
\eqref{eq:criterion} falls at or below $\tfrac12$, as it must. The second clause
$d\le\tfrac12 n_\theta$ holds far more widely, in $1451$ of the $2418$, so the
agreement rests on the first clause, and on this sweep. The last bar, labelled
$\le4$, holds everything above $2$; the largest value of the abscissa anywhere
is $2.70$, so there is no $>4$ bin.

\paragraph{Column scaling in experiment~18.}
The wedge matrix of
\eqref{eq:wedge} has column norms spanning many decades, far wider than the
Jacobian itself, because each entry is a product of two compliance entries. With
a relative cutoff of $10^{-9}$ applied to the unnormalised matrix, genuinely
independent directions fall below the cut, so experiment~18 normalises the
columns before taking singular values, which leaves the span unchanged.
Normalisation has a failure mode of its own. A column that is zero in exact
arithmetic but not in floating point, such as the two bonds of a degree-two
corner node of a disordered lattice, which carry no tension under any load
applied elsewhere, is scaled up to a unit vector of round-off and counted as
a direction.

The first deposited run of experiment~18 (release 1.3.0 of the
archive named under Code availability, since superseded) did
this, and reported a passive rank above the bound of Theorem~\ref{thm:main}
in $153$ of the $1008$ deposited rows, and a wedge span of $252$ of $252$ where
the count is $251$. Those three counts belong to three different sets. The
$1008$ rows of experiment~18 (Sec.~\ref{sec:numerics}) are $576$ on Delaunay
networks of $16$ to $40$ nodes and $432$ on disordered triangular lattices,
with $33\le n_b\le108$. They are $945$ distinct configurations: at $m=3$ the
overlap $\lfloor m/2\rfloor$ coincides with $p=1$ and the script runs it twice
with the same seed, so $63$ rows ($36$ Delaunay, $27$ lattice) repeat that
case. The $252$ are the full-overlap rows among the
$1008$, one per configuration and none of them repeated; the sweep has $252$
rows at each of its four overlap settings, $p=0$ among them. The $251$ are the full-overlap
configurations whose wedges span, which is why
\texttt{exp18\_sparse.csv} carries $251$ rows and not $252$.

The deposited script now zeroes every column below $10^{-12}$ of
the largest before normalising, asserts the bound of Theorem~\ref{thm:main}
in every row, and selects the sparse set by a pivoted $QR$ of the
unnormalised matrix; the counts quoted in Sec.~\ref{sec:numerics} are from
that run (release 1.3.1), and every data row of its table is unchanged in
every release since. The threshold is not delicate: re-running the
deposited script with it set to $10^{-8}$, $10^{-10}$, $10^{-11}$, $10^{-13}$
and $10^{-14}$ reproduces every result table byte for byte. Every discarded column has a
norm below $10^{-14}$ of the largest except in the rank-greedy comparison arm,
where such columns reach $7\times10^{-13}$; kept at $10^{-13}$ and $10^{-14}$,
they change no rank and no selection in the tables. No retained column is
smaller than $2\times10^{-8}$, and no rank moves when the threshold does. The three-dimensional count of experiment~18b, $60$
of $60$, is likewise from that release: its own three-dimensional pass,
re-pinned, replaces the $53$ of $53$ of an earlier script that was not
deposited. The sweeps of
Secs.~\ref{sec:numerics} and \ref{sec:consequences} other than experiment~18 do
not normalise, and do not need to, since their gap statistics are those reported
above.

\paragraph{The $8$-node networks of experiment~20.}
On $8$-node networks ($n_b=15$ or $16$, so that $n_b$ binds at $p=6$ and
lies within two of $\tfrac12p(p+1)-1=14$ at $p=5$) the same three families miss
the ceiling of \eqref{eq:multirank} in $34$ of $66$ configurations, at $p=4$,
$5$ and $6$, by one to five, and in every one of them
$\rank\partial\mathbf{D}_{\mathcal{T}}/\partial k=\rank\partial C_P/\partial
k-1$ exactly, so that \eqref{eq:multirank2} is attained where
\eqref{eq:multirank} is not. In $31$ of the $34$ the shortfall is in the map
$k\mapsto C_P$, whose rank falls below $n_b$; in the other three, at $p=6$,
that rank equals $n_b=16$ and the miss is the $-1$ of scale invariance
alone. In the $16$ misses at $p=4$ and $p=5$, where $\tfrac12p(p+1)\le n_b$,
the rank of $k\mapsto C_P$ falls below $\tfrac12p(p+1)$ as well.

The
count $n_b$ overstates the independent columns of $\partial C_P/\partial k$
whenever a node whose degrees of freedom are all hidden from $P$ has degree
three or four. Such a node is seen by $P$ only through its Schur complement
(Kron reduction), and the reduced star of $n_\star$ central-force bonds in
the plane has rank $n_\star-2$, hence at most
$\min\bigl(n_\star,\tfrac12(n_\star-2)(n_\star-1)\bigr)$ independent entries:
$1$, $3$ or $n_\star$ for $n_\star=3$, $n_\star=4$ and $n_\star\ge5$. On the
larger networks the same nodes exist, but $n_b$ is far from binding.

\paragraph{What experiment~21 checks in Sec.~\ref{sec:tasklist}.}
Experiment~21 checks the task-list results of Sec.~\ref{sec:tasklist}
numerically, over the $7245$ task lists of its first three parts and the $1637$
floor instances of its fourth. Its first part draws lists over six families
(coordinate, mixed, self-adjoint, supported and unsupported Gaussian patterns,
and lists read on real Delaunay networks) at every overlap from disjoint driven
and read supports to identical ones. It finds the two readings of
$\delta_{\mathcal{L}}$ in \eqref{eq:deltaL} agreeing with each other and with
$\dim(\mathcal{M}\cap\operatorname{Skew})$ computed as an honest subspace
intersection and cross-checked by principal angles, rather than through the
rank--nullity identity that would make the check circular. On the network
family it also tests the rank bound \eqref{eq:taskbound} and its
non-reciprocal counterpart. Its second confirms the reversed-pair count
\eqref{eq:digon} on $4800$ random coordinate lists over six families. Among
them are one augmented with every self-loop, which must not and does not move
$\delta_{\mathcal{L}}$, and one of self-loops only, where
$\delta_{\mathcal{L}}$ must be zero. A further $786$ paired tests confirm that
deleting and adding self-loops moves nothing, and the named two-task lists of
Sec.~\ref{sec:tasklist} are rows of their own. Its third recovers the
full-block values of Proposition~\ref{prop:taskblock} in exact rational
arithmetic, with no floating point in that part at all. Its fourth checks
the floor \eqref{eq:taskfloor}, the orthogonal splitting \eqref{eq:tasksplit}
and attainment over symmetric $C$ against a brute-force least squares over a
basis of $\Sym(n)$. It closes with three certificates: the two examples of
Sec.~\ref{sec:tasklist} on what the list version does not deliver, and the
single task at $n=1$ used in the proof of
Proposition~\ref{prop:tasklist}(iv), where at $y^{*}=0$ the infimum over
symmetric invertible $C$ is not attained.

\paragraph{The three-bond arms of experiment~21, part~8.}
These two arms support the statement of Sec.~\ref{sec:wedges} that the count of
$p(p-1)/2$ odd bonds holds with the stiffnesses held and not with $k$ free. The
stiffness-free arm lets $k$ move on $40$-node networks (three geometries) at
$m=p=5$, where the count asks for ten bonds. In that arm, three odd bonds
realise a prescribed antisymmetric block of relative size $10^{-2}$ to a
relative residual of at most $10^{-12}$ in nine of nine (network, target)
pairs, with every $k_b>0$ and the symmetric part of $K_{\!f\!f}$ still
positive definite. Which three matters: the first three pivots of the pivoted
$QR$ of the wedge matrix served in six of the nine pairs, and the other three
pairs were solved only by another triple, disjoint from the $QR$ one. Three
bonds suffice, but not every triple does. The same part runs the $QR$ triple
again at fixed $k$ as its negative control, at the smaller target $10^{-4}$
rather than $10^{-2}$. There it solves none of the nine, and it stalls exactly
at the component of the target orthogonal to the span of its three wedges
(median ratio $1.0000$, worst deviation $8.4\times10^{-5}$). The ten-bond
fixed-$k$ arm of the same part is reported in Remark~\ref{rem:submersion}.

\paragraph{The two clauses of \eqref{eq:ratio}, counted apart.}
Proposition~\ref{prop:ratio} asserts an inequality where $P$ is a proper
subset of the free set and an \emph{equality} where $P$ is the whole of it, so
experiment~22 tests the two differently and never adds their counts. Each
abstract draw carries one choice of $P$, $1872$ of them a proper subset
and $1128$ the whole set, whereas every odd-bond case is read at both choices,
once on its small shared set and once on the whole free set; the $5232$
proper-subset readings are those $1872$ together with all $3360$. On a
proper subset a case counts as a failure when the measured
$\eta(C_P)/\eta(K_{\!f\!f})$ exceeds $1+10^{-9}$; over all $5232$ such cases
there are none, the largest ratio being $0.9997$. On the whole free set the
only meaningful test is $|\eta(C_P)/\eta(K_{\!f\!f})-1|\le10^{-3}$, and a
ratio a few parts in $10^{5}$ above $1$ is round-off in an equality and not a
counterexample to it. That tolerance is set by what $\eta$ costs to evaluate,
two inverse square roots on top of an inversion, so that its accuracy degrades
with $\operatorname{cond}K_s$, which reaches $3.8\times10^{5}$ on this
ensemble. Over the $1128$ abstract equality cases the worst deviation is
$2.64\times10^{-5}$ and the median is $4\times10^{-15}$, which is round-off and
is quoted to its order alone; eleven of the
$1128$ lie above $1+10^{-9}$, at $\operatorname{cond}K_s$ between
$3.4\times10^{2}$ and $1.4\times10^{5}$. That lower figure is no threshold:
$249$ of the $1128$ lie above it and $238$ of those stay at or below
$1+10^{-9}$. The deposited table records
them under a column named for the numerical fact, not as violations of the
equality. The maximum over all $6360$ cases
therefore prints as $1.0000264$ rather than as $1$. Inside the
odd-bond model the same equality holds, to within $2.9\times10^{-11}$ in
$3360$ of $3360$. The rest of experiment~22 follows the paper's conventions
unchanged: ranks at the relative cutoff of $10^{-9}$, and a guard that would
discard a case with $\operatorname{cond}K_{\!f\!f}>10^{13}$, which never
binds, so all $3000$ and all $3360$ cases are reported.

\paragraph{The passive substitution in \eqref{eq:ratio}, counted.}
Sec.~\ref{sec:price} shows why the passive $K_{\!f\!f}(k,0)$ cannot stand in
for the active $K_s$ in \eqref{eq:ratio}; these are the counts. With this
passive substitution the inequality fails, and it fails worst on the whole free
set, where \eqref{eq:ratio} is an equality and the substitute is asked to do
exactly the work of $K_s$. There it is violated in $3066$ of the $3360$
odd-bond cases, already in $285$ of the $480$ at the weakest coupling tested,
by factors of up to $46$. Only on a small shared set, where the inequality has
slack of its own, does the substitution mostly survive: there it fails in $65$
of the $3360$ cases, by a factor of up to $3.51$. On that small shared set it
fails in none of the $480$ a tenth of the way to the boundary of
$\{\operatorname{sym}K_{\!f\!f}\succ0\}$ and in $25$ of the $480$ at $0.999$
of the way.

\paragraph{The normalisation behind Corollary~\ref{cor:price}.}
The equivalence with maximising
$\sigma_{\min}(\mathsf{W}_{\mathcal{B}})$ stated in Sec.~\ref{sec:price} holds
at first order, for Frobenius-normalised $B_a$ and Euclidean $\|a\|$. The
factor $\sqrt2$ in \eqref{eq:astar} comes from $\operatorname{vec}_<$,
which keeps one entry of
each antisymmetric pair, so that
$\|\operatorname{vec}_<(B_a)\|_2=\|B_a\|_F/\sqrt2$; the identification of
$\mathcal{A}_P$ with $\Lambda^2\R^p$ is itself a Frobenius isometry. If
$\Lambda^2\R^p$ is instead normalised by $\|e_r\wedge e_c\|=1$, the bounds
read $2\|B_a\|_{\Lambda}/\sigma$. Pricing the coupling as
$\|(a_b/k_b)_b\|$, or by the largest per-bond ratio $\max_b|a_b|/k_b$, or
normalising $B_a$ by $\|B_a\|_2$, changes the criterion, and a different
selection can become optimal.

\paragraph{Spanning, solving, and which price is quoted.}
The counts behind Proposition~\ref{prop:torquefree} are kept apart in the same
way. Whether the torque-free couplings span,
$\mathsf{W}\mathcal{K}=\R^{\delta_p}$, is a property of the configuration
alone, decided at the same cutoff of $10^{-9}$: $316$ of the $336$ span. The
$20$ that do not are reported, not excluded, and every
one has $\dim\ker Q^\T<\delta_p$, so what fails is the hypothesis of
clause~(ii) and not the clause. Whether the target is then \emph{solved} is a
property of the run. The Newton iteration stops at a relative residual of
$10^{-13}$, but a configuration is counted among the $316$, $314$ and $314$
only when its final relative residual is below $10^{-10}$, which is
experiment~19's threshold, so that the two experiments are read on one
footing. The two configurations that miss at $\varepsilon=10^{-3}$ are the two
that miss at $10^{-2}$: the $24$-node geometry at seed~$2$ with $m=p=6$, at
each of its two stiffness draws. Each stops after seven to sixteen iterations,
well short of the cap of $30$, on the iteration's slow-progress guard: the
step is accepted, but four consecutive accepted steps each fail to cut the
residual by so much as a thousandth of itself, and the loop breaks. The
deposited table records the iteration count and the single boolean
\texttt{stalled}, not which of the iteration's three stall exits was taken;
the exit was identified by re-running the deposited script with it logged, a
check that leaves every table unchanged and is not itself deposited. Either
way those are limits of the solver and not statements about existence.

The price quoted in
Sec.~\ref{sec:price}, a median $4.73$ at $\varepsilon=10^{-4}$, is
$\|a_{\mathcal{K}}\|_2/\|a^{*}\|_2$ for the Newton solution
$a_{\mathcal{K}}$ over the $316$ configurations that both span and solve; its
minimum, tenth percentile, ninetieth percentile and maximum there are $1.60$,
$2.55$, $11.6$ and $93.2$. The first-order least-norm element of $\mathcal{K}$
gives the same median to three figures, but it needs no solve and is defined
wherever the configuration spans, so the deposited summary quotes it twice,
over the solved rows and over all spanning rows; at $\varepsilon=10^{-4}$ the
two sets coincide, since all $316$ solve, and at $10^{-3}$ and $10^{-2}$ they
differ by the two configurations above. The denominator $\|a^{*}\|_2$ is the
unconstrained least-norm coupling of Corollary~\ref{cor:price}; it
carries a net moment of its own, a median $9.0\times10^{-6}$ relative at
$\varepsilon=10^{-4}$, against at most $2.8\times10^{-17}$ on the torque-free
solutions. The relative net moment of $a^{*}$ is proportional to
$\varepsilon$ (the medians at
the three target sizes are $9.0\times10^{-6}$, $8.8\times10^{-5}$ and
$9.4\times10^{-4}$), so the deposited summary's single
\texttt{median\_net\_moment\_rel\_astar}, $8.9\times10^{-5}$, is the median
pooled over all three and describes no one of them. The bound
$2.8\times10^{-17}$ needs no such care: it holds at every $\varepsilon$.

\paragraph{Criterion \eqref{eq:criterion} requires genericity.}
On a
perfect $4\times4$ square lattice with both diagonals in every cell and uniform
stiffnesses (integer node coordinates, so that $\kappa_b=k_b/L_b^2$ is
rational and the rank is exact over $\mathbb{Q}$ with no floating-point cutoff
at all, and the same pinning rule as everywhere else), attainment fails
\emph{inside} the criterion. It fails for particular choices of the driven and
read-out coordinates and not for all of them, which is what makes the criterion
empirical rather than simply wrong. At $m=3$, $p=2$ the bound is $8$ and there
are index sets at which the
exact rank is $7$. Over index sets drawn uniformly at random on
that lattice, the exact rank falls below the bound in $9$ of the $322$ draws
that satisfy \eqref{eq:criterion} (experiment~13, part~5). The deposited
sweeps of experiments~03 and~06 contain no such case among their
per-configuration medians (per draw, experiment~03 has the $18$ single-draw
deficits of the paragraph on the median over stiffness draws), and that is a
statement about those draws and not about the ensembles behind them: enlarging the
Delaunay arm of experiment~08, which draws from the same family, from three
seeds to eight produces one. On the twelve-node network at seed~$7$, where
$n_{\mathrm{free}}=21$ and $n_\theta=25$, the draw at $m=3$, $p=1$ has $d=9$
inside \eqref{eq:criterion} and measured rank $7$, with a singular-value gap
of $9.5\times10^{13}$ and the same rank at every stiffness redraw we tried.
The reason is the one given in Sec.~\ref{sec:numerics}: that network carries a
vertex of degree
two, and on it $341$ of the $5985$ full-overlap index sets of four free
coordinates fall short of $d=10$, all but one of them at exactly
$m(m+1)/2-3=7$, with the count unchanged over five independent stiffness
draws. Among the $252$ full-overlap
configurations of experiment~18, one $5\%$-disorder lattice is such a case
(Sec.~\ref{sec:numerics}), and a reader
who builds a regular lattice with equal springs should not expect
\eqref{eq:criterion} to hold.

The perfect triangular lattices of the main rank sweep are the comparison
quoted in Sec.~\ref{sec:numerics}. Within the lattice arm the bound is
attained in $135/186$ cases at zero disorder against $292/372$ and $290/372$ at disorder
$0.05$ and $0.15$, which is $73\%$ against $78\%$ and $78\%$. Together, the two
disordered arms attain in $582$ of $744$ against $135$ of $186$ at zero
disorder, a gap of $5.6$ percentage points. A pooled two-proportion test gives
$z=1.64$ and a two-sided observed significance level of $0.10$; Fisher's exact
test gives $0.12$. At zero disorder the lattice
generator ignores its seed, so that arm has half as many distinct geometries as
the others; the counts above reflect that.

\paragraph{The hard limits in the code, and whether they bind.}
Three bounds in the implementation are numerical rather than physical: the
Levenberg--Marquardt log-stiffnesses are clamped to $|u|\le12$, the
coupled-learning stiffnesses to $k\in[10^{-3},10^{3}]$, and the coupled-learning
rate is floored at $\chi\ge10^{-5}$ (Appendix~\ref{app:params}). A reported run
sitting against one of these would be reporting the bound and not the physics,
so we instrumented every training run, and wherever a limit turned out to be
active we re-ran the experiment with it moved. The instrumentation and the
re-runs are one-line variants of the deposited scripts, not deposited as
separate outputs, and so is the natural-frequency range of
Appendix~\ref{app:params}; Code availability lists them with the other such
variants.

\emph{The Levenberg--Marquardt clamp is active, and most reported errors do not
move with it.} It is active: $1534$ of the $1872$ fits of experiment~04,
$419$ of the $432$ of experiment~13 and $1016$ of the $1080$ of experiment~14
end with at least one log-stiffness within $10^{-9}$ of $\pm12$, and that
includes the winning
restart of $134$ of the $144$ reported forbidden-target runs. What saturates are
bonds the target barely sees. Re-running all of experiment~04 with it \emph{tightened} to $|u|\le8$,
a stiffness range narrower by a factor $e^{4}=55$ at each end, leaves the
median reached-over-floor ratio unchanged to eight decimals, at a median
relative change of $3.4\times10^{-9}$. The forbidden arm is still within $1\%$
of the floor in $138$ of $144$ against $143$, the odd control is still below
$10^{-6}$ in all $144$, and the random-target median, over the $36$ distinct
measurements that arm contains, moves from $1.400$ to $1.389$. Six runs of the
$144$ end more than $1\%$ above the floor at $|u|\le8$. One was already outside
that band at $|u|\le12$ and is genuinely sensitive, ending at $1.755$ instead
of $1.023$; the other five, four of them within $0.1\%$ of the floor at
$|u|\le12$ and the fifth within $0.3\%$, end between $1.0105$ and $1.439$
times the floor. \emph{Widening} it to $|u|\le24$ instead makes everything worse, for
a reason that has nothing to do with capacity. The admissible stiffnesses then
span $21$ orders of magnitude, and at the states the optimiser reaches,
$\operatorname{cond}\,K_{\!f\!f}$ over the $144$ winning restarts is the
following at the three clamp settings, with four runs at $|u|\le24$ pinned at
the code's own admissibility guard of $10^{13}$:

\begin{center}
\small
\begin{tabular}{@{}lccc@{}}
\toprule
clamp & smallest & median & largest \\
\midrule
$|u|\le8$  & $3.1\times10^{4}$  & $4.6\times10^{5}$  & $5.7\times10^{7}$ \\
$|u|\le12$ & $7.7\times10^{6}$  & $2.3\times10^{7}$  & $7.1\times10^{10}$ \\
$|u|\le24$ & $7.1\times10^{11}$ & $3.2\times10^{12}$ & $1.0\times10^{13}$ \\
\bottomrule
\end{tabular}
\end{center}

\noindent
At $|u|\le24$ the residual is then formed with about four significant digits
left. The output carries the
signature: at $|u|\le24$, $96$ of the $144$ runs end \emph{below} the exact
floor, the worst by $2.9\times10^{-2}$ relative, which no statement about a
reachable set permits and which is lost precision instead of a
counterexample. The clamp at $12$ is a conditioning guard placed where the
arithmetic is still trustworthy. The bonds that reach it are ones the target
barely sees, and tightening it to $|u|\le8$ leaves the median ratio unchanged.

\emph{The coupled-learning stiffness clamp is never approached from above, and
from below it helps rather than hurts.} The largest stiffness reached anywhere
along the $51$ trajectories of experiment~07 is $4.88$ (they are $24$
forbidden-target, $24$ allowed-target and the three of the nudge-amplitude
sweep, one of which repeats a forbidden-target trajectory, so $50$ are
distinct), and along the $24$ of
experiment~09 is $3.79$, so the upper clamp of $10^{3}$ is inactive by factors
of $205$ and $264$. The lower clamp is touched: $4$ of the $24$
forbidden-target runs of experiment~07 and $13$ of the $24$ runs of
experiment~09 drive some stiffness onto $10^{-3}$, and the run that ends at
$1.0882$ is one of the four. Widening the clamp to $[10^{-6},10^{6}]$ and
re-running leaves the other $20$ forbidden-target runs of experiment~07
unchanged to four
decimals and makes all four of the touching ones \emph{worse}:
$1.0000\to13.04$, $1.0000\to130.1$, $1.0010\to5.55$ and $1.0882\to54.58$.
Experiment~09 behaves the same way, its median reached-over-floor going from
$1.000000$ to $1.023$ and its median overlapping-to-disjoint error ratio from
$7.0\times10^{5}$ to $2.5\times10^{4}$. The clamp is holding the network away
from singularity, not holding the error up: every reported value is at or below
what the looser clamp gives, so no run's error is being propped above its floor
by it, and no run of either experiment ends below its floor under either
setting.

\emph{The rate floor is reached, by design, and changes nothing.} The schedule
halves $\chi$ every $120$ non-improving steps, so a run that has settled anneals
onto the floor and stays there; $18$ of the $24$ forbidden-target runs of
experiment~07 and $17$ of the $24$ of experiment~09 do. Lowering the floor to
$10^{-8}$ and re-running experiment~07 changes no reached-over-floor ratio by
more than $4.0\times10^{-6}$ relative, leaves $22$ of $24$ within $0.1\%$ of the
floor exactly as before, and leaves the two runs that end more than $1\%$
above the floor, at
$1.0204$ and $1.0882$, unchanged to six decimals. Six runs never reach the rate
floor, all six among the eleven forbidden-target runs of experiment~07 that
exhaust the step budget, and the two runs above $1\%$ are among them. What
stopped those eleven is the budget of $40\,000$ steps.

\paragraph{Two readings of Sec.~\ref{sec:dirichlet} that we first got wrong.}
Both are our own, from the drafting of this paper. We first read the necessary
conditions \eqref{eq:dbalance}
and \eqref{eq:dcycle} as sufficient; they are not, for the reason stated after
Theorem~\ref{thm:dirichlet}. And we first took the rank in \eqref{eq:dbound}
by central finite differences, which put a noise floor at the cutoff and
inflated the count; the reported ranks come from the exact Jacobian.

\paragraph{Exclusions.}
There are three. First, the three room-limited configurations of the subspace
comparison in Sec.~\ref{sec:numerics}, where the passive image does not reach
$\mathcal{S}_P$ and the comparison is therefore undefined. Second, the cases with $n_\theta<m^2$ in the full-overlap deficit law of
Fig.~\ref{fig:law}b, where the parameter count is not comfortably clear of the
target. The rule is wider than its reason, and deliberately so: what would put
the rank in the hands of the bond count rather than the symmetry at full
overlap is $n_\theta<m(m+1)/2$, the predicted rank there, and $12$ of the $21$
excluded cases satisfy $n_\theta<m^2$ without satisfying that, $8$ of them
meeting the law exactly. They are excluded with the rest rather than argued
case by case. In the excluded cases the measured
forfeited \emph{dimension} fraction is never smaller than predicted, and in
thirteen of the twenty-one it is strictly larger: $0.470$, $0.470$ and
$0.490$ at $m=10$ against a predicted $0.450$; up to $0.632$ at $m=12$ against
$0.458$; and up to $0.719$ at $m=14$ against $0.464$. In the remaining eight
(six at $m=12$ and two at $m=14$, the ones named above as meeting the law
exactly) it equals the prediction. Every one is a further
loss or none at all, not a violation; the full sweep is shipped with the data. Third, one of the $60$ planned large-deformation cases, at $12$ nodes and the
smallest non-zero load, where the perturbed finite-difference equilibrium failed
to converge; $59$ are reported.


\section{Reproducing the numbers}
\label{app:repro}

\begin{table}[ht]
  \centering
  \scriptsize
  \setlength{\tabcolsep}{4pt}
  \renewcommand{\arraystretch}{0.9}
  \begin{tabular}{@{}lllr@{}}
    \toprule
    Script & Cases & What it establishes & Time \\
    \midrule
    \texttt{exp01\_rank\_vs\_bound.py} & 640 &
      the two extremes, $T=S$ and $T\cap S=\emptyset$ & 1.7\,s \\
    \texttt{exp02\_overlap\_law.py} & 540 &
      deficit $=p(p-1)/2$ at fixed room; $(m-1)/(2m)$ & 2.9\,s \\
    \texttt{exp03\_theorem\_check.py} & 19737, 120, 2418 &
      mechanism, subspace, bound, criterion & 11.6\,s \\
    \texttt{exp04\_learning.py} & 144 rows, 1872 fits &
      the error floor of Theorem~\ref{thm:floor} & 422\,s \\
    \texttt{exp05\_duality.py} & 336 &
      the identity of Theorem~\ref{thm:duality} & 25\,s \\
    \texttt{exp06\_certificate.py} & 432 &
      certificates over $\mathbb{F}_\ell$, failures over $\mathbb{Q}$ & 90\,s \\
    \texttt{exp07\_coupled.py} & 24 pairs, 3 $\beta$ &
      the floor under coupled learning; nudge sweep & 475\,s \\
    \texttt{exp08\_dimension.py} & 432 &
      3D and bending elements & 1.2\,s \\
    \texttt{exp09\_design\_rule.py} & 12 pairs &
      the design rule, run & 205\,s \\
    \texttt{exp10\_large\_deform.py} & 59 &
      the tangent bound at large strain & 190\,s \\
    \texttt{exp11\_frequency.py} & 216 $+$ 18 &
      the bound at finite frequency & 0.5\,s \\
    \texttt{exp12\_layout\_budget.py} & 315 &
      the layout law at fixed instrumentation & 0.5\,s \\
    \texttt{exp13\_passivity\_floor.py} & 3200, 72, 108, 5, 322$^{\ddagger}$ &
      the floor with passivity; the square lattice & 193\,s \\
    \texttt{exp14\_layout\_task.py} & 360 &
      $\Lambda(p)$ enumerated; the layout on a task & 443\,s \\
            \texttt{exp15\_layouts\_lit.py} & 19 $+$ 5 &
      forfeited norm fraction in the literature & 0.02\,s \\
        \texttt{exp16\_dirichlet.py} & 14 quantities &
      the imposed-displacement law & 21\,s \\
        \texttt{exp17\_du\_replication.py} & 531 &
      the obstruction in a published experiment & 1\,s \\
        \texttt{exp18\_wedge\_recovery.py} & 1008$^{\dagger}$ + 251 + 243 &
      what non-reciprocity recovers; the multi-clamp law & 6\,s \\
        \texttt{exp18b\_independent\_check.py} & 300; 120 $+$ 60 &
      the same two claims, by a second implementation & 8\,s \\
        \texttt{exp19\_submersion\_check.py} & 336, 327, 100 &
      Corollary~\ref{cor:submersion}: Newton solves; negative control & 67\,s \\
        \texttt{exp20\_multiclamp\_rank.py} & 720 &
      Proposition~\ref{prop:multirank}: attainment of the ceiling & 3\,s \\
        \texttt{exp21\_task\_deficit.py} & 7245, 1637, 790, 36, 43 &
      the deficit of a task list; the digon law & 415\,s \\
        \texttt{exp22\_price\_of\_nonreciprocity.py} & 40, 3000, 3360, 10, 336 &
      the ratio bound; the least coupling; torque-free & 17\,s \\
        \texttt{exp23\_running\_example.py} & 80; 760, 400 &
      one network: a reversed pair and an odd bond & 55\,s \\
    \bottomrule
  \end{tabular}
  \caption{The twenty-four experiment scripts, what each establishes, and its
    runtime; $^{\dagger}$rows, of which $945$ are distinct configurations
    (Sec.~\ref{sec:numerics}). Experiment~19's first two entries are
    configurations, at $12$ and $2$ Newton solves each, and its third counts
    rows, for $4786$ solves in all. Experiment~18b's first count is the
    multi-clamp law and its second and third the wedge law, in two dimensions
    and in three. $^{\ddagger}$Experiment~13's five parts in order; its first
    and fifth print their counts and deposit no table. Experiment~23's first count is the runs of its four pre-registered arms, the other two its post-hoc sweeps.}
  \label{tab:scripts}
\end{table}

\noindent
The twenty-four experiment scripts are listed in Table~\ref{tab:scripts}. The
determinism claimed under Code availability rests on one convention: every seed
in every script is derived arithmetically, from integers alone, and
Python's \texttt{hash()} of a string, salted per process, appears nowhere. The
reference environment is Python 3.11 with NumPy 2.4, SciPy 1.17 and OpenBLAS;
the OpenBLAS build is not pinned, because the first of the two caveats below
concerns it,
and Matplotlib is not pinned because no number reported in this paper passes
through it. We
do not claim byte-identical output across machines, LAPACK builds or BLAS
threadings for the floating-point scripts: twelve of the $2418$ rank cases have
a singular-value gap below $10^{9}$, and at those points a different LAPACK build
could place the $10^{-9}$ cut differently; all twelve lie outside criterion
\eqref{eq:criterion}. Experiment~06 is exact throughout, over
$\mathbb{F}_\ell$ for its certificates and over $\mathbb{Q}$ for its failures,
and does not carry the first caveat. It does carry the second: its
\emph{inputs} depend on the random streams of the NumPy version used, which is why the
certified configurations are shipped as data instead of regenerated.

\end{document}